\let\manuscriptKernelLabel\label
\documentclass[aps,prb,10pt,twocolumn,superscriptaddress,floatfix,longbibliography]{revtex4-2}
\let\label\manuscriptKernelLabel

\usepackage{amsmath,amssymb,amsthm,bbold,mathrsfs}
\usepackage{etoolbox}
\usepackage{needspace}
\usepackage{graphicx}
\usepackage{xcolor}
\usepackage{multirow}
\usepackage{array}[=2016-10-06]
\usepackage{longtable}
\usepackage{booktabs}
\usepackage{tikz-cd}

\makeatletter
\g@addto@macro\@mkpream@relax{\let\@endpbox\relax}
\makeatother

\newcommand{\spd}[1]{(#1+1)d}
\newcommand{\Z}{\mathbb{Z}}
\newcommand{\R}{\mathbb{R}}
\newcommand{\C}{\mathbb{C}}
\newcommand{\ii}{\mathrm{i}}
\newcommand{\ee}{\mathrm{e}}
\newcommand{\cA}{ {\cal A} }
\newcommand{\cC}{ {\cal C} }
\newcommand{\cD}{ {\cal D} }

\newcommand{\cH}{ {\cal H} }
\newcommand{\cM}{ {\cal M} }
\newcommand{\cN}{ {\cal N} }
\newcommand{\cO}{\mathcal{O}}
\newcommand{\cR}{ {\cal R} }
\newcommand{\cS}{ {\cal S} }
\newcommand{\cT}{ {\cal T} }
\newcommand{\cV}{ {\cal V} }
\newcommand{\cW}{ {\cal W} }

\newcommand{\cY}{ {\cal Y} }

\newcommand{\Hilb}{\mathsf{Hilb}}

\newcommand{\Ising}{\mathsf{Ising}}
\newcommand{\Fib}{\mathsf{Fib}}
\newcommand{\one}{\mathbb{1}}
\newcommand{\Tr}{\mathrm{Tr}}
\newcommand{\tTr}{\mathrm{tTr}}
\newcommand{\ev}{\mathrm{ev}}
\newcommand{\coev}{\mathrm{coev}}
\newcommand{\End}{\mathrm{End}}
\newcommand{\Fun}{\mathsf{Fun}}
\newcommand{\Rep}{\mathsf{Rep}}
\newcommand{\Mod}{\mathrm{Mod}}
\newcommand{\id}{\mathrm{id}}
\newcommand{\Irr}{\mathrm{Irr}}

\newcommand{\Tube}{\mathsf{Tube}}
\newcommand{\Torus}{\mathsf{Torus}}
\newcommand{\Tubev}{\mathsf{Tube}_{\mathrm v}}
\newcommand{\Tubeh}{\mathsf{Tube}_{\mathrm h}}
\newcommand{\Torusv}{\mathsf{Torus}_{\mathrm v}}
\newcommand{\Torush}{\mathsf{Torus}_{\mathrm h}}

\newcommand{\fZ}{\mathfrak{Z}}
\newcommand{\scrA}{\mathscr{A}}
\newcommand{\scrB}{\mathscr{B}}
\newcommand{\scrC}{\mathscr{C}}

\newcommand{\scrL}{\mathscr{L}}

\newcommand{\scrS}{\mathscr{S}}

\newcommand{\scrV}{\mathscr{V}}
\newcommand{\scrW}{\mathscr{W}}
\newcommand{\dd}{\mathrm{d}}
\newcommand{\xt}{\times}
\newcommand{\la}{\langle}
\newcommand{\ra}{\rangle}
\newcommand{\fgt}{\mathrm{fgt}}
\newcommand{\Mat}{\mathsf{Mat}}

\newcommand{\rev}{\mathrm{rev}}
\newcommand{\sym}{\mathrm{sym}}
\newcommand{\diagsector}{\mathrm{diag}}
\newcommand{\localmodules}{\mathrm{loc}}
\newcommand{\closedsector}{\mathrm{cl}}
\newcommand{\cutlabel}{\mathrm{cut}}
\newcommand{\contraction}{\mathrm{ct}}
\newcommand{\fiberfunctor}{\mathrm{ff}}
\newcommand{\ptr}{\operatorname{ptr}}
\newcommand{\swap}{\operatorname{swap}}
\newcommand{\JonesKauffman}{\mathsf{JK}}

\makeatletter
\newcommand{\condensationHookFill@}[1]{%
  \arrowfill@\relbar\relbar{\relbar\joinrel\rhook}#1%
}
\newcommand{\condensesTo}[1]{%
  \mathrel{\mathpalette\condensesTo@{#1}}%
}
\newcommand{\condensesTo@}[2]{%
  \setbox\z@\hbox{$\m@th#1\xrightarrow{\mkern5mu#2\mkern5mu}$}%
  \ooalign{%
    \copy\z@\cr
    \hbox to\wd\z@{%
      \kern.75ex
      \hbox to\dimexpr\wd\z@-1.5ex\relax{\condensationHookFill@{#1}\hfil}%
      \kern.75ex\hfil
    }\cr
  }%
}
\makeatother

\theoremstyle{definition}
\newtheorem{question}{Question}
\newtheorem{definition}{Definition}[section]
\newtheorem{proposition}[definition]{Proposition}

\newtheorem{lem}[definition]{Lemma}

\theoremstyle{remark}
\newtheorem{remark}{Remark}

\allowdisplaybreaks[4]
\makeatletter
\pretocmd{\l@part}{\Needspace{8\baselineskip}}{}{}
\makeatother

\definecolor{CUPurple}{RGB}{117,15,109}
\usepackage[
  colorlinks=true,
  linkcolor=CUPurple,
  citecolor=CUPurple,
  urlcolor=CUPurple,
]{hyperref}

\tikzcdset{
  condensesTo/.code={%
    \pgfsetarrows{tikzcd cap-tikzcd to}%
    \tikzset{postaction={
      draw,
      arrows={tikzcd cap-tikzcd left hook},
      shorten <=1ex,
      shorten >=1ex
    }}%
  }
}

\begin{document}

\title{Algebraic Tensor Network Renormalization and Holographic duality}
\author{Chenqi Meng}
\affiliation{Department of Physics, The Chinese University of Hong Kong, Shatin, New Territories, Hong Kong}
\author{Tian Lan}
\affiliation{Department of Physics, The Chinese University of Hong Kong, Shatin, New Territories, Hong Kong}
\author{Zhengcheng Gu}
\email[Corresponding author: ]{zcgu@phy.cuhk.edu.hk}
\affiliation{Department of Physics, The Chinese University of Hong Kong, Shatin, New Territories, Hong Kong}
\date{\today}

\begin{abstract}
Emergent generalized symmetries and the holographic principle play very important roles in understanding quantum phase transitions. In recent years, tensor-network renormalization with generalized symmetry and the corresponding fixed-point tensor formulation have been proposed as a discrete spacetime framework for reformulating conformal field theory. Nevertheless, a complete holographic formulation within this framework is still lacking, and the current formulation does not directly accommodate sector partition functions that are modular covariant in general. In this work, we formulate $\cR$-TNR, which builds the categorical constraints of a fusion category $\cR$ into local tensor spaces and preserves these constraints under Loop-TNR. For suitable choices of $\cR$, the initial tensors constructed here flow under $\cR$-Loop-TNR to stable gapless fixed points. In the Ising, tricritical Ising, and three-state Potts conformal field theory examples, a broad range of initial parameters leads to the same gapless fixed point in each example. The method resolves conformal towers by anyon sector and yields scaling dimensions, conformal spins, and magnitudes of sector-resolved structure constants with exteremly high accuracy compared to previous methods. These results realize a \emph{topological bootstrap} based on \textit{stable symmetry–topological-order} in which finite categorical and geometric input, together with an initial tensor that need not be fine-tuned, suffices to recover the infrared conformal field theory and its sector-resolved conformal data.
\end{abstract}

\maketitle

\tableofcontents

\section{Introduction}\label{sec:introduction}
The renormalization group (RG) is a central concept in modern physics, especially in the study of critical systems, where behavior changes with scale. In condensed matter physics, RG fixed points give rise to universality classes, which characterize the universal properties of continuous phase transitions. In high-energy physics, RG flows describe how interactions change with energy scale, as exemplified by asymptotic freedom in the strong interaction.
In quantum field theory, renormalizability allows predictions to remain independent of the ultraviolet cutoff and is a key property of the Standard Model. Nevertheless, the conventional momentum-space RG calculations rely on perturbative expansions and are difficult to apply at strong coupling. Tensor nework renormalization (TNR) offers a complementary real-space RG approach that acts directly on lattice partition functions. 
To tackle the challenges of real-space RG, recent advances in quantum information theory have introduced powerful new ideas. Concepts such as long-range entanglement~\cite{Chen_2010} and entanglement renormalization~\cite{Vidal_2007,Gu_2008,Gu_2009} have led to RG approaches applicable to both weakly and strongly interacting systems.

The tensor renormalization group (TRG)~\cite{Levin_2007} provides a first tensor-network formulation of real-space coarse graining for lattice models, enabling accurate calculations of free energies, specific heats, and phase-transition points~\cite{Hinczewski_2008}. Beyond these thermodynamic quantities, incorporating entanglement filtering allows the RG flow to reveal fixed-point tensor structures of gapped phases with symmetry~\cite{Gu_2009}. Despite these successes for gapped phases, TRG performs poorly at criticality. At finite bond dimension, its local, tree-like truncations leave short-range correlations trapped in closed loops, allowing these correlations and the associated truncation errors to accumulate along the RG flow. Entanglement filtering removes local correlations but does not by itself fully resolve this difficulty at criticality. TNR~\cite{Evenbly_2015} and, in particular, Loop-TNR~\cite{Yang_2017} address this limitation by treating local correlations more systematically. Loop-TNR jointly optimizes the tensors around a plaquette, so that errors in one tensor are constrained by the surrounding tensors. This collective optimization substantially improves the numerical stability of critical RG flows and enables precise extraction of conformal spectra~\cite{Bao:2019hfc,wei2023tensornetworkrenormalizationapplication}, RG flows of coupling constants~\cite{Ueda_2021,Ueda_2023_finite_bond_dimension}, and structure constants~\cite{Li_2022,ueda2023fixedpointtensorfourpointfunction}.

Nevertheless, the \emph{numerical} stability of Loop-TNR does not, however, imply the \emph{physical} stability of a critical point. A CFT represented accurately by Loop-TNR may still admit relevant perturbations allowed by its microscopic global symmetry. Reaching such a CFT then requires fine-tuning the corresponding microscopic couplings. In finite-bond-dimension calculations of an unprotected critical point, the computed spectrum may therefore drift as renormalization proceeds.
Thus, it is very important to incorporate emergent genralized symmetry into these real-space RG schemes. For ordinary global symmetries, tensor-network bonds can be labeled by irreducible representations, with fusion rules determined by the symmetry group. Generalized symmetries require a broader categorical description that also captures their fusion and recoupling data. Strange-correlator constructions~\cite{Bal_2018,Vanhove_2022} provide a useful perspective on the relation between such symmetries and a topological theory in one higher dimension. Recent fixed-point tensor constructions for rational CFTs also admit a strange-correlator interpretation ~\cite{Cheng:2023kxh}. These developments motivate a tensor RG scheme in which categorical symmetry constraints are imposed directly on local tensors and preserved at every coarse-graining step.

Mathematically, these symmetry constraints can be understood through the
symmetry--topological-order (symTO)
correspondence~\cite{Ji_2019,Ji_2020,Kong_2020_math,Kong_2020_phy,Chatterjee_2023,
Aasen_2016,Aasen_2020,Chen_2020}.
In symTO, a $(1+1)$-dimensional system with symmetry constraints provided
by a unitary fusion category (UFC) $\cR$ can be regarded as a boundary of
the $(2+1)$-dimensional topological order $\fZ(\cR)$, where $\fZ(\cR)$ is
the Drinfeld center of $\cR$.
When the low-energy boundary theory is a CFT, symmetry-preserving relevant
local operators correspond to relevant operators in its trivial bulk-anyon
sector. If this sector admits no relevant perturbations, the CFT can be
realized without fine-tuning. We call
$\fZ(\cR)$ a \emph{stable symTO for this CFT}
(Definition~\ref{def:stable-symto}).
This mechanism of symmetry protection has been studied both through the
symTO description and in lattice models with fusion-category
symmetry~\cite{chatterjee2024emergentgeneralizedsymmetrymaximal,Roychowdhury_2026}.

To impose the symmetry constraints of a unitary fusion category $\cR$ within Loop-TNR, we develop an algebraic TNR framework in which local tensors are morphisms in $\cR$. We call these morphisms $\cR$-tensors and networks built from them $\cR$-tensor networks. We extend tensor contractions and low-rank approximations to these tensors and apply the Loop-TNR scheme. The resulting algorithm, which we call $\cR$-Loop-TNR, preserves the symmetry constraints throughout the RG flow. With a suitable choice of $\cR$, the initial tensors constructed here can flow to a stable gapless fixed point without fine-tuning. At a gapless fixed point of $\cR$-Loop-TNR, the local $\cR$-tensors carry information about the CFT on the boundary of the bulk topological order $\fZ(\cR)$. To extract scaling dimensions resolved by bulk-anyon sectors, we further generalize the transfer-matrix method of Ref.~\cite{Gu_2009} to $\cR$-tensors. The corresponding sector partition functions form a modular-covariant vector, as predicted by the symTO framework~\cite{Ji_2019}.

In particular, constructing a transfer matrix requires closing two opposite legs of a
local tensor to form a tube. For an $\cR$-tensor, this closure must be
defined without assuming a braiding on $\cR$. As pointed out in
Refs.~\cite{lan2025tubecategorytensorrenormalization,YBWL2605.06661}, this
closure is induced by an appropriate coend. Such a coend can be
thought of as summing up the spaces of symmetric local tensors over the repeated
charges labeling the contracted physical bonds. 
Ref.~\cite{YBWL2605.06661}
proposed calling the arrays of spaces of tensors indexed by charges
as \emph{pro-tensors} (as opposed to ordinary tensors, which are arrays of numbers
indexed by basis vectors in charges). Therefore, the correct symmetric
transfer matrices are computed via pro-tensor contraction/trace. In general,
the pro-tensor trace differs from both the ordinary trace and the quantum trace.
When the symmetry is trivial, it reduces to the ordinary trace~\cite{YBWL2605.06661}.
These coend spaces, together with the composition induced by tensor contraction, form the tube category associated with $\cR$ \cite{HARDIMAN_2019,Hardiman_2019_graphical,lan2025tubecategorytensorrenormalization}.
The transfer matrix is therefore a morphism in this category.
By diagonalizing it, we obtain the conformal spectrum and the
corresponding anyonic eigenstates in each bulk-anyon sector.
We further show how to extract structure constants resolved by
bulk-anyon sectors from the fixed-point $\cR$-tensors.

Finally, we would like to introduce a novel concpet: \textit{the stable symTO}, that is, for certain choices of $\cR$, the categorical symmetry can protect a
stable gapless phase described by a
CFT~\cite{chatterjee2024emergentgeneralizedsymmetrymaximal}.
We will introduce a universal construction for initial
$\cR$-tensors that guaranteen a RG flow towards the
corresponding gapless fixed point.
For the examples studied here, we choose $\cR$ to be the underlying
fusion category of a unitary modular tensor category (UMTC).
Using finite categorical and geometric input, we construct a family
of initial $\cR$-tensors with a small number of trial parameters.
This construction does not require the full conformal data of the
target CFT.
These tensors may flow under $\cR$-Loop-TNR to a gapless fixed point
stabilized by $\cR$.
When such a gapless flow is obtained, we find that the conformal
spectrum remains stable over many RG steps.
Moreover, varying the initial tensor over a broad range of trial
parameters leads to the same gapless fixed point.
From the fixed-point tensors, we resolve conformal towers by anyon
sector and extract sector-resolved structure constants, obtaining
more accurate conformal data than in calculations with weaker
categorical constraints.
We call this procedure of recovering the infrared CFT and its
conformal data from finite categorical and geometric input the
\emph{topological bootstrap}%
~\cite{bao2025tensorcomplextopologicalbootstrap}.

We study the Ising CFT protected by
$\cR=\Ising$, the tricritical Ising CFT protected by
$\cR=\mathsf{6}_{7/10}$, and the three-state Potts CFT protected by
$\cR=\mathsf{6}_{4/5}$.
For comparison, we study non-stable symTOs
whose trivial sectors contain relevant perturbations. The $\Z_2$-symmetric
Ising model, the tricritical Ising CFT with $\Fib$ symmetry, and the three-state
Potts model with $\Rep(\Z_3)$ symmetry therefore require fine-tuning of their
interactions.
The data for these categories are collected in
Appendix~\ref{app:categorical-data}.
Our numerical implementation benefits from the \texttt{TensorKit.jl}
ecosystem developed by the Ghent tensor-network group and collaborators.
We use its tools for fusion-category
data~\cite{TensorKitSectors_code,CategoryData_code}, categorical tensor
operations~\cite{Devos_2025_TensorKit,jutho_2025_15243493}, and numerical
linear algebra~\cite{haegeman_2025_14911568,MatrixAlgebraKit_code}, together
with the Loop-TNR implementation in
\texttt{TNRKit.jl}~\cite{Vanthilt_2026_TNRKit,TNRKit_code}.
The categorical operations and contraction schemes are specified explicitly
in this paper, allowing the calculations to be reproduced using other
numerical implementations.
However, the algorithms for diagonalizing the transfer matrix in the tube
category and computing structure constants have not yet been implemented
in \texttt{TNRKit.jl}.

The paper is organized as follows. Section~\ref{sec:r-tnr} introduces
$\cR$-tensors and $\cR$-tensor networks, their fusion-tree representation,
and the $\cR$-Loop-TNR algorithm. Section~\ref{sec:holographic-data}
develops the holographic Turaev--Viro interpretation and explains how to
extract conformal spectra and structure constants, using the tube category
to resolve the topological sectors and derive modular covariance. Section~\ref{sec:ordinary-and-generalized-symmetry-examples} presents the fine-tuned
$\Z_2$-symmetric Ising, Fibonacci-symmetric tricritical Ising, and
$\Z_3$-symmetric three-state Potts models.
Section~\ref{sec:stable-symto} studies realizations of the same CFTs with
stable symTOs and demonstrates the topological bootstrap.
Section~\ref{sec:cft-fixed-point-mechanism} analyzes the fixed-point conditions and the
mechanism underlying the stability of these constructions.
Section~\ref{sec:conclusion-and-outlook} summarizes the results and discusses
open questions and future directions. Readers unfamiliar with symTO are encouraged to begin with
Appendix~\ref{app:symto-review}, which provides a brief overview of
the symmetry--boundary correspondence and the relations between gapped and
gapless phases. Appendix~\ref{app:categorical-data} collects the categorical
data and factorizations used in the examples, while
Appendix~\ref{app:four-by-four-structure-constants} describes the
half-$4\xt4$ patch calculation of structure constants.
Appendix~\ref{app:geometric-factors} derives the geometric factors entering
the tensor and structure-constant constructions.
Appendix~\ref{app:directional-tube-traces} gives the basis-independent
change-of-cut proof of modular $S$-covariance, and
Appendix~\ref{app:notation-table} collects the notation used throughout
the paper.

\section{\texorpdfstring{$\cR$}{R}-tensors, \texorpdfstring{$\cR$}{R}-tensor networks, and \texorpdfstring{$\cR$}{R}-TNR}
\label{sec:r-tnr}

This section assembles the local algebraic and numerical ingredients of the
$\cR$-TNR framework. Subsection~\ref{sec:r-tensor-basis} defines $\cR$-tensors
and $\cR$-tensor networks, fixes the fusion-tree basis and duality conventions,
and reviews the decompositions and truncations used in numerical calculations.
Subsection~\ref{sec:r-loop-tnr} then describes how these operations are combined
with Loop-TNR to preserve the fusion-category charge structure throughout the
RG flow.

The underlying categorical tensor operations are implemented in the
\texttt{TensorKit.jl} ecosystem, including
\texttt{CategoryData.jl}~\cite{CategoryData_code},
\texttt{TensorKitSectors.jl}~\cite{TensorKitSectors_code},
\texttt{TensorKit.jl}~\cite{Devos_2025_TensorKit,jutho_2025_15243493},
\texttt{KrylovKit.jl}~\cite{haegeman_2025_14911568}, and
\texttt{MatrixAlgebraKit.jl}~\cite{MatrixAlgebraKit_code}. The
$\cR$-Loop-TNR algorithm is implemented in
\texttt{TNRKit.jl}~\cite{Vanthilt_2026_TNRKit,TNRKit_code}. Our purpose here
is to establish a self-contained notation and algorithmic interface for the
holographic constructions and physical observables developed in
Sec.~\ref{sec:holographic-data}.

\subsection{\texorpdfstring{$\cR$}{R}-tensor network and \texorpdfstring{$\cR$}{R}-tensor operations in basis}\label{sec:r-tensor-basis}

This section fixes the categorical and numerical conventions used
throughout the paper. We first define $\cR$-tensors as morphisms in a unitary
fusion category, then construct $\cR$-tensor networks and present their
fusion-tree basis representation and basic numerical operations.

\subsubsection{\texorpdfstring{$\cR$-tensor}{R-tensor}}\label{sec:r-tensor}

Let $\cR$ be a unitary fusion category (UFC). We call a morphism in $\cR$
an $\cR$-tensor and use ``tensor'' and ``morphism'' interchangeably
throughout this work. For objects $X,Y\in\cR$, we write
$\cR(X,Y):=\operatorname{Hom}_{\cR}(X,Y)$ for the corresponding morphism
space.

We represent $\cR$-tensors using string diagrams, with morphisms read from
top to bottom \cite{Turaev:2017uxl,Devos_2025_TensorKit,Lan2024_Intro_cat}.
For example, the morphisms $f:X\to Y$ and
$g:X_1\otimes X_2\to Y_1\otimes Y_2$ are represented as
\[
    \raisebox{-28.65274pt}{\includegraphics{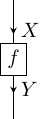}},\quad 
    \raisebox{-28.65274pt}{\includegraphics{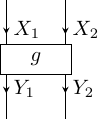}}.
\]

The tensor product $f\otimes g$ is represented by placing the diagrams
side by side.
\[
\raisebox{-28.65274pt}{\includegraphics{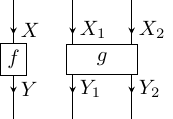}}.
\]
Composition is represented by stacking the diagrams vertically. For
$f:X\to Y$ and $h:Y\to Z$, the composite $h\circ f$ is
\[
\raisebox{-35.76593pt}{\includegraphics{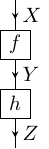}}.
\]

The dagger of $f:X\to Y$ is $f^\dagger:Y\to X$.
To bend tensor legs, we use the evaluation and coevaluation morphisms
\[
\ev_X:X\otimes X^*\longrightarrow\one,
\qquad
\coev_X:\one\longrightarrow X^*\otimes X,
\]
represented by
\[
\ev_X :=
\raisebox{-14.42636pt}{\includegraphics{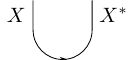}},\qquad
\coev_X :=
\raisebox{-14.42636pt}{\includegraphics{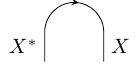}}.
\]
These satisfy the snake identities
\begin{align*}
(\ev_X\otimes\id_X)
\circ(\id_X\otimes\coev_X)
&=\id_X,\\
(\id_{X^*}\otimes\ev_X)
\circ(\coev_X\otimes\id_{X^*})
&=\id_{X^*}.
\end{align*}
The oppositely oriented coevaluation and
evaluation are $\ev_X^\dagger$ and $\coev_X^\dagger$, respectively.
We choose the standard spherical normalization, so that the two closures
of an endomorphism $f:X\to X$ give the same trace.
\begin{align*}
\Tr_X(f)
&=\ev_X\circ(f\otimes\id_{X^*})\circ\ev_X^\dagger\\
&=\coev_X^\dagger
  \circ(\id_{X^*}\otimes f)
  \circ\coev_X.
\end{align*}

\subsubsection{\texorpdfstring{$\cR$-tensor network}{R-tensor network}}\label{sec:r-tensor-network}

An $\cR$-tensor network is formed by connecting $\cR$-tensors along legs
labeled by matching objects. For concreteness, consider a square network
in the plane. Each edge carries an object of $\cR$, and each vertex
carries an $\cR$-tensor whose domain and codomain are determined by its
incoming and outgoing edges. Contracting the network means composing
these morphisms according to the string-diagram rules. In the diagram
below, morphisms are read from upper-right to lower-left.
\[
\raisebox{-61.03795pt}{\includegraphics{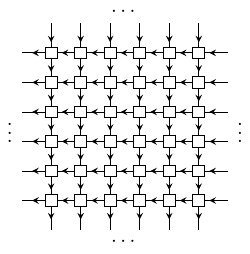}}.
\]

For an infinite tensor network, the contraction of infinitely many tensors
requires a careful definition, which we do not provide here. In practice, we
perform finite local contractions within the infinite tensor network. In numerical TNR calculations, the final physical
data, such as partition functions, correlation functions, and energy spectra,
are always computed from a very large but finite tensor network.

Physically, a 2D $\cR$-tensor network can represent the partition function of a
2D classical (or (1+1)-dimensional quantum) symmetric system. Here $\cR$ is the
system's category of symmetry charges, which we refer to simply as the
\emph{charge category}~\cite{Lan:2023oyu,Lan_2024,Meng:2024nxx,Kong_2020_phy}.
For example, for a 2D statistical system without symmetry, $\cR$ is chosen to be
the category of finite-dimensional Hilbert spaces $\Hilb$. Then each edge is
assigned a finite-dimensional Hilbert space, and each morphism is an ordinary
tensor. This reduces to an ordinary tensor network. For a system with
$G$-symmetry, where $G$ is a finite group, $\cR$ is chosen to be the category
of finite-dimensional unitary $G$-representations, denoted by $\Rep(G)$.
In this case, each bond of the
tensor network is a $G$-representation, i.e.\ carries a $G$-charge, and each
morphism in the tensor network is a $G$-symmetric operator. For a system with
anomaly-free fusion category symmetry whose symmetry algebra is the Hopf
algebra $\mathsf{H}$, the charge category is chosen to be $\Rep(\mathsf{H})$
\cite{Inamura_2022,Meng:2024nxx}.

This tensor-network construction also extends to compact Lie groups $G$.
Although $\Rep(G)$ may have infinitely many simple isomorphism classes,
each bond is a finite direct sum of simple objects with finite multiplicities.
Local tensor contractions therefore involve only finitely many simple sectors,
and low-rank approximations retain finitely many of them.
Thus, each step of the computation takes place in finite-dimensional morphism spaces.

For most statistical models, the partition function can be represented by a
translationally invariant 2D $\cR$-tensor network in which the same morphism is
assigned to every vertex and the same object to every edge. We
consider only translationally invariant tensor networks throughout this work.

\subsubsection{\texorpdfstring{Presentation of $\cR$-tensors}{Presentation of R-tensors}}\label{sec:r-tensor-presentation}

To make the $\cR$-tensor network numerically tractable, we fix a set of
representatives of simple objects, labeled by $i,j,k,l,m,n,\ldots$,
and expand all morphisms in a chosen basis.
Every object $X\in\cR$ admits a decomposition
\[
X\cong\bigoplus_{i\in\Irr\cR}\C^{N^X_i}\odot i,
\]
where $N^X_i=\dim\cR(i,X)$ and $\odot$ denotes the copower, so that
$\C^{N^X_i}\odot i$ consists of $N^X_i$ copies of $i$.
We choose orthonormal embeddings \cite{Lan2024_Intro_cat,Devos_2025_TensorKit}
\[
\left\{q^X_{i,\alpha}\in\cR(i, X):\ i\in\Irr\cR,\ \alpha = 1,2,\cdots, N^X_i\right\}
\]
satisfying
\begin{equation}\label{eq:orthonormal-completeness}
q_{i,\alpha}^{X\dagger} \circ q_{j,\beta}^X = \delta_{ij}\delta_{\alpha\beta},\quad \sum_{i\in\Irr\cR}\sum_{\alpha = 1}^{N^X_i}q_{i,\alpha}^X\circ q_{i,\alpha}^{X\dagger} = \id_{X}.
\end{equation}
Graphically, we use triangles to represent orthonormal embeddings:
\[
q_{i,\alpha}^X :=
\raisebox{-28.75275pt}{\includegraphics{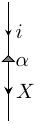}}.
\]
The orthonormality and completeness conditions are then expressed as
\[
\raisebox{-28.65274pt}{\includegraphics{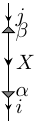}} = \delta_{ij}\delta_{\alpha\beta}\xt
\raisebox{-28.65274pt}{\includegraphics{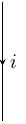}}
,\quad \sum_{i,\alpha}
\raisebox{-28.85274pt}{\includegraphics{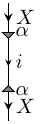}} =
\raisebox{-28.85274pt}{\includegraphics{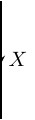}}.
\]

With this choice of basis, any morphism in $\cR$,
\[
f: A \rightarrow B
\]
can be expanded by inserting the identity decompositions of $\id_A$ and
$\id_B$:
\[
f = \sum_{i,\alpha}\sum_{j,\beta}\left(q^B_{j,\beta}\circ q_{j,\beta}^{B\dagger}\right)\circ f\circ \left(q^A_{i,\alpha}\circ q_{i,\alpha}^{A\dagger}\right).
\]
Since $i$ and $j$ are simple objects, the following morphism is proportional to
the identity morphism:
\[
q_{j,\beta}^{B\dagger} \circ f\circ q^A_{i,\alpha} = \delta_{ij}f^i_{\beta\alpha}\times \id_{i},
\]
where $f^i_{\beta\alpha}$ is a scalar. It follows that $f$ can be expanded as
\[
f = \sum_{i,\alpha\beta}f^i_{\beta\alpha} \times q^B_{i,\beta} \circ q_{i,\alpha}^{A\dagger}.
\]
Thus, a tensor (morphism) in $\cR$ is essentially a collection of matrices
$\{f^m\}_{m}$ graded by the simple objects in $\Irr\cR$. When morphisms are
composed,
\[
A \xrightarrow{ f}B \xrightarrow{g}C,
\]
the corresponding $\Irr(\cR)$-graded matrices are multiplied under the basis
expansion. The graded multi-matrix is the raw data stored on a computer.

To express orthonormal embeddings for higher-rank tensors, we fix a basis of
trivalent embedding maps:
\[
q^{i\otimes j}_{k,\alpha} = \raisebox{-28.65274pt}{\includegraphics{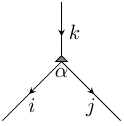}},\quad i, j, k\in\Irr\cR.
\]
We also introduce an isotopy-invariant basis defined by
\begin{equation}\label{eq:isotopy-invariant-fusion-basis}
\raisebox{-28.65274pt}{\includegraphics{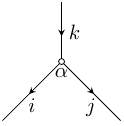}} =
\left(\frac{d_i d_j}{d_k}\right)^{\frac{1}{4}}
\raisebox{-28.65274pt}{\includegraphics{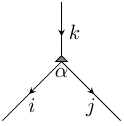}}.
\end{equation}
The inner product of this basis is defined using the quantum spherical trace,
\[
\la f|g\ra := \Tr(f^\dagger \circ g).
\]
Graphically, this reads
\[
\vcenter{\hbox{\includegraphics{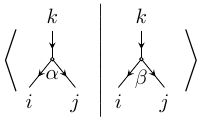}}}
=
\raisebox{-27.68953pt}{\includegraphics{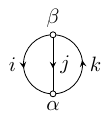}}
=
\delta_{\alpha\beta}\sqrt{d_id_jd_{k^*}}.
\]

Any higher-rank tensor can then be expanded in terms of fusion trees. For
example, a rank-4 $\cR$-tensor is a morphism in $\cR$:
\[
T: X\otimes Y \rightarrow Z \otimes W,\quad T =
\raisebox{-28.85274pt}{\includegraphics{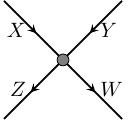}}.
\]
Unlike in the rank-2 case, the orthonormal embeddings
$\left\{q^{X\otimes Y}_{m,I}\right\}$ and
$\left\{q^{Z\otimes W}_{m,J}\right\}$ can be expanded sequentially by choosing
a basis for each bond. Each trivalent graph used as a basis vector is called a
fusion tree. We first insert an identity decomposition on each bond:
\[
\raisebox{-24.58499pt}{\includegraphics{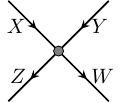}} = \sum_{ijkl \in\Irr\cR}\sum_{\alpha\beta\gamma\delta}
\raisebox{-36.62492pt}{\includegraphics{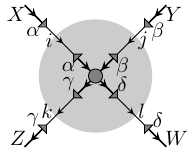}}.
\]

For each fixed choice of $i,j,k,l,\alpha,\beta,\gamma,\delta$, we can expand the
shaded part in the orthonormal fusion-tree basis:
\begin{equation}\label{eq:rank-4-fusion-tree-expansion}
\raisebox{-34.3432pt}{\includegraphics{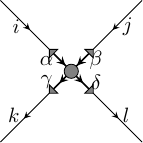}}= \sum_{m,\mu\nu}T^{m}_{(\mu,i\alpha, j\beta), (\nu,k\gamma, l\delta)}
\raisebox{-32.20934pt}{\includegraphics{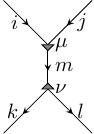}}.
\end{equation}
Therefore, in the fusion-tree basis, the tensor is expressed as
\begin{align}
    &\raisebox{-43.0791pt}{\includegraphics{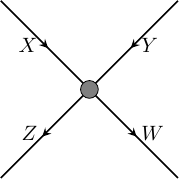}} \nonumber\\
    &= \sum_{m}\sum_{\mu,i\alpha, j\beta} \sum_{\nu,k\gamma,l\delta}T^m_{(\mu,i\alpha, j\beta), (\nu,k\gamma, l\delta)} \nonumber\\
    &\qquad \times \raisebox{-45.87216pt}{\includegraphics{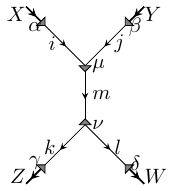}}
\label{eq:rank-4-tensor-decomposition}
\end{align}

A general $\cR$-tensor
\[
f:X := X_1\otimes X_2\otimes \cdots \otimes X_n \rightarrow Y_1\otimes Y_2\otimes \cdots \otimes Y_m =: Y
\]
can then be represented as a block-diagonal matrix in the embedding-tree basis:
\[
f = \sum_{m}\sum_{I, J}f^m_{J, I}\xt q^{Y}_{m, J} \circ q^{X\dagger}_{m, I},
\]
where $I$ and $J$ label embedding-tree basis elements for $X$ and $Y$,
respectively. For example,
\[
q^{X}_{m, I} := \raisebox{-35.86592pt}{\includegraphics{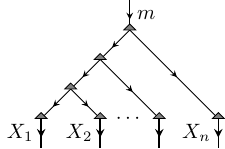}}.
\]
Here, $I$ denotes the collection of indices associated with an embedding tree.
As in the rank-2 tensor case, the $\Irr\cR$-graded multi-matrix constitutes the
raw data stored on a computer~\cite{Devos_2025_TensorKit}. Composition of
morphisms again corresponds to multiplication of the associated block matrices.

Different embedding trees differ by an $F$-move:
    \begin{multline}
    \raisebox{-33.48965pt}{\includegraphics{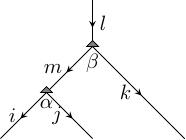}} \\
    = \sum_{n\mu\nu}F^{ijk}_{l, (n,\mu,\nu),(m,\alpha,\beta)}
    \raisebox{-33.48965pt}{\includegraphics{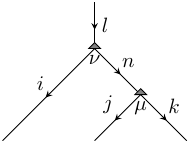}}.
    \end{multline}
Therefore the $F$-move can be used to transform between different embedding orders. 
The pentagon identity ensures that all transformation paths give a consistent result.
In UFC, $F$-matrix is further required to be unitary.

We have so far defined only vertical tensor contractions. In a tensor network,
however, tensors are not always composed vertically. Tensor transposition, or
equivalently, bending tensor legs, is therefore needed in general to define
arbitrary tensor contractions. In fusion category theory, such bending is
defined by duality morphisms.

General transpositions can be constructed from several elementary moves. Since
all $\cR$-tensors are expanded in a fusion-tree basis, it suffices to express
the duality maps for a fixed choice of simple objects and trivalent embedding
trees. We review the algorithm implemented in \texttt{TensorKit.jl} (see its
documentation).

We fix a set of representatives $\Irr\cR$ of the simple objects in $\cR$ and all tri-valent embedding basis.
Let
$i \in \cR$ be a simple object. Since $i^*$ is also simple, there is a unique
$\bar i\in\Irr\cR$ for which a unitary isomorphism
$Z_i:i^* \simeq \bar i$ exists. We denote this isomorphism by
\[
Z_i \ =\ \raisebox{-28.65274pt}{\includegraphics{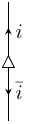}}.
\]

The unitary isomorphism $Z_i$ is defined only up to a phase, which can be chosen
so that the following condition holds:
\begin{equation}\label{eq:first-duality-identity}
    \raisebox{-22.96227pt}{\includegraphics{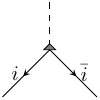}} = \frac{1}{\sqrt{d_i}} \xt
\raisebox{-21.53955pt}{\includegraphics{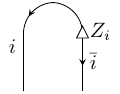}}
\end{equation}
Taking the dagger and the dual of this isomorphism yields two further
isomorphisms:
\[
\raisebox{-28.65274pt}{\includegraphics{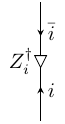}},\quad 
\raisebox{-28.65274pt}{\includegraphics{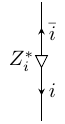}}\ =\ 
\raisebox{-28.65274pt}{\includegraphics{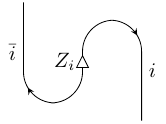}}.
\]
Since $Z_i$ can be defined for every $i$, we find that $Z_i^*$ is proportional
to $Z_{\bar{i}}$. We denote the proportionality phase by $\varkappa_i$:
\[
\raisebox{-28.65274pt}{\includegraphics{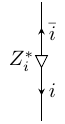}}\ = \ 
\varkappa_i\times
\raisebox{-28.65274pt}{\includegraphics{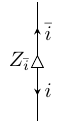}}.
\]
Consequently, Eq.~\eqref{eq:first-duality-identity} implies the other two
identities:
\begin{equation}\label{eq:other-two-dual}
    \raisebox{-26.61208pt}{\includegraphics{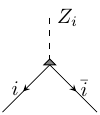}} 
=
\frac{1}{\sqrt{d_i}}
\raisebox{-21.53955pt}{\includegraphics{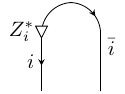}}
=
\frac{\varkappa_i}{\sqrt{d_i}}
\raisebox{-21.53955pt}{\includegraphics{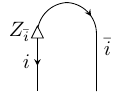}}.
\end{equation}
These identities relate embedding maps to duality maps.

Using these relations, we can fix the phase factor by considering the diagram
\[
\raisebox{-85.55823pt}{\includegraphics{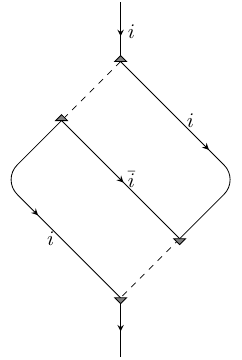}}.
\]
On the one hand, using the $F$-move, the diagram evaluates to
$F^{i\bar{i}i}_{i, (\one, 1,1), (\one, 1, 1)} \xt \id_{i}$. On the other hand,
taking the dagger of
\[
\raisebox{-22.96227pt}{\includegraphics{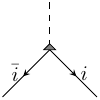}} = \frac{1}{\sqrt{d_i}} \xt
\raisebox{-21.53955pt}{\includegraphics{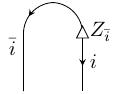}},
\]
the diagram simplifies to
\[
\frac{\varkappa_i}{d_i}\xt
\raisebox{-85.55823pt}{\includegraphics{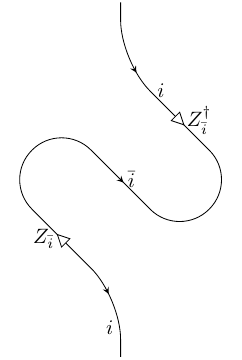}} = \frac{\varkappa_i}{d_i}\xt \id_i.
\]
Therefore, $\varkappa_i$ is the phase of the $F$-symbol
$F^{i\bar{i}i}_{i, (\one, 1,1), (\one, 1, 1)}$, while $d_i$ is the inverse of
its absolute value. The phase factor $\varkappa_i$ is called the
Frobenius--Schur phase.

We now consider how the duality map transforms the basic trivalent embedding
map. The first transformation is called the $B$-move in Kitaev's
work~\cite{Kitaev_2006}.
\[
\raisebox{-28.65274pt}{\includegraphics{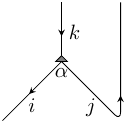}} = 
\sum_{\beta}\sqrt{\frac{d_k}{d_i}}B^{ij}_{k, \beta\alpha}
\raisebox{-31.90134pt}{\includegraphics{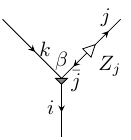}}.
\]
The normalization factor $\sqrt{\frac{d_k}{d_i}}$ is chosen by convention so
that the $B$-symbol is a unitary matrix.

The $B$-symbol can be obtained from the $F$-symbol as described in
\texttt{TensorKit.jl}'s documentation. Taking the dagger of
Eq.~\eqref{eq:first-duality-identity} gives
\begin{equation}\label{eq:duality-no-fs}
    \raisebox{-21.23157pt}{\includegraphics{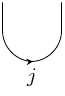}} 
 = \sqrt{d_j}\xt
\raisebox{-28.65274pt}{\includegraphics{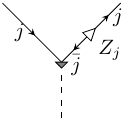}}.
\end{equation}

Substituting this expression into the left-hand side and using the $F$-move and
its unitarity gives
\[
\sqrt{\frac{d_k}{d_i}}B^{ij}_{k,\beta\alpha} = \sqrt{d_j}F^{ij\bar{j}}_{i, (\one, 1, 1), (k, \alpha, \beta)}.
\]

Similarly, the $A$-move is defined by
    \[
    \raisebox{-28.65274pt}{\includegraphics{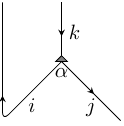}} = \sum_{\beta}\sqrt{\frac{d_k}{d_j}}A^{ij}_{k, \beta\alpha}
    \raisebox{-28.65274pt}{\includegraphics{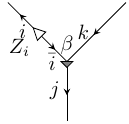}}.
    \]
Taking the dagger of Eq.~\eqref{eq:other-two-dual} and then substituting
$i \leftarrow \bar{i}$ gives
\begin{equation}\label{eq:duality-fs}
    \raisebox{-19.91212pt}{\includegraphics{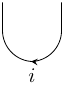}} 
 = \sqrt{d_{\bar{i}}}\varkappa_{\bar{i}}\xt
\raisebox{-28.65274pt}{\includegraphics{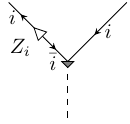}}.
\end{equation}
Substituting this expression into the left-hand side and using the dagger of the
$F$-move gives
\[
\sqrt{\frac{d_k}{d_j}}A^{ij}_{k,\beta\alpha} = \sqrt{d_i}\varkappa_{\bar{i}} \left(F^{\bar{i}ij}_{j, (k, \alpha,\beta), (\one, 1, 1)}\right)^*.
\]

When the bent leg is already dualized, we take the dagger of the inverse
bend. Unitarity gives the inverse matrices $A^\dagger$ and $B^\dagger$,
and taking the dagger of the diagram also conjugates their coefficients.
The resulting expressions therefore contain $A$ and $B$ with exchanged
multiplicity indices.
\begin{align*}
\raisebox{-28.65274pt}{\includegraphics{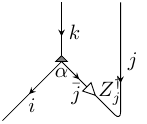}}
&=\sum_{\beta}\sqrt{\frac{d_k}{d_i}}B^{kj}_{i,\alpha\beta}
\raisebox{-28.65274pt}{\includegraphics{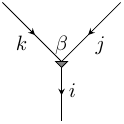}},\\
\raisebox{-30.35362pt}{\includegraphics{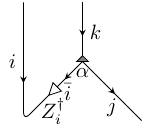}}
&=\sum_{\beta}\sqrt{\frac{d_k}{d_j}}A^{ik}_{j,\alpha\beta}
\raisebox{-28.65274pt}{\includegraphics{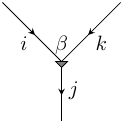}}.
\end{align*}
Here the vertices on the right are $(q^{k\otimes j}_{i,\beta})^\dagger$
and $(q^{i\otimes k}_{j,\beta})^\dagger$, respectively. Equivalently,
Eqs.~\eqref{eq:duality-fs} and~\eqref{eq:duality-no-fs} give
\begin{align*}
\sqrt{\frac{d_k}{d_i}}B^{kj}_{i,\alpha\beta}
&=\sqrt{d_j}\varkappa_{\bar j}
F^{i\bar j j}_{i,(\one,1,1),(k,\alpha,\beta)},\\
\sqrt{\frac{d_k}{d_j}}A^{ik}_{j,\alpha\beta}
&=\sqrt{d_i}
\left(F^{i\bar i j}_{j,(k,\alpha,\beta),(\one,1,1)}\right)^*.
\end{align*}
Thus the right-leg bend carries a Frobenius--Schur phase
$\varkappa_{\bar j}$, whereas the left-leg bend has no explicit phase factor
in its $F$-symbol expression.

All basis transformations associated with the elementary duality moves are
therefore expressed in terms of $F$-symbols.

For a generic $\cR$-tensor
\[
f:X_1\otimes X_2\otimes \cdots \otimes X_n \rightarrow Y_1\otimes Y_2\otimes \cdots \otimes Y_m,
\]
all transpositions of $f$ can be generated by duality transformations:
\[
\raisebox{-19.96135pt}{\includegraphics{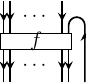}},\quad
\raisebox{-19.96135pt}{\includegraphics{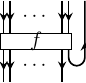}},\quad
\raisebox{-19.96135pt}{\includegraphics{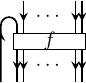}},\quad
\raisebox{-19.96135pt}{\includegraphics{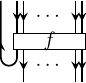}}.
\]
In the fusion-tree basis, these transpositions become the following four types
of moves in \texttt{TensorKit.jl}.
\begin{itemize}
    \item \texttt{bendright}:
    \[
    \raisebox{-64.41866pt}{\includegraphics{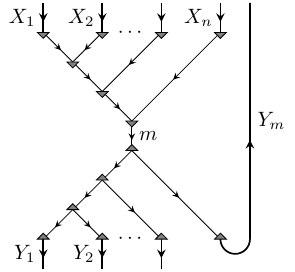}}.
    \]
This elementary move can be calculated using a single $B$-move.
    \item \texttt{bendleft}:
    \[
    \raisebox{-64.41866pt}{\includegraphics{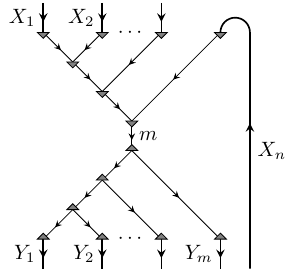}}.
    \]
This move is also elementary and can be calculated by taking the dagger of a
$B$-move.
\item \texttt{foldright}:
\[
    \raisebox{-64.41866pt}{\includegraphics{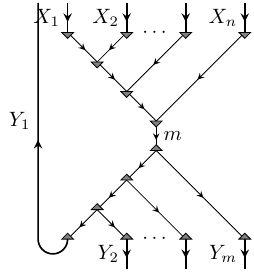}}.
    \]
This move is more involved. One first performs a sequence of $F$-moves to leave
the first leg dangling, then performs an $A$-move, and finally applies a
sequence of daggered $F$-moves to return the fusion tree to its standard form.
    \item \texttt{foldleft}:
  \[
    \raisebox{-64.41866pt}{\includegraphics{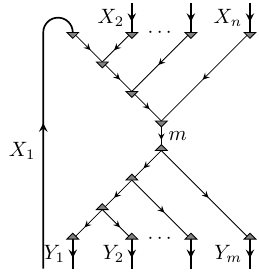}}.
    \]
    Similarly, \texttt{foldleft} can be calculated using a sequence of daggered
    $F$-moves, a single daggered $A$-move, and a sequence of $F$-moves.
\end{itemize}

Having established how to calculate duality maps, we can evaluate arbitrary
$\cR$-tensor contractions by applying a tensor transposition, performing a
vertical tensor contraction, and then applying another tensor transposition.
For example,
\[
\raisebox{-21.53955pt}{\includegraphics{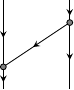}}\  = \ 
\raisebox{-21.53955pt}{\includegraphics{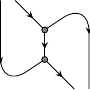}}.
\]

\subsubsection{Decomposition and truncation}\label{sec:decomposition-and-truncation}

Another basic tensor operation is tensor decomposition. For simplicity, we
consider only singular value decomposition (SVD) here. For further details, see
the implementations in \texttt{MatrixAlgebraKit.jl}~\cite{MatrixAlgebraKit_code}.

Suppose an $\cR$-tensor $f:A \rightarrow B$ is expressed in multi-matrix form as
\[
f = \sum_{m,IJ}f^m_{I,J} \xt q^B_{m, I} \circ q^{A,\dagger}_{m, J}.
\]
The SVD of $f$ is defined by performing a sectorwise SVD on each $f^m$:
\[
f^m_{I, J} = \sum_{K = 1}^{\mathrm{rank}(f^m)}U^m_{I,K}\xt \sigma^m_{K}\xt V^m_{K,J}.
\]
We then define a new object
\[
M = \bigoplus_{m\in\Irr\cR}\C^{\mathrm{rank}(f^m)} \odot m
\]
and the morphisms
\begin{align*}
    U & := \sum_{m,IK}U^m_{I,K} \xt q^B_{m, I} \circ q^{M\dagger}_{m, K},\\
    S & := \sum_{m, K}\sigma^m_K \xt q^{M}_{m, K} \circ q^{M\dagger}_{m, K},\\
    V & := \sum_{m, KJ}V^m_{K,J} \xt q^{M}_{m, K}\circ q^{A\dagger}_{m, J},
\end{align*}
satisfying
\begin{gather*}
f = U \circ S\circ V, \quad U^\dagger\circ U = \id_{M},\\
V\circ V^\dagger = \id_M,\quad S^\dagger = S > 0.
\end{gather*}
Graphically, the decomposition is
\[
\raisebox{-46.725pt}{\includegraphics{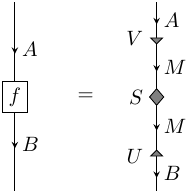}}.
\]
Other decompositions, such as the QR decomposition, can be defined similarly by
performing a sectorwise QR decomposition and then combining the decomposed
matrices from different sectors.

Bond dimension is an important measure of the complexity of a tensor-network
algorithm, and serves as an input constraint when performing tensor decomposition.
We therefore also need to define bond dimension for an $\cR$-tensor.
Here, the quantum dimension of the bond objects replaces the usual bond
dimension of an ordinary tensor.

To see why, consider an $\cR$-tensor with $n$ legs, each labeled by an object
$V\in\cR$. For any finite $n$, we have the following well-known inequality:
\[
\dim\cR(\one, V^{\otimes n}) \leq d_V^n,
\]
where $d_V$ is the quantum dimension of $V\in\cR$. Moreover, as $n$ goes to infinity, one has the equation
\[
\limsup_{n\rightarrow \infty}\left(\dim\cR(\one, V^{\otimes n})\right)^{\frac{1}{n}} = d_V.
\]
Thus, as $n$ tends to infinity, the dimension of an $\cR$-tensor scales as
$d_V^n$. Therefore, by
analogy with ordinary tensors, the quantum dimension correctly captures the
scaling of the number of degrees of freedom of an $\cR$-tensor as the number of
legs increases. In practical calculations, however, $n$ is typically not very
large, so the number of degrees of freedom is generally smaller than $d_V^n$.

Recall that the purpose of the SVD is to obtain the optimal low-rank
approximation of a matrix in the Frobenius norm. To generalize the truncated SVD to an
$\cR$-matrix, we first compute the Frobenius norm, which is defined by the
quantum trace:
\[
\|f\|^2_F := \Tr(f^\dagger \circ f) = \Tr(S^2) = \sum_{m,K}d_m\xt \left(\sigma^m_{K}\right)^2.
\]
Graphically,
\[
\|f\|^2_F =
\raisebox{-67.914pt}{\includegraphics{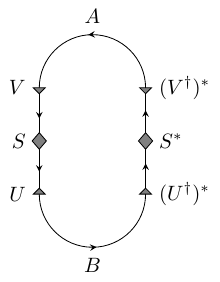}}
=\raisebox{-16.85pt}{\includegraphics{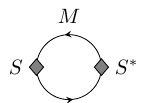}}.
\]
The right-hand strand carries the dual morphisms $(V^\dagger)^*$,
$S^*=(S^\dagger)^*$, and $(U^\dagger)^*$. The second equality uses
$U^\dagger\circ U=\id_M$ and $V\circ V^\dagger=\id_M$, leaving the
closure that evaluates to $\Tr(S^2)$.
Therefore, the optimal rank-$\chi$ approximation is obtained by the following procedure.
\begin{enumerate}
    \item For each sector $m$, collect the singular values $\sigma_K^m$ and
    sort them in descending order.
    \item Suppose we retain the first $r_m$ singular values in each sector
    $m$. The squared Frobenius norm of the truncated tensor is then
    \[
    \|f_{\text{trunc}}\|_F^2
    =\sum_m d_m\sum_{K=1}^{r_m}(\sigma_K^m)^2.
    \]
    \item Choose the retained ranks by solving
    \[
    \{r_m\}_{m\in\Irr\cR}
    \in\underset{\{r_m\}:\,\sum_m d_mr_m\leq\chi}{\arg\max}
    \|f_{\text{trunc}}\|_F^2.
    \]
    Define the truncated space
    \[
    M_{\text{trunc}} = \bigoplus_{m\in\Irr\cR}\C^{r_m} \odot m
    \]
    and the truncated tensors
\begin{align*}
    U_{\text{trunc}} & := \sum_{m,I}\sum_{K = 1}^{r_m}U^m_{I,K} \xt q^B_{m, I} \circ q^{M_{\text{trunc}}\dagger}_{m, K},\\
    S_{\text{trunc}} & := \sum_{m}\sum_{K = 1}^{r_m}\sigma^m_K \xt q^{M_{\text{trunc}}}_{m, K} \circ q^{M_{\text{trunc}}\dagger}_{m, K},\\
    V_{\text{trunc}} & := \sum_{m, J}\sum_{K = 1}^{r_m}V^m_{K,J} \xt q^{M_{\text{trunc}}}_{m, K}\circ q^{A\dagger}_{m, J}.
\end{align*}
\item The squared truncation error is
\[
\|f-f_{\text{trunc}}\|_F^2
=\sum_m d_m\sum_{K=r_m+1}^{\mathrm{rank}(f^m)}(\sigma_K^m)^2,
\]
where $f_{\text{trunc}}=U_{\text{trunc}}\circ S_{\text{trunc}}\circ V_{\text{trunc}}$.
\end{enumerate}

This procedure could be improved by a lazy algorithm that computes the
truncated SVD directly, avoiding the dense sectorwise SVDs followed by
truncation.

\subsection{\texorpdfstring{$\cR$}{R}-Loop-TNR}\label{sec:r-loop-tnr}
We now generalize Loop-TNR to $\cR$-tensor networks.
Introduced in Ref.~\cite{Yang_2017}, Loop-TNR is a state-of-the-art
tensor-network renormalization scheme that yields accurate and stable RG flows
for critical theories.
\texttt{TensorKit.jl}~\cite{Devos_2025_TensorKit,jutho_2025_15243493}
provides the general $\cR$-tensor operations described above.
Its interface allows computations with general $\cR$-tensor networks to be
written in much the same way as computations with ordinary tensor networks.
\texttt{TNRKit.jl}~\cite{Vanthilt_2026_TNRKit,TNRKit_code} builds on these
operations, so that the same Loop-TNR implementation applies to a general
UFC $\cR$ once its categorical data are supplied through the TensorKit
interface. Following Julia's generic programming model, the TensorKit
ecosystem separates the data defining $\cR$ from the numerical algorithms
that act on its tensors.

A coarse-graining step of $\cR$-Loop-TNR consists of two substeps. The first
approximates a two-site-unit-cell MPS on a loop by a lower-rank
single-site-unit-cell MPS.
The corresponding loop-optimization problem is
\begin{equation}\label{eq:loop}
\operatorname*{argmin}_{S_1,S_2,\cdots,S_8}
\left\|
\raisebox{-32.12853pt}{\includegraphics{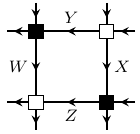}}-
\raisebox{-36.13801pt}{\includegraphics{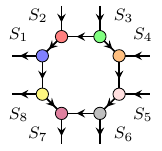}}\right\|_F^2.
\end{equation}

The loop-optimization problem can be solved approximately using alternating
least squares (ALS). We initialize the ALS tensors by decomposing each rank-4
tensor with an SVD in each of the two lattice directions and truncating the
resulting singular-value spectra to the cutoff bond dimension $\chi$.

In the first direction, the truncated SVD of a rank-4 tensor $T$ is obtained by
applying the procedure in Sec.~\ref{sec:decomposition-and-truncation}:
\[
\raisebox{-28.85274pt}{\includegraphics{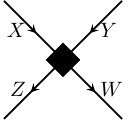}} \approx
\raisebox{-44.07918pt}{\includegraphics{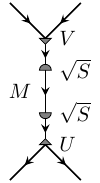}}
\]
where
\[
U^\dagger \circ  U = \id_{M},\quad V \circ V^\dagger = \id_{M},\quad S^\dagger = S,\quad S> 0.
\]

To perform the SVD in the orthogonal direction, we first transpose the fusion
tree into the corresponding basis:
\begin{align}
 &\quad \sum_{m}\sum_{\mu,i\alpha, j\beta} \sum_{\nu,k\gamma,l\delta}T^m_{(\mu,i\alpha, j\beta), (\nu,k\gamma, l\delta)} \nonumber\\
 &\qquad \times \raisebox{-65.43341pt}{\includegraphics{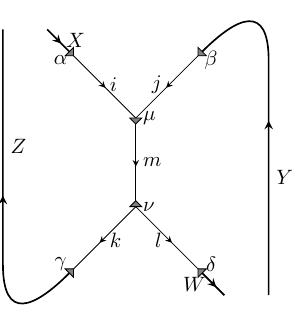}} \nonumber\\
 &= \sum_{m}\sum_{\mu,i\alpha, j\beta} \sum_{\nu,k\gamma,l\delta}T^m_{(\mu,i\alpha, j\beta), (\nu,k\gamma, l\delta)} \nonumber\\
 &\qquad \times \sum_{\xi}\sqrt{\frac{d_m}{d_l}}A^{kl}_{m, \xi\nu}\sum_{n,\zeta,\lambda}F^{\bar{k}ij}_{l, (m, \mu, \xi), (n,\zeta,\lambda)} \nonumber\\
 &\qquad \times \sum_{\epsilon}\sqrt{\frac{d_l}{d_n}}\left(B^{nj}_{l,\epsilon\lambda}\right)^* \nonumber\\
 &\qquad \times \raisebox{-65.43341pt}{\includegraphics{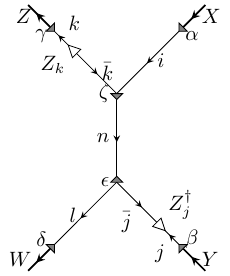}}
\end{align}

We therefore define a set of matrices $\tilde{T}^n$ with elements
\begin{multline}
    \tilde{T}^n_{(\zeta, \bar{k}\gamma, i\alpha), (\epsilon, l\delta, \bar{j}\beta)} := \sum_{m,\mu\nu,\xi\lambda}\sqrt{\frac{d_m}{d_n}}T^m_{(\mu,i\alpha, j\beta), (\nu,k\gamma, l\delta)}\\
    A^{kl}_{m, \xi\nu}F^{\bar{k}ij}_{l, (m,\mu,\xi), (n, \zeta, \lambda)}\left(B^{nj}_{l,\epsilon\lambda}\right)^*.
\end{multline}
Taking an SVD of each matrix $\tilde{T}^n$ yields a second decomposition, with a
new emergent object $\tilde{M}\in\cR$ on the internal leg:
\[
\raisebox{-28.85274pt}{\includegraphics{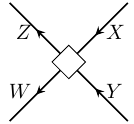}} \approx
\raisebox{-44.07918pt}{\includegraphics{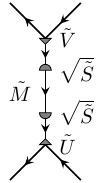}}.
\]
The objects $M$ and $\tilde{M}$ arise from decompositions of the same tensor
along two orthogonal directions and need not coincide.
The morphisms $\tilde{U}$, $\tilde{V}$, and $\tilde{S}$ satisfy the analogous identities
\[
\tilde{U}^\dagger \circ  \tilde{U} = \id_{\tilde{M}},\quad \tilde{V} \circ \tilde{V}^\dagger = \id_{\tilde{M}},\quad \tilde{S}^\dagger = \tilde{S},\quad \tilde{S}> 0.
\]

After the SVD initialization, we minimize the loop cost function using ALS, as
in ordinary Loop-TNR. Varying the cost function in Eq.~\eqref{eq:loop} with
respect to one $S_i$ while holding the other seven tensors fixed yields a linear
equation for $S_i$. We solve these equations successively for
$i=1,2,\ldots,8$. For example, the equation for $S_4$ is
\begin{align}
&\raisebox{-37.31089pt}{\includegraphics{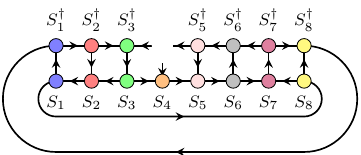}} \nonumber\\
&=
\raisebox{-37.31089pt}{\includegraphics{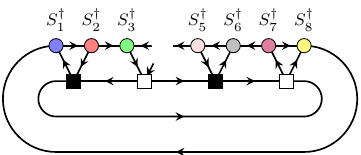}}
\end{align}
As an equation between morphisms in $\cR$, this becomes
\[
\cN_4\circ S_{4} = \cW_4,
\]
where the morphisms $\cN_4$, $S_4$, and $\cW_4$ are specified by the
following diagrams, including the transpositions of their external legs:
\begin{align*}
\cN_4 &= 
\raisebox{-52.61557pt}{\includegraphics{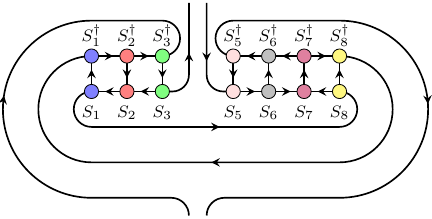}}\ ,\\
S_4 &= 
\raisebox{-29.85253pt}{\includegraphics{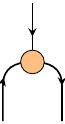}},\\
\cW_4 &= \raisebox{-52.61557pt}{\includegraphics{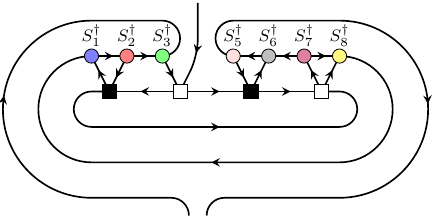}}.
\end{align*}
Decomposing $\cN_4$, $S_4$, and $\cW_4$ into their sectorwise matrices reduces
the morphism equation to
\[
\cN_4^i S_4^i = \cW_4^i,\quad \forall i \in\Irr\cR,
\]
which can be solved sector by sector using standard linear-algebra routines.

In the second and final substep, four rank-3 tensors produced by the loop
optimization are contracted into a new rank-4 tensor. For $\cR$-tensors, these
contractions are simply compositions of morphisms:
\begin{equation}\label{eq:next-ta-tb}
    \begin{aligned}
\raisebox{-28.85274pt}{\includegraphics{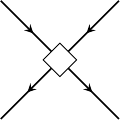}} &=
    \raisebox{-28.85274pt}{\includegraphics{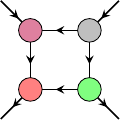}},\\
\raisebox{-28.85274pt}{\includegraphics{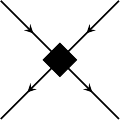}} &=
    \raisebox{-28.85274pt}{\includegraphics{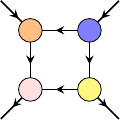}}.
\end{aligned}
\end{equation}

\section{Holographic interpretation and physical data}
\label{sec:holographic-data}

The local $\cR$-tensor operations and the $\cR$-TNR algorithm are collected in
Sec.~\ref{sec:r-tnr}.  We now show that $\cR$-TNR has a natural global holographic interpretation
that motivates the extraction of universal bulk observables.

Section~\ref{sec:holographic-picture} reorganizes the global contraction of an
$\cR$-tensor network as a Turaev--Viro evolution acting on a specific state defined by the dynamics, thereby identifying the network as
a boundary condition for $\fZ(\cR)$ and its relative partition function as a
vector in the torus ground-state space.

The remaining constructions are operational consequences of this picture.
In particular, an $\cR$-tensor
network on a tube corresponds to a time-evolution operator on a solid tube,
with the topological order $\fZ(\cR)$ in its bulk.
Section~\ref{sec:transfer-morphism-diagonalization} generalizes the transfer
matrices of ordinary $\Hilb$-tensor networks to tube-category morphisms for
$\cR$-tensors. Bulk anyon insertions resolve their spectra into
$\Irr\fZ(\cR)$ sectors, giving an algebraic realization of this picture.

Section~\ref{sec:modular-covariance}
derives the modular $T$- and $S$-covariance of these sector-resolved partition
functions. Finally, Sec.~\ref{sec:structure-constants} extends fixed-point
three-point-function methods to $\cR$-tensor networks, yielding sector-resolved structure constants.

\subsection{Holographic picture}\label{sec:holographic-picture}
We now describe the holographic interpretation of an $\cR$-tensor network,
which also appears in the TQFT and symTFT tensor-network constructions of
Refs.~\cite{chen2024cftdtqftd1holographictensor,Cheng:2023kxh,
HungCFTFactory2025,JiChenYangHung2024}.

Instead of contracting an $\cR$-tensor network locally in $\cR$-TNR, we change
the viewpoint by reordering the evaluation. We evaluate the graph states
globally. To avoid
subtleties associated with an infinite tensor network, we consider a finite
network on a torus.
This viewpoint makes explicit
that the $\cR$-tensor network defines a boundary condition for the
$\fZ(\cR)$ topological order, while its relative partition function is a vector
in the corresponding torus ground-state space.

Let $\Gamma$ denote a square network on the torus.
The $\cR$-tensor network on $\Gamma$ can be represented schematically by
\[
Z^{\text{sym}}_{\cR}=
\raisebox{-49.9923pt}{\includegraphics{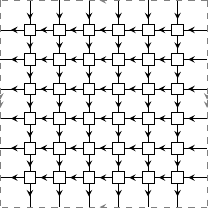}}.
\]
The dashed gray arrows indicate how opposite boundaries are identified to form
the torus, and tensor legs on identified boundaries are contracted pairwise.

For simplicity, consider a tensor network with fewer $\cR$-tensors.
Inserting the completeness relation in
Eq.~\eqref{eq:orthonormal-completeness} on every bond gives
\[
\vcenter{\hbox{\includegraphics[width=.96\columnwidth]{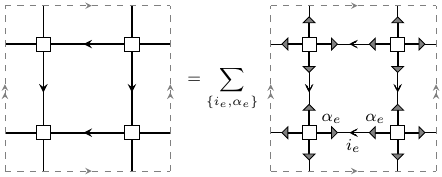}}}.
\]
On each bond $e$ carrying $X_e$, the paired triangles represent
$q^{X_e}_{i_e,\alpha_e}$ and its dagger. The sum runs over the simple
labels $i_e$ and their corresponding multiplicity indices $\alpha_e$
on all bonds. Only one bond's labels are shown. Opposite boundaries
are identified as in the preceding diagram.

As in Eq.~\eqref{eq:rank-4-fusion-tree-expansion}, we further decompose each rank-4
vertex in the fusion-tree basis. To match the convention used in the
Turaev--Viro state sum, we use the isotopy-invariant basis defined in
Eq.~\eqref{eq:isotopy-invariant-fusion-basis}:
\begin{equation}\label{eq:holographic-isotopy-fusion-tree-expansion}
\raisebox{-34.3432pt}{\includegraphics{figures/tikz/r-tensor-channel-expansion-lhs.pdf}}
=\sum_{m,\mu\nu}\widetilde T^{m}_{(\mu,i\alpha,j\beta),(\nu,k\gamma,l\delta)}
\vcenter{\hbox{\includegraphics{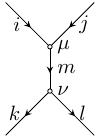}}}.
\end{equation}
The coefficients in the two bases are related by
\[
\widetilde T^{m}_{(\mu,i\alpha,j\beta),(\nu,k\gamma,l\delta)}
=\frac{\sqrt{d_m}}{(d_i d_j d_k d_l)^{1/4}}
T^{m}_{(\mu,i\alpha,j\beta),(\nu,k\gamma,l\delta)}.
\]

For the bulk time evolution discussed below, it is convenient to draw
the network in three dimensions, with a surface germ attached to each
string. This germ is used to indicate the framing of strings in the time
direction. These surface germs indicate the world sheets swept out by the
string-net strings during time evolution. In this representation, a
four-valent vertex labeled by $(m,\mu,\nu)$ resolves into two junctions
labeled by $\mu$ and $\nu$, joined by a face labeled by $m$,
\[
\vcenter{\hbox{\includegraphics[scale=.90]{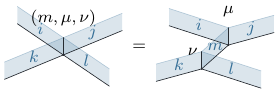}}}.
\]
For each fixed choice of $(i,j,k,l,m,\mu,\nu)$, we collect the
fusion-tree expansion coefficients of
Eq.~\eqref{eq:holographic-isotopy-fusion-tree-expansion} into an ordinary
rank-4 tensor $\mathbf{T}^{m\mu\nu}_{ijkl}$ with multiplicity
indices $\alpha,\beta,\gamma,\delta$:
\begin{equation}\label{eq:ordinary-rank-4-tensor}
\left(\mathbf{T}^{m\mu\nu}_{ijkl}\right)_{\alpha\beta\gamma\delta}
:= \widetilde T^{m}_{(\mu,i\alpha,j\beta),(\nu,k\gamma,l\delta)}.
\end{equation}

Let $\Gamma$ denote the graph of the original square tensor network on
the torus, with a vertex for each tensor and an edge for each bond.
After resolving each vertex in the fusion-tree basis, the resulting
colored trivalent graph is specified by
\[
\Lambda=\bigl(\{i_e\}_e,\{(m_v,\mu_v,\nu_v)\}_v\bigr).
\]
Here $i_e\in\Irr\cR$ labels the original bond $e$.
At each four-valent vertex $v$ of $\Gamma$, $m_v\in\Irr\cR$
labels the internal edge introduced by the resolution, and $\mu_v,\nu_v$
label the two trivalent basis morphisms.

Each colored trivalent graph defines a graph state $|\Gamma,\Lambda\ra$.
The four-valent notation for $|\Gamma,\Lambda\ra$ and its trivalent
resolution are illustrated below, with opposite resolution directions at
neighboring vertices. Labels are shown for two representative neighboring
vertices and their incident bonds.
\[
\vcenter{\hbox{\includegraphics[width=.96\columnwidth]{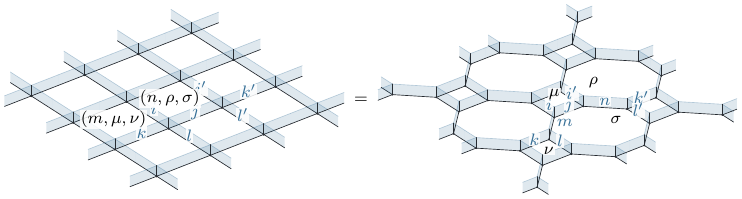}}}
\]
These states span the colored graph space with
a fixed underlying square graph $\Gamma$:
\begin{equation}
\cH_{\Gamma}:=\mathrm{span}_{\C}
\bigl\{|\Gamma,\Lambda\ra: \Lambda\bigr\}.
\end{equation}

For fixed $\Lambda$, let $\mathbf T_v(\Lambda)$ denote the tensor
$\mathbf T^{m_v\mu_v\nu_v}_{ijkl}$ defined in
Eq.~\eqref{eq:ordinary-rank-4-tensor}, where $i,j,k,l$ are the labels
assigned by $\Lambda$ to the four bonds incident on $v$, ordered according
to the chosen local fusion tree.
Contracting the multiplicity indices along the bonds of $\Gamma$ gives
\begin{equation}
A_\Gamma(\Lambda)
:=\tTr_\Gamma\!\left[
\bigotimes_{v\in V(\Gamma)}\mathbf T_v(\Lambda)
\right].
\end{equation}

The tensor network therefore defines a state frozen to the graph $\Gamma$:
\begin{equation}\label{eq:frozen-graph-state}
|Z_\Gamma\ra
=\sum_{\Lambda}
A_\Gamma(\Lambda)
|\Gamma,\Lambda\ra
\in\cH_{\Gamma}.
\end{equation}

This decomposition into graph states and their coefficients is illustrated by
\[
\vcenter{\hbox{\includegraphics[width=.96\columnwidth]{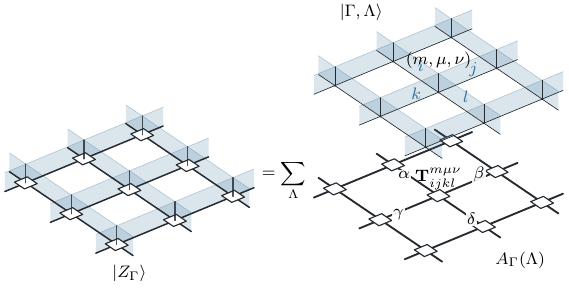}}}.
\]
The graph states $|\Gamma,\Lambda\ra$ form a spanning set of
$\cH_\Gamma$. For each fixed $\Lambda$, the ordinary tensor network in
the lower layer contracts to $A_\Gamma(\Lambda)$, the coefficient of
$|\Gamma,\Lambda\ra$ in the expansion of $|Z_\Gamma\ra$.

At this stage, however, the contraction is only partially complete.
In the graphical picture, composing $\cR$-tensors amounts to composing
colored graph states. The resulting evaluation map~\cite{Lan_2014} is the
wave-function renormalization map of the $\cR$ string-net model
\cite{Levin_2005,Chen_2010,Lan_2014}. These renormalization moves are
generated by local $F$-moves (or $2$--$2$ moves) and bubble moves
(or $2$--$0$ moves).

An $F$-move acts on two neighboring trivalent vertices. In the bulk
picture, it is implemented by a local Turaev--Viro evolution operator
associated with a tetrahedron,
\[
\vcenter{\hbox{\includegraphics[scale=.90]{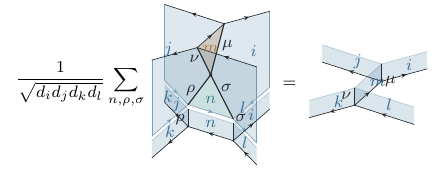}}}.
\]
In this relation, we sum over the internal channel $(n,\rho,\sigma)$ on the
left and fix the output channel $(m,\mu,\nu)$ on the right.
The embedding-tree basis is orthogonal but not normalized with respect to the
metric induced by the spherical trace. We compute the inner product by
gluing one tree to the adjoint of the other and taking the spherical trace,
\begin{equation}\label{eq:holographic-fusion-tree-metric}
\vcenter{\hbox{\includegraphics[width=.88\columnwidth]{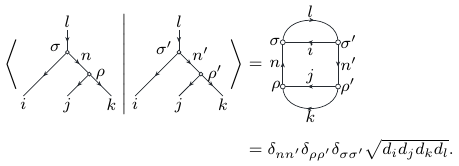}}}
\end{equation}
Each two-vertex tree has squared norm $\sqrt{d_i d_j d_k d_l}$.
We therefore include the inverse metric $1/\sqrt{d_i d_j d_k d_l}$ when
summing over the intermediate basis states.
The labeled tetrahedral branching is evaluated through its spherical link,
\[
\vcenter{\hbox{\includegraphics[width=.96\columnwidth]{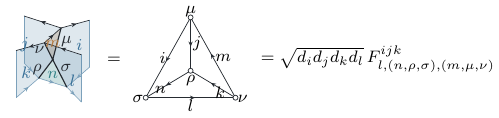}}}.
\]
Here we use the embedding-tree convention $l\to i\otimes j\otimes k$
for the graph states. With the gluing weight above,
the tetrahedral amplitude gives the coefficient
$F^{ijk}_{l,(n,\rho,\sigma),(m,\mu,\nu)}$ in the first diagram.

A bubble move replaces two trivalent vertices connected by two edges with
a single edge joining their remaining external legs. In the
handle-decomposition description of the Turaev--Viro state sum, the
corresponding local evolution is represented by a $1$-handle
\cite{Aasen_2019}.
The three-dimensional local piece in bubble removal is evaluated through
its theta link,
\[
\vcenter{\hbox{\includegraphics[scale=.90]{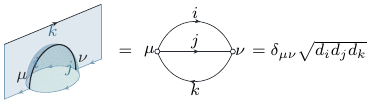}}}.
\]
The closed graphs use the isotopy-invariant vertices of
Eq.~\eqref{eq:isotopy-invariant-fusion-basis}.

For fixed $k$, the open bubble evaluates to
\[
\vcenter{\hbox{\includegraphics[scale=.90]{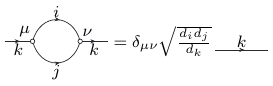}}}.
\]
Taking the spherical trace after composing one bubble's adjoint with the
other gives the pairing
\[
\vcenter{\hbox{\includegraphics[width=.96\columnwidth]{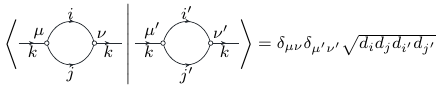}}}.
\]
This Gram matrix has rank one, since the nonzero bubbles are all
proportional to $\id_k$. Thus these diagrams do not form an independent
basis of $\cR(k,k)$.

Multiplying the theta amplitude above by the open bubble gives
$\delta_{\mu\nu}d_i d_j\,\id_k$.
Using rigidity and multiplicativity of quantum dimensions,
\[
\sum_{i,j}N^{ij}_k d_i d_j=d_k D^2,
\qquad D^2:=\sum_{i\in\Irr\cR}d_i^2.
\]
We therefore sum over $i,j\in\Irr\cR$ and
$\mu,\nu=1,\ldots,N^{ij}_k$, keeping $k$ fixed, and divide by $d_kD^2$,
\[
\vcenter{\hbox{\includegraphics[scale=.90]{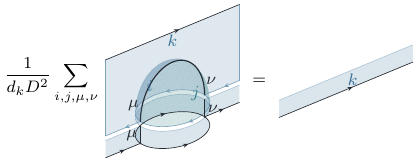}}}.
\]

Repeated application of these local Turaev--Viro evolutions gives a bulk
evolution in time that coarse-grains the lattice degrees of freedom.
Let $\Gamma_0,\Gamma_1,\Gamma_2$ denote the graphs before and after the
first two coarse-graining steps, and let $U_1$ and $U_2$ denote the
corresponding evolution operators,
\[
\vcenter{\hbox{\includegraphics[width=.96\columnwidth]{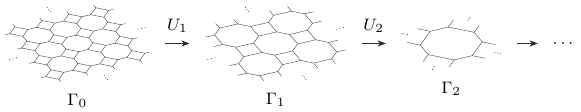}}}.
\]
Each coarse-graining step consists of a layer of $2$--$2$ moves followed
by a layer of $2$--$0$ moves,
\[
U_1=\vcenter{\hbox{\includegraphics[width=.78\columnwidth]{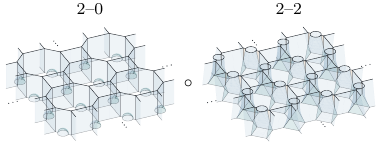}}},
\]
\[
U_2=\vcenter{\hbox{\includegraphics[width=.60\columnwidth]{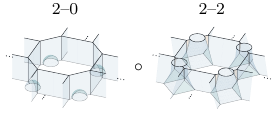}}}.
\]
The second round is shown in a square window rotated by $45^\circ$.

For a finite $\cR$-tensor network, coarse-graining proceeds until the
graph admits no further local reduction. Each $2$--$0$ move decreases
the number of vertices, so only finitely many reducing steps are required.
On the torus, we choose the terminal graph $\Gamma_N$ to consist of a
single four-valent vertex with two edges winding around the two
noncontractible cycles. As in the initial torus network, the dashed gray
boundaries with matching single or double arrows are identified,
\[
\Gamma_N=
\vcenter{\hbox{\includegraphics[scale=.90]{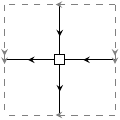}}}.
\]
The resulting bulk evolution is
\[
\begin{aligned}
Z_{\cR\text{-TV}}(t_N,t_0)&=U_N\circ\cdots\circ U_1\\
&:\cH_{\Gamma_0}\longrightarrow\cH_{\Gamma_N},
\end{aligned}
\]
where $\Gamma_0=\Gamma$.

Reaching this terminal graph does not exhaust the string-net relations.
Local moves and isotopies on the torus impose further equivalences among
its labeled graph states. Taking the quotient by these relations gives
the torus ground-state space $\cH_{T^2}$ of the $\cR$ string-net model.
The tube-category construction and the corresponding torus quotient are
described in Appendix~\ref{app:directional-tube-traces}.

Denote this quotient map by
\[
Z_{\cR\text{-TV}}(\mathrm{anyon},t_N):
\cH_{\Gamma_N}\longrightarrow\cH_{T^2}.
\]
Composing it with the finite coarse-graining evolution gives the evaluation
map~\cite{Lan_2014,Turaev:2017uxl},
\[
\begin{aligned}
Z_{\cR\text{-TV}}
&:=Z_{\cR\text{-TV}}(\mathrm{anyon},t_N)
\circ Z_{\cR\text{-TV}}(t_N,t_0)\\
&:\cH_{\Gamma_0}\longrightarrow\cH_{T^2}.
\end{aligned}
\]

The ground-state space $\cH_{T^2}$ has dimension $|\Irr\fZ(\cR)|$.
There is a basis of $\cH_{T^2}$ labeled by the simple anyons in
$\fZ(\cR)$~\cite{Aasen_2019}.
Let $x=(X,\gamma_{X,-})\in\Irr\fZ(\cR)$, where $X\in\cR$ and
$\gamma_{X,-}$ is a half-braiding, and denote the corresponding torus ground
state by $|x\ra$. We define the Turaev--Viro evaluation amplitudes
$E_\Gamma^x(\Lambda)$ by
\begin{equation}\label{eq:tv-evaluation-amplitude}
Z_{\cR\text{-TV}}|\Gamma,\Lambda\ra
=\sum_{x\in\Irr\fZ(\cR)}
E_\Gamma^x(\Lambda)|x\ra.
\end{equation}
Each $E_\Gamma^x(\Lambda)$ is a Turaev--Viro state-sum amplitude built from $F$-symbols
and quantum dimensions. The $F$-symbols supply tetrahedral weights, while the
quantum dimensions supply internal-edge weights.

Applying the evaluation map to Eq.~\eqref{eq:frozen-graph-state} gives
\begin{align}
Z^{\text{sym}}_{\cR}&:=Z_{\cR\text{-TV}}|Z_{\Gamma}\ra
=\sum_{x\in\Irr\fZ(\cR)}Z_x|x\ra,\label{eq:evaluated-partition-state}\\
Z_x&=\sum_{\Lambda}
A_\Gamma(\Lambda)
E_\Gamma^x(\Lambda) = \la x|Z_{\cR\text{-TV}}|Z_\Gamma\ra.
\end{align}
Thus, the $\cR$-tensor network defines a relative partition function. On a
closed surface, it is a vector in the ground state space rather than a scalar
\cite{Ji_2019,Ji_2020,FreedTeleman2014,Fiorenza_2015}.

Equation
\eqref{eq:evaluated-partition-state} also separates two ingredients of the
contraction. The coefficient $A_\Gamma$ is an ordinary tensor-network contraction, whereas
$E_\Gamma^x$ gives the amplitude of the bulk Turaev--Viro evolution.

The construction above applies to arbitrarily large finite networks. As the boundary network grows, the bulk evolution can contain arbitrarily many coarse-graining layers.
In the thermodynamic limit, the number of coarse-graining layers grows
without bound, and the bulk evolution extends to $t\to+\infty$.
With the boundary at $t=0$, formally,
\begin{equation}\label{eq:symmetry-to}
    \boxed{
    Z^{\sym}_{\cR} = Z_{\cR\text{-TV}}(+\infty,0)|Z_{\Gamma}\ra.
    }
\end{equation}
Thus the symmetric partition function $Z^{\sym}_{\cR}$ is the partition function of the bulk topological
order $\fZ(\cR)$ on the upper half-space, with boundary condition specified by
$|Z_\Gamma\ra$. In this sense, a system constrained by internal symmetry
intrinsically contains a bulk topological order.

Equation~\eqref{eq:symmetry-to} provides the starting point for the
symmetry-topological-order duality (symTO)
\cite{Ji_2019,Ji_2020,Chatterjee_2023,Kong_2020_math,Kong_2020_phy}. 
The symTO framework establishes a correspondence between a system
with categorical symmetry $\cR$ and a boundary of the non-chiral topological
order $\fZ(\cR)$, see Appendix~\ref{app:symto-review}. Our construction derives this correspondence explicitly from
the composition rules of the $\cR$-tensor network.

For a finite contractible patch of the $\cR$-tensor network,
the coarse-graining evolution can be represented by a
Turaev--Viro tent with a common bulk apex $v_*$, as illustrated
in Fig.~\ref{fig:holographic-dual-bulk}. Since the bulk evolution
is topological, the successive layers can be reorganized using
the $2$--$3$ and $2$--$0$ relations into this tent representation
of the renormalization process. In this picture, the
$\cR$-tensor network intrinsically carries an $\cR$-Turaev--Viro
TQFT in its bulk. The symmetry-constrained theory is therefore
naturally described as a relative boundary theory, with a
non-invertible gravitational anomaly characterized by the bulk
topological order $\fZ(\cR)$~\cite{Ji_2019,Ji_2020}.

We emphasize that the holographic duality here is symmetry-topological-order
holography, and at least not apparently, an AdS/CFT holography.

\begin{figure}[htb]
    \centering
    \includegraphics[width=.94\columnwidth]{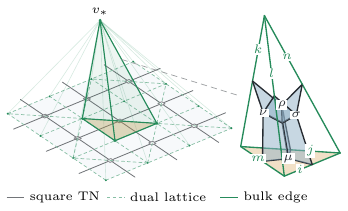}
    \caption{A Turaev--Viro tent over the dual triangulation of the tensor
    network.}
    \label{fig:holographic-dual-bulk}
    \label{fig:holographic-surface-dual-tetrahedron}
\end{figure}

\subsection{Diagonalizing a transfer morphism in the tube category}\label{sec:transfer-morphism-diagonalization}
As reviewed in Refs.~\cite{Vanthilt_2026_TNRKit,TNRKit_code}, at an RG fixed
point, a single local tensor approximates the partition function of a large
block of the microscopic lattice in the continuum limit~\cite{Levin_2007}. Its
indices (legs) carry boundary conditions assigned to the edges of that block.

As an application of this correspondence, contracting two opposite legs corresponds
to placing the tensor network on a tube~\cite{Gu_2009}. At the fixed point, the transfer
matrix becomes an imaginary-time evolution operator on a tube of a CFT or TQFT, and its
eigensystem encodes the energy spectrum and the corresponding eigenstates on a
circle~\cite{Gu_2009}. The spectrum and eigenstates at the RG fixed point are
important gauge-invariant quantities of the phase.

In practice, we divide the transfer-matrix eigenvalues in all sectors by
the same vacuum eigenvalue, shifting the lowest energy to zero and removing
the overall normalization factor. The central charge can be extracted
separately by comparing the leading eigenvalues, before this rescaling,
for different tube geometries constructed from the same fixed-point
tensors, as reviewed in Ref.~\cite{Vanthilt_2026_TNRKit}.

To extract CFT or TQFT data from an $\cR$-tensor network, we likewise need to
construct a transfer matrix from a fixed-point $\cR$-tensor and compute its
spectrum. At first sight, contracting two opposite legs of a rank-4
$\cR$-tensor inevitably causes the tensor legs to cross. If $\cR$ is not
endowed with a braiding, such a contraction goes beyond planar diagrams and
appears to be disallowed in general.

The solution is dictated by the intrinsic holographic interpretation of an
$\cR$-tensor. In Sec.~\ref{sec:holographic-picture}, the apparently nonplanar
torus contraction was completed by the bulk Turaev--Viro evolution rather than
by introducing a braiding in $\cR$. Therefore a braiding is not needed in the first place.

As emphasized in
Refs.~\cite{lan2025tubecategorytensorrenormalization,YBWL2605.06661},
the annular closure naturally takes values in the quotient obtained by
imposing the equivalence relation
\begin{equation}\label{eq:tube-equivalence-relation}
    \scalebox{0.85}{%
  $\raisebox{-61.5575pt}{\includegraphics{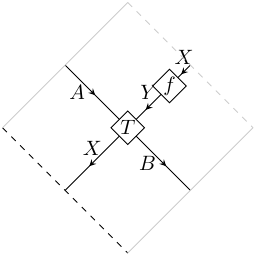}}
  \sim
  \raisebox{-61.5575pt}{\includegraphics{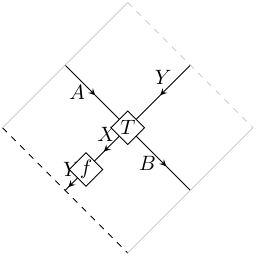}}$%
}
\end{equation}
for all $X,Y\in\cR$, $T\in\cR(A\otimes Y,X\otimes B)$, and
$f\in\cR(X,Y)$. This is precisely the dinaturality relation defining the
coend
\begin{equation*}
    \int^{V\in\cR}\cR(A\otimes V,V\otimes B).
\end{equation*}
Equipping these coend spaces with the composition induced by gluing annuli
gives the tube category. We review this construction in
App.~\ref{app:directional-tube-traces}.

For semisimple $\cR$, a concrete realization of the tube category is given in \cite{HARDIMAN_2019}.
We use the annular orientation convention of
Appendix~\ref{app:directional-tube-traces}.  It is the categorical counterpart of Ocneanu's
tube algebra and was developed in
Refs.~\cite{Ocneanu_1994,mueger2001subfactorscategoriestopologyii,HARDIMAN_2019,Lan_2014,lan2025tubecategorytensorrenormalization}.
\begin{definition}[Vertical tube category~\cite{HARDIMAN_2019}]
    Let $\cR$ be a unitary fusion category. The tube category associated with
    $\cR$, denoted by $\Tubev(\cR)$, is a unitary linear category defined by
    the following data.
    \begin{itemize}
        \item[-] The objects of $\Tubev(\cR)$ are the objects of $\cR$;
        \item[-] For $A,B\in\cR$, the morphism space in $\Tubev(\cR)$ is
        defined by
        \[
        \Tubev(\cR)(A, B) := \bigoplus_{i\in\Irr\cR}\cR(A\otimes i, i\otimes B).
        \]
        That is, any morphism $u$ in $\Tubev(\cR)(A,B)$ is an
        $\Irr\cR$-graded vector of the form
        \[
        \left(\raisebox{-28.65274pt}{\includegraphics{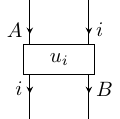}}\right)_{i\in\Irr\cR},
        \]
        where each $u_i$ is a morphism in $\cR(A\otimes i, i\otimes B)$;
        \item[-] For $u\in \Tubev(\cR)(B,C)$ and
        $v\in \Tubev(\cR)(A,B)$, their composition is defined by
        \[
        (u\circ v)_k = \sum_{i,j\in\Irr\cR}\sum_{\alpha = 1}^{N^{ji}_k}
        \raisebox{-57.10547pt}{\includegraphics{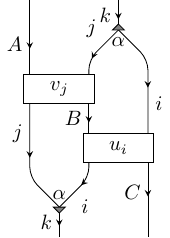}}.
        \]
        Here
        \[
            \left\{q_\alpha\in \cR(k, j\otimes i): \alpha = 1,2,\ldots, N^{ji}_k\right\}
        \]
        is an orthonormal basis of $\cR(k, j\otimes i)$;
        \item[-] The dagger structure~\cite{Green_2024} is defined by
        \[
        \left(\raisebox{-28.65274pt}{\includegraphics{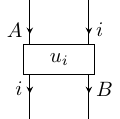}}\right)^\dagger :=
    \raisebox{-28.65274pt}{\includegraphics{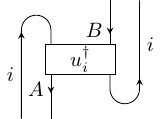}}.
        \]
    \end{itemize}
\end{definition}

We denote the universal dinatural map by
\begin{equation*}
    \pi_V^{A,B}:\cR(A\otimes V,V\otimes B)\longrightarrow\Tubev(\cR)(A,B).
\end{equation*}
Concretely, an arbitrary $\cR$-tensor
\[
T: A\otimes V\rightarrow V\otimes B,
\]
defines a tube closure $\pi_V^{A,B}(T)\in\Tubev(\cR)(A,B)$ obtained by
``contracting'' the legs labeled by $V$:
\[
\pi_V^{A,B}(T) :=
\left(\sum_{\alpha = 1}^{N^{V}_i}
\raisebox{-42.8791pt}{\includegraphics{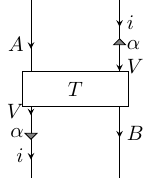}}\right)_{i\in\Irr\cR}.
\]
Here $N_i^V:=\dim\cR(i,V)$, and the triangles labeled by $\alpha$ denote an
orthonormal basis of $\cR(i,V)$ and its adjoint.
When $\cR=\Hilb$, we have $\Tubev(\Hilb)\cong\Hilb$, and
$\pi_V^{A,B}$ reduces to the usual trace over the Hilbert space $V$.
Thus, opposite legs of a rank-4 $\cR$-tensor can be contracted even when
$\cR$ is not braided. The result, however, is no longer an $\cR$-tensor but a
``$\Tubev(\cR)$-matrix,'' that is, a morphism in $\Tubev(\cR)$.

The problem of finding the energy spectrum of an $\cR$-tensor network on a
tube therefore reduces to defining what it means to diagonalize a
$\Tubev(\cR)$-matrix. For every $A\in\Tubev(\cR)$, the endomorphism algebra
$\Tubev(\cR)(A,A)$ is a semisimple $*$-algebra and is therefore $*$-isomorphic
to a direct sum of matrix algebras:
\begin{equation}\label{eq:tube-alg-iso}
    \Tubev(\cR)(A,A)\simeq \bigoplus_{s = 1}^{N_A}\Mat_{n_s(A)}(\C).
\end{equation}
Here $s$ labels the matrix summands, $N_A$ is the number of direct summands,
and $n_s(A)$ is the dimension of the defining representation of the $s$th
matrix algebra $\Mat_{n_s(A)}(\C)$.

To diagonalize a transfer morphism
$f=(f_i)_{i\in\Irr\cR}\in\Tubev(\cR)(A,A)$, we diagonalize its image in
\emph{each} matrix summand $\Mat_{n_s(A)}(\C)$. Any two such
$*$-isomorphisms differ only by a permutation of isomorphic matrix blocks and
unitary conjugation within each block, so the resulting collection of spectra
is independent of the chosen isomorphism.

We therefore seek all irreducible $*$-representations of
$\Tubev(\cR)(A,A)$ and diagonalize the corresponding representation matrices
of $f$. The matrix blocks are canonically labeled by simple objects of
the Drinfeld center $\fZ(\cR)$~\cite{Lan_2014,HARDIMAN_2019,Kitaev_2012, Green_2024}.

Following Hardiman~\cite{Hardiman_2019_graphical}, we will show how objects of
$\fZ(\cR)$ give rise to representations of $\Tubev(\cR)(A,A)$, with simple
anyons labeling the irreducible representations that occur. We first review
the ingredients of Hardiman's representation functor.

An object of $\fZ(\cR)$ is a pair
$x = (X,\gamma_{X,-})$, where $X\in\cR$ and
$\gamma_{X,A}:X\otimes A\rightarrow A\otimes X$ is a natural isomorphism, natural in
$A\in\cR$, called a half-braiding.
For $V\in\cR$, we represent the half-braiding by the positive crossing
\[
\gamma_{X,V}=
\raisebox{-30.95274pt}{\includegraphics{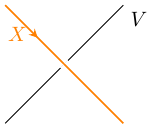}}.
\]
For every $A\in\cR$, we can equip the space $\cR(X,A)$ with the inner product
\begin{equation}\label{eq:inner-product-def}
    \la \phi|\psi \ra:= \frac{1}{d_X}\Tr(\phi^\dagger\circ \psi)
\end{equation}
for $\phi,\psi\in \cR(X,A)$.

We write $\ptr^{\ell}_V$ for the left partial trace, obtained by closing the
left $V$-leg using the duality conventions of Sec.~\ref{sec:r-tensor}.
Its explicit formula is given in Appendix~\ref{app:directional-tube-traces}.

\begin{definition}[Anyon-insertion functor]
For $x=(X,\gamma_{X,-})$ and $A \in \cR$, define $x^\sharp(A):=\cR(X,A)$. For
$A,B\in\cR$, define the action of $x^\sharp$ on tube morphisms by the linear
map
\begin{equation}\label{eq:anyon-insertion-rep}
    x^\sharp_{A,B}:\Tubev(\cR)(A,B)
    \rightarrow\Hilb\bigl(\cR(X,A),\cR(X,B)\bigr)
\end{equation}
determined as follows. For $V\in\cR$ and
$u:A\otimes V\rightarrow V\otimes B$, its action on
$\phi:X\rightarrow A$ is defined algebraically by~\cite{Lan_2014,Green_2024}
\begin{equation}\label{eq:anyon-insertion-partial-trace}
    x^\sharp_{A,B}\left(\pi_V^{A,B}(u)\right)(\phi)
    :=\ptr^{\ell}_V\!\left[
        u\circ(\phi\otimes\id_V)\circ\gamma_{X,V}^{-1}
    \right].
\end{equation}
This action is graphically represented by
\[
\raisebox{-52.21042pt}{\includegraphics{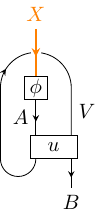}}.
\]
Equation~\eqref{eq:anyon-insertion-partial-trace} uses the inverse half-braiding to move the
$V$-string across the $X$-string.
We call $x^\sharp$ the \emph{anyon-insertion representation functor}, associated with $x$.
\end{definition}

The following result establishes the functoriality of this construction and
relates the resulting representation functors to the Drinfeld center.

\begin{proposition}[Unitary version of Hardiman~\cite{Hardiman_2019_graphical}]
\label{prop:hardiman-tube-center-equivalence}
For every $x\in\fZ(\cR)$, the construction above defines a $*$-functor
$x^\sharp:\Tubev(\cR)\rightarrow\Hilb$. Moreover, the assignment
$x\mapsto x^\sharp$ extends to a unitary equivalence
\begin{equation}
    (-)^\sharp: \fZ(\cR)^{\mathrm{op}}
    \cong \Fun_*\bigl(\Tubev(\cR),\Hilb\bigr),
\end{equation}
where the right-hand side is the category of $*$-functors to finite-dimensional
Hilbert spaces and natural transformations.
\end{proposition}

The physical intuition behind this representation is precisely the
holographic picture of Sec.~\ref{sec:holographic-picture}, expressed in terms
of the string-net renormalization of Refs.~\cite{Lan_2014,Kitaev_2012}, where
\begin{itemize}
    \item an object $x = (X, \gamma_{X, -})$ in $\fZ(\cR)$ is an anyon in the bulk;
    \item a state in that sector is the boundary condition on the plane prepared by the anyon;
    \item a morphism in the tube category is a layer of string-net wave-function renormalization, together with some dynamics-dependent data, acting on that prepared state.
\end{itemize}
When the closed string $V$ of the tube is moved past the inserted anyon, we apply the inverse half-braiding $\gamma_{X,V}^{-1}$ to move the anyon
past the string in Eq.~\eqref{eq:anyon-insertion-partial-trace}.
\[
\raisebox{-21.63954pt}{\includegraphics{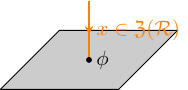}}
\xrightarrow{\mathrm{Log}}
\raisebox{-22.96227pt}{\includegraphics{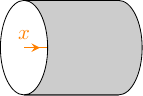}}
\]

We now express the matrix-block decomposition in Eq.~\eqref{eq:tube-alg-iso}
in terms of these representations.

\begin{proposition}[Sector decomposition of tube morphisms]
For all $A,B\in\cR$, the anyon-insertion maps associated with the simple
objects of $\fZ(\cR)$ assemble into a linear isomorphism
\begin{equation}\label{eq:tube-sector-decomposition}
\begin{aligned}
    \Phi_{A,B}:\Tubev(\cR)(A,B)
    &\xrightarrow{\ \simeq\ }
    \bigoplus_{z\in\Irr\fZ(\cR)}
    \Hilb\bigl(z^\sharp(A),z^\sharp(B)\bigr),\\
    u&\longmapsto
    \bigl(z^\sharp_{A,B}(u)\bigr)_{z}.
\end{aligned}
\end{equation}
For $A=B$, this is an isomorphism of finite-dimensional $C^*$-algebras.
\end{proposition}
\begin{proof}
By Proposition~\ref{prop:hardiman-tube-center-equivalence}, the
anyon-insertion functors $z^\sharp$, with $z\in\Irr\fZ(\cR)$, form a complete
set of simple tube representations. The semisimple Yoneda
decomposition of the tube category then gives Eq.~\eqref{eq:tube-sector-decomposition}.
Under this isomorphism, composition and the dagger operation act
componentwise.
\end{proof}

Setting $B=A$ in Eq.~\eqref{eq:tube-sector-decomposition} gives
\begin{equation}\label{eq:tube-alg-sector-iso}
    \Tubev(\cR)(A,A)\simeq
    \bigoplus_{\substack{z\in\Irr\fZ(\cR)\\z^\sharp(A)\neq0}}
    \End\bigl(z^\sharp(A)\bigr).
\end{equation}
Thus, whenever $z^\sharp(A)\neq0$, the representation
$z^\sharp_{A,A}$ is the projection onto a matrix summand and is
therefore irreducible. Conversely, every irreducible representation of
$\Tubev(\cR)(A,A)$ arises in this way, up to unitary equivalence.

Thus the matrix blocks indexed by $s$ in Eq.~\eqref{eq:tube-alg-iso}
can be labeled by simple objects $z$ in $\fZ(\cR)$ for which
$z^\sharp(A)\neq0$, with
$n_s(A)=d_{z^\sharp(A)}$. We refer to every
$z\in\Irr\fZ(\cR)$ as a \emph{sector}. For a fixed $A$, only the
sectors with $z^\sharp(A)\neq0$ occur in Eq.~\eqref{eq:tube-alg-sector-iso}.

More generally, one may insert a non-simple object
$x\in\fZ(\cR)$. The corresponding anyon-insertion functor $x^\sharp$
decomposes as the direct sum of those associated with the simple summands of
$x$. It therefore suffices to consider simple objects of $\fZ(\cR)$ when
resolving sectors.

Non-simple insertions are nevertheless useful when describing a change of
symmetry through anyon condensation
\cite{Ji_2020,Kong_2020_phy,Chatterjee_2023,
christian2023latticemodelcondensationlevinwen,Huston_2023}.
As reviewed in Sec.~\ref{appsec:changing-symmetry}, let
$A\in\fZ(\cR)$ be a condensable algebra, and let $\cR_{A}$
denote the resulting charge category. By the boundary--bulk relation
\cite{kong2015boundarybulkrelationtopologicalorders,Kong_2017,
jones2025localtopologicalorderboundary,
jones2025holographybulkboundarylocaltopological}, its bulk excitations are
labeled by
\[
    \fZ(\cR_{A})
    \cong \fZ(\cR)_{A}^{\localmodules}.
\]
Thus, inserting an object of $\fZ(\cR_{A})$ into an
$\cR_{A}$-tube is equivalently described, in the original
$\cR$-tube category, by inserting the corresponding local
$A$-module
\cite{christian2023latticemodelcondensationlevinwen,Huston_2023,Green_2024}.

In practical calculations, after each RG step of Loop-TNR, consider the
following unit cell of the tensor network:
\[
\raisebox{-28.65274pt}{\includegraphics{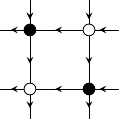}},
\]
From this unit cell, for each $x=(X,\gamma_{X,-})\in\Irr\fZ(\cR)$, we construct
sector-resolved transfer operators by applying the anyon-insertion functor
$x^\sharp$ to the corresponding tube morphisms. Evaluating these morphisms
yields operators on the space $\cR(X,A)$ of junctions between the anyon
and the $\cR$-tensor. This construction generalizes
the transfer matrices of
Refs.~\cite{Yang_2017,Bao:2019hfc,wei2023tensornetworkrenormalizationapplication,Gu_2009}:
\begin{equation}\label{eq:transfer-morphism-action}
    \raisebox{-51.1423pt}{\includegraphics{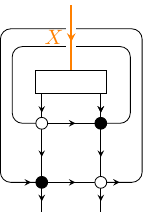}},\quad 
\raisebox{-52.56512pt}{\includegraphics{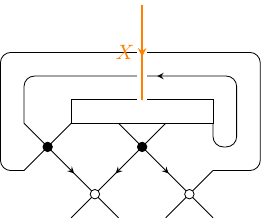}}.
\end{equation}

Two linear maps defined in Eq.~\eqref{eq:transfer-morphism-action} serve as inputs to
\texttt{KrylovKit.jl}~\cite{haegeman_2025_14911568} for computing the eigensystem,
directly generalizing the diagonalization of transfer matrices in
Refs.~\cite{Gu_2009,Yang_2017,wei2023tensornetworkrenormalizationapplication,Bao:2019hfc}.
To suppress irrelevant perturbations, we usually choose the second transfer
morphism during the RG flow.

For a $G$-symmetric system with $\cR=\Rep(G)$, where $G$ is finite,
$\fZ(\Rep(G))$ contains a canonical symmetric fusion subcategory
$\Rep(G)$ consisting of pure charges. For an object
$x=(X,\gamma_{X,-})$ in this subcategory, the half-braiding
$\gamma_{X,V}$ is given by the trivial symmetric braiding of $\Rep(G)$.
Thus, in the anyon-insertion representation, we can use this braiding to
move the $V$-leg across the $X$-leg and then take the partial trace over $V$.
This recovers the charge-resolved transfer matrices used in
\texttt{TNRKit.jl} for the bosonic case. However, a general implementation
of the anyon-insertion construction for resolving spectra in all sectors
$x\in\Irr\fZ(\cR)$ is not yet available in \texttt{TNRKit.jl}.

For a gapless system, the spectrum of a transfer morphism in a sector
$(X,\gamma_{X,-})\in\Irr(\fZ(\cR))$ encodes the conformal spectrum in the
corresponding sector of the symTO
theory, see review in Appendix~\ref{app:symto-review}.

\begin{remark}[Relation to previous sector-resolution methods]
The half-braiding construction for resolving topological sectors was
implemented numerically for the non-braided Haagerup fusion category
$\mathcal H_3$ in Ref.~\cite{Vanhove_2022}. An alternative procedure using
central idempotents of the MPO tube algebra was applied to
strange-correlator transfer matrices in Ref.~\cite{Bal_2018}, building on
Ref.~\cite{Bultinck_2017}. Here we formulate the anyon-insertion construction
as the $*$-functor $x^\sharp:\Tubev(\cR)\to\Hilb$, establish the sector
decomposition of general tube morphisms, and apply it to transfer morphisms
produced by Loop-TNR.
\label{rem:hardiman-numerics}
\end{remark}

\begin{remark}[Relation to defect-line TNR]
Hauru et al.~\cite{Hauru_2016} extracted sector-resolved conformal spectra
using symmetry defects and charge resolution. For $\cR=\Rep(\Z_n)$, the
defect and charge labels $(g,\chi)$ correspond to the simple objects of
$\fZ(\Rep(\Z_n))$. The present construction extends this description to
fusion-category symmetries. The anyon is inserted when constructing the
transfer morphism, so the same renormalized tensors can be used to resolve
all sectors. We do not follow the RG flow of a lattice defect operator,
such as the Kramers--Wannier duality operator.
\label{rem:defect-line-tnr}
\end{remark}

Following the convention in Refs.~\cite{Vanthilt_2026_TNRKit,TNRKit_code}, we denote the geometry
of a lattice tube by $[h,L,\xi]$, where $h$ is its Euclidean-time height, $L$
is its circumference, and $\xi$ is the horizontal displacement used when
gluing its ends (Fig.~\ref{fig:tube-shape-geometry}). The left transfer morphism in Eq.~\eqref{eq:transfer-morphism-action} has shape $[2,2,0]$, while the
right one has shape
$[\sqrt{2},2\sqrt{2},0]$.

\begin{figure}[htb]
    \centering
    \includegraphics[width=0.95\linewidth]{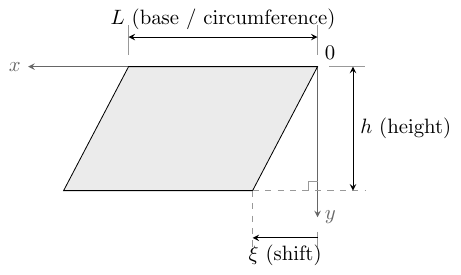}
    \caption{Geometry of a lattice tube with shape $[h,L,\xi]$.
    The positive $x$- and $y$-directions point left and down, respectively.
    The parallelogram has height $h$ and base length $L$. Its lower edge is
    displaced by $\xi$ in the positive $x$-direction.}
    \label{fig:tube-shape-geometry}
\end{figure}

When calculating conformal towers, it is necessary
to resolve higher conformal spins. As reviewed in
Ref.~\cite{Vanthilt_2026_TNRKit}, we use the transfer morphism with shape
$[\frac{4}{\sqrt{10}},2\sqrt{10},\frac{2}{\sqrt{10}}]$. This transfer
morphism, however, requires an additional RG step with a cutoff larger than the
one used during the RG flow. Here $\chi_{\text{RG}}$ denotes the cutoff used
during the RG flow, and $\chi_{\text{transfer}}$ denotes the cutoff used to
construct the transfer morphism. Typically,
$\chi_{\text{transfer}}\approx2\chi_{\text{RG}}$ is needed to maintain high
precision.

\subsection{Modular covariance}\label{sec:modular-covariance}
For a boundary theory defined relative to a $(2+1)$-dimensional topological order,
the sector-resolved torus partition functions form a modular-covariant vector.
This should be contrasted with the partition function of a standalone,
anomaly-free $(1+1)$-dimensional theory, which is modular invariant.
Such modular covariance was derived in Ref.~\cite{Ji_2019}.
Here we derive both the $T$- and $S$-covariance directly from the $\cR$-tensor network construction.

Before discussing modular covariance, we first define the partition function in each sector.
\begin{definition}[Partition vector]
    Let $u \in \Tubev(\cR)(A, A)$ be a morphism. For each
    $x\in\Irr\fZ(\cR)$, define the sector partition function by
    \[
    Z_x(u):=\Tr_{x^\sharp(A)}\left[x^\sharp_{A,A}(u)\right].
    \]
    The partition vector of $u$ is the vector indexed by $\Irr\fZ(\cR)$,
    \[
    Z(u):=\left(Z_x(u)\right)_{x\in\Irr\fZ(\cR)}.
    \]
    Equivalently, $Z(u)$ is obtained by applying the ordinary trace
    componentwise to the sector decomposition $\Phi_{A,A}(u)$.
\end{definition}

See App.~\ref{app:directional-tube-traces} for a more abstract definition.

\subsubsection{Modular $T$-covariance}

We first review how the modular $T$-constraint arises in a $\Hilb$-tensor network.
For a conventional translation-invariant tensor network,
the transfer matrix commutes with the translation operator $\cT$,
as can be illustrated graphically by
\[
\left[
\raisebox{-14.42636pt}{\includegraphics{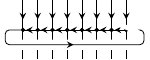}},\ 
\raisebox{-14.42636pt}{\includegraphics{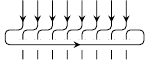}}
\right] = 0.
\]

In this situation, the transfer matrix and the translation operator can be simultaneously diagonalized.
Thus, as in band theory, the energy spectrum is resolved by the eigenvalues of the translation operator.
Moreover, since the translation operator $\cT$ generates a cyclic group $\Z_N$, where $N$ is the total number of bonds (equal to eight in the above graph),
its eigenvalues must be of the form $\ee^{\frac{2\pi \ii m}{N}}$ for $m = 0,\, 1,\,\cdots,\, N-1$.
A larger $N$ yields a finer resolution of the energy spectrum.
In a conformal scaling limit, this quantization implies that all conformal spins $h-\bar h$ are integers.

In the $\cR$-tensor network, however,
the translation operator $\cT$ is promoted to a morphism in $\Tubev(\cR)$ and need not satisfy $\cT^N = \id$.

We now analyze the eigenvalues of the translation operator in the anyon-insertion representation.
Suppose each vertical bond of the transfer operator is labeled by $A\in\cR$.
For $x=(X,\gamma_{X,-}) \in \Irr\fZ(\cR)$, consider the anyon-insertion representation $x^\sharp_{A^{\otimes N},A^{\otimes N}}$ of $\Tubev\left(\cR\right)\left(A^{\otimes N},A^{\otimes N}\right)$.
Let $\phi\in x^\sharp(A^{\otimes N})$ be a nonzero translation eigenvector:
\[
x^\sharp_{A^{\otimes N},A^{\otimes N}}(\cT)\phi = \lambda_N \phi.
\]
Here the subscript $N$ indicates the number of lattice bonds. To constrain
$\lambda_N$, we evaluate the $N$-fold translation:
\[
x^\sharp_{A^{\otimes N},A^{\otimes N}}(\cT^N)\phi =
\raisebox{-31.66034pt}{\includegraphics{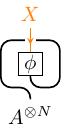}} =\theta_{x} \phi.
\]
Here
\begin{align*}
    \raisebox{-30.95274pt}{\includegraphics{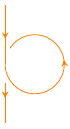}} &= \theta_{x}\id_{x}\\
&= \ee^{2\pi\ii s_{x}}\id_{x},
\end{align*}
which, with the orientation convention shown in the diagram, is the twist of the simple anyon
$x = (X,\gamma_{X,-})\in\fZ(\cR)$, and $s_{x}\in\mathbb{R}/\Z$ is its topological spin.

By functoriality of $x^\sharp$,
\[
\lambda_N^N\phi
= \left[x^\sharp_{A^{\otimes N},A^{\otimes N}}(\cT)\right]^N\phi
= x^\sharp_{A^{\otimes N},A^{\otimes N}}(\cT^N)\phi
= \theta_x\phi,
\]
and hence $\lambda_N^N=\theta_x$.

To relate this constraint to conformal spins, write
$\lambda_N=\ee^{\ii k_N}$, where $k_N$ is the dimensionless lattice momentum.
After choosing a representative of $s_x$, we have
\[
\ee^{\ii Nk_N}=\ee^{2\pi\ii s_x},
\qquad
\frac{Nk_N}{2\pi}=s_x+n_N,\quad n_N\in\Z.
\]
Since $k_N$ is defined modulo $2\pi$, the integer $n_N$ is determined only
modulo $N$. This gives $N$ possible translation eigenvalues.

Now suppose that the critical tensor network defining the transfer morphism
$u$ admits a conformal scaling limit. Let $a$ be the lattice spacing, and
take $a\to0$ and $N\to\infty$ with the circumference $L=Na$ fixed. Consider
a sequence of low-energy lattice states converging to a continuum momentum
eigenstate with finite momentum $p$. Identifying a one-site lattice
translation with a spatial displacement $a$ in the continuum requires a
choice of phase representatives such that
\[
\frac{k_N}{a}\longrightarrow p.
\]
It follows that
\[
Nk_N=L\frac{k_N}{a}\longrightarrow Lp.
\]
Thus, the exact lattice constraint implies
\[
\ee^{\ii Lp}
=\lim_{N\to\infty}\ee^{\ii Nk_N}
=\theta_x
=\ee^{2\pi\ii s_x},
\]
or equivalently,
\[
\frac{Lp}{2\pi}\in s_x+\Z.
\]
This expresses the fact that translating a state once around the spatial
circle produces the topological phase $\theta_x$.

In the continuum CFT, a state with conformal weights $(h,\bar h)$ on a
circle of circumference $L$ has momentum
\[
p=\frac{2\pi}{L}(h-\bar h).
\]
Substituting this relation into the momentum constraint gives
\[
h-\bar h\in s_x+\Z.
\]
Therefore, the fractional part of the conformal spin in sector $x$ is fixed
by the topological spin $s_x$.

Consequently, the sector partition functions obey
\begin{align*}
    Z_x(\tau+1,\bar\tau+1) &= Z_x\left(u \circ \cT^N\right)\\
    &=\theta_{x} Z_x(u)\\
    &= \theta_{x} Z_x(\tau, \bar\tau).
\end{align*}
This is the modular $T$-covariance of the sector partition functions \cite{Ji_2019, Ji_2020}.
We will see this feature in examples given later.

For $\cR=\Hilb$, the Drinfeld center $\fZ(\Hilb)\cong\Hilb$ has only the trivial simple sector,
for which $\theta=1$ and $s=0$ modulo $1$.
The above relations then reduce to $\cT^N=\id$ and $h-\bar h\in\Z$, as expected for the anomaly-free case.

\subsubsection{Modular $S$-covariance}

The modular $S$-invariance of a $\Hilb$-tensor network can be understood from
two evaluations of the same closed tensor contraction:
\[
\raisebox{-28.65274pt}{\includegraphics{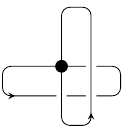}}.
\]
This scalar can be evaluated either as evolution along a vertical cylinder,
followed by gluing its upper and lower edges, or as evolution along a horizontal
cylinder, followed by gluing its left and right edges.  The equality of the two
evaluations is precisely modular $S$-invariance.

For a boundary theory with a nontrivial bulk, the partition vector is generally
not invariant under the modular $S$-transformation but is instead modular
\emph{covariant}~\cite{Ji_2019}.  We now formulate this covariance within the
$\cR$-tensor-network framework.

Orient the meridional and longitudinal cycles
as $(m,\ell)$.  With the tube convention above, changing the cut implements
$(m,\ell)\mapsto(\ell,-m)$ and hence the standard transformation
$\tau\mapsto-1/\tau$.

For a closed tensor contraction, given $T\in\cR(X\otimes A, A\otimes X)$ and $y=(Y,\gamma_{Y,-})\in\Irr\fZ(\cR)$, the partition function obtained from the
vertical cut of the solid torus is \footnote{Slightly different from App.~\ref{app:directional-tube-traces}, here we implicitely embed usual tensors to $\Torus$ in the partition vector notation.}
\[
Z_y^{\mathrm v}(T) =
\frac{1}{d_Y}\sum_{J}
\raisebox{-38.06593pt}{\includegraphics{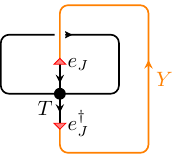}}
.
\]
where $\{e_J\}$ is an orthonormal basis with respect to the inner product in
Eq.~\eqref{eq:inner-product-def}.  Such a basis can be chosen in the form
\begin{equation}\label{eq:forget-embedding}
    \raisebox{-28.85274pt}{\includegraphics{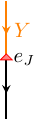}}
\ =\ 
\raisebox{-28.85274pt}{\includegraphics{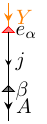}}.
\end{equation}
where $q_{j,\beta}^A$ is normalized as Eq.~\eqref{eq:orthonormal-completeness} and
$e_{\alpha}$ is normalized with respect to
Eq.~\eqref{eq:inner-product-def}.  Thus $J=(j,\alpha,\beta)$.

As in Appendix~\ref{app:directional-tube-traces}, the horizontal tube category is
$\Tubeh(\cR)=\Tubev(\cR^{\rev})$.  Hence a sector
$x=(X,\gamma_{X,-})\in\Irr\fZ(\cR)$ is identified in the horizontal direction
with
$\kappa(x):=(X,\gamma_{X,-}^{-1})\in\Irr\fZ(\cR^{\rev})$.
The horizontal partition function is therefore
\begin{align*}
Z_{\kappa(x)}^{\mathrm h}(T) &=
\frac{1}{d_X}
\sum_I
\raisebox{-42.71556pt}{\includegraphics{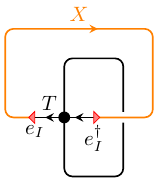}}.
\end{align*}

Modular $S$-covariance of the partition vector was formulated through the
mapping-class-group analysis of Ref.~\cite{Ji_2019}.  We now derive it directly
from the $\cR$-tensor-network construction.  We use the unnormalized Hopf-link
convention for the $S$-matrix~\cite{Turaev:2017uxl,MUGER2003159}:
\[
S_{xy} := \Tr_{X\otimes Y}
      \!\left(\gamma_{Y,X}\circ\gamma_{X,Y}\right).
\]
\begin{proposition}[Modular $S$-covariance]
For every $x=(X,\gamma_{X,-})\in\Irr\fZ(\cR)$ and
$T\in\cR(A\otimes B, B \otimes A)$,
\begin{equation*}
    Z_{\kappa(x)}^{\mathrm h}\!\left(T
    \right)
    =\frac{1}{D_\cR^2}
      \sum_{y\in\Irr\fZ(\cR)}S_{xy}Z_y^{\mathrm v}(T).
\end{equation*}
Graphically, this reads
\begin{multline}
\frac{1}{D_\cR^2}\sum_{y\in\Irr\fZ(\cR)}\frac{S_{xy}}{d_Y}\\
\times \sum_{I}
\raisebox{-25.61763pt}{\includegraphics{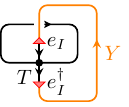}} \\
= \frac{1}{d_X}\sum_J
\raisebox{-30.6559pt}{\includegraphics{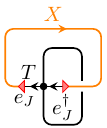}}.
\end{multline}
\end{proposition}

Equivalently,
\begin{multline}
\Tr\Bigl(T\circ \frac{1}{D_\cR^2}\sum_{y}\frac{S_{xy}}{d_Y}\sum_{I} \Bigr. \\
\Bigl. \raisebox{-28.85274pt}{\includegraphics{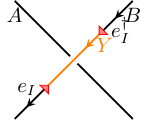}}\Bigr) \\
    = \Tr\Bigl(T\circ \frac{1}{d_X}\sum_J \Bigr. \\
    \Bigl. \raisebox{-28.85274pt}{\includegraphics{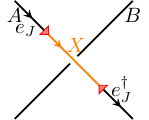}}\Bigr),\quad \forall T.
\end{multline}
Since the trace pairing $(T,M)\mapsto\Tr(T\circ M)$ is nondegenerate on the
relevant finite-dimensional matrix blocks, it suffices to prove the following
lemma.
\begin{lem}
\begin{multline}
\frac{1}{D_\cR^2}\sum_{y}\frac{S_{xy}}{d_Y}\sum_{I}
\raisebox{-28.85274pt}{\includegraphics{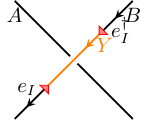}} \\
    = \frac{1}{d_X}\sum_J
\raisebox{-28.85274pt}{\includegraphics{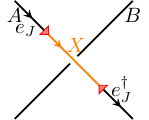}}.
\end{multline}
holds for all $A,B\in\cR$ and
$x=(X,\gamma_{X,-})\in\Irr\fZ(\cR)$.
\end{lem}
\begin{proof}
Using the basis $I=(i,\alpha,\beta)$ from
Eq.~\eqref{eq:forget-embedding}, it suffices to consider simple objects
$A=a$ and $B=b$.  Regarding both sides as morphisms in $\Tubeh(\cR)$ and taking
the weighted direct sum $\bigoplus_a d_a/D_\cR^2$, the desired identity becomes
\begin{multline}
\bigoplus_a\frac{d_a}{D_\cR^2}\sum_\alpha\frac{1}{D_\cR^2}
\sum_{y}\frac{S_{xy}}{d_Y} \\
\times \raisebox{-28.65274pt}{\includegraphics{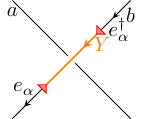}} \\
    =\frac{1}{D_\cR^2}
\raisebox{-30.95274pt}{\includegraphics{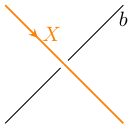}}.
\end{multline}

By Proposition~\ref{prop:hardiman-tube-center-equivalence}, the
anyon-insertion representations jointly detect morphisms in $\Tubeh(\cR)$.
For $z=(Z,\gamma_{Z,-})\in\Irr\fZ(\cR)$, both sides act on
\[
\raisebox{-28.75273pt}{\includegraphics{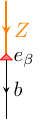}},
\]
as
\[
\frac{1}{D_\cR^2}\frac{S_{xz}}{d_Z}\xt
\raisebox{-28.75273pt}{\includegraphics{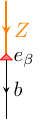}}.
\]
Here we have used the graphical center projector of
Lemma~3.10 of Ref.~\cite{mueger2001subfactorscategoriestopologyii}:
\begin{equation}\label{eq:center-projector}
    \raisebox{-28.85274pt}{\includegraphics{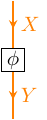}}\mapsto
\sum_{i\in\Irr\cR}\frac{d_i}{D_\cR^2}
\raisebox{-30.95274pt}{\includegraphics{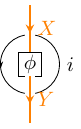}}.
\end{equation}
It projects $\cR(X,Y)$ onto
$\fZ(\cR)\bigl((X,\gamma_{X,-}),(Y,\gamma_{Y,-})\bigr)$.
\end{proof}

We emphasize that the modular covariance is kinematical. It holds for any
$\cR$-tensor $T$, independent of the specific values of the tensor elements, and
is determined completely by the charge category $\cR$. It is not a dynamical
condition that depends on how the interaction is defined.
A more abstract, basis-independent proof in terms of directional tube traces
and the change of cut is given in Appendix~\ref{app:directional-tube-traces}.

\subsection{Computing structure constants}\label{sec:structure-constants}
The numerical extraction of bulk CFT structure constants from a rank-three
fixed-point tensor, without assuming an internal symmetry, was first introduced in Ref.~\cite{Li_2022}.
In Ref.~\cite{Li_2022}, however, the relation between the rank-three fixed-point tensor and structure constants involved an undetermined parameter, termed the \emph{fundamental length}, which had to be fitted numerically.

Ref.~\cite{ueda2023fixedpointtensorfourpointfunction} derived the primary-field
three-point amplitude for a closed-string pair-of-pants geometry, obtaining a
fundamental length $x_S=\ee^{\pi/4}$. The local-coordinate maps and the calculation
of the corresponding geometric factor are reviewed in
Appendix~\ref{appsec:closed-string-pants}.

The extraction starts from the geometric interpretation of a fixed-point
tensor introduced in Ref.~\cite{Levin_2007} and reviewed in
Ref.~\cite{Vanthilt_2026_TNRKit}. At an RG fixed point, a rank-three tensor
represents the partition function of a coarse-grained triangular block, while
its three legs carry the boundary state spaces associated with the three
edges. Graphically,
\begin{equation}\label{eq:three-point-function-tensor}
    \raisebox{-24.58478pt}{\includegraphics{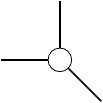}}
=
Z\left(
\raisebox{-31.564pt}{\includegraphics{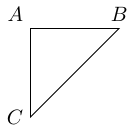}}
\right)
\end{equation}

Contracting each edge of the coarse-grained triangle with a transfer-matrix
eigenstate rolls that edge into a cylindrical end. Under these three
contractions, the vertices $A$, $B$, and $C$ of the triangle in
Eq.~\eqref{eq:three-point-function-tensor} are identified, while the edges
$AB$, $BC$, and $AC$ become the boundary circles of the three cylindrical
ends. The resulting surface is a closed-string pair of pants, with the right
isosceles triangle forming the central region that connects its three
cylindrical ends.

We now generalize this technique to fixed-point $\cR$-tensor networks. Using
the method introduced in Sec.~\ref{sec:transfer-morphism-diagonalization}, states in
each anyon sector of $\Irr\fZ(\cR)$ are obtained by diagonalizing the
corresponding tube transfer morphism. Contracting these sector-resolved states
with the rank-three fixed-point $\cR$-tensor gives access to additional
structure constants that are invisible in a $\Hilb$-tensor network.

Holographically, this contraction glues three anyon-threaded transfer tubes to
a pair of pants. The three bulk worldlines meet at a junction in
$\fZ(\cR)(x\otimes y,z)$. Therefore the resulting amplitude is resolved not
only by the three anyon sectors but also by their fusion-space index.
In examples below, all fusion spaces are one-dimensional, therefore we can omit the fusion-space index when labeling the resulting amplitudes.

Here, using the conventions of
Ref.~\cite{Vanthilt_2026_TNRKit,TNRKit_code}, we choose the transfer geometry with shape $[1, 2, 1]$ to determine the
conformal spin modulo~$2$.

We contract the transfer-morphism eigenstates $|x;a\ra$, $|y;b\ra$, and
$|z;c\ra$, obtained in the three directions, with the three legs of the
fixed-point tensor. The normalization of these states uses closed contraction divided by the quantum dimension of the anyon. Here $x=(X,\gamma_{X,-})$, $y=(Y,\gamma_{Y,-})$, and
$z=(Z,\gamma_{Z,-})$ are anyon sectors in $\Irr\fZ(\cR)$, while $a$, $b$,
and $c$ label eigenstates in the respective sectors. We then apply the center
projector in Eq.~\eqref{eq:center-projector}, which maps the resulting
morphism in $\cR(X\otimes Y,Z)$ to $\fZ(\cR)(x\otimes y,z)$.

With the incoming
and outgoing orientations shown in Fig.~\ref{fig:three-point-func-process},
the result is
\begin{equation}\label{eq:structure-amplitude-hom-space}
    \cA_{abc^*}
    \in\fZ(\cR)(x\otimes y,z).
\end{equation}
We interpret this numerical amplitude in terms of the nonchiral closed-string
three-point function
\[
\la\phi_a^{\closedsector}\phi_b^{\closedsector}\phi_{c^*}^{\closedsector}\ra_{\mathrm{cpants}}.
\]

\begin{figure}[tb]
  \centering
    \[
\raisebox{-33.05421pt}{\includegraphics{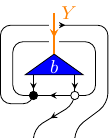}} = \lambda^y_b 
\raisebox{-18.11661pt}{\includegraphics{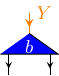}},\quad 
\raisebox{-25.29604pt}{\includegraphics{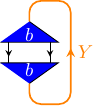}} = d_{Y}.
\]
  \quad
    \[
\raisebox{-26.09192pt}{\includegraphics{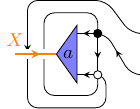}} = \lambda_a^x
\raisebox{-14.14171pt}{\includegraphics{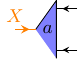}},\quad \raisebox{-25.29604pt}{\includegraphics{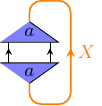}} = d_X.
\]
  \quad
    \[
\raisebox{-46.0006pt}{\includegraphics{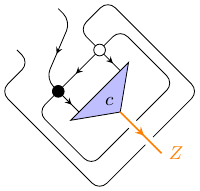}}
=
\lambda_c^z
\raisebox{-23.09233pt}{\includegraphics{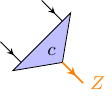}}
,\quad \raisebox{-31.07678pt}{\includegraphics{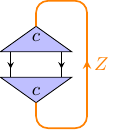}} = d_Z.
\]
\quad
    \[\cA_{abc^*} = 
\sum_{i\in\Irr\cR}\frac{d_i}{D_\cR^2}
\raisebox{-55.46689pt}{\includegraphics{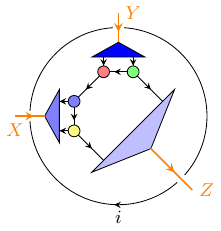}}
\]
  
  \caption{Construction of the bulk three-point amplitude. The first three
  panels show eigenstates of the anyon-insertion representations of the tube
  morphisms together with their normalization conventions, corresponding to primary or descendant states $a$, $b$, and $c$
  in the sectors $x$, $y$, and $z$ in $\fZ(\cR)$, respectively. The third panel uses the
  fixed-point tensor at the next Loop-TNR step to match the gauge of the bond
  entering the three-point contraction. The last panel evaluates
  $\cA_{abc^*}\in
  \fZ(\cR)(x\otimes y,z)$ for the fields
  $\phi_a$, $\phi_b$, and $\phi_c$ in sector $x$, $y$, and $z$, respectively.}
  \label{fig:three-point-func-process}
\end{figure}

In practice, each eigenstate obtained by diagonalizing a tube morphism has a
phase ambiguity. We assume throughout that the relevant transfer eigenspaces
are nondegenerate, so this ambiguity is only a $U(1)$ phase for each
eigenstate. Without a $C_4$-symmetric Loop-TNR implementation, the phase of a
given bulk field may be chosen independently in the up-, left-, and
tilted-contracted states. We therefore keep only the absolute values.

Choose the isotopy-invariant trivalent basis
$\{v_{xy}^{z,\alpha}\}_{\alpha}$ of
$\fZ(\cR)(x\otimes y,z)$ defined by
Eq.~\eqref{eq:isotopy-invariant-fusion-basis}. We decompose the numerical
amplitude into basis as
\begin{equation}\label{eq:structure-amplitude-decomposition}
    \cA_{abc^*}
    =\sum_\alpha \cA_{abc^*; \alpha}v_{xy}^{z,\alpha}.
\end{equation}
Motivated by the closed-string pair-of-pants geometry, we adopt the following
prescription to estimate the magnitudes of primary-field structure constants
from the numerical fixed-point tensor.
\begin{equation}\label{eq:extract-closed-structure-constants}
    \left|\frac{\cA_{abc^*;\alpha}}{\cA_{000}}\right|
    \simeq G_{\mathrm{cpants}}(\Delta_a,\Delta_b,\Delta_{c^*})
      \left|C^{\closedsector}_{abc^*;\alpha}\right|,
\end{equation}
where the complete derivation of
\begin{equation*}
G_{\mathrm{cpants}}(\Delta_a,\Delta_b,\Delta_c)
=\ee^{-\frac{\pi}{4}(\Delta_a+\Delta_b+\Delta_c)}
 (\sqrt{2})^{-\Delta_a-\Delta_b+\Delta_c}
\end{equation*}
is given in Appendix~\ref{appsec:closed-string-pants}. See in particular
Eq.~\eqref{eq:pants-gauge-app}. Here $\mathcal A_{000}$ is the amplitude of
the vacuum primary in the trivial anyon sector and its unique vacuum fusion
channel. Dividing by it removes the overall tensor normalization, including
the nonuniversal free-energy-density contribution.

The geometric factor is derived for the specified continuum geometry and
local-coordinate convention. Its application to the numerical fixed-point
tensor is supported by the agreement observed in our examples. However, the
matching has not been verified for generic states, including Virasoro
descendants.

Using the fundamental length $x_S=\ee^{\pi/4}$ gives good agreement with our numerical
results. However, the agreement is less precise for amplitudes involving higher scaling dimension, and for structure constants in some
sectors. The pair-of-pants description computed \emph{analytically} may require further corrections when
applied to \emph{numerical} fixed-point tensors. In particular, possible contributions from local fields at the conical
singularity with opening angle $4\pi$ or at sewing junctions remain to be
understood. Such fields may be induced by truncation in the algorithm.

The coefficients $C_{abc;\alpha}^{\closedsector}$ are closed-string structure
constants of the full nonchiral CFT, resolved in the bulk anyon sectors of
$\fZ(\cR)$. They depend on the chosen basis of the fusion space and throughout we
use the isotopy-invariant convention of
Eq.~\eqref{eq:isotopy-invariant-fusion-basis}. Equation~\eqref{eq:extract-closed-structure-constants}
determines only their magnitudes. Recovering relative phases requires a
consistent phase convention for the three transfer directions.

For the numerical implementation, we use the enlarged half-$4\xt4$
triangular patch detailed in
Appendix~\ref{app:four-by-four-structure-constants}. This suppresses
irrelevant perturbations while keeps the computation complexity to $\cO(\chi^6)$.

\Needspace{4\baselineskip}
Naively inverting the environment matrices $\cN_i^a$ of
Sec.~\ref{sec:r-loop-tnr} in the local RG updates can substantially degrade
the rotational and reflection symmetries of the renormalized tensors because
division by very small eigenvalues amplifies numerical errors.
The resulting loss of symmetry reduces the accuracy of the extracted
structure constants $C_{abc;\alpha}^{\closedsector}$.
In our tests, Krylov solvers reduce this loss of symmetry and yield
more accurate primary-field structure constants. The local solves can
converge slowly due to the ill-conditioning of the environment matrices, but their updates can still lower the loop cost even when
the prescribed residual tolerance is not reached. In these tests, incomplete
convergence mainly affects the accuracy of the descendant spectrum, while the primary scaling dimension is still correct.

\section{Simple examples with ordinary symmetry and generalized symmetry}\label{sec:ordinary-and-generalized-symmetry-examples}

In this section, we introduce several $\cR$-tensor-network models and compute
their CFT data at criticality. All of these critical points require
fine-tuning. We begin with the standard Ising model protected by $\Z_2$
symmetry in Subsection~\ref{sec:z2-statistical-ising}, using it to illustrate the
methods developed above, including the construction of a $\Rep(\Z_2)$-tensor
network, the extraction of $\fZ(\Rep(\Z_2))$-anyon-resolved spectra, and the
computation of structure constants.

We then consider random-cluster models with generalized symmetries in
Subsection~\ref{sec:random-cluster-models}. This class generalizes the Potts
random-cluster model by replacing the integer parameter $q$ with the quantum
dimension of a Q-system $Q$ in a UFC. The three-state Potts model and the
Fibonacci-symmetric random-cluster model are two examples.
This Q-system construction is closely related to the CFT-factory framework,
where critical lattice models arise from competing anyon condensates and their
couplings are organized by Frobenius-algebra data~\cite{HungCFTFactory2025}.

\subsection{\texorpdfstring{$\Z_2$}{Z2}-protected 2D statistical Ising model on square lattice}
\label{sec:z2-statistical-ising}
We begin with the simplest example, the two-dimensional statistical Ising model
on a square lattice, protected by $\Z_2$ symmetry. We formulate this model as a
$\Rep(\Z_2)$-tensor network and apply $\Rep(\Z_2)$-Loop-TNR to it. We then
extract the conformal towers and structure constants resolved by the sectors
of $\fZ(\Rep(\Z_2))$.

\paragraph{The model}
The degrees of freedom are Ising spins located at the sites
$\mathbf{i}\in\Z^2$ of the square lattice. Each spin $s_\mathbf{i}$ takes a
value in $\{-1,1\}$. The configuration space is therefore the Cartesian product
\[
\scrC = \prod_{\mathbf{i}\in\Z^2}\{-1, 1\},
\]
and a spin configuration is an element of $\scrC$:
\[
\mathbf{s} = (s_{\mathbf{i}})_{\mathbf{i}\in\Z^2} \in \scrC .
\]
The partition function is the sum over all spin configurations weighted by
their Boltzmann factors:
\[
Z(\beta) = \sum_{\mathbf{s}\in \scrC }\prod_{\langle \mathbf{i}\mathbf{j}\rangle}\ee^{\beta s_\mathbf{i}s_\mathbf{j}}.
\]
The critical inverse temperature is
\[
\beta_c = \frac{1}{2}\log(1+\sqrt{2}).
\]

\paragraph{The tensor-network representation}
We use the following tensor-network representation of the partition function
$Z(\beta)$:
\[
\raisebox{-67.43977pt}{\includegraphics{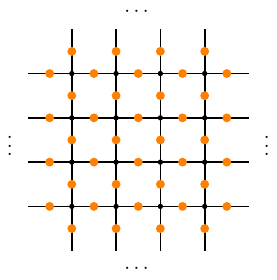}}.
\]
Each black four-leg tensor represents the $\delta$-function map
\begin{gather*}
    \raisebox{-14.62636pt}{\includegraphics{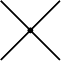}}:\C\{0, 1\} \otimes \C\{0, 1\} \rightarrow\C\{0, 1\} \otimes \C\{0, 1\}\\
e_{i}\otimes e_j \mapsto \delta_{ij}e_i\otimes e_i,\quad i,j = 0,1.
\end{gather*}
The binary labels $i\in\{0,1\}$ encode the Ising spins as
$s_i=(-1)^i$.

Each orange tensor, placed on a bond, represents the linear map
\begin{gather*}
    b(\beta):=\raisebox{-14.62636pt}{\includegraphics{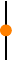}}:\C\{0, 1\} \rightarrow \C\{0, 1\},\\
e_i \mapsto \sum_{j = 0}^1\ee^{\beta s_is_j}e_j,\quad s_j = (-1)^{j}.
\end{gather*}

Combining these maps gives an equivalent representation in terms of the local
rank-four tensor
\[
T(\beta) = \raisebox{-28.85274pt}{\includegraphics{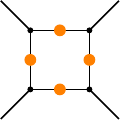}},
\]
where each black three-leg tensor represents the $\delta$-function map
\begin{gather*}
    m := \raisebox{-14.62636pt}{\includegraphics{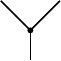}}: \C\{0, 1\} \otimes \C\{0, 1\} \rightarrow \C\{0, 1\},\\
e_{i}\otimes e_j \mapsto \delta_{ij}e_i.
\end{gather*}

\paragraph{The symmetric basis}
The vector space $\C\{0,1\}$ carries a $\Z_2$ representation $\rho$ defined by
\begin{gather*}
    \rho:\Z_2\rightarrow \End(\C\{0,1\}),\\
    i\mapsto \left(e_j\mapsto e_{i+j}\right),
\end{gather*}
where the indices are added modulo $2$. Both $m$ and $b(\beta)$ are
$\Z_2$ intertwiners with respect to this representation.

Consequently, $\C\{0,1\}$ decomposes into the direct sum of the trivial and
nontrivial irreducible $\Z_2$ representations:
\begin{gather*}
    e_+ = \frac{1}{\sqrt{2}}\left(e_0 + e_1\right),\ \mathbf{0} = \C e_+\in \Rep(\Z_2),\\
    e_- = \frac{1}{\sqrt{2}}\left(e_0 - e_1\right),\ \mathbf{1} = \C e_-\in \Rep(\Z_2), \\
    \C\{0,1\} = \mathbf{0} \oplus \mathbf{1} \in \Rep(\Z_2).
\end{gather*}

In this symmetry-adapted basis, we represent the sector $\mathbf{0}$ by dashed
lines and the sector $\mathbf{1}$ by solid lines. The three-leg tensor $m$ then
takes the form
\begin{multline}
m = \frac{1}{\sqrt{2}}\Bigl(
\raisebox{-21.73953pt}{\includegraphics{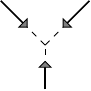}}
+
\raisebox{-21.73953pt}{\includegraphics{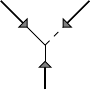}} \\
+
\raisebox{-21.73953pt}{\includegraphics{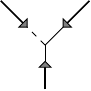}}
+
\raisebox{-21.73953pt}{\includegraphics{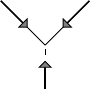}}\Bigr).
\end{multline}
In the same basis, the bond map becomes
\[
b(\beta) = 2\cosh{\beta}\xt \ 
\raisebox{-28.85274pt}{\includegraphics{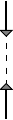}}
+
2\sinh{\beta}\xt\ 
\raisebox{-28.85274pt}{\includegraphics{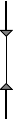}}.
\]

\paragraph{The symmetric partition function}
In the symmetry-adapted basis, the local tensors enforce $\Z_2$ charge
conservation, so only charge-conserving configurations contribute to the
partition function. The lines carrying nontrivial $\Z_2$ charge must therefore
form closed loops. Let $E$ denote the set of edges of the square lattice. The
relevant configuration space is therefore the constrained subset
\[
\scrC^{\sym}
:=\left\{
c\in\{0,1\}^{E}\;\middle|\;
\substack{\text{the occupied edges of $c$}\\
          \text{form closed loops}}
\right\}.
\]

On a closed manifold, $\scrC^{\sym}$ further decomposes into
topological sectors because local tensor contractions and recouplings preserve
the homology classes of the loops. On the torus, the two $\Z_2$ winding
parities label four sectors, and the partition vector is
\begin{multline}
|Z(\beta)\ra = Z_{00}(\beta)\xt \left|
\raisebox{-14.42638pt}{\includegraphics{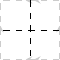}}\right\ra \\
+ Z_{01}(\beta)\xt\left|
\raisebox{-14.42638pt}{\includegraphics{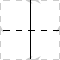}}\right\ra \\
+ Z_{10}(\beta)\xt\left|
\raisebox{-14.42638pt}{\includegraphics{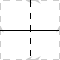}}\right\ra \\
+ Z_{11}(\beta)\xt\left|
\raisebox{-14.42638pt}{\includegraphics{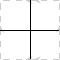}}\right\ra.
\end{multline}
Up to an overall normalization, the coefficient of the basis state $|ij\ra$ is
\[
Z_{ij}(\beta) \propto \sum_{c\in \scrC^{\sym}_{ij}}\tanh^{l(c)}{\beta}.
\]
Here $\scrC^{\sym}_{ij}\subset
\scrC^{\sym}$ denotes the sector with winding parities $(i,j)$,
and $l(c)$ is the total number of occupied links in the loop configuration
$c$. We emphasize that the four components $Z_{ij}$ are distinguished by
$\Z_2$ \emph{charge} lines winding around the two noncontractible cycles,
rather than by the $\Z_2$ \emph{defect} insertions used to define conventional
twisted partition functions.

The states $|00\ra$, $|01\ra$, $|10\ra$, and $|11\ra$ form a
basis for the four-dimensional ground-state subspace of the toric code on the
torus. They are related by a basis transformation to the toric-code anyon
basis $|\one\ra$, $|e\ra$, $|m\ra$, and $|f\ra$. This constrained
configuration space concretely realizes the holographic picture of an
$\cR$-tensor network described in Sec.~\ref{sec:holographic-picture} and
App.~\ref{app:directional-tube-traces}.

\begin{figure*}[tb]
    \centering
    \includegraphics[width=0.98\textwidth]{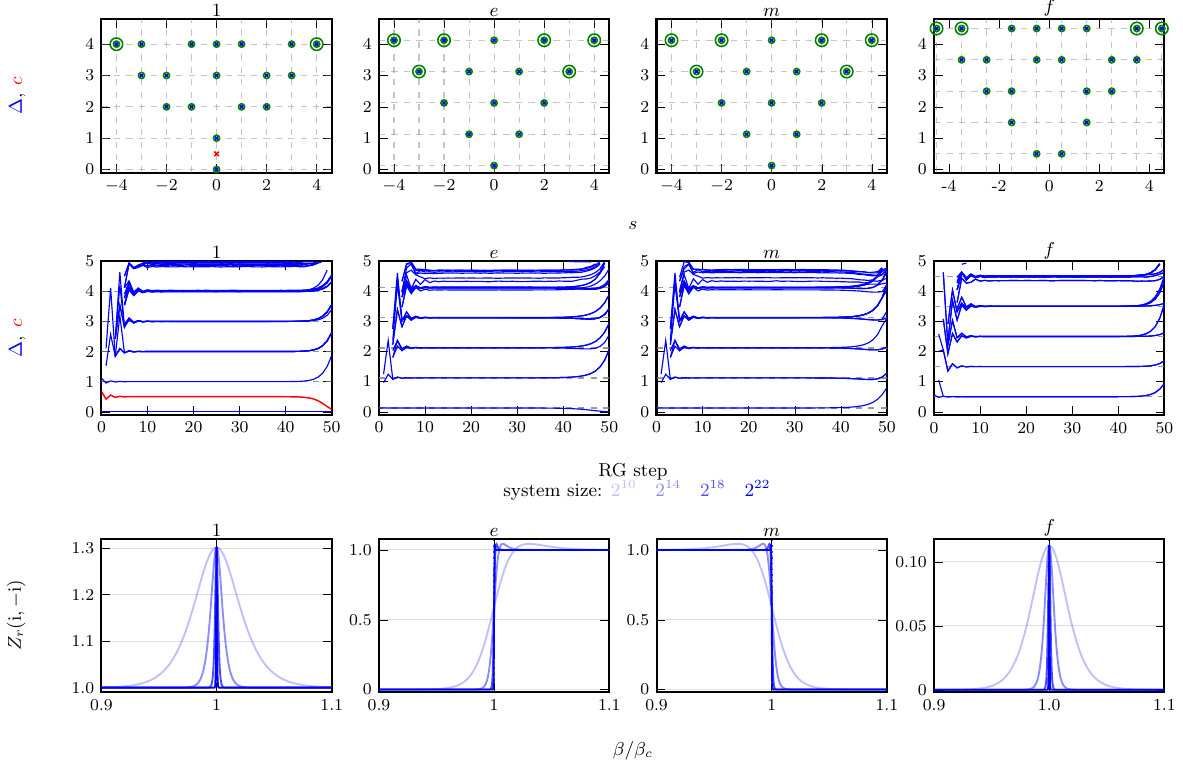}
    \caption{Sector-resolved conformal data of the critical classical Ising
    model with $\Z_2$ symmetry in
    $\Irr\fZ(\Rep(\Z_2))=\{\one,e,m,f=em\}$.
    The top row shows the spin-resolved spectrum at the 20th RG step.
    Its sector-dependent dashed grid marks the exact allowed values of $s$ and
    $\Delta$, and green concentric rings show exact degeneracies. The middle
    row traces the numerical scaling dimensions through the Loop-TNR flow, with
    dashed horizontal lines marking the exact dimensions.
    The third row shows the partition function $Z_x(\ii,-\ii)$ on a square
    across the thermal transition for systems with $2^{10},2^{14},2^{18}$,
    and $2^{22}$ spins. The first two rows use
    $\chi_{\text{RG}}=32$, with $\chi_{\text{transfer}}=40$ for the top row,
    whereas the partition functions in the third row use
    $\chi_{\text{RG}}=16$.}
    \label{fig:ising-statistics}
    \label{fig:ising-phase-diagram}
\end{figure*}

\begin{table*}[tb]
\centering
\caption{Magnitudes $|C^{\closedsector}_{abc}|$ of the closed-string structure
constants of the critical $\Z_2$-symmetric Ising model at
$\chi_{\mathrm{RG}}=32$ and $N_{\mathrm{RG}}=20$, using the $4\times4$
patch with cutoff 32 in both UTD decompositions.
All 56 symmetry-allowed ordered triples are shown.
Mean denotes the average of the two numerical magnitudes related by
$e\leftrightarrow m$.
We use the shorthand $x\bar y$ to denote the primary field in the double Ising
sector $x\boxtimes y$. The subscripts $e$ and $m$ on
$(\sigma\bar\sigma)_{e,m}$ distinguish the two local $A$-module structures.}
\label{tab:structure-constants-statistical-ising}
\begingroup
\fontsize{8.0}{9.6}\selectfont
\setlength{\tabcolsep}{1.2pt}
\renewcommand{\arraystretch}{1.02}
\begin{minipage}[t]{0.486\textwidth}
\vspace{0pt}
\begin{tabular*}{\linewidth}{@{\extracolsep{\fill}}lcrr@{}}
\toprule
$(a,b,c)$ & Theory & Numerical & Abs.\ error \\
\midrule
\multicolumn{4}{c}{$\one\times\one\times\one$} \\
\midrule
$\one\bar\one,\, \one\bar\one,\, \one\bar\one$ & \multirow{1}{*}{$1$} & 1.000000 & 0 \\
\cmidrule(lr){1-4}
$\one\bar\one,\, \one\bar\one,\, \psi\bar\psi$ & \multirow{3}{*}{$0$} & 0.022213 & $2.22\times10^{-2}$ \\
$\one\bar\one,\, \psi\bar\psi,\, \one\bar\one$ &  & 0.044530 & $4.45\times10^{-2}$ \\
$\psi\bar\psi,\, \one\bar\one,\, \one\bar\one$ &  & 0.042181 & $4.22\times10^{-2}$ \\
\cmidrule(lr){1-4}
$\one\bar\one,\, \psi\bar\psi,\, \psi\bar\psi$ & \multirow{3}{*}{$1$} & 1.003632 & $3.63\times10^{-3}$ \\
$\psi\bar\psi,\, \one\bar\one,\, \psi\bar\psi$ &  & 0.994521 & $5.48\times10^{-3}$ \\
$\psi\bar\psi,\, \psi\bar\psi,\, \one\bar\one$ &  & 0.993694 & $6.31\times10^{-3}$ \\
\cmidrule(lr){1-4}
$\psi\bar\psi,\, \psi\bar\psi,\, \psi\bar\psi$ & \multirow{1}{*}{$0$} & 0.021318 & $2.13\times10^{-2}$ \\
\midrule
\multicolumn{4}{c}{$\one\times f\times f$} \\
\midrule
$\one\bar\one,\, \one\bar\psi,\, \one\bar\psi$ & \multirow{3}{*}{$1$} & 1.002066 & $2.07\times10^{-3}$ \\
$\one\bar\psi,\, \one\bar\one,\, \one\bar\psi$ &  & 0.997438 & $2.56\times10^{-3}$ \\
$\one\bar\psi,\, \one\bar\psi,\, \one\bar\one$ &  & 0.997732 & $2.27\times10^{-3}$ \\
\cmidrule(lr){1-4}
$\one\bar\one,\, \one\bar\psi,\, \psi\bar\one$ & \multirow{6}{*}{$0$} & 0.031764 & $3.18\times10^{-2}$ \\
$\one\bar\one,\, \psi\bar\one,\, \one\bar\psi$ &  & 0.031764 & $3.18\times10^{-2}$ \\
$\one\bar\psi,\, \one\bar\one,\, \psi\bar\one$ &  & 0.030837 & $3.08\times10^{-2}$ \\
$\one\bar\psi,\, \psi\bar\one,\, \one\bar\one$ &  & 0.043846 & $4.38\times10^{-2}$ \\
$\psi\bar\one,\, \one\bar\one,\, \one\bar\psi$ &  & 0.030837 & $3.08\times10^{-2}$ \\
$\psi\bar\one,\, \one\bar\psi,\, \one\bar\one$ &  & 0.043846 & $4.38\times10^{-2}$ \\
\cmidrule(lr){1-4}
$\one\bar\one,\, \psi\bar\one,\, \psi\bar\one$ & \multirow{3}{*}{$1$} & 1.002066 & $2.07\times10^{-3}$ \\
$\psi\bar\one,\, \one\bar\one,\, \psi\bar\one$ &  & 0.997438 & $2.56\times10^{-3}$ \\
$\psi\bar\one,\, \psi\bar\one,\, \one\bar\one$ &  & 0.997732 & $2.27\times10^{-3}$ \\
\cmidrule(lr){1-4}
$\one\bar\psi,\, \one\bar\psi,\, \psi\bar\psi$ & \multirow{3}{*}{$0$} & 0.022135 & $2.21\times10^{-2}$ \\
$\one\bar\psi,\, \psi\bar\psi,\, \one\bar\psi$ &  & 0.031097 & $3.11\times10^{-2}$ \\
$\psi\bar\psi,\, \one\bar\psi,\, \one\bar\psi$ &  & 0.030339 & $3.03\times10^{-2}$ \\
\cmidrule(lr){1-4}
$\one\bar\psi,\, \psi\bar\one,\, \psi\bar\psi$ & \multirow{6}{*}{$1$} & 0.999838 & $1.62\times10^{-4}$ \\
$\one\bar\psi,\, \psi\bar\psi,\, \psi\bar\one$ &  & 0.999675 & $3.25\times10^{-4}$ \\
$\psi\bar\one,\, \one\bar\psi,\, \psi\bar\psi$ &  & 0.999838 & $1.62\times10^{-4}$ \\
$\psi\bar\one,\, \psi\bar\psi,\, \one\bar\psi$ &  & 0.999675 & $3.25\times10^{-4}$ \\
$\psi\bar\psi,\, \one\bar\psi,\, \psi\bar\one$ &  & 0.995069 & $4.93\times10^{-3}$ \\
$\psi\bar\psi,\, \psi\bar\one,\, \one\bar\psi$ &  & 0.995069 & $4.93\times10^{-3}$ \\
\cmidrule(lr){1-4}
$\psi\bar\one,\, \psi\bar\one,\, \psi\bar\psi$ & \multirow{3}{*}{$0$} & 0.022135 & $2.21\times10^{-2}$ \\
$\psi\bar\one,\, \psi\bar\psi,\, \psi\bar\one$ &  & 0.031097 & $3.11\times10^{-2}$ \\
$\psi\bar\psi,\, \psi\bar\one,\, \psi\bar\one$ &  & 0.030339 & $3.03\times10^{-2}$ \\
\bottomrule
\end{tabular*}
\end{minipage}\hfill
\begin{minipage}[t]{0.486\textwidth}
\vspace{0pt}
\begin{tabular*}{\linewidth}{@{\extracolsep{\fill}}lcrrr@{}}
\toprule
$(a,b,c)$ & Theory & Numerical & Mean & Abs.\ error \\
\midrule
\multicolumn{5}{c}{$\one\times e\times e\quad/\quad\one\times m\times m$} \\
\midrule
$\one\bar\one,\, (\sigma\bar\sigma)_e,\, (\sigma\bar\sigma)_e$ & \multirow{6}{*}{$1$} & 0.985019 & \multirow{2}{*}{1.000478} & \multirow{2}{*}{$4.78\times10^{-4}$} \\
$\one\bar\one,\, (\sigma\bar\sigma)_m,\, (\sigma\bar\sigma)_m$ & & 1.015938 & & \\
$(\sigma\bar\sigma)_e,\, \one\bar\one,\, (\sigma\bar\sigma)_e$ &  & 0.983856 & \multirow{2}{*}{0.999341} & \multirow{2}{*}{$6.59\times10^{-4}$} \\
$(\sigma\bar\sigma)_m,\, \one\bar\one,\, (\sigma\bar\sigma)_m$ & & 1.014827 & & \\
$(\sigma\bar\sigma)_e,\, (\sigma\bar\sigma)_e,\, \one\bar\one$ &  & 0.977223 & \multirow{2}{*}{0.999448} & \multirow{2}{*}{$5.52\times10^{-4}$} \\
$(\sigma\bar\sigma)_m,\, (\sigma\bar\sigma)_m,\, \one\bar\one$ & & 1.021674 & & \\
\cmidrule(lr){1-5}
$(\sigma\bar\sigma)_e,\, (\sigma\bar\sigma)_e,\, \psi\bar\psi$ & \multirow{6}{*}{$1/2$} & 0.489200 & \multirow{2}{*}{0.500385} & \multirow{2}{*}{$3.85\times10^{-4}$} \\
$(\sigma\bar\sigma)_m,\, (\sigma\bar\sigma)_m,\, \psi\bar\psi$ & & 0.511570 & & \\
$(\sigma\bar\sigma)_e,\, \psi\bar\psi,\, (\sigma\bar\sigma)_e$ &  & 0.461449 & \multirow{2}{*}{0.500627} & \multirow{2}{*}{$6.27\times10^{-4}$} \\
$(\sigma\bar\sigma)_m,\, \psi\bar\psi,\, (\sigma\bar\sigma)_m$ & & 0.539805 & & \\
$\psi\bar\psi,\, (\sigma\bar\sigma)_e,\, (\sigma\bar\sigma)_e$ &  & 0.459041 & \multirow{2}{*}{0.496408} & \multirow{2}{*}{$3.59\times10^{-3}$} \\
$\psi\bar\psi,\, (\sigma\bar\sigma)_m,\, (\sigma\bar\sigma)_m$ & & 0.533775 & & \\
\midrule
\multicolumn{5}{c}{$e\times m\times f$} \\
\midrule
$(\sigma\bar\sigma)_e,\, (\sigma\bar\sigma)_m,\, \one\bar\psi$ & \multirow{6}{*}{$1/\sqrt{2}$} & 0.707216 & \multirow{2}{*}{0.707482} & \multirow{2}{*}{$3.75\times10^{-4}$} \\
$(\sigma\bar\sigma)_m,\, (\sigma\bar\sigma)_e,\, \one\bar\psi$ & & 0.707748 & & \\
$(\sigma\bar\sigma)_e,\, \one\bar\psi,\, (\sigma\bar\sigma)_m$ &  & 0.696211 & \multirow{2}{*}{0.707982} & \multirow{2}{*}{$8.75\times10^{-4}$} \\
$(\sigma\bar\sigma)_m,\, \one\bar\psi,\, (\sigma\bar\sigma)_e$ & & 0.719752 & & \\
$\one\bar\psi,\, (\sigma\bar\sigma)_e,\, (\sigma\bar\sigma)_m$ &  & 0.694469 & \multirow{2}{*}{0.705354} & \multirow{2}{*}{$1.75\times10^{-3}$} \\
$\one\bar\psi,\, (\sigma\bar\sigma)_m,\, (\sigma\bar\sigma)_e$ & & 0.716238 & & \\
\cmidrule(lr){1-5}
$(\sigma\bar\sigma)_e,\, (\sigma\bar\sigma)_m,\, \psi\bar\one$ & \multirow{6}{*}{$1/\sqrt{2}$} & 0.707216 & \multirow{2}{*}{0.707482} & \multirow{2}{*}{$3.75\times10^{-4}$} \\
$(\sigma\bar\sigma)_m,\, (\sigma\bar\sigma)_e,\, \psi\bar\one$ & & 0.707748 & & \\
$(\sigma\bar\sigma)_e,\, \psi\bar\one,\, (\sigma\bar\sigma)_m$ &  & 0.696211 & \multirow{2}{*}{0.707982} & \multirow{2}{*}{$8.75\times10^{-4}$} \\
$(\sigma\bar\sigma)_m,\, \psi\bar\one,\, (\sigma\bar\sigma)_e$ & & 0.719752 & & \\
$\psi\bar\one,\, (\sigma\bar\sigma)_e,\, (\sigma\bar\sigma)_m$ &  & 0.694469 & \multirow{2}{*}{0.705354} & \multirow{2}{*}{$1.75\times10^{-3}$} \\
$\psi\bar\one,\, (\sigma\bar\sigma)_m,\, (\sigma\bar\sigma)_e$ & & 0.716238 & & \\
\bottomrule
\end{tabular*}
\end{minipage}
\endgroup
\end{table*}

\paragraph{Anyon-resolved spectra}
Applying the algorithm described in Subsec.~\ref{sec:r-loop-tnr} to the critical
$\Z_2$-symmetric classical Ising model, with an RG cutoff
$\chi_{\text{RG}}=32$, we obtain the sector-resolved evolution of the scaling
dimensions along the Loop-TNR flow. At the 20th RG step, we compute the
spin-resolved conformal towers in the symTO sectors~\cite{Ji_2019,Ji_2020}
using the transfer morphism of shape
$[\frac{4}{\sqrt{10}},2\sqrt{10},\frac{2}{\sqrt{10}}]$ in the $[h,L,\xi]$
convention of Fig.~\ref{fig:tube-shape-geometry}, with cutoff
$\chi_{\text{transfer}}=40$. See Fig.~\ref{fig:ising-statistics}.

For $\cR=\Rep(G)$ with $G$ finite Abelian, the simple objects of
$\fZ(\cR)$ are labelled by pairs $(g,\chi)\in G\times\widehat G$ and
factorize into charge and flux objects. For non-Abelian $G$, they are
instead labelled by pairs $([g],\pi)$, where $[g]$ is a conjugacy class
and $\pi$ is an irreducible representation of the centralizer $C_G(g)$.
They therefore do not, in general, factorize into independent charge and
flux objects, reflecting the fact that a flux labelled by an individual
$g$ is not $G$-symmetric~\cite{davydov2010modular}.

Since all $F$-symbols of $\Rep(\Z_2)$ can be gauge-fixed to unity, the
bulk Turaev--Viro evolution described in
Sec.~\ref{sec:holographic-picture} is almost trivial. Its remaining
nontrivial content comes from the fusion rules, which take the form of
the nonlocal constraints defining the configuration space
$\scrC^{\sym}$.

Before topological Wick rotation, the condensation from
$\Ising\boxtimes\overline{\Ising}$ to
$\fZ(\Rep(\Z_2))$ is described by the local-module category~\cite{Ji_2020}:
\[
\left(\Ising\boxtimes\overline{\Ising}\right)^{\localmodules}_{A}
\cong \fZ(\Rep(\Z_2)),
\qquad
A=\one\boxtimes\one\oplus\psi\boxtimes\psi.
\]
The simple objects of $\fZ(\Rep(\Z_2))$ are realized as local $A$-modules in
$\Ising\boxtimes\overline{\Ising}$:
\begin{equation}\label{eq:double-ising-toric-code-condensation}
\begin{aligned}
    \one &= \one\boxtimes\one\oplus\psi\boxtimes\psi,\\
    e &= \sigma\boxtimes\sigma,\\
    m &= (\sigma\boxtimes\sigma)_{\mathrm{tw}},\\
    f &= \psi\boxtimes\one\oplus\one\boxtimes\psi.
\end{aligned}
\end{equation}
The label $\mathrm{tw}$ marks the alternative local $A$-module structure, which means
$e$ and $m$ share the underlying object but not the module structure.

After topological Wick rotation, these local $A$-modules are interpreted as
gapless modes on a circle at the Ising critical
point~\cite{KONG2018140,Kong_2020,Kong_2021,Chen_2020},
as reviewed in Appendix~\ref{app:symto-review}.
Let $\scrV^{\mathrm I}_i$ denote the irreducible Ising VOA module labeled by
$i$, and $\scrV_x$ the full CFT Hilbert space in anyon sector $x$.
The sector Hilbert spaces decompose as
\begin{equation}\label{eq:z2-ising-sector-hilbert-spaces}
\begin{aligned}
    \scrV_{\one}
    &= \scrV^{\mathrm I}_{\one}\otimes\overline{\scrV^{\mathrm I}_{\one}}
    \oplus\scrV^{\mathrm I}_{\psi}\otimes\overline{\scrV^{\mathrm I}_{\psi}},\\
    \scrV_e
    &= \scrV^{\mathrm I}_{\sigma}\otimes\overline{\scrV^{\mathrm I}_{\sigma}},\\
    \scrV_m
    &= \left(\scrV^{\mathrm I}_{\sigma}\otimes
       \overline{\scrV^{\mathrm I}_{\sigma}}\right)_{\mathrm{tw}},\\
    \scrV_f
    &= \scrV^{\mathrm I}_{\psi}\otimes\overline{\scrV^{\mathrm I}_{\one}}
    \oplus\scrV^{\mathrm I}_{\one}\otimes\overline{\scrV^{\mathrm I}_{\psi}}.
\end{aligned}
\end{equation}
Here each left-hand side is a full sector Hilbert space and each right-hand
term is an entire chiral--antichiral conformal family of the Ising VOA
$\scrV^{\mathrm I}$.
The spin-resolved spectrum computed using the method of
Sec.~\ref{sec:transfer-morphism-diagonalization} and shown in
Fig.~\ref{fig:ising-statistics} verifies this decomposition into conformal families.

Taking characters of these complete conformal families gives the
sector partition functions~\cite{Ji_2020,Chen_2020,Ji_2019}:
\begin{equation}\label{eq:z2-ising-sector-characters}
\begin{aligned}
    Z_\one &= \chi_0\overline{\chi_0}
    + \chi_{\frac12}\overline{\chi_{\frac12}},\\
    Z_e &= \chi_{\frac1{16}}\overline{\chi_{\frac1{16}}},\\
    Z_m &= \chi_{\frac1{16}}\overline{\chi_{\frac1{16}}},\\
    Z_f &= \chi_0\overline{\chi_{\frac12}}
    + \chi_{\frac12}\overline{\chi_0}.
\end{aligned}
\end{equation}
The vector formed by these sector partition functions satisfies the
modular covariance relations discussed in Sec.~\ref{sec:modular-covariance}.

The conformal tower of the trivial sector also reveals how the
$\Z_2$ symmetry enhances perturbative stability. Excluding the identity,
the only operator in the trivial sector with scaling dimension
$\Delta<2$ is $\phi_\psi\otimes\overline{\phi_\psi}$, so the theory has a single
symmetry-preserving relevant perturbation. In the absence of an imposed
symmetry, the relevant operator $\phi_\sigma\otimes\phi_\sigma$ also belongs
to the single unconstrained sector. Imposing the $\Z_2$ symmetry moves
this conformal family to the $e$ sector, thereby excluding it from
symmetry-preserving perturbations and enhancing the stability of the
critical point~\cite{Chatterjee_2023,
chatterjee2024emergentgeneralizedsymmetrymaximal}.

Varying the inverse temperature $\beta$ across its critical value
$\beta_c$, we compute, with cutoff $\chi_{\text{RG}}=16$, the
sector-resolved torus partition functions
$Z_x(\ii,-\ii)$, with $x\in\{\one,e,m,f\}$, shown in
Fig.~\ref{fig:ising-phase-diagram}. They exhibit the change of condensed
Lagrangian algebra across the phase transition. For $\beta>\beta_c$, the
Ising model is in the $\Z_2$-symmetry-breaking phase, characterized by
$e$ condensation and the Lagrangian algebra $\one\oplus e$. For
$\beta<\beta_c$, it is in the $\Z_2$-symmetric phase, characterized by
$m$ condensation and the Lagrangian algebra
$\one\oplus m$~\cite{Ji_2020,Chen_2020}.

Using the tensor obtained after $N_{\text{RG}}=20$ Loop-TNR steps with
cutoff $\chi_{\text{RG}}=32$, we compute the magnitudes of the closed-string
structure constants for all 56 symmetry-allowed ordered triples using
the $4\times4$ patch of App.~\ref{app:four-by-four-structure-constants},
with cutoff 32 in both UTD decompositions. Table~\ref{tab:structure-constants-statistical-ising}
pairs entries related by exchanging $e$ and $m$
and gives their arithmetic mean. The $\one\times\one\times\one$ and
$\one\times f\times f$ channels are listed separately.

\subsection{Random cluster models with fusion category symmetry}
\label{sec:random-cluster-models}

Let $(Q,m,\iota)$ be a $Q$-system in the charge category $\cR$. We can construct a $\cR$-tensor network from the data $(Q,m,\iota)$ as follows:
\[
\raisebox{-67.43977pt}{\includegraphics{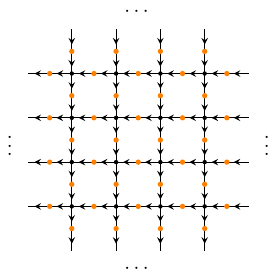}}.\]
The set of all edges of the network is denoted $E$, and each edge $e\in E$ represents the object $Q\in\cR$.
The multiplication map and its dagger are represented graphically as
\[
m =\raisebox{-14.42636pt}{\includegraphics{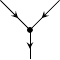}}
,\quad m^\dagger =
\raisebox{-14.42636pt}{\includegraphics{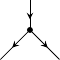}}.
\]
The $Q$-system relations are represented graphically as
\begin{gather*}
\raisebox{-21.53955pt}{\includegraphics{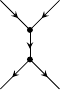}}
=
\raisebox{-21.53955pt}{\includegraphics{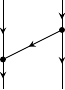}}
=
\raisebox{-21.53955pt}{\includegraphics{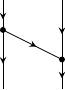}}
\equiv
\raisebox{-14.42636pt}{\includegraphics{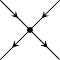}}
,\\
\raisebox{-21.53955pt}{\includegraphics{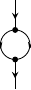}} = \quad
\raisebox{-21.53955pt}{\includegraphics{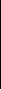}}\quad ,\\
\raisebox{-14.42636pt}{\includegraphics{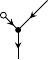}} =
\raisebox{-14.42636pt}{\includegraphics{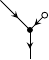}}
=\quad
\raisebox{-14.42636pt}{\includegraphics{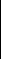}}\quad,\\
\raisebox{-15.84908pt}{\includegraphics{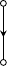}} = d_Q.
\end{gather*}
Each orange tensor is a weighted superposition of two tensors
\[
\raisebox{-14.42636pt}{\includegraphics{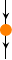}}
 =
 p\xt \raisebox{-14.42636pt}{\includegraphics{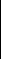}} +
    (1-p)\xt \raisebox{-14.42636pt}{\includegraphics{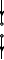}}\ ,
\]
Let $\Omega:E\to\{0,1\}$ denote a bond configuration, where $\Omega(e)=1$ and $0$ correspond to an open and a closed edge, respectively. We take $p\in[0,1]$ to be the bond-occupation parameter, so that the connected and disconnected tensors carry weights $p$ and $1-p$, respectively. The configuration space is therefore $\{0,1\}^{E}$.

Expanding every orange tensor gives a sum over bond configurations
$\Omega\in\{0,1\}^{E}$, weighted by the corresponding bond weights.
By the $Q$-system relations, each connected component evaluates to
$d_Q$. Hence, a configuration with $C(\Omega)$ clusters contributes a
factor $d_Q^{C(\Omega)}$, and the partition function is
\cite{Fortuin_Kasteleyn_1972,Random_Cluster}
\begin{align*}
Z(p)
&=
\sum_{\Omega\in\{0,1\}^{E}}
d_Q^{C(\Omega)}
\prod_{e\in E}p^{\Omega(e)}(1-p)^{1-\Omega(e)}
\\
&=
(1-p)^{|E|}
\sum_{\Omega\in\{0,1\}^{E}}
d_Q^{C(\Omega)}v^{l(\Omega)},
\end{align*}
where $v=p/(1-p)$ and
$l(\Omega)=\sum_{e\in E}\Omega(e)$ is the number of open edges.
This is the random-cluster partition function with cluster weight
$q=d_Q$.

Algebraically, the random-cluster model can be defined for
arbitrary complex $q$ \cite{Sokal_2005}. By contrast, the $Q$-system
construction restricts $q$ to categorical dimensions $d_Q$, which are
algebraic numbers---in fact, positive algebraic integers
\cite{Etingof_Nikshych_Ostrik_2005}. On a planar lattice, the
random-cluster model admits an equivalent Temperley--Lieb loop
formulation with loop weight $\sqrt q$
\cite{Temperley:1971iq,Baxter_1973}.

An important consequence is that the untwisted partition function of
the planar tensor network depends on the $Q$-system $(Q,m,\iota)$ only
through its categorical dimension $d_Q$. Thus, $Q$-systems with the
same value of $d_Q$, even when defined in different charge categories,
give the same random-cluster partition function. Although invisible
to the untwisted partition function, the categorical data
reappear in sector-resolved observables.

One possible four-leg tensor representation of the random-cluster
partition function is
\[
    T =
    \raisebox{-28.85274pt}{\includegraphics{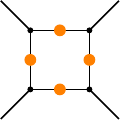}}.
\]

A particularly useful case is the matrix $Q$-system
$Q\simeq W\otimes W^*$, for which $d_Q=d_W^2$, where $W$ is another object in $\cR$. Its normalized
multiplication and unit are
\begin{align*}
m&:W\otimes W^* \otimes W\otimes W^*
\xrightarrow{\id_W \otimes \frac{1}{\sqrt{d_W}}\coev_W^\dagger \otimes \id_{W^*}}
W\otimes W^*,
\\
\iota&:\one \xrightarrow{\sqrt{d_W}\ev_W^\dagger} W\otimes W^*,
\end{align*}
respectively. For this $Q$-system, the random-cluster representation
becomes a loop model with loop weight
$d_W=\sqrt{q}$ \cite{Temperley:1971iq,Baxter_1973}. The four-valent
$Q$-system tensor takes the following corner-double-line form
\cite{Gu_2009}:
\[
\raisebox{-14.42636pt}{\includegraphics{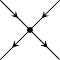}}
\propto
\raisebox{-17.27159pt}{\includegraphics{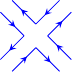}},\quad
\raisebox{-14.42636pt}{\includegraphics{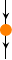}}
= p
\raisebox{-11.58113pt}{\includegraphics{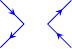}} + (1-p)d_W
\raisebox{-17.27159pt}{\includegraphics{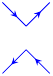}}
\]
Each blue strand is labeled by $W$, and every closed loop evaluates
to $d_W$.

Expanding the orange tensors therefore gives a fully packed
loop model on the medial square lattice:
\[
\raisebox{-60.557pt}{\includegraphics{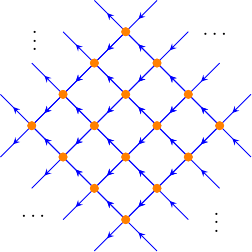}}.
\]
Lattice duality exchanges the two local connectivities. Equalizing
their weights gives the self-dual transition point
\cite{Baxter_1973},
\[
p_c=(1-p_c)d_W,
\qquad
p_c = \frac{d_W}{1+d_W} = \frac{\sqrt{d_Q}}{1+\sqrt{d_Q}}.
\]

The local degree of freedom at each vertex of the medial lattice
selects one of the two noncrossing pairings. At the self-dual point
the two pairings have equal weight, and, after removing an overall
factor, the partition function becomes
\[
Z(p_c)\propto\sum_{\mathcal L}d_W^{N(\mathcal L)},
\]
where the sum runs over fully packed loop configurations $\mathcal L$
and $N(\mathcal L)$ is the number of closed loops.

Through the Temperley--Lieb construction, the square-lattice loop
model is mapped to a staggered six-vertex model
\cite{Temperley:1971iq,Baxter_1973}. Baxter's exact solution at the
self-dual point shows that the transition is continuous for
$d_W\leq2$, equivalently $d_Q\leq4$, and first order for $d_W>2$,
equivalently $d_Q>4$ \cite{Baxter_1973}. Since the loop/vertex
formulation depends on $q$ only through the loop weight
$d_W=\sqrt{q}$, this result applies equally to the noninteger
categorical values $q=d_Q$.

On a pair of adjacent strands ordered as $W\otimes W^*$, define the local
Temperley--Lieb generator by
\[
e_i:=\id^{\otimes(i-1)}\otimes
\bigl(\ev_W^\dagger\circ\ev_W\bigr)\otimes
\id^{\otimes(N-i-1)},
\]
using $\coev_W\circ\coev_W^\dagger$ for the opposite ordering
$W^*\otimes W$. These generators obey
\begin{align*}
e_i^2&=d_W e_i,
&e_i e_{i\pm1}e_i&=e_i,\\
[e_i,e_j]&=0,
&|i-j|&>1.
\end{align*}
The one-dimensional Temperley--Lieb Hamiltonian arises from the
anisotropic limit of the row-to-row transfer matrix of the loop model
\cite{Temperley:1971iq,Baxter_1973}. Up to an overall scale and an
additive constant, it takes the form
\[
H_{\mathrm{TL}}=-\sum_i e_i.
\]
For $d_W\leq2$, where the transition is continuous, the quantum chain
and the isotropic loop model therefore have the same continuum CFT.

For integers $k\geq2$, consider the antiferromagnetic
Temperley--Lieb Hamiltonian in its $A_{k+1}$ RSOS representation,
with loop weight
\[
d_W = [2]_k = \frac{\sin\left(\frac{2\pi}{k+2}\right)}{\sin\left(\frac{\pi}{k+2}\right)} = 2 \cos\left(\frac{\pi}{k+2}\right).
\]
Here the notation $[m]_k$ is defined for $1\leq m\leq k+1$ by
\[
[m]_k=
\frac{\sin\left(\frac{m\pi}{k+2}\right)}
     {\sin\left(\frac{\pi}{k+2}\right)}.
\]
This is the quantum dimension of the spin-$(m-1)/2$ simple object of
$\mathrm{SU}(2)_k$. In this case, the Temperley--Lieb Hamiltonian is integrable, and its continuum limit is the unitary
Virasoro minimal model $\cM(k+2,k+1)$
\cite{Andrews_Baxter_Forrester_1984,Pasquier_1987}.
The first three cases are listed below.
\begin{itemize}
    \item For $k=2$, $d_W=[2]_2=\sqrt{2}$. Choosing
    $W=\sigma\in\Ising$ gives
    \[
    Q=W\otimes W^*=\sigma\otimes\sigma=\one\oplus\psi,
    \]
    and the continuum theory is the critical Ising CFT
    $\cM(4,3)$.

    \item For $k=3$, $d_W=[2]_3=\varphi$, where
    $\varphi=(1+\sqrt{5})/2$. Choosing $W=\tau\in\Fib$ gives
    \[
    Q=W\otimes W^*=\one\oplus\tau,
    \]
    reproducing the golden-chain construction
    \cite{Feiguin_2007}, whose continuum theory is the tricritical
    Ising CFT $\cM(5,4)$.

    \item For $k=4$, $d_W=[2]_4=\sqrt{3}$. Choosing the simple object
    labeled $1$ (or $3$) in $\mathrm{SU}(2)_4$ or
    $\JonesKauffman_4$ gives
    \begin{equation}
    \label{eq:k4-q-system}
    Q=1\otimes1=0\oplus2,
    \end{equation}
    and the continuum theory is the three-state Potts CFT
    $\cM(6,5)$.
\end{itemize}

The symmetry actually realized by the tensor network is the fusion
subcategory generated by the simple summands of $Q$. For $k=2$, this
is the $\Rep(\Z_2)$ subcategory of $\Ising$ generated by $\psi$. For $k=4$, it is
the even subcategory $\{0,2,4\}\cong\Rep(S_3)$ of
$\mathrm{SU}(2)_4$ or $\JonesKauffman_4$. By contrast, for $k=3$ the
$Q$-system $Q=\one\oplus\tau$ generates the full Fibonacci category and
has noninteger dimension $d_Q=\varphi^2$. Since every object of
$\Rep(G)$ has integer categorical dimension, this case cannot be
reduced to an ordinary finite-group symmetry and is genuinely
categorical. We therefore treat it first.

\subsubsection{Fibonacci-symmetric random-cluster model}
\begin{figure*}
    \centering
    \includegraphics[width=\linewidth]{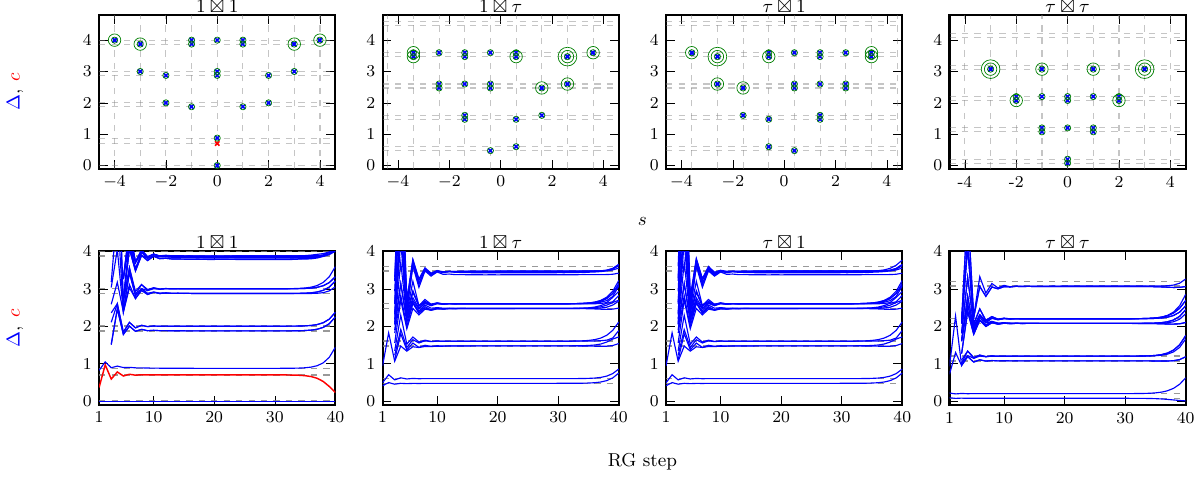}
    \caption{Sector-resolved conformal data of the Fibonacci-symmetric
    random-cluster model at the self-dual point.
    The upper row shows the spin-resolved
    conformal spectrum obtained from the $\frac{4}{\sqrt{10}}\xt
    2\sqrt{10}$ transfer morphism at $N_{\text{RG}}=15$, with
    $\chi_{\text{transfer}}=72$. The lower row shows the RG flow of
    the scaling dimensions as continuous level trajectories, obtained from
    the $\sqrt{2}\xt2\sqrt{2}$ transfer morphism with
    $\chi_{\text{RG}}=36$.}
    \label{fig:fib-self-dual}
\end{figure*}

Taking $\cR=\Fib$ and $W=\tau$, with the conventions of
App.~\ref{appsec:fibonacci-category}, the matrix $Q$-system
\[
Q=W\otimes W^*=\one\oplus\tau
\]
defines a Fibonacci-symmetric random-cluster tensor network on the
square lattice. At the self-dual point
\[
p_c=\frac{\varphi}{1+\varphi},
\]
the $\cR$-Loop-TNR flow converges to the tricritical Ising CFT. See
Fig.~\ref{fig:fib-self-dual}.

The continuum theory is the minimal model $\cM(5,4)$, with central
charge $c=7/10$ and six Virasoro primary sectors. The categorical data
of its rank-$6$ UMTC, including its factorization and the correspondence
between primaries and simple objects, are reviewed in
App.~\ref{appsec:tricritical-ising-category} and summarized in
Tab.~\ref{tab:triising-mtc}~\cite{Wen_2015}.

The categorical origin of the $\Fib$ symmetry follows from the UMTC
factorization
\[
\begin{aligned}
\fZ\left(\mathsf{6}_{7/10}\right)
&\cong \mathsf{6}_{7/10}\boxtimes\overline{\mathsf{6}_{7/10}}\\
&\cong
\overline{\Fib}\boxtimes\Fib
\boxtimes\mathsf{3}_{7/2}\boxtimes\overline{\mathsf{3}_{7/2}}.
\end{aligned}
\]
In the last two factors, let
\[
L_{\diagsector}
=\one\boxtimes\one\oplus\psi\boxtimes\psi\oplus\sigma\boxtimes\sigma
\in\mathsf{3}_{7/2}\boxtimes\overline{\mathsf{3}_{7/2}}
\]
be the diagonal Lagrangian algebra
and define the condensable algebra
\[
A=\one\boxtimes\one\boxtimes L_{\diagsector}
\in\fZ\left(\mathsf{6}_{7/10}\right).
\]
Condensing $A$ removes the doubled Ising-type factor and gives \cite{Kong_2020,Kong_2021,
chatterjee2024emergentgeneralizedsymmetrymaximal}
\[
\fZ\left(\mathsf{6}_{7/10}\right)_{A}^{\localmodules}
\cong\overline{\Fib}\boxtimes\Fib
\cong\fZ(\Fib).
\]
Before the topological Wick rotation reviewed in
Sec.~\ref{appsec:enriched-category-description}, this equation describes an
anyon-condensation domain wall from
$\fZ(\mathsf{6}_{7/10})$ to $\fZ(\Fib)$.

After topological Wick rotation, the four simple objects of
$\fZ(\Fib)\cong\Fib\boxtimes\overline{\Fib}$ label four CFT Hilbert
spaces.  Let $\scrV^{\mathrm{TCI}}$ denote the tricritical Ising VOA and
$\scrV^{\mathrm{TCI}}_i$ its irreducible module labeled by the primary
$i$.  Using the primary notation of Tab.~\ref{tab:triising-mtc} and
the symTO interpretation reviewed in Appendix~\ref{app:symto-review},
the sector decomposition underlying Fig.~\ref{fig:fib-self-dual} is
\begin{align*}
\scrV_{\one\boxtimes\one}
&=\scrA
=\scrV^{\mathrm{TCI}}_{\one}\otimes\overline{\scrV^{\mathrm{TCI}}_{\one}}
 \oplus\scrV^{\mathrm{TCI}}_{\psi'}\otimes\overline{\scrV^{\mathrm{TCI}}_{\psi'}}\\
&\quad\oplus\scrV^{\mathrm{TCI}}_{\sigma'''}\otimes\overline{\scrV^{\mathrm{TCI}}_{\sigma'''}},\\
\scrV_{\one\boxtimes\tau}
&=\scrV^{\mathrm{TCI}}_{\sigma'}\otimes\overline{\scrV^{\mathrm{TCI}}_{\one}}
 \oplus\scrV^{\mathrm{TCI}}_{\epsilon}\otimes\overline{\scrV^{\mathrm{TCI}}_{\psi'}}\\
&\quad\oplus\scrV^{\mathrm{TCI}}_{\sigma''}\otimes\overline{\scrV^{\mathrm{TCI}}_{\sigma'''}},\\
\scrV_{\tau\boxtimes\one}
&=\scrV^{\mathrm{TCI}}_{\one}\otimes\overline{\scrV^{\mathrm{TCI}}_{\sigma'}}
 \oplus\scrV^{\mathrm{TCI}}_{\psi'}\otimes\overline{\scrV^{\mathrm{TCI}}_{\epsilon}}\\
&\quad\oplus\scrV^{\mathrm{TCI}}_{\sigma'''}\otimes\overline{\scrV^{\mathrm{TCI}}_{\sigma''}},\\
\scrV_{\tau\boxtimes\tau}
&=\scrV^{\mathrm{TCI}}_{\sigma'}\otimes\overline{\scrV^{\mathrm{TCI}}_{\sigma'}}
 \oplus\scrV^{\mathrm{TCI}}_{\epsilon}\otimes\overline{\scrV^{\mathrm{TCI}}_{\epsilon}}\\
&\quad\oplus\scrV^{\mathrm{TCI}}_{\sigma''}\otimes\overline{\scrV^{\mathrm{TCI}}_{\sigma''}}.
\end{align*}

In the enriched-category notation of
Refs.~\cite{Kong_2021,KONG2018140,Kong_2020}, this gapless phase is
\[
\left(\scrV^{\mathrm{TCI}}\otimes\overline{\scrV^{\mathrm{TCI}}},
{}^{\fZ(\mathsf{3}_{7/2})}\!\mathsf{3}_{7/2}
\boxtimes{}^{\fZ(\Fib)}\!\fZ(\Fib)\right).
\]

\subsubsection{$\Z_3$-symmetric random-cluster model}
\begin{figure*}[htbp]
    \centering
    \includegraphics[width=0.7\linewidth]{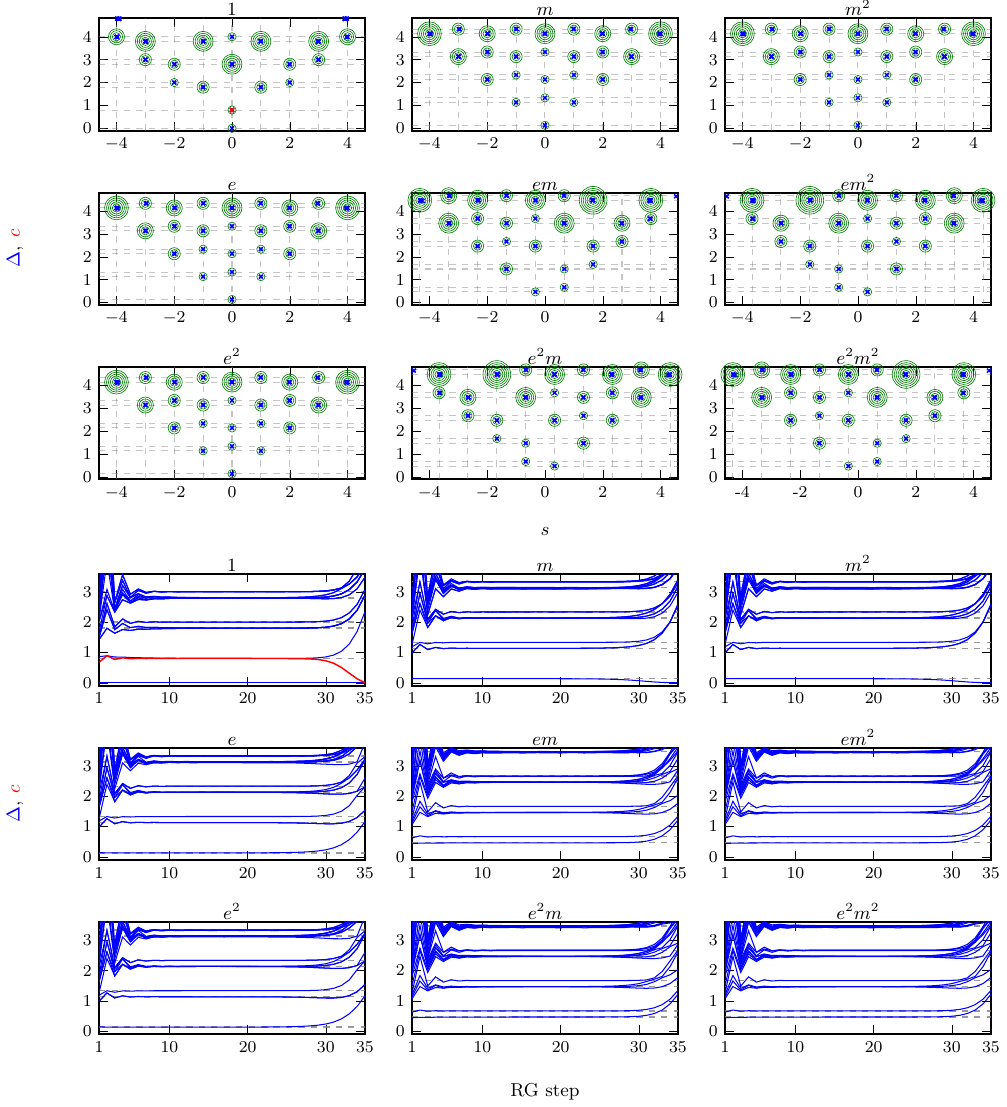}
    \caption{Sector-resolved conformal data of the $\Z_3$-symmetric
    random-cluster model at the self-dual point.
    The upper block shows the spin-resolved
    conformal spectra obtained from the
    $\frac{4}{\sqrt{10}}\xt2\sqrt{10}$ transfer morphism at
    $N_{\text{RG}}=15$, with $\chi_{\text{transfer}}=72$.
    The lower block shows the RG flow of the scaling dimensions as
    continuous level trajectories, obtained from the
    $\sqrt{2}\xt2\sqrt{2}$ transfer morphism with $\chi_{\text{RG}}=36$.}
    \label{fig:z3-potts}
\end{figure*}
As reviewed in the $k=4$ example above,
Eq.~\eqref{eq:k4-q-system}, the simple objects of
$\mathrm{SU}(2)_4$ or $\JonesKauffman_4$ are labeled by
$0,1,2,3,4$. Choosing $W=1$ gives the $Q$-system
\[
    Q=W\otimes W^*=0\oplus2.
\]
The simple summands of $Q$ generate the even fusion subcategory
\[
    \{0,2,4\}\cong\Rep(S_3).
\]
We restrict this charge category along the subgroup
$\Z_3\subset S_3$. Under the restriction functor, $Q$ becomes
\[
    Q\simeq\Fun(\Z_3)
    \simeq\one\oplus\chi\oplus\chi^2\in\Rep(\Z_3),
\]
where $\chi$ and $\chi^2$ are the two nontrivial one-dimensional
representations of $\Z_3$.

Since $d_Q=3$, the resulting categorical
random-cluster model is the standard three-state random-cluster model.
Its self-dual point is
\[
p_c=\frac{\sqrt{3}}{1+\sqrt{3}},
\]
and its continuum limit is the critical three-state Potts CFT. We do
not analyze here the full $S_3$-resolved sector structure, which is
governed by $\fZ(\Rep(S_3))$
\cite{chatterjee2024emergentgeneralizedsymmetrymaximal}.

The critical three-state Potts CFT has central charge $c=4/5$.
Its chiral algebra $\cW_3$ extends the Virasoro algebra of the
minimal model $\cM(6,5)$ by a primary field of conformal weight $3$.
Let $\scrV^{\mathrm P}_h$ denote the irreducible Virasoro module
of conformal weight $h$ in this minimal model, where the
superscript $\mathrm P$ refers to the Potts theory.
The extended vacuum module is
\[
    \scrW^{\mathrm P}_{\one}
    =\scrV^{\mathrm P}_0\oplus\scrV^{\mathrm P}_3.
\]

The extended algebra has six irreducible modules
$\scrW^{\mathrm P}_i$, labeled by
\[
    i\in\{\one,\epsilon,\psi_1,\psi_2,\sigma_1,\sigma_2\}.
\]
Their conformal weights and decompositions into Virasoro modules
are listed in Tab.~\ref{tab:potts-mtc}.
These modules form the UMTC denoted by $\mathsf{6}_{4/5}$.
Its factorization and the correspondence between the primary
labels and simple objects are given in
App.~\ref{appsec:potts-categories}.
The chiral extension and the subsequent condensation producing
the full Potts modular invariant are reviewed in
Sec.~\ref{appsec:partition-vector-and-modular-covariance}.

Since the microscopic charge category is $\Rep(\Z_3)$, the
corresponding symTO is $\fZ(\Rep(\Z_3))$. In the gapless-boundary
construction reviewed in Appendix~\ref{app:symto-review}, we realize
this symTO by a condensable algebra
\[
    A\in\fZ(\mathsf{6}_{4/5})
\]
such that
\[
    \fZ(\mathsf{6}_{4/5})_{A}^{\localmodules}
    \cong\fZ(\Rep(\Z_3))
\]
\cite{Kong_2020,Kong_2021,Ji_2020,
chatterjee2024emergentgeneralizedsymmetrymaximal}.

Let $\mathsf{3}_2$ denote the pointed UMTC with simple objects
$a\in\Z_3$, fusion rule $a\otimes b=a+b\pmod 3$, and twists
$\theta_a=\ee^{2\pi\ii a^2/3}$.
Its braided reverse is denoted by $\overline{\mathsf{3}_2}$.
Forgetting the braiding gives the equivalence of fusion categories
\[
    \overline{\mathsf{3}_2}\cong\Rep(\Z_3).
\]
See App.~\ref{appsec:potts-categories}. Their Drinfeld centers are therefore
braided equivalent. We fix the following correspondence between simple objects:
\begin{equation}
\label{eq:z3-iso}
\begin{aligned}
    \fZ\left(\overline{\mathsf{3}_2}\right)
    &\cong
    \overline{\mathsf{3}_2}\boxtimes\mathsf{3}_2
    \cong\fZ(\Rep(\Z_3)),\\
    i\boxtimes j&\longmapsto e^{i+j}m^{i-j},
    \qquad i,j\in\Z_3.
\end{aligned}
\end{equation}
Here $e$ and $m$ denote the electric and magnetic generators,
respectively, and all exponents are understood modulo $3$. We use the
convention
\[
    \Irr\fZ(\Rep(\Z_3))
    =\{e^a m^b\mid a,b\in\Z_3\}.
\]

Using
\[
    \mathsf{6}_{4/5}
    \cong\Fib\boxtimes\overline{\mathsf{3}_2},
\]
we order the factors of its Drinfeld center as
\[
    \fZ(\mathsf{6}_{4/5})
    \cong
    \Fib\boxtimes\overline{\Fib}
    \boxtimes\overline{\mathsf{3}_2}\boxtimes\mathsf{3}_2.
\]
The diagonal Lagrangian algebra in the doubled Fibonacci factor is
\[
    L_{\diagsector}
    =\one\boxtimes\one\oplus\tau\boxtimes\tau
    \in\Fib\boxtimes\overline{\Fib}.
\]
The required condensable algebra is therefore
\[
    A=L_{\diagsector}\boxtimes0\boxtimes0
    \in\fZ(\mathsf{6}_{4/5}).
\]
Condensing $A$ removes the doubled Fibonacci factor and leaves
\[
    \fZ(\mathsf{6}_{4/5})_A^{\localmodules}
    \cong
    \overline{\mathsf{3}_2}\boxtimes\mathsf{3}_2
    \cong\fZ(\Rep(\Z_3)),
\]
where the last isomorphism is the one fixed in
Eq.~\eqref{eq:z3-iso}.

Before the topological Wick rotation reviewed in
Sec.~\ref{appsec:enriched-category-description}, the condensable algebra $A$
defines an anyon-condensation from
$\fZ(\mathsf{6}_{4/5})$ to $\fZ(\Rep(\Z_3))$. In the doubled
Fibonacci factor, condensation of $L_{\diagsector}$ gives the
canonical gapped boundary of $\Fib\boxtimes\overline{\Fib}$. When the
doubled theory is unfolded, this boundary becomes the transparent identity domain wall
of the Fibonacci topological order. Equivalently, it is the folding
line. The factor
\[
    \overline{\mathsf{3}_2}\boxtimes\mathsf{3}_2
\]
passes through the wall unchanged.

After topological Wick rotation, each simple object
$x\in\Irr\fZ(\Rep(\Z_3))$ labels a CFT Hilbert space $\scrV_x$, as
reviewed in Appendix~\ref{app:symto-review}. Using the $\cW_3$ primary
notation of Tab.~\ref{tab:potts-mtc}, the sector decomposition
underlying Fig.~\ref{fig:z3-potts} is
\begin{align*}
    \scrV_{\one} &= \scrW^{\mathrm P}_{\one} \otimes \overline{\scrW^{\mathrm P}_\one} \oplus \scrW^{\mathrm P}_\epsilon \otimes \overline{\scrW^{\mathrm P}_\epsilon},\\
    \scrV_{e} &= \scrW^{\mathrm P}_{\psi_2}\otimes\overline{\scrW^{\mathrm P}_{\psi_2}} \oplus \scrW^{\mathrm P}_{\sigma_2}\otimes \overline{\scrW^{\mathrm P}_{\sigma_2}},\\
    \scrV_{e^2} &= \scrW^{\mathrm P}_{\psi_1}\otimes\overline{\scrW^{\mathrm P}_{\psi_1}} \oplus \scrW^{\mathrm P}_{\sigma_1}\otimes \overline{\scrW^{\mathrm P}_{\sigma_1}},\\
    \scrV_{m} &= \scrW^{\mathrm P}_{\psi_2}\otimes \overline{\scrW^{\mathrm P}_{\psi_1}}\oplus \scrW^{\mathrm P}_{\sigma_2}\otimes \overline{\scrW^{\mathrm P}_{\sigma_1}},\\
    \scrV_{em} & = \scrW^{\mathrm P}_{\psi_1} \otimes \overline{\scrW^{\mathrm P}_\one}\oplus \scrW^{\mathrm P}_{\sigma_1} \otimes \overline{\scrW^{\mathrm P}_{\epsilon}},\\
    \scrV_{e^2m} & = \scrW^{\mathrm P}_{\one} \otimes \overline{\scrW^{\mathrm P}_{\psi_2}}\oplus \scrW^{\mathrm P}_\epsilon \otimes \overline{\scrW^{\mathrm P}_{\sigma_2}},\\
    \scrV_{m^2} &= \scrW^{\mathrm P}_{\psi_1}\otimes \overline{\scrW^{\mathrm P}_{\psi_2}}\oplus \scrW^{\mathrm P}_{\sigma_1}\otimes \overline{\scrW^{\mathrm P}_{\sigma_2}},\\
    \scrV_{em^2} &= \scrW^{\mathrm P}_\one \otimes \overline{\scrW^{\mathrm P}_{\psi_1}} \oplus \scrW^{\mathrm P}_\epsilon \otimes \overline{\scrW^{\mathrm P}_{\sigma_1}},\\
    \scrV_{e^2m^2} &= \scrW^{\mathrm P}_{\psi_2}\otimes \overline{\scrW^{\mathrm P}_{\one}} \oplus \scrW^{\mathrm P}_{\sigma_2}\otimes \overline{\scrW^{\mathrm P}_{\epsilon}}.
\end{align*}
On the left-hand side, the subscript
$x\in\Irr\fZ(\Rep(\Z_3))$ labels the full sector Hilbert space. On the
right-hand side, $\scrW^{\mathrm P}_i$ denotes the irreducible chiral
$\cW_3$ module labeled by the primary $i$, and
$\overline{\scrW^{\mathrm P}_i}$ denotes
its anti-chiral counterpart. Under the equivalence
\[
    \fZ(\mathsf{6}_{4/5})_A^{\localmodules}
    \cong\fZ(\Rep(\Z_3)),
\]
the Hilbert space $\scrV_x$ is the local $A$-module associated with
the sector $x$.

Therefore, in the theory by \cite{Kong_2021,Kong_2020}, taking the local quantum symmetry as the doubled $\cM(6, 5)$-VOA $\scrV^{\mathrm P}_0 \otimes \overline{\scrV^{\mathrm P}_0}$,
the pair describing the gapless boundary is
\[
\left(\scrV^{\mathrm P}_{0} \otimes \overline{\scrV^{\mathrm P}_0}, {}^{\fZ(\Fib)}\!\Fib \boxtimes {}^{\JonesKauffman_4}\!\left(\JonesKauffman_4\right)_{0 \oplus 4} \boxtimes {}^{\overline{\JonesKauffman_4}}\!\left(\overline{\JonesKauffman_4}\right)_{0 \oplus 4} \right)
\]

The Loop-TNR result of the spin-resolved conformal spectrum is plotted in Fig.~\ref{fig:z3-potts}. For the Fibonacci-symmetric tricritical Ising and $\Z_3$-symmetric Potts
models, some numerically extracted structure constants remain appreciably
nonzero even when their CFT values vanish. This motivates studying
realizations with stronger categorical constraints, as considered in the
following section.

\section{Stable symTO and topological bootstrap}
\label{sec:stable-symto}

In Sec.~\ref{sec:z2-statistical-ising}, we studied the
$\Z_2$-symmetric Ising model with charge category
$\cR=\Rep(\Z_2)$. Its trivial topological sector contains the
$\Z_2$-even thermal operator $\phi_\psi\boxtimes \overline{\phi_\psi}$ and scaling dimension $\Delta=1$.
Since this relevant
perturbation preserves the $\Z_2$ symmetry, its coupling must be tuned
to zero to reach the Ising fixed point.

As reviewed in Sec.~\ref{appsec:ffa-extensions-and-emergent-symto}, the possible
symmetry-preserving local perturbations are encoded by the gapless modes
of the trivial sector of the associated symTO
$\fZ(\cR)$, so the relevant operators in this sector determine the amount of
fine-tuning required. Therefore, we define a stable gapless edge and a
stable symTO below, restating
Definition~\ref{def:stable-full-field-algebra}.

\begin{definition}[Stable symTO]
\label{def:stable-symto}
Let $\cT$ be a \spd{1} CFT realized as a boundary theory of the symTO
$\fZ(\cR)$. We call $\fZ(\cR)$ a \emph{stable symTO for $\cT$} if, apart
from the identity operator, the trivial bulk-anyon sector of $\cT$ contains
no relevant local operator. Equivalently, it contains no nonidentity
local operator with scaling dimension $\Delta<2$.
\end{definition}

Brehm and Runkel's proposed stability condition concerns spinless fields
because their circular hole boundary state is rotationally
invariant~\cite{BrehmRunkel2024}. In our Loop-TNR setting, spatial rotation
symmetry is not explicitly enforced, so symmetry-preserving perturbations
need not be spinless. We therefore impose the absence of relevant operators
in the trivial bulk-anyon sector without restricting their conformal spin.

Marginal perturbations are not excluded by this definition and may produce
slow drift near a fixed point. And their eventual RG behavior requires separate
analysis. For the examples studied here, the only dimension-two operators
in the trivial sector are $T$ and $\overline T$, which deform the background
geometry without opening a gap.

The stability criterion above motivates the \emph{topological
bootstrap} \cite{bao2025tensorcomplextopologicalbootstrap,Cheng:2023kxh}.
\begin{question}[Topological bootstrap]
Given a charge category $\cR$, which rational CFTs can arise as stable
IR fixed points of an $\cR$-symmetric system without fine-tuning?
\end{question}
For a unitary minimal model, one natural choice is to take $\cR$ to be the
underlying fusion category of its chiral UMTC. However, this choice need
not be unique, since a given CFT may admit more than one stable symTO.

We now describe a class of solutions to the topological bootstrap
problem realized by $\cR$-TNR. The construction has the following three stages:
specifying the categorical and geometric input, performing the
symmetry-preserving RG flow, and extracting CFT data from the resulting
fixed point.

\paragraph{Input}
For the class of solutions considered here, we choose a charge category
$\cR$ obtained from a UMTC by forgetting its braiding structure. The
categorical input consists of the simple objects, fusion rules,
$F$-symbols, and quantum dimensions of $\cR$, together with the
topological spins inherited from the chosen UMTC. Constructing the
initial $\cR$-tensor also requires a triangulation, a choice of
local-coordinate geometry, and integer lifts of the topological spins.

Let
\[
    V = \bigoplus_{i\in\Irr(\cR)} i.
\]
We construct a rank-3 $\cR$-tensor
\[
    T\in\cR(V\otimes V,V),
\]
defined by
\begin{multline}\label{eq:initial-tensor-bootstrap}
    T :=
    \raisebox{-18.8943pt}{\includegraphics{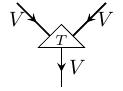}} = \sum_{i,j,k\in\Irr(\cR)}\sum_{\alpha=1}^{N^{ij}_{\bar{k}}}
        \cN_{ijk;\alpha}\\
        \;\xt\;
        G_{\triangle}(h_i,\;h_j,\;h_k)\;\xt\;
        \raisebox{-20.62036pt}{\includegraphics{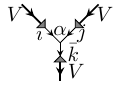}},
\end{multline}
where the ingredients obey the following conditions.
\begin{itemize}
    \item $\cN_{ijk;\alpha}$ are continuous trial parameters assigned to
        admissible triples $(i,j,k)$ and are invariant under $S_3$
        permutations of the three slots;
    \item $G_{\triangle}(h_i,h_j,h_k)$ is the geometric factor
        determined by the chosen local-coordinate geometry and chiral
        weights. With the topological spin of $i$ given by
        \[
            \theta_i=\ee^{2\pi\ii s_i},
        \]
        we choose an integer lift
        \[
            h_i=s_i+n_i,\qquad n_i\in\Z.
        \]
        The integers $n_i$ are additional discrete input because the
        topological spins determine the chiral weights only modulo~1.
\end{itemize}
The tensors $T$ with other orientations are defined analogously.

\begin{figure}[tb]
    \centering
    \raisebox{-40.03366pt}{\includegraphics{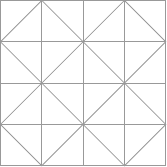}}
    \qquad
    \raisebox{-27.52252pt}{\includegraphics{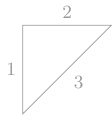}}
    \par\medskip
    \includegraphics{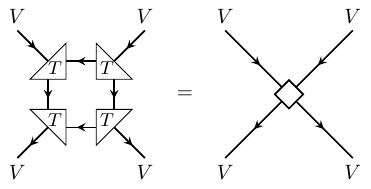}
    \caption{Top: triangulation of the plane by isosceles right triangles, with the
        single building-block triangle shown on the right. Bottom: contraction of
        four copies of the rank-3 tensor $T$ into the initial rank-4 tensor,
        leaving four external legs.}
    \label{fig:triangulation}
\end{figure}

For the fixed triangulation in Fig.~\ref{fig:triangulation}, different
choices of local-coordinate geometry give different geometric factors
$G_{\triangle}$ and hence different microscopic tensors. We use the
nut and pants gauges derived in App.~\ref{app:geometric-factors}, with
geometric factors given in Eqs.~\eqref{eq:nut-gauge-app} and
\eqref{eq:flat-pants-gauge-app}, respectively. Both choices flow to the same
IR fixed point in the examples below and therefore lie in the same basin
of attraction. All numerical results presented in this section use the
nut geometry.

The only continuous freedom in this initial ansatz is the set of trial
parameters $\cN_{ijk;\alpha}$. Their number is much smaller than the number
of degrees of freedom in a generic finite-bond-dimension $\cR$-tensor. The topological
bootstrap therefore starts from a low-dimensional, but not fine-tuned,
family of symmetric tensors.

When the target CFT data are known, one may set $\cN_{ijk;\alpha}$ equal to
the chiral three-point coefficients in the same fusion basis and normalization. As explained in
Sec.~\ref{sec:two-two-move}, this choice gives the primary-state restriction
of the CFT-input tensor and therefore provides a natural initialization near
the fixed point. This matching is
not required for the bootstrap. In our numerical experiments, trial
parameters that differ substantially from these coefficients can still lie
in the same basin of attraction and flow to the same infrared fixed point.

\paragraph{RG flow}
We first contract four copies of the rank-3 tensor $T$ around a square,
leaving four external legs, to obtain the initial rank-4 tensor shown in
Fig.~\ref{fig:triangulation}.
We coarse-grain this rank-4 tensor using the $\cR$-Loop-TNR algorithm
of Subsec.~\ref{sec:r-loop-tnr}, with a finite cutoff $\chi_{\mathrm{RG}}$.

\paragraph{Output}
From the resulting fixed-point tensor, we extract the following data.
\begin{itemize}
    \item the sector-resolved conformal spectra, including scaling
        dimensions and conformal spins in each
        $x\in\Irr\fZ(\cR)$ sector;
    \item the fusion-space-resolved closed-string structure constants
        $C^{\closedsector}_{abc;\alpha}$ of the doubled CFT.
\end{itemize}

When the associated symTO $\fZ(\cR)$ is stable for the target CFT in the
sense of Definition~\ref{def:stable-symto}, we
find in the examples below that a finite region of the trial-parameter
space flows to the same IR fixed point without fine-tuning. Increasing
$\chi_{\mathrm{RG}}$ improves the numerical accuracy while leaving the
IR universality class unchanged. In this sense, the
construction is a \emph{bootstrap}. Finite categorical and geometric
input, together with a low-dimensional family of initial tensors, is
sufficient for the RG flow to recover the fixed-point CFT data.

The full CFT data are not part of the input. Although the integer
lifts $h_i$ entering the geometric factors are specified, neither the
complete sector-resolved conformal spectrum, including descendant
towers and degeneracies, nor the closed-string structure constants are
imposed. These data emerge from the fixed point reached by $\cR$-TNR.

The categorical notation and data used in the examples below are
collected in App.~\ref{app:categorical-data}.

In the remainder of this section, we illustrate the topological bootstrap
with three concrete examples: the Ising CFT with stable symTO
$\fZ(\Ising)$, the tricritical Ising CFT with stable symTO
$\fZ(\overline{\Fib}\boxtimes\mathsf{3}_{7/2})$, and the three-state Potts
model as $\cW_3$-algebra modules with stable symTO
$\fZ(\Fib\boxtimes\overline{\mathsf{3}_{2}})$, corresponding respectively to
the charge categories
$\cR=\Ising$, $\cR=\overline{\Fib}\boxtimes\mathsf{3}_{7/2}$, and
$\cR=\Fib\boxtimes\overline{\mathsf{3}_{2}}$.
The Potts model example demonstrates that the bootstrap also works for full field algebras extended from unitary minimal models.

\subsection{Ising CFT with a stable symTO}
\label{sec:stable-ising}

\begin{figure*}[tb]
    \centering
    \includegraphics[width=0.6\textwidth]{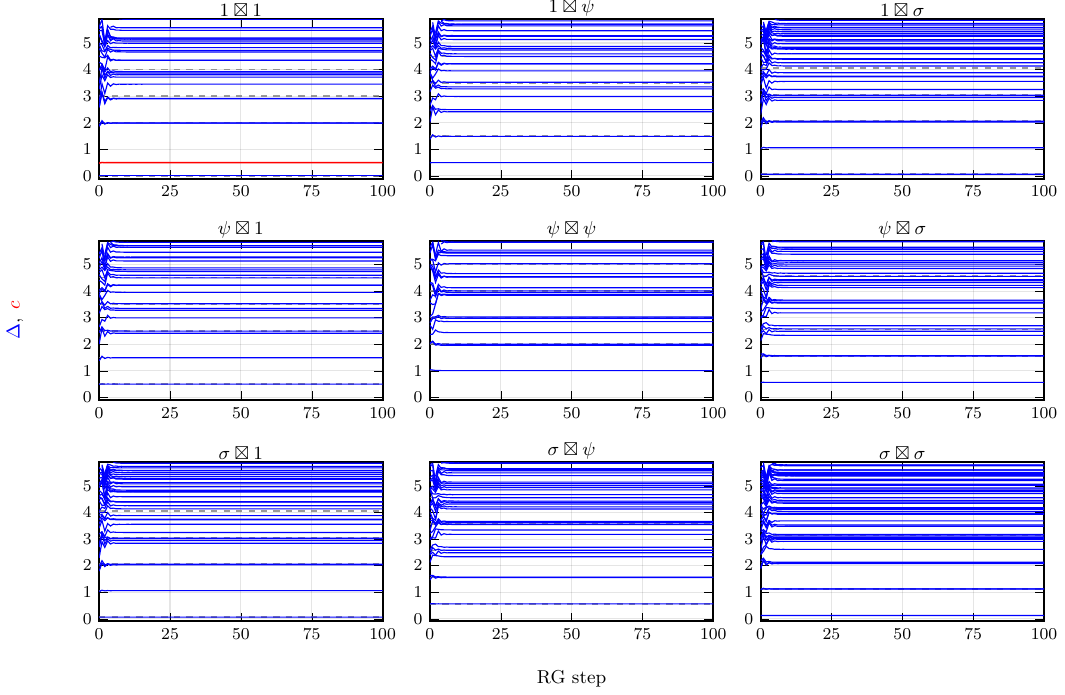}
    \caption{RG flow of the conformal spectra in the nine
    $\fZ(\Ising)\cong\Ising\boxtimes\overline{\Ising}$ sectors, computed with
    $\chi_{\text{RG}}=5$. Each panel is labeled by its sector
    $x\boxtimes y$. The first 100 RG steps are shown.}
    \label{fig:stable-ising-rg-spec}
\end{figure*}

\begin{figure*}[tb]
    \centering
    \includegraphics[width=0.6\textwidth]{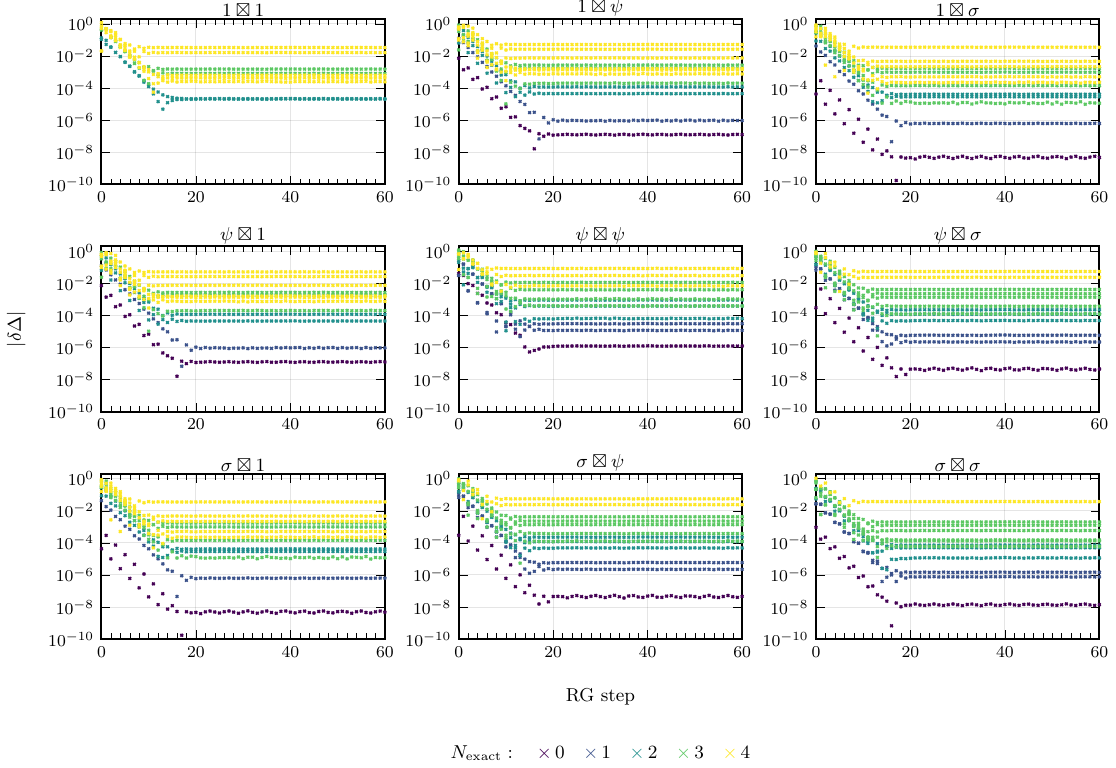}
    \caption{Absolute errors $|\delta\Delta|$ of the scaling dimensions
    relative to the exact Ising CFT values in the nine $\fZ(\Ising)$
    sectors. Each panel is labeled by its sector $x\boxtimes y$.
    The shared key maps the color of each curve to its
    exact descendant level $N_{\mathrm{exact}}$, where
    $\Delta=\Delta_{\mathrm{primary}}+N$.
    The spectra are obtained using nuclear-norm-regularized
    $\Ising$-Loop-TNR~\cite{Homma_2024} with $\chi_{\text{RG}}=32$.}
    \label{fig:stable-ising-cft-errors}
\end{figure*}

\begin{figure*}[tb]
    \centering
    \includegraphics[width=0.6\textwidth]{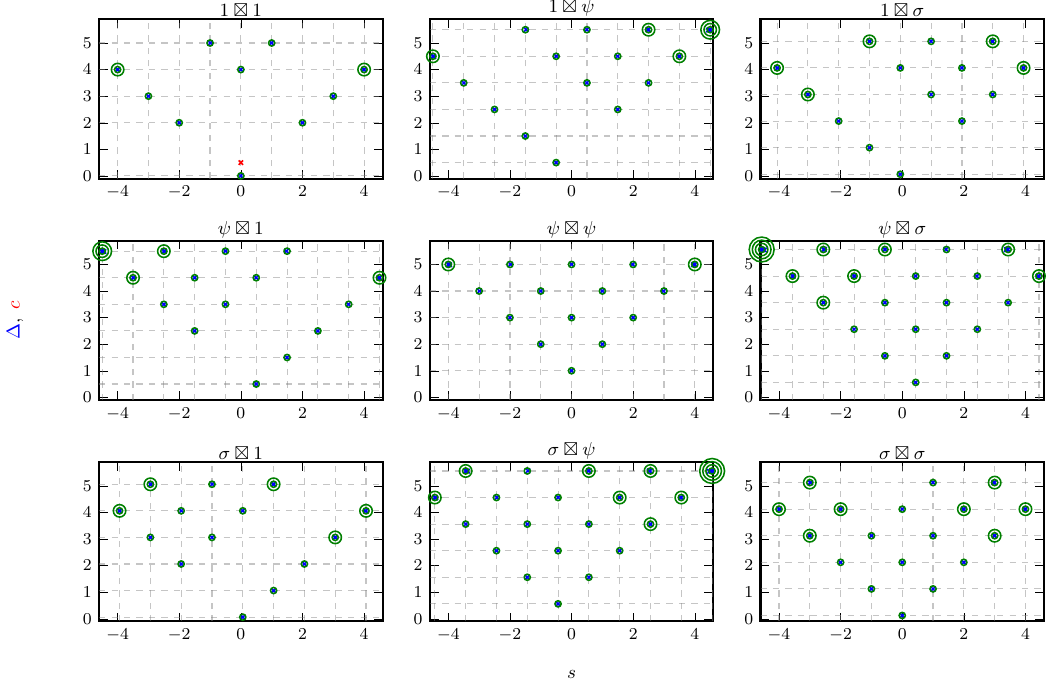}
    \caption{Spin-resolved conformal spectrum of the doubled Ising CFT in
    the nine $\fZ(\Ising)$ sectors at the RG fixed point. Each panel is
    labeled by its sector $x\boxtimes y$. Green concentric rings mark the
    exact CFT positions and their
    multiplicities. In each panel, the dashed grid marks its exact allowed
    values of $s$ and $\Delta$ ($\chi_{\text{RG}}=16$, $\chi_{\text{transfer}} = 40$,
    $N_{\text{RG}}=20$).}
    \label{fig:stable-ising-spin-spec}
\end{figure*}

\begin{table*}[tb]
    \centering
    \scriptsize
    \caption{Magnitudes $|C^{\text{cl}}_{abc}|$ of the closed-string structure constants of the Ising CFT
    with stable symTO $\fZ(\Ising)$ ($\cR=\Ising$), at the RG fixed point
    ($\chi_{\text{RG}}=36$, $N_{\text{RG}}=20$).
    All nine doubled primary fields in
    $\fZ(\Ising)\cong\Ising\boxtimes\overline{\Ising}$ are resolved. The
    full-CFT primary field associated with the sector $x\boxtimes y$ is denoted by
    $x\bar y$. Numerical values include the effective collar correction
    of Eq.~\eqref{eq:ising-collar-amplitude}. The ordered labels $a,b,c$
    correspond to the left, up, and tilted directions. Asterisks mark the
    24 entries used to determine $l_i$. The vacuum channel is normalized
    to unity and is excluded from the fit.}
    \label{tab:structure-constants-stable-ising}

    \begin{minipage}[t]{0.3200\textwidth}
        \vspace{0pt}
        \centering
        \begin{tabular}{cccc}
            \toprule
            $|C^{\text{cl}}_{a,b,c}|$ & Theory & Numerical & Abs. error \\
            \midrule
            $\psi\bar\psi, \sigma\bar\sigma, \sigma\bar\sigma$ & \multirow{3}{*}{$1/2$} & $0.500424$ & $0.000424$ \\
            $\sigma\bar\sigma, \psi\bar\psi, \sigma\bar\sigma$ & & $0.499737$ & $0.000263$ \\
            $\sigma\bar\sigma, \sigma\bar\sigma, \psi\bar\psi$ & & $0.500524$ & $0.000524$ \\
            \midrule
            $\psi\bar\sigma, \sigma\bar\psi, \sigma\bar\sigma$ & \multirow{6}{*}{$1/2$} & $0.499620$ & $0.000380$ \\
            $\psi\bar\sigma, \sigma\bar\sigma, \sigma\bar\psi$ & & $0.499889$ & $0.000111$ \\
            $\sigma\bar\psi, \psi\bar\sigma, \sigma\bar\sigma$ & & $0.499620$ & $0.000380$ \\
            $\sigma\bar\psi, \sigma\bar\sigma, \psi\bar\sigma$ & & $0.499889$ & $0.000111$ \\
            $\sigma\bar\sigma, \psi\bar\sigma, \sigma\bar\psi$ & & $0.499534$ & $0.000466$ \\
            $\sigma\bar\sigma, \sigma\bar\psi, \psi\bar\sigma$ & & $0.499534$ & $0.000466$ \\
            \midrule
            $\one\bar\psi, \one\bar\sigma, \one\bar\sigma$ & \multirow{3}{*}{$1/\sqrt{2}$} & $0.706914$ & $0.000193$ \\
            $\one\bar\sigma, \one\bar\psi, \one\bar\sigma$ & & $0.706446$ & $0.000661$ \\
            $\one\bar\sigma, \one\bar\sigma, \one\bar\psi$ & & $0.707280$ & $0.000174$ \\
            \midrule
            $\one\bar\psi, \psi\bar\sigma, \psi\bar\sigma$ & \multirow{3}{*}{$1/\sqrt{2}$} & $0.706661$ & $0.000446$ \\
            $\psi\bar\sigma, \one\bar\psi, \psi\bar\sigma$ & & $0.706120$ & $0.000987$ \\
            $\psi\bar\sigma, \psi\bar\sigma, \one\bar\psi$ & & $0.706308$ & $0.000799$ \\
            \midrule
            $\one\bar\psi, \sigma\bar\sigma, \sigma\bar\sigma$ & \multirow{3}{*}{$1/\sqrt{2}$} & $0.706883$ & $0.000224$ \\
            $\sigma\bar\sigma, \one\bar\psi, \sigma\bar\sigma$ & & $0.706419$ & $0.000688$ \\
            $\sigma\bar\sigma, \sigma\bar\sigma, \one\bar\psi$ & & $0.707137$ & $0.000030$ \\
            \midrule
            $\one\bar\sigma, \psi\bar\psi, \psi\bar\sigma$ & \multirow{6}{*}{$1/\sqrt{2}$} & $0.706855$ & $0.000252$ \\
            $\one\bar\sigma, \psi\bar\sigma, \psi\bar\psi$ & & $0.707302$ & $0.000195$ \\
            $\psi\bar\psi, \one\bar\sigma, \psi\bar\sigma$ & & $0.707357$ & $0.000250$ \\
            $\psi\bar\psi, \psi\bar\sigma, \one\bar\sigma$ & & $0.707523$ & $0.000416$ \\
            $\psi\bar\sigma, \one\bar\sigma, \psi\bar\psi$ & & $0.707309$ & $0.000203$ \\
            $\psi\bar\sigma, \psi\bar\psi, \one\bar\sigma$ & & $0.707028$ & $0.000079$ \\
            \midrule
            $\one\bar\sigma, \sigma\bar\psi, \sigma\bar\sigma$ & \multirow{6}{*}{$1/\sqrt{2}$} & $0.706498$ & $0.000609$ \\
            $\one\bar\sigma, \sigma\bar\sigma, \sigma\bar\psi$ & & $0.707284$ & $0.000177$ \\
            $\sigma\bar\psi, \one\bar\sigma, \sigma\bar\sigma$ & & $0.706982$ & $0.000125$ \\
            $\sigma\bar\psi, \sigma\bar\sigma, \one\bar\sigma$ & & $0.706970$ & $0.000137$ \\
            $\sigma\bar\sigma, \one\bar\sigma, \sigma\bar\psi$ & & $0.707297$ & $0.000190$ \\
            $\sigma\bar\sigma, \sigma\bar\psi, \one\bar\sigma$ & & $0.706500$ & $0.000607$ \\
            \midrule
            $\psi\bar\one, \sigma\bar\one, \sigma\bar\one$ & \multirow{3}{*}{$1/\sqrt{2}$} & $0.706914$ & $0.000193$ \\
            $\sigma\bar\one, \psi\bar\one, \sigma\bar\one$ & & $0.706446$ & $0.000661$ \\
            $\sigma\bar\one, \sigma\bar\one, \psi\bar\one$ & & $0.707280$ & $0.000174$ \\
            \bottomrule
        \end{tabular}
    \end{minipage}
    \hfill
    \begin{minipage}[t]{0.3200\textwidth}
        \vspace{0pt}
        \centering
        \begin{tabular}{cccc}
            \toprule
            $|C^{\text{cl}}_{a,b,c}|$ & Theory & Numerical & Abs. error \\
            \midrule
            $\psi\bar\one, \sigma\bar\psi, \sigma\bar\psi$ & \multirow{3}{*}{$1/\sqrt{2}$} & $0.706661$ & $0.000446$ \\
            $\sigma\bar\psi, \psi\bar\one, \sigma\bar\psi$ & & $0.706120$ & $0.000987$ \\
            $\sigma\bar\psi, \sigma\bar\psi, \psi\bar\one$ & & $0.706308$ & $0.000799$ \\
            \midrule
            $\psi\bar\one, \sigma\bar\sigma, \sigma\bar\sigma$ & \multirow{3}{*}{$1/\sqrt{2}$} & $0.706883$ & $0.000224$ \\
            $\sigma\bar\sigma, \psi\bar\one, \sigma\bar\sigma$ & & $0.706419$ & $0.000688$ \\
            $\sigma\bar\sigma, \sigma\bar\sigma, \psi\bar\one$ & & $0.707137$ & $0.000030$ \\
            \midrule
            $\psi\bar\psi, \sigma\bar\one, \sigma\bar\psi$ & \multirow{6}{*}{$1/\sqrt{2}$} & $0.707357$ & $0.000250$ \\
            $\psi\bar\psi, \sigma\bar\psi, \sigma\bar\one$ & & $0.707523$ & $0.000416$ \\
            $\sigma\bar\one, \psi\bar\psi, \sigma\bar\psi$ & & $0.706855$ & $0.000252$ \\
            $\sigma\bar\one, \sigma\bar\psi, \psi\bar\psi$ & & $0.707302$ & $0.000195$ \\
            $\sigma\bar\psi, \psi\bar\psi, \sigma\bar\one$ & & $0.707028$ & $0.000079$ \\
            $\sigma\bar\psi, \sigma\bar\one, \psi\bar\psi$ & & $0.707309$ & $0.000203$ \\
            \midrule
            $\psi\bar\sigma, \sigma\bar\one, \sigma\bar\sigma$ & \multirow{6}{*}{$1/\sqrt{2}$} & $0.706982$ & $0.000125$ \\
            $\psi\bar\sigma, \sigma\bar\sigma, \sigma\bar\one$ & & $0.706970$ & $0.000137$ \\
            $\sigma\bar\one, \psi\bar\sigma, \sigma\bar\sigma$ & & $0.706498$ & $0.000609$ \\
            $\sigma\bar\one, \sigma\bar\sigma, \psi\bar\sigma$ & & $0.707284$ & $0.000177$ \\
            $\sigma\bar\sigma, \psi\bar\sigma, \sigma\bar\one$ & & $0.706500$ & $0.000607$ \\
            $\sigma\bar\sigma, \sigma\bar\one, \psi\bar\sigma$ & & $0.707297$ & $0.000190$ \\
            \midrule
            $\one\bar\one, \one\bar\one, \one\bar\one$ & $1$ & $1.000000$ & $0.000000$ \\
            \midrule
            $\one\bar\one, \one\bar\psi, \one\bar\psi$ & \multirow{3}{*}{$1$} & $0.999707^{*}$ & $0.000293$ \\
            $\one\bar\psi, \one\bar\one, \one\bar\psi$ & & $0.999653^{*}$ & $0.000347$ \\
            $\one\bar\psi, \one\bar\psi, \one\bar\one$ & & $0.999140^{*}$ & $0.000860$ \\
            \midrule
            $\one\bar\one, \one\bar\sigma, \one\bar\sigma$ & \multirow{3}{*}{$1$} & $0.999963^{*}$ & $0.000037$ \\
            $\one\bar\sigma, \one\bar\one, \one\bar\sigma$ & & $0.999975^{*}$ & $0.000025$ \\
            $\one\bar\sigma, \one\bar\sigma, \one\bar\one$ & & $0.999862^{*}$ & $0.000138$ \\
            \midrule
            $\one\bar\one, \psi\bar\one, \psi\bar\one$ & \multirow{3}{*}{$1$} & $0.999707^{*}$ & $0.000293$ \\
            $\psi\bar\one, \one\bar\one, \psi\bar\one$ & & $0.999653^{*}$ & $0.000347$ \\
            $\psi\bar\one, \psi\bar\one, \one\bar\one$ & & $0.999140^{*}$ & $0.000860$ \\
            \midrule
            $\one\bar\one, \psi\bar\psi, \psi\bar\psi$ & \multirow{3}{*}{$1$} & $1.000519^{*}$ & $0.000519$ \\
            $\psi\bar\psi, \one\bar\one, \psi\bar\psi$ & & $1.000594^{*}$ & $0.000594$ \\
            $\psi\bar\psi, \psi\bar\psi, \one\bar\one$ & & $1.001572^{*}$ & $0.001572$ \\
            \midrule
            $\one\bar\one, \psi\bar\sigma, \psi\bar\sigma$ & \multirow{3}{*}{$1$} & $0.999809^{*}$ & $0.000191$ \\
            $\psi\bar\sigma, \one\bar\one, \psi\bar\sigma$ & & $0.999787^{*}$ & $0.000213$ \\
            $\psi\bar\sigma, \psi\bar\sigma, \one\bar\one$ & & $0.999409^{*}$ & $0.000591$ \\
            \bottomrule
        \end{tabular}
    \end{minipage}
    \hfill
    \begin{minipage}[t]{0.3200\textwidth}
        \vspace{0pt}
        \centering
        \begin{tabular}{cccc}
            \toprule
            $|C^{\text{cl}}_{a,b,c}|$ & Theory & Numerical & Abs. error \\
            \midrule
            $\one\bar\one, \sigma\bar\one, \sigma\bar\one$ & \multirow{3}{*}{$1$} & $0.999963^{*}$ & $0.000037$ \\
            $\sigma\bar\one, \one\bar\one, \sigma\bar\one$ & & $0.999975^{*}$ & $0.000025$ \\
            $\sigma\bar\one, \sigma\bar\one, \one\bar\one$ & & $0.999862^{*}$ & $0.000138$ \\
            \midrule
            $\one\bar\one, \sigma\bar\psi, \sigma\bar\psi$ & \multirow{3}{*}{$1$} & $0.999809^{*}$ & $0.000191$ \\
            $\sigma\bar\psi, \one\bar\one, \sigma\bar\psi$ & & $0.999787^{*}$ & $0.000213$ \\
            $\sigma\bar\psi, \sigma\bar\psi, \one\bar\one$ & & $0.999409^{*}$ & $0.000591$ \\
            \midrule
            $\one\bar\one, \sigma\bar\sigma, \sigma\bar\sigma$ & \multirow{3}{*}{$1$} & $0.999944^{*}$ & $0.000056$ \\
            $\sigma\bar\sigma, \one\bar\one, \sigma\bar\sigma$ & & $0.999970^{*}$ & $0.000030$ \\
            $\sigma\bar\sigma, \sigma\bar\sigma, \one\bar\one$ & & $0.999776^{*}$ & $0.000224$ \\
            \midrule
            $\one\bar\psi, \psi\bar\one, \psi\bar\psi$ & \multirow{6}{*}{$1$} & $0.999402$ & $0.000598$ \\
            $\one\bar\psi, \psi\bar\psi, \psi\bar\one$ & & $0.999481$ & $0.000519$ \\
            $\psi\bar\one, \one\bar\psi, \psi\bar\psi$ & & $0.999402$ & $0.000598$ \\
            $\psi\bar\one, \psi\bar\psi, \one\bar\psi$ & & $0.999481$ & $0.000519$ \\
            $\psi\bar\psi, \one\bar\psi, \psi\bar\one$ & & $0.999465$ & $0.000535$ \\
            $\psi\bar\psi, \psi\bar\one, \one\bar\psi$ & & $0.999465$ & $0.000535$ \\
            \midrule
            $\one\bar\psi, \sigma\bar\one, \sigma\bar\psi$ & \multirow{6}{*}{$1$} & $0.999622$ & $0.000378$ \\
            $\one\bar\psi, \sigma\bar\psi, \sigma\bar\one$ & & $0.999187$ & $0.000813$ \\
            $\sigma\bar\one, \one\bar\psi, \sigma\bar\psi$ & & $0.999687$ & $0.000313$ \\
            $\sigma\bar\one, \sigma\bar\psi, \one\bar\psi$ & & $0.999655$ & $0.000345$ \\
            $\sigma\bar\psi, \one\bar\psi, \sigma\bar\one$ & & $0.999201$ & $0.000799$ \\
            $\sigma\bar\psi, \sigma\bar\one, \one\bar\psi$ & & $0.999604$ & $0.000396$ \\
            \midrule
            $\one\bar\sigma, \psi\bar\one, \psi\bar\sigma$ & \multirow{6}{*}{$1$} & $0.999687$ & $0.000313$ \\
            $\one\bar\sigma, \psi\bar\sigma, \psi\bar\one$ & & $0.999655$ & $0.000345$ \\
            $\psi\bar\one, \one\bar\sigma, \psi\bar\sigma$ & & $0.999622$ & $0.000378$ \\
            $\psi\bar\one, \psi\bar\sigma, \one\bar\sigma$ & & $0.999187$ & $0.000813$ \\
            $\psi\bar\sigma, \one\bar\sigma, \psi\bar\one$ & & $0.999604$ & $0.000396$ \\
            $\psi\bar\sigma, \psi\bar\one, \one\bar\sigma$ & & $0.999201$ & $0.000799$ \\
            \midrule
            $\one\bar\sigma, \sigma\bar\one, \sigma\bar\sigma$ & \multirow{6}{*}{$1$} & $0.999939$ & $0.000061$ \\
            $\one\bar\sigma, \sigma\bar\sigma, \sigma\bar\one$ & & $0.999836$ & $0.000164$ \\
            $\sigma\bar\one, \one\bar\sigma, \sigma\bar\sigma$ & & $0.999939$ & $0.000061$ \\
            $\sigma\bar\one, \sigma\bar\sigma, \one\bar\sigma$ & & $0.999836$ & $0.000164$ \\
            $\sigma\bar\sigma, \one\bar\sigma, \sigma\bar\one$ & & $0.999849$ & $0.000151$ \\
            $\sigma\bar\sigma, \sigma\bar\one, \one\bar\sigma$ & & $0.999849$ & $0.000151$ \\
            \bottomrule
        \end{tabular}
    \end{minipage}
\end{table*}

\begin{figure*}[tb]
    \centering
    \includegraphics[width=\linewidth]{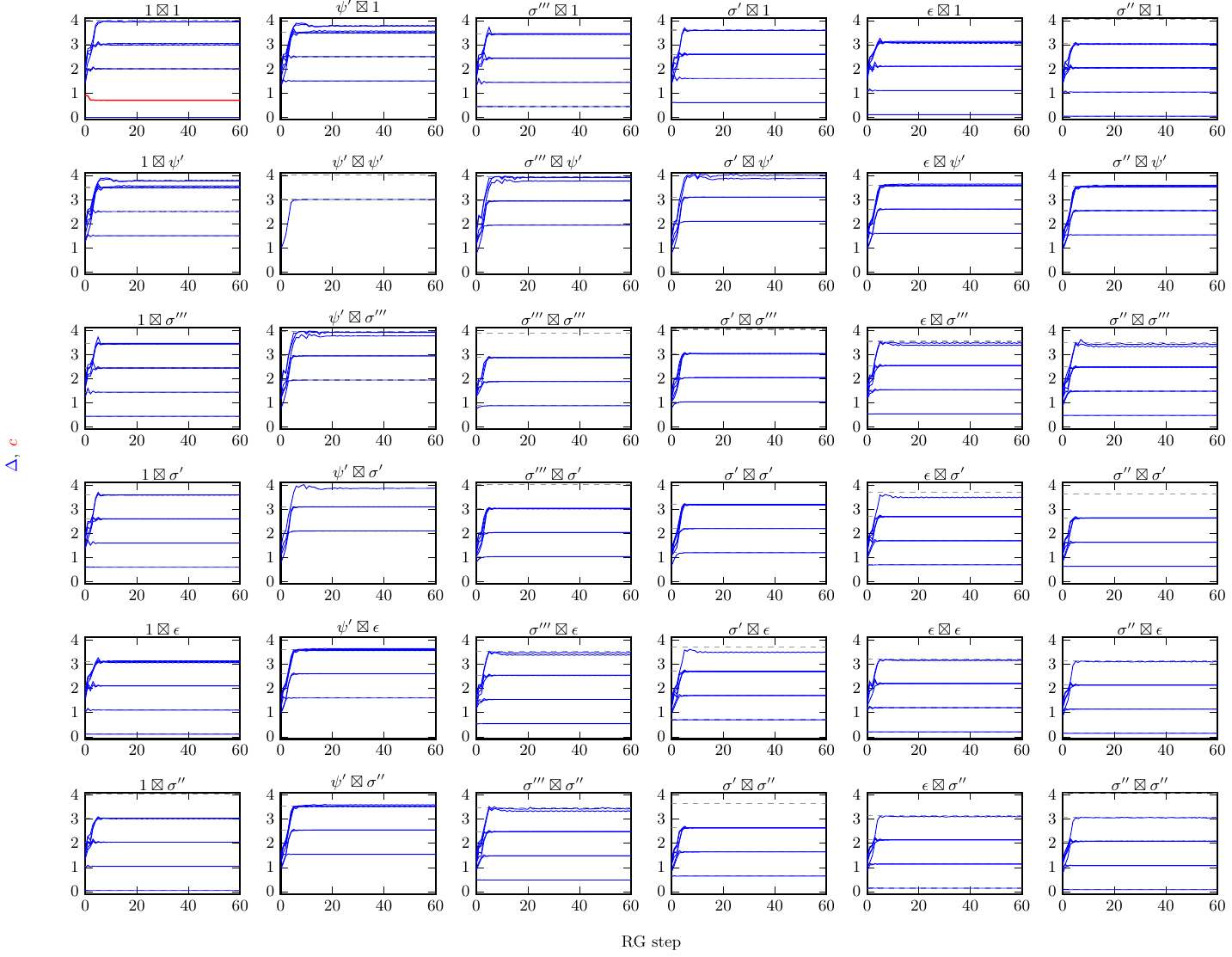}
    \caption{RG flow of the scaling dimensions in all 36 sectors of the
    doubled tricritical Ising CFT
    $\fZ(\cC_{5,4})\cong
    \cC_{5,4}\boxtimes\overline{\cC_{5,4}}$. Each panel is labeled by its sector $x\boxtimes y$.
    Solid blue curves trace the numerical levels, gray dashed lines mark
    the exact scaling dimensions, and red denotes the central charge.
    Rapidly fluctuating high-lying trajectories are omitted.
    The flow is computed with $\chi_{\text{RG}}=48$ and remains stable
    through $N_{\text{RG}}=60$ RG steps.}
    \label{fig:triising-rg-spec}
\end{figure*}

\begin{figure*}[tb]
    \centering
    \includegraphics[width=\linewidth]{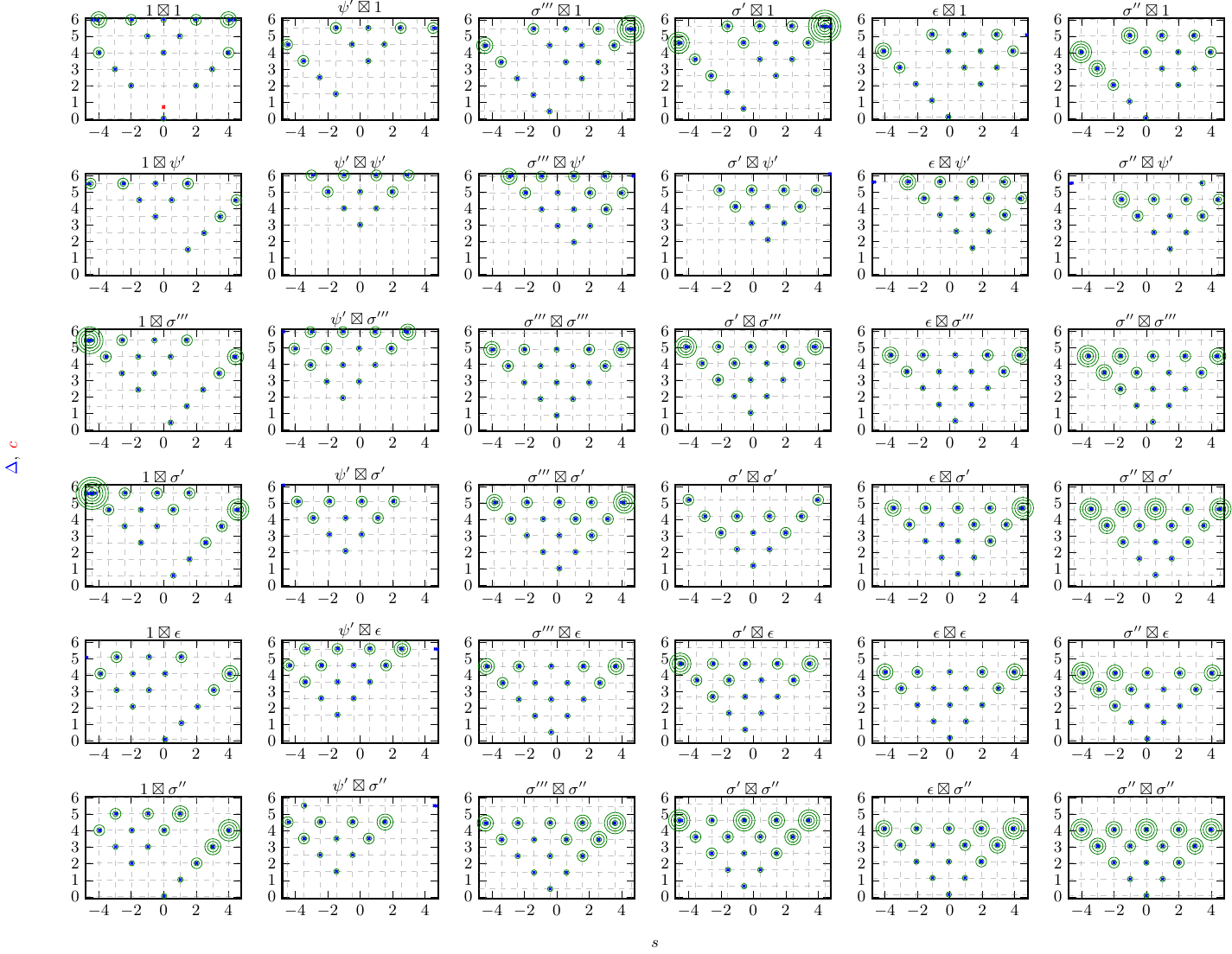}
    \caption{Spin-resolved conformal spectra in all 36 sectors of the doubled
    tricritical Ising CFT $\fZ(\cC_{5,4})$. Each panel is labeled by its sector $x\boxtimes y$. The
    sector-dependent gray dashed grid marks the exact values of $s$ and
    $\Delta$, green concentric circles indicate degenerate levels, and red
    denotes the central charge. The spectra are extracted from the RG
    fixed-point tensor at $N_{\text{RG}}=20$ with $\chi_{\text{RG}}=48$ and
    $\chi_{\text{transfer}}=100$.}
    \label{fig:triising-spin-spec}
\end{figure*}

\begin{table*}[tb]
    \centering
    \footnotesize
    \caption{Magnitudes of the closed-string structure constants
    $|C^{\text{cl}}_{abc}|$ of the tricritical Ising CFT
    $\cM(5,4)$ with stable symTO $\fZ(\cC_{5,4})$, using $\chi_{\text{RG}}=50$ at $N_{\text{RG}}=20$.
    Only diagonal full-CFT primary fields $x\bar x$ are listed.
    The half-$4\times4$ triangular patch was used, with cutoff 50 in
    both subsequent TNR optimizations. Numerical values include the
    effective collar correction of Eq.~\eqref{eq:ising-collar-amplitude}.
    The ordered labels $a,b,c$ correspond to the left, up, and tilted
    directions. Asterisks mark the 12 entries used to determine $l_i$.
    Daggers mark the 15 channels containing $\psi'\bar\psi'$ with
    $\Delta=3$, which are retained in the table but excluded from the fit.
    The vacuum channel is normalized to unity and is excluded from the fit.}
    \label{tab:structure-constants-triising}
    \begin{minipage}[t]{0.4800\textwidth}
        \vspace{0pt}
        \centering
        \begin{tabular}{cccc}
        \toprule
        $|C^{\text{cl}}_{a,b,c}|$ & Theory & Numerical & Abs. error \\
        \midrule
        $\one\bar\one,\one\bar\one,\one\bar\one$ & $1$ & $1.000000$ & $0.000000$ \\
\midrule
$\one\bar\one,\sigma''\bar\sigma'',\sigma''\bar\sigma''$  & \multirow{3}{*}{$1$}  & $1.000218^{*}$ & $0.000218$ \\
$\sigma''\bar\sigma'',\one\bar\one,\sigma''\bar\sigma''$  &  & $1.000276^{*}$ & $0.000276$ \\
$\sigma''\bar\sigma'',\sigma''\bar\sigma'',\one\bar\one$  &  & $1.000569^{*}$ & $0.000569$ \\
\midrule
$\sigma''\bar\sigma'',\sigma''\bar\sigma'',\epsilon\bar\epsilon$  & \multirow{3}{*}{$A_{\mathrm{TCI}}$}  & $0.915905$ & $0.000452$ \\
$\sigma''\bar\sigma'',\epsilon\bar\epsilon,\sigma''\bar\sigma''$  &  & $0.916219$ & $0.000766$ \\
$\epsilon\bar\epsilon,\sigma''\bar\sigma'',\sigma''\bar\sigma''$  &  & $0.916386$ & $0.000933$ \\
\midrule
$\one\bar\one,\epsilon\bar\epsilon,\epsilon\bar\epsilon$  & \multirow{3}{*}{$1$}  & $1.000510^{*}$ & $0.000510$ \\
$\epsilon\bar\epsilon,\one\bar\one,\epsilon\bar\epsilon$  &  & $1.000638^{*}$ & $0.000638$ \\
$\epsilon\bar\epsilon,\epsilon\bar\epsilon,\one\bar\one$  &  & $1.001296^{*}$ & $0.001296$ \\
\midrule
$\sigma''\bar\sigma'',\epsilon\bar\epsilon,\sigma'''\bar\sigma'''$  & \multirow{6}{*}{$1/2$}  & $0.500192$ & $0.000192$ \\
$\sigma''\bar\sigma'',\sigma'''\bar\sigma''',\epsilon\bar\epsilon$  &  & $0.499636$ & $0.000364$ \\
$\epsilon\bar\epsilon,\sigma''\bar\sigma'',\sigma'''\bar\sigma'''$  &  & $0.500050$ & $0.000050$ \\
$\epsilon\bar\epsilon,\sigma'''\bar\sigma''',\sigma''\bar\sigma''$  &  & $0.499772$ & $0.000228$ \\
$\sigma'''\bar\sigma''',\sigma''\bar\sigma'',\epsilon\bar\epsilon$  &  & $0.500856$ & $0.000856$ \\
$\sigma'''\bar\sigma''',\epsilon\bar\epsilon,\sigma''\bar\sigma''$  &  & $0.501060$ & $0.001060$ \\
\midrule
$\sigma''\bar\sigma'',\sigma''\bar\sigma'',\sigma'\bar\sigma'$  & \multirow{3}{*}{$A_{\mathrm{TCI}}/6$}  & $0.152381$ & $0.000195$ \\
$\sigma''\bar\sigma'',\sigma'\bar\sigma',\sigma''\bar\sigma''$  &  & $0.151968$ & $0.000608$ \\
$\sigma'\bar\sigma',\sigma''\bar\sigma'',\sigma''\bar\sigma''$  &  & $0.152713$ & $0.000137$ \\
\midrule
$\epsilon\bar\epsilon,\sigma'\bar\sigma',\epsilon\bar\epsilon$  & \multirow{3}{*}{$2A_{\mathrm{TCI}}/3$}  & $0.608672$ & $0.001630$ \\
$\sigma'\bar\sigma',\epsilon\bar\epsilon,\epsilon\bar\epsilon$  &  & $0.611204$ & $0.000902$ \\
$\epsilon\bar\epsilon,\epsilon\bar\epsilon,\sigma'\bar\sigma'$  &  & $0.610490$ & $0.000188$ \\
\midrule
$\one\bar\one,\sigma'''\bar\sigma''',\sigma'''\bar\sigma'''$  & \multirow{3}{*}{$1$}  & $1.000503^{*}$ & $0.000503$ \\
$\sigma'''\bar\sigma''',\one\bar\one,\sigma'''\bar\sigma'''$  &  & $1.000585^{*}$ & $0.000585$ \\
$\sigma'''\bar\sigma''',\sigma'''\bar\sigma''',\one\bar\one$  &  & $1.001048^{*}$ & $0.001048$ \\
        \bottomrule
        \end{tabular}
    \end{minipage}
    \hfill
    \begin{minipage}[t]{0.4800\textwidth}
        \vspace{0pt}
        \centering
        \begin{tabular}{cccc}
        \toprule
        $|C^{\text{cl}}_{a,b,c}|$ & Theory & Numerical & Abs. error \\
        \midrule
$\sigma''\bar\sigma'',\sigma'''\bar\sigma''',\sigma'\bar\sigma'$  & \multirow{6}{*}{$3/4$}  & $0.751068$ & $0.001068$ \\
$\sigma''\bar\sigma'',\sigma'\bar\sigma',\sigma'''\bar\sigma'''$  &  & $0.749317$ & $0.000683$ \\
$\sigma'''\bar\sigma''',\sigma''\bar\sigma'',\sigma'\bar\sigma'$  &  & $0.750436$ & $0.000436$ \\
$\sigma'\bar\sigma',\sigma''\bar\sigma'',\sigma'''\bar\sigma'''$  &  & $0.750516$ & $0.000516$ \\
$\sigma'''\bar\sigma''',\sigma'\bar\sigma',\sigma''\bar\sigma''$  &  & $0.749412$ & $0.000588$ \\
$\sigma'\bar\sigma',\sigma'''\bar\sigma''',\sigma''\bar\sigma''$  &  & $0.750756$ & $0.000756$ \\
\midrule
$\one\bar\one,\sigma'\bar\sigma',\sigma'\bar\sigma'$  & \multirow{3}{*}{$1$}  & $0.999535^{*}$ & $0.000465$ \\
$\sigma'\bar\sigma',\one\bar\one,\sigma'\bar\sigma'$  &  & $0.999450^{*}$ & $0.000550$ \\
$\sigma'\bar\sigma',\sigma'\bar\sigma',\one\bar\one$  &  & $0.998985^{*}$ & $0.001015$ \\
\midrule
$\sigma''\bar\sigma'',\sigma''\bar\sigma'',\psi'\bar\psi'$  & \multirow{3}{*}{$1/56$}  & $0.017653^{\dagger}$ & $0.000205$ \\
$\sigma''\bar\sigma'',\psi'\bar\psi',\sigma''\bar\sigma''$  &  & $0.017268^{\dagger}$ & $0.000590$ \\
$\psi'\bar\psi',\sigma''\bar\sigma'',\sigma''\bar\sigma''$  &  & $0.017587^{\dagger}$ & $0.000270$ \\
\midrule
        $\sigma'\bar\sigma',\sigma'\bar\sigma',\sigma'\bar\sigma'$ & $2A_{\mathrm{TCI}}/3$ & $0.612321$ & $0.002019$ \\
\midrule
$\epsilon\bar\epsilon,\sigma'\bar\sigma',\psi'\bar\psi'$  & \multirow{6}{*}{$3/7$}  & $0.429427^{\dagger}$ & $0.000855$ \\
$\epsilon\bar\epsilon,\psi'\bar\psi',\sigma'\bar\sigma'$  &  & $0.418328^{\dagger}$ & $0.010244$ \\
$\sigma'\bar\sigma',\epsilon\bar\epsilon,\psi'\bar\psi'$  &  & $0.426360^{\dagger}$ & $0.002211$ \\
$\sigma'\bar\sigma',\psi'\bar\psi',\epsilon\bar\epsilon$  &  & $0.421086^{\dagger}$ & $0.007486$ \\
$\psi'\bar\psi',\epsilon\bar\epsilon,\sigma'\bar\sigma'$  &  & $0.423267^{\dagger}$ & $0.005304$ \\
$\psi'\bar\psi',\sigma'\bar\sigma',\epsilon\bar\epsilon$  &  & $0.425256^{\dagger}$ & $0.003315$ \\
\midrule
$\sigma'''\bar\sigma''',\sigma'''\bar\sigma''',\psi'\bar\psi'$  & \multirow{3}{*}{$7/8$}  & $0.877602^{\dagger}$ & $0.002602$ \\
$\sigma'''\bar\sigma''',\psi'\bar\psi',\sigma'''\bar\sigma'''$  &  & $0.859250^{\dagger}$ & $0.015750$ \\
$\psi'\bar\psi',\sigma'''\bar\sigma''',\sigma'''\bar\sigma'''$  &  & $0.867781^{\dagger}$ & $0.007219$ \\
\midrule
$\one\bar\one,\psi'\bar\psi',\psi'\bar\psi'$  & \multirow{3}{*}{$1$}  & $0.985805^{\dagger}$ & $0.014195$ \\
$\psi'\bar\psi',\one\bar\one,\psi'\bar\psi'$  &  & $0.984646^{\dagger}$ & $0.015354$ \\
$\psi'\bar\psi',\psi'\bar\psi',\one\bar\one$  &  & $1.023516^{\dagger}$ & $0.023516$ \\
\bottomrule
        \end{tabular}
    \end{minipage}
\end{table*}

The first example is the Ising CFT, for which $\fZ(\Ising)$ is a stable symTO
in the sense of Definition~\ref{def:stable-symto}.  It is also
chiral-factorized and maximal~\cite{chatterjee2024emergentgeneralizedsymmetrymaximal} in the sense of
Sec.~\ref{appsec:ffa-extensions-and-emergent-symto}. We choose the charge category $\cR=\Ising$, obtained
from the Ising UMTC by forgetting its braiding structure. The relevant
categorical data are reviewed in App.~\ref{appsec:ising-category}.

Complementary Hamiltonian realizations have also been studied. Roychowdhury
and Wang recently found a symmetric critical phase in an Ising-fusion-category
lattice model whose low-energy theory lies in the Ising universality
class~\cite{Roychowdhury_2026}. A related precursor is the $G$-graded
fusion-category construction of Ning, Mao, and Wang, whose Ising-category
example likewise realizes the critical Ising CFT at low
energy~\cite{Ning:2023qbv}.

We build the initial tensor $T$ according to
Eq.~\eqref{eq:initial-tensor-bootstrap} with
$V=\one\oplus\psi\oplus\sigma$. The integer lifts $n_i$ are fixed so
that the conformal weights $h_i=s_i+n_i$ are
\[
    \left\{0,\frac{1}{2},\frac{1}{16}\right\}.
\]
We apply $\Ising$-Loop-TNR with cutoff $\chi_{\text{RG}}=5$ for
$N_{\text{RG}}=200$ RG steps.

The stability of the theory is supported by five complementary observations.
\begin{itemize}
    \item The RG flow remains stable for 200 steps without divergence.
    The first 100 steps are shown in Fig.~\ref{fig:stable-ising-rg-spec}.
    \item With $\chi_{\text{RG}}=32$, nuclear-norm-regularized Loop-TNR
    gives small absolute errors for the low-lying levels relative to the
    exact CFT spectrum, as shown in Fig.~\ref{fig:stable-ising-cft-errors}.
    \item The initial parameters $\cN_{ijk}$ can be chosen from a large
    region of parameter space while the flow still converges to the Ising
    CFT fixed point. In particular, $\cN_{ijk}=1$ for every admissible
    triple $i,j,k\in\Irr(\Ising)$ is already a valid choice.
    \item After symmetry-preserving random perturbations are added to the
    fixed-point tensor, the ensuing Loop-TNR flow returns to the same fixed
    point.
    \item The spin-resolved transfer spectrum separates all nine
    $\fZ(\Ising)$ sectors into the expected doubled Ising conformal families,
    as shown in Fig.~\ref{fig:stable-ising-spin-spec}.
\end{itemize}

The spin-resolved spectrum obeys the doubled-sector relations
\begin{equation}\label{eq:doubled-sector-spectrum}
    \begin{aligned}
    \Delta_{x\boxtimes y}&=h_x+h_y+n+\bar n,\\
    s_{x\boxtimes y}&=h_x-h_y+n-\bar n
    \equiv s_x-s_y \pmod 1,
    \end{aligned}
\end{equation}
where $n,\bar n\geq0$ are the integer descendant levels of the states in
the two chiral modules. The spin relation confirms modular $T$-covariance. We use
Eq.~\eqref{eq:doubled-sector-spectrum} as the common diagnostic for the
remaining examples.

The observed RG stability can be understood from the spectrum in the trivial
sector $\one\boxtimes\one$.  The level-one vacuum descendants are null, so
every nonidentity local operator has scaling dimension $\Delta\geq2$.
Consequently, there is no nontrivial relevant symmetric perturbation, in
agreement with the criterion of Definition~\ref{def:stable-symto}.

For the stable Ising boundary in Tab.~\ref{tab:gapless-boundary-examples}, the corresponding
condensable algebra in
$\fZ(\Ising)\cong\Ising\boxtimes\overline{\Ising}$ is the tensor unit
$A=\one\boxtimes\one$, hence no nontrivial anyon is condensed before topological Wick rotation.  The stable,
chiral-factorized, and maximal symTOs~\cite{chatterjee2024emergentgeneralizedsymmetrymaximal} reviewed in
Sec.~\ref{appsec:ffa-extensions-and-emergent-symto} therefore all coincide
with $\fZ(\Ising)$ in this example.

We extract the magnitudes of the closed-string structure constants using
the method introduced in Sec.~\ref{sec:structure-constants}. We find that
introducing three additional effective collar lengths $l_i$, with
$i=1,2,3$ corresponding to the left, up, and tilted directions,
substantially improves the agreement with the CFT values. We replace
Eq.~\eqref{eq:extract-closed-structure-constants} by
\begin{equation}\label{eq:ising-collar-amplitude}
\begin{aligned}
\left|\frac{\mathcal A_{abc^*}}{\mathcal A_{000}}\right|
&\simeq \left|C^{\mathrm{cl}}_{abc^*}\right|
e^{-\frac{\pi}{4}(\Delta_a+\Delta_b+\Delta_c)}\\
&\quad\times(\sqrt{2})^{-\Delta_a-\Delta_b+\Delta_c}
e^{-l_1\Delta_a-l_2\Delta_b-l_3\Delta_c}.
\end{aligned}
\end{equation}
Here $0$ denotes the vacuum, $a^*$ is the conjugate of $a$, and
$\Delta_a$ is its scaling dimension. The parameters $l_i$ are
dimensionless shifts relative to the reference pair-of-pants geometry
and are allowed to have either sign.

We first determine $l_i$ from the channels $0aa^*$, $a0a^*$, and
$aa^*0$. Inserting the identity reduces these amplitudes to two-point
functions, for which the normalized structure constants are unity.
Since $\Delta_{a^*}=\Delta_a$, we obtain
\begin{equation}\label{eq:ising-collar-identity-channels}
\begin{aligned}
\left|\frac{\mathcal A_{0aa^*}}{\mathcal A_{000}}\right|
&\simeq e^{-(\pi/2+l_2+l_3)\Delta_a},\\
\left|\frac{\mathcal A_{a0a^*}}{\mathcal A_{000}}\right|
&\simeq e^{-(\pi/2+l_1+l_3)\Delta_a},\\
\left|\frac{\mathcal A_{aa^*0}}{\mathcal A_{000}}\right|
&\simeq e^{-(\pi/2+\log 2+l_1+l_2)\Delta_a}.
\end{aligned}
\end{equation}
We collect the nonvacuum primary fields from all sectors into a set
$\mathcal P$ and determine the three parameters jointly by minimizing
\begin{equation}\label{eq:ising-collar-fit}
\begin{aligned}
& (\widehat l_1,\widehat l_2,\widehat l_3)
=\underset{l_1,l_2,l_3\in\mathbb R}{\operatorname{arg\,min}}\\
&\quad\sum_{a\in\mathcal P}\Bigg\{
\left[
\log\left|\frac{\mathcal A_{0aa^*}}{\mathcal A_{000}}\right|
+\left(\frac{\pi}{2}+l_2+l_3\right)\Delta_a
\right]^2\\
&\qquad+
\left[
\log\left|\frac{\mathcal A_{a0a^*}}{\mathcal A_{000}}\right|
+\left(\frac{\pi}{2}+l_1+l_3\right)\Delta_a
\right]^2\\
&\qquad+
\left[
\log\left|\frac{\mathcal A_{aa^*0}}{\mathcal A_{000}}\right|
+\left(\frac{\pi}{2}+\log 2+l_1+l_2\right)\Delta_a
\right]^2
\Bigg\}.
\end{aligned}
\end{equation}
We use the exact primary scaling dimensions and assign equal weight to
each logarithmic residual. For the eight nonvacuum primaries of the
maximal Ising theory, this gives 24 constraints on three parameters.

Using the half-$4\times4$ patch after 20 RG steps, with cutoff 36 in
both the RG flow and the two subsequent TNR optimizations, we obtain
\begin{equation}\label{eq:ising-collar-lengths}
\begin{aligned}
(\widehat l_1,\widehat l_2,\widehat l_3)
&=(0.00557785,\ 0.00759716,\\
&\qquad -0.00401785).
\end{aligned}
\end{equation}
We then keep these parameters fixed and extract the structure constants
in all sectors. For the 75 ordered channels containing no identity
field, none of which enter the fit, the root-mean-square relative error
is $0.496\%$ when using the geometric factor in
Eq.~\eqref{eq:extract-closed-structure-constants} without collar
corrections. Including the fitted collar corrections reduces this
error to $0.0599\%$.

Table~\ref{tab:structure-constants-stable-ising} summarizes the resulting
100 ordered channels involving all nine doubled primary fields. We
denote the full-CFT primary in the $x\boxtimes y$ sector by $x\bar y$.
Entries used to determine $l_i$ are marked by an asterisk. The vacuum
channel fixes the overall normalization and is excluded from the fit.

\subsection{Tricritical Ising CFT with a stable symTO}
\label{sec:stable-tricritical-ising}

The second example is the tricritical Ising CFT $\cM(5,4)$, with
$c=7/10$.  As in the Ising example of
Sec.~\ref{para:ising-extension-tower}, the Virasoro vacuum module has no
nonidentity state below conformal weight two, and the minimal full field
subalgebra coincides with the maximal chiral-factorized subalgebra.
Consequently, the stable, chiral-factorized, and maximal symTOs~\cite{chatterjee2024emergentgeneralizedsymmetrymaximal} all coincide
with $\fZ(\cC_{5,4})$.
We take $\cR=\cC_{5,4}\cong\mathsf{6}_{7/10}$. Its six simple objects,
their conformal weights, and the UMTC factorization are collected in
App.~\ref{appsec:tricritical-ising-category} and Tab.~\ref{tab:triising-mtc}. Thus
$\fZ(\cC_{5,4})$ has 36 doubled sectors. We apply
nuclear-norm-regularized $\cR$-Loop-TNR~\cite{Homma_2024} with
$\chi_{\text{RG}}=48$.

The numerical results are summarized by three observations.
\begin{itemize}
    \item The RG flow remains stable for 60 steps and produces
    converged spectra in all 36 $\fZ(\cC_{5,4})$ sectors, as shown in
    Fig.~\ref{fig:triising-rg-spec}.
    \item At $N_{\text{RG}}=20$, the spin-resolved spectrum contains one
    conformal family in each sector, reproduces the corresponding doubled
    $\cM(5,4)$ tower, and satisfies
    Eq.~\eqref{eq:doubled-sector-spectrum}. See
    Fig.~\ref{fig:triising-spin-spec}.
    \item Within the tested range, varying the initial parameters
    $\cN_{ijk}$ does not change the fixed point reached by the flow.
\end{itemize}

We extract the magnitudes of the structure constants for diagonal primary
fields using the collar correction introduced for the Ising example in
Eq.~\eqref{eq:ising-collar-amplitude}. We use the half-$4\times4$ patch
after 20 RG steps, with cutoff 50 in both the RG flow and the two
subsequent TNR optimizations. We display only diagonal primary fields
$x\bar x$ to keep the table manageable, since listing all allowed
channels in the 36 doubled sectors would make it considerably longer.

The channels involving $\psi'\bar\psi'$, whose scaling dimension is
$\Delta=3$, are less accurate than those involving lower-dimensional
fields in the present calculation. We therefore exclude this field
when determining the effective collar lengths. Applying the joint fit
in Eq.~\eqref{eq:ising-collar-fit} to the 12 identity channels associated
with the remaining four nonvacuum diagonal primaries gives
\begin{equation}\label{eq:triising-collar-lengths}
\begin{aligned}
(\widehat l_1,\widehat l_2,\widehat l_3)
&=(0.01707758,\ 0.03078702,\\
&\qquad -0.01175532).
\end{aligned}
\end{equation}
We then keep these parameters fixed and extract all 50 ordered channels
among the six diagonal primary fields, including those involving
$\psi'\bar\psi'$.

For the 22 ordered channels containing neither the identity nor the
dimension-three field, none of which enter the fit, the root-mean-square
relative error decreases from $2.486\%$ without collar corrections to
$0.157\%$ with the fitted corrections. The maximum relative error
decreases from $5.047\%$ to $0.398\%$. Including the dimension-three
field gives a root-mean-square relative error of $0.943\%$ over all 34
nonidentity channels, compared with $5.004\%$ without the correction.
Since the higher-dimensional field is less reliably resolved, we use
the lower-dimensional channels as the main accuracy benchmark.

Table~\ref{tab:structure-constants-triising} summarizes the resulting
structure constants. Here we define the constant~\cite{DOTSENKO1985291}
\begin{equation}
    A_{\mathrm{TCI}}
    =\sqrt{\frac{\Gamma(4/5)\Gamma(2/5)^3}
    {\Gamma(1/5)\Gamma(3/5)^3}}
    =0.915453111\ldots .
\label{eq:tci-structure-constant}
\end{equation}
Entries used to determine $l_i$ are marked by an asterisk. The vacuum channel fixes the overall
normalization and is excluded from the fit.

\subsection{Three-state Potts model as \texorpdfstring{$\cW_3$}{W3} algebra modules}
\label{sec:potts-w3}

\begin{figure*}[tb]
    \centering
    \includegraphics[width=\linewidth]{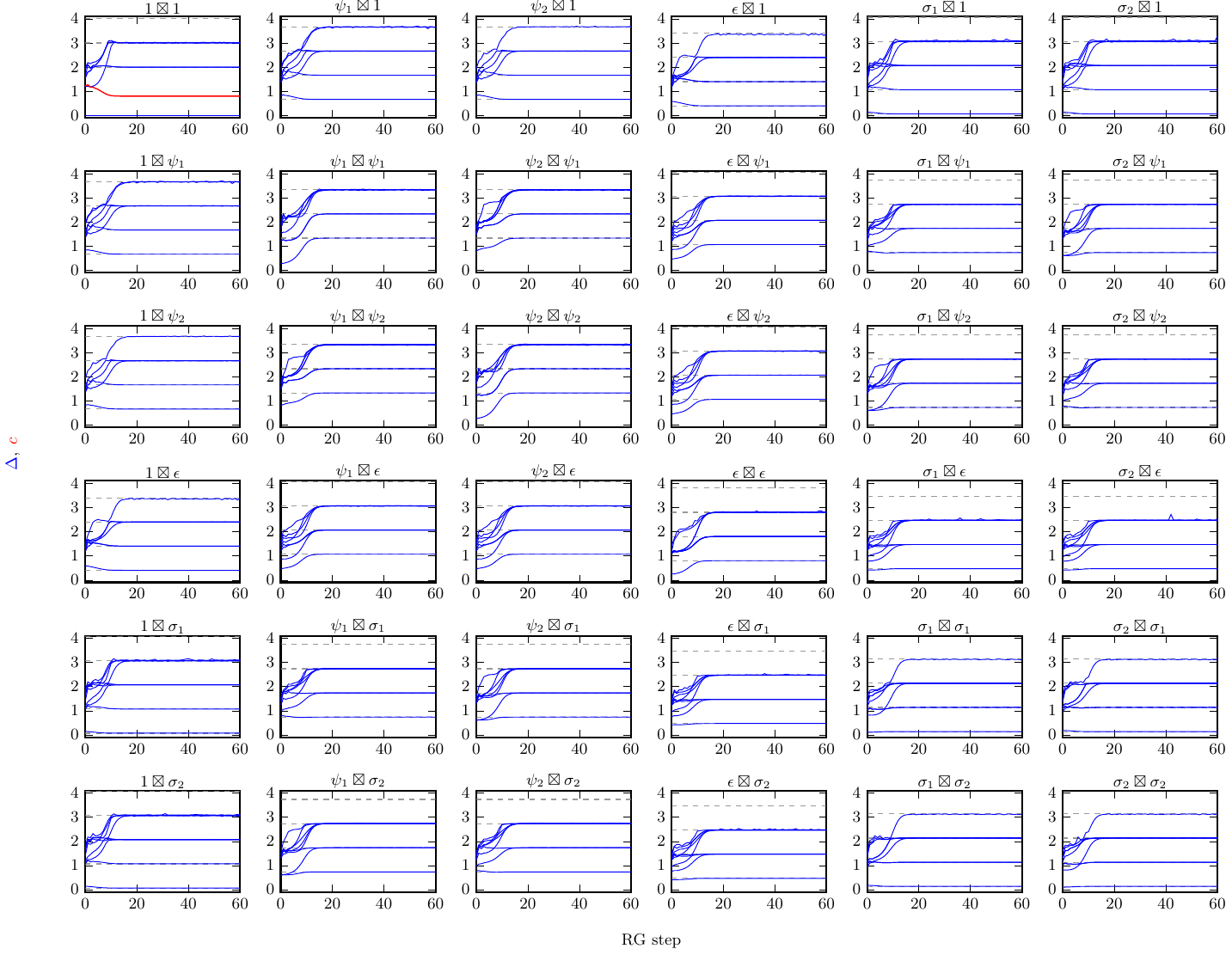}
    \caption{RG flow of the scaling dimensions in all 36 sectors of the
    doubled $\cW_3$-algebra CFT
    $\fZ(\mathsf{6}_{4/5})\cong
    \mathsf{6}_{4/5}\boxtimes\overline{\mathsf{6}_{4/5}}$.
    Each panel is labeled by its sector $x\boxtimes y$. Solid blue curves trace the numerical levels, gray dashed
    lines mark the exact scaling dimensions, and red denotes the central
    charge. Rapidly fluctuating high-lying trajectories are omitted.
    The flow is computed with $\chi_{\text{RG}}=48$ and remains stable
    through $N_{\text{RG}}=60$ RG steps.}
    \label{fig:w3-rg-spec}
\end{figure*}

\begin{figure*}[tb]
    \centering
    \includegraphics[width=\linewidth]{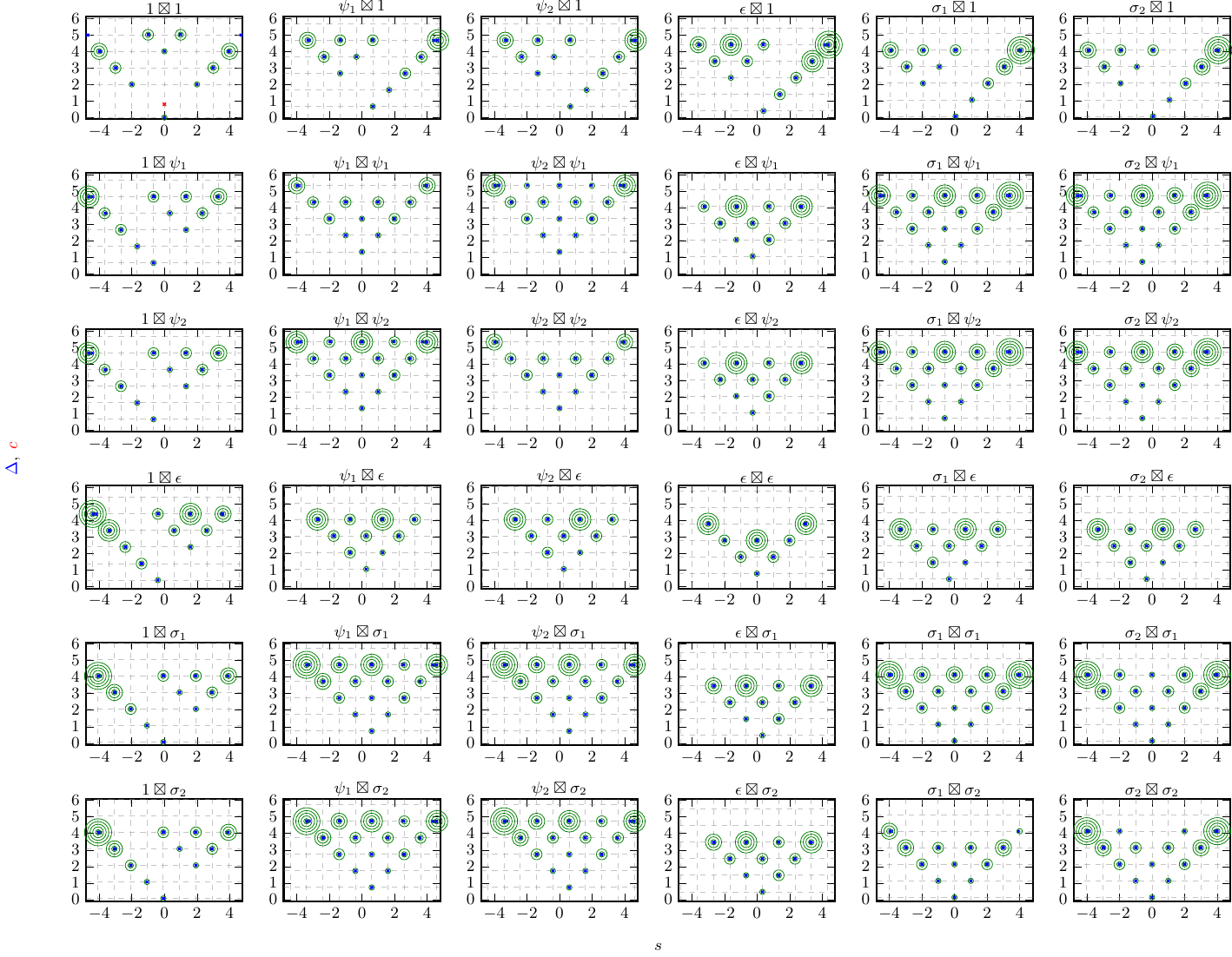}
    \caption{Spin-resolved conformal spectra in all 36 sectors of the doubled
    $\cW_3$-algebra CFT $\fZ(\mathsf{6}_{4/5})$. Each panel is labeled by its sector $x\boxtimes y$. The
    sector-dependent gray dashed grid marks the exact values of $s$ and
    $\Delta$, green concentric circles indicate degenerate levels, and red
    denotes the central charge. The spectra are extracted from the RG
    fixed-point tensor at $N_{\text{RG}}=23$ with $\chi_{\text{RG}}=48$ and
    $\chi_{\text{transfer}}=100$.}
    \label{fig:w3-spin-spec}
\end{figure*}

\begin{table*}[tb]
    \centering
    \scriptsize
    \caption{Magnitudes of the closed-string structure constants
    $|C^{\text{cl}}_{abc}|$ of the three-state Potts CFT in the $\cW_3$-primary
    basis, using $\chi_{\text{RG}}=50$ at $N_{\text{RG}}=20$. Only diagonal
    full-CFT primary fields $x\bar x$ are listed. The half-$4\times4$ triangular patch was used, with cutoff 50 in
    both subsequent TNR optimizations. Numerical values include the
    effective collar correction of Eq.~\eqref{eq:ising-collar-amplitude}.
    The ordered labels $a,b,c$ correspond to the left, up, and tilted
    directions. Asterisks mark the 15 entries used to determine $l_i$.
    The vacuum channel is normalized to unity and is excluded from the fit.}
    \label{tab:structure-constants-potts}
    \begin{minipage}[t]{0.4800\textwidth}
        \vspace{0pt}
        \centering
        \begin{tabular}{cccc}
        \toprule
        $|C^{\text{cl}}_{a,b,c}|$ & Theory & Numerical & Abs. error \\
        \midrule
        $\epsilon\bar\epsilon,\epsilon\bar\epsilon,\epsilon\bar\epsilon$ & $0$ & $0.000186$ & $0.000186$ \\
        \midrule
        $\psi_1\bar\psi_1,\sigma_1\bar\sigma_1,\sigma_1\bar\sigma_1$ & \multirow{3}{*}{$1/3$} & $0.334528$ & $0.001194$ \\
        $\sigma_1\bar\sigma_1,\psi_1\bar\psi_1,\sigma_1\bar\sigma_1$ & & $0.332169$ & $0.001164$ \\
        $\sigma_1\bar\sigma_1,\sigma_1\bar\sigma_1,\psi_1\bar\psi_1$ & & $0.331551$ & $0.001782$ \\
        \midrule
        $\psi_2\bar\psi_2,\sigma_2\bar\sigma_2,\sigma_2\bar\sigma_2$ & \multirow{3}{*}{$1/3$} & $0.334525$ & $0.001192$ \\
        $\sigma_2\bar\sigma_2,\psi_2\bar\psi_2,\sigma_2\bar\sigma_2$ & & $0.332152$ & $0.001181$ \\
        $\sigma_2\bar\sigma_2,\sigma_2\bar\sigma_2,\psi_2\bar\psi_2$ & & $0.331553$ & $0.001780$ \\
        \midrule
        $\epsilon\bar\epsilon,\sigma_1\bar\sigma_1,\sigma_2\bar\sigma_2$ & \multirow{6}{*}{$A_{\mathrm P}$} & $0.549040$ & $0.002863$ \\
        $\epsilon\bar\epsilon,\sigma_2\bar\sigma_2,\sigma_1\bar\sigma_1$ & & $0.549039$ & $0.002862$ \\
        $\sigma_1\bar\sigma_1,\epsilon\bar\epsilon,\sigma_2\bar\sigma_2$ & & $0.547606$ & $0.001428$ \\
        $\sigma_1\bar\sigma_1,\sigma_2\bar\sigma_2,\epsilon\bar\epsilon$ & & $0.545665$ & $0.000512$ \\
        $\sigma_2\bar\sigma_2,\epsilon\bar\epsilon,\sigma_1\bar\sigma_1$ & & $0.547607$ & $0.001430$ \\
        $\sigma_2\bar\sigma_2,\sigma_1\bar\sigma_1,\epsilon\bar\epsilon$ & & $0.545664$ & $0.000513$ \\
        \midrule
        $\epsilon\bar\epsilon,\psi_1\bar\psi_1,\sigma_2\bar\sigma_2$ & \multirow{6}{*}{$2/3$} & $0.666999$ & $0.000332$ \\
        $\epsilon\bar\epsilon,\sigma_2\bar\sigma_2,\psi_1\bar\psi_1$ & & $0.667369$ & $0.000703$ \\
        $\psi_1\bar\psi_1,\epsilon\bar\epsilon,\sigma_2\bar\sigma_2$ & & $0.669846$ & $0.003180$ \\
        $\psi_1\bar\psi_1,\sigma_2\bar\sigma_2,\epsilon\bar\epsilon$ & & $0.669175$ & $0.002508$ \\
        $\sigma_2\bar\sigma_2,\epsilon\bar\epsilon,\psi_1\bar\psi_1$ & & $0.668211$ & $0.001544$ \\
        $\sigma_2\bar\sigma_2,\psi_1\bar\psi_1,\epsilon\bar\epsilon$ & & $0.666803$ & $0.000137$ \\
        \midrule
        $\epsilon\bar\epsilon,\epsilon\bar\epsilon,\one\bar\one$ & \multirow{3}{*}{$1$} & $1.005334^{*}$ & $0.005334$ \\
        $\epsilon\bar\epsilon,\one\bar\one,\epsilon\bar\epsilon$ & & $1.001890^{*}$ & $0.001890$ \\
        $\one\bar\one,\epsilon\bar\epsilon,\epsilon\bar\epsilon$ & & $1.001731^{*}$ & $0.001731$ \\
        \midrule
        $\one\bar\one,\one\bar\one,\one\bar\one$ & $1$ & $1.000000$ & $0.000000$ \\
        \bottomrule
        \end{tabular}
    \end{minipage}
    \hfill
    \begin{minipage}[t]{0.4800\textwidth}
        \vspace{0pt}
        \centering
        \begin{tabular}{cccc}
        \toprule
        $|C^{\text{cl}}_{a,b,c}|$ & Theory & Numerical & Abs. error \\
        \midrule
        $\epsilon\bar\epsilon,\psi_2\bar\psi_2,\sigma_1\bar\sigma_1$ & \multirow{6}{*}{$2/3$} & $0.666973$ & $0.000306$ \\
        $\epsilon\bar\epsilon,\sigma_1\bar\sigma_1,\psi_2\bar\psi_2$ & & $0.667367$ & $0.000700$ \\
        $\psi_2\bar\psi_2,\epsilon\bar\epsilon,\sigma_1\bar\sigma_1$ & & $0.669837$ & $0.003170$ \\
        $\psi_2\bar\psi_2,\sigma_1\bar\sigma_1,\epsilon\bar\epsilon$ & & $0.669169$ & $0.002502$ \\
        $\sigma_1\bar\sigma_1,\epsilon\bar\epsilon,\psi_2\bar\psi_2$ & & $0.668216$ & $0.001549$ \\
        $\sigma_1\bar\sigma_1,\psi_2\bar\psi_2,\epsilon\bar\epsilon$ & & $0.666766$ & $0.000099$ \\
        \midrule
        $\one\bar\one,\psi_1\bar\psi_1,\psi_2\bar\psi_2$ & \multirow{6}{*}{$1$} & $0.999432^{*}$ & $0.000568$ \\
        $\one\bar\one,\psi_2\bar\psi_2,\psi_1\bar\psi_1$ & & $0.999380^{*}$ & $0.000620$ \\
        $\psi_1\bar\psi_1,\one\bar\one,\psi_2\bar\psi_2$ & & $0.999360^{*}$ & $0.000640$ \\
        $\psi_1\bar\psi_1,\psi_2\bar\psi_2,\one\bar\one$ & & $0.998147^{*}$ & $0.001853$ \\
        $\psi_2\bar\psi_2,\one\bar\one,\psi_1\bar\psi_1$ & & $0.999343^{*}$ & $0.000657$ \\
        $\psi_2\bar\psi_2,\psi_1\bar\psi_1,\one\bar\one$ & & $0.998184^{*}$ & $0.001816$ \\
        \midrule
        $\one\bar\one,\sigma_1\bar\sigma_1,\sigma_2\bar\sigma_2$ & \multirow{6}{*}{$1$} & $1.000751^{*}$ & $0.000751$ \\
        $\one\bar\one,\sigma_2\bar\sigma_2,\sigma_1\bar\sigma_1$ & & $1.000752^{*}$ & $0.000752$ \\
        $\sigma_1\bar\sigma_1,\one\bar\one,\sigma_2\bar\sigma_2$ & & $1.000825^{*}$ & $0.000825$ \\
        $\sigma_1\bar\sigma_1,\sigma_2\bar\sigma_2,\one\bar\one$ & & $1.002405^{*}$ & $0.002405$ \\
        $\sigma_2\bar\sigma_2,\one\bar\one,\sigma_1\bar\sigma_1$ & & $1.000826^{*}$ & $0.000826$ \\
        $\sigma_2\bar\sigma_2,\sigma_1\bar\sigma_1,\one\bar\one$ & & $1.002404^{*}$ & $0.002404$ \\
        \midrule
        $\sigma_1\bar\sigma_1,\sigma_1\bar\sigma_1,\sigma_1\bar\sigma_1$ & $2A_{\mathrm P}$ & $1.094652$ & $0.002297$ \\
        \midrule
        $\sigma_2\bar\sigma_2,\sigma_2\bar\sigma_2,\sigma_2\bar\sigma_2$ & $2A_{\mathrm P}$ & $1.094652$ & $0.002297$ \\
        \midrule
        $\psi_1\bar\psi_1,\psi_1\bar\psi_1,\psi_1\bar\psi_1$ & $4/3$ & $1.345729$ & $0.012395$ \\
        \midrule
        $\psi_2\bar\psi_2,\psi_2\bar\psi_2,\psi_2\bar\psi_2$ & $4/3$ & $1.345640$ & $0.012307$ \\
        \bottomrule
        \end{tabular}
    \end{minipage}
\end{table*}

The third example is the critical three-state Potts CFT, regarded as a
$c=4/5$ rational CFT with chiral algebra $\cW_3$. The center of
$\cR=\mathsf{6}_{4/5}\cong
\Fib\boxtimes\overline{\mathsf{3}_2}$~\cite{Wen_2015} is a stable symTO for
this CFT in the sense of Definition~\ref{def:stable-symto}, since the factorized
$\cW_3$ vacuum algebra has no nonidentity local operator of scaling dimension
below two.

Unlike the Ising and tricritical Ising examples, this stable
symTO is the minimal chiral-factorized symTO,
whereas the maximal symTO~\cite{chatterjee2024emergentgeneralizedsymmetrymaximal} is
$\cC_{6,5}\boxtimes\overline{\cC_{6,5}}$, as shown in
Sec.~\ref{para:potts-extension-tower}. The correspondence between
the six simple objects of $\cR$ and the $\cW_3$ primaries is given in
App.~\ref{appsec:potts-categories} and Tab.~\ref{tab:potts-mtc}, and its center again
contains 36 doubled sectors.

The same calculation can also be performed with $\cR=\cC_{6,5}$, whose
center is the maximal symTO~\cite{chatterjee2024emergentgeneralizedsymmetrymaximal} and satisfies the same stability criterion.
We do not present this case here because its 100 simple objects give too
many sectors to display.

Figures~\ref{fig:w3-rg-spec} and~\ref{fig:w3-spin-spec} show that the flow is
stable through 60 steps with $\chi_{\text{RG}}=48$ and, at
$N_{\text{RG}}=23$, resolves one doubled $\cW_3$ conformal family in each
sector, with dimensions and spins satisfying
Eq.~\eqref{eq:doubled-sector-spectrum}.

Unlike the preceding examples, the relevant chiral algebra is not the
Virasoro minimal-model VOA of $\cM(6,5)$, but its $\cW_3$ extension. The VOA
extension and the associated two-step condensation are reviewed in
Sec.~\ref{appsec:partition-vector-and-modular-covariance}. Accordingly,
we compare three-point coefficients in the $\cW_3$-primary basis of
Tab.~\ref{tab:potts-mtc}.

We extract the magnitudes of the structure constants in the $\cW_3$-primary
basis using the collar correction introduced for the Ising example in
Eq.~\eqref{eq:ising-collar-amplitude}. We use the half-$4\times4$ patch
after 20 RG steps, with cutoff 50 in both the RG flow and the two
subsequent TNR optimizations. We display only diagonal primary fields
$x\bar x$ to keep the table manageable, since listing all allowed
channels in the 36 doubled sectors would make it considerably longer.
Table~\ref{tab:structure-constants-potts} reports 45 ordered channels
involving the six diagonal primary fields. We define~\cite{McCabe_1998}
\begin{equation}
\begin{aligned}
    A_{\mathrm P}
    &=\frac{\sin^{3/2}(2\pi/5)}
    {2\sin(\pi/5)}
    \frac{\Gamma(3/5)^2}{\Gamma(2/5)\Gamma(4/5)} \\
    &\quad\times
    \left[\sin(2\pi/5)+\sin(\pi/5)\right]^{-1/2} \\
    &=0.546177618\ldots .
\end{aligned}
\label{eq:potts-structure-constant}
\end{equation}

The five nonvacuum diagonal primaries provide 15 identity channels for
the joint fit in Eq.~\eqref{eq:ising-collar-fit}. Including both members
of each conjugate pair, we obtain
\begin{equation}\label{eq:potts-collar-lengths}
\begin{aligned}
(\widehat l_1,\widehat l_2,\widehat l_3)
&=(0.02437255,\ 0.04272924,\\
&\qquad -0.01573058).
\end{aligned}
\end{equation}
We then keep these parameters fixed when extracting the remaining
structure constants.

For the 28 nonidentity channels with nonzero theoretical coefficients,
none of which enter the fit, the root-mean-square relative error
decreases from $3.836\%$ without collar corrections to $0.399\%$ with
the fitted corrections. The maximum relative error decreases from
$7.125\%$ to $0.930\%$. All 28 channels improve after the correction.
The channel involving three $\epsilon\bar\epsilon$ fields has a
vanishing theoretical coefficient and is evaluated separately. Its
corrected numerical magnitude is $1.86\times10^{-4}$.

Entries used to determine $l_i$ are marked by an asterisk. The vacuum
channel fixes the overall normalization and is excluded from the fit.

\section{The mechanism of CFT fixed point tensor: A category theory perspective}
\label{sec:cft-fixed-point-mechanism}

The examples above show that tensors initialized by the topological bootstrap
flow under $\cR$-Loop-TNR to long-lived approximate fixed points whose transfer
spectra and three-point coefficients reproduce those of the target RCFTs. This
agreement raises two distinct questions. Why is the bootstrap tensor close to
satisfying the local fixed-point equations, and why does the resulting RG flow
remain stable under repeated coarse-graining?

To address these questions, we construct ideal fixed-point tensors using
CFT three-point functions on a pair of pants. Their bond spaces contain
infinite towers of conformal descendants, reflecting the presence of gapless
modes.

The fixed-point tensors should satisfy both 2–2 and 2–0 move invariance.
In our construction, the CFT \(F\)-matrix associated with a change of conformal-block channel cancels the symmetry \(F\)-matrix associated with charge recoupling,
yielding 2–2 move invariance.
The \(\mathcal R\)-tensor contractions automatically generate a cloaking boundary condition,
yielding 2–0 move invariance in the shrinking-hole limit.

However, the hole-shrinking argument alone does not explain the numerical
stability observed in Sec.~\ref{sec:stable-symto}, including the robustness
against numerical roundoff errors. Using the killing-ring identity of UMTCs,
we further argue that the $\cR$ symmetry constraints project numerical
perturbations onto the vacuum conformal family, thereby excluding relevant
perturbations in the examples studied here.

Throughout this discussion, tensor-contraction identities are understood
up to an overall nonzero scalar independent of the external state labels.
Such factors are absorbed into the tensor normalization and divided out
in Loop-TNR.

\subsection{2--2 move invariance and the diagonal Lagrangian algebra}
\label{sec:two-two-move}

In a chiral rational CFT, sewing two three-point blocks produces a four-point
block, with the sewn leg carrying the intermediate chiral sector. The $s$- and
$t$-channel sewings give two bases of four-point blocks, related by the
$F$-matrix $F^{\mathrm{CFT}}$, which is part of the Moore--Seiberg data of the
RCFT~\cite{Moore:1988qv}.
The conformal families and their conformal blocks form a unitary modular
tensor category $\cR$. Its $F$-symbols give precisely the $F$-matrix
$F^{\cR}=F^{\mathrm{CFT}}$ relating the $s$- and $t$-channel bases above.

In Ref.~\cite{Cheng:2023kxh}, the fixed-point tensors are constructed by pairing conformal three-point functions with
boundary structure constants to ensure 2-2 move invariance. Sewing two such pairs and summing over the
intermediate labels produces a four-point function. Under a change of channel,
the boundary structure constants contribute an $F$-matrix that cancels its
conformal counterpart.

Here, we simplify this construction by directly pairing categorical fusion
channels with the corresponding conformal families. Before presenting this
construction, we first review the canonical algebra in a fusion category.
\begin{proposition}[Canonical (Longo--Rehren) Q-system~\cite{LongoRehren1995}]\label{prop:diagonal-algebra}
    Let $\cC$ be a unitary fusion category, and let $\cC^{\rev}$
    denote the category with reversed tensor product,
    $X \otimes^{\rev} Y := Y \otimes X$. Set
    \[
        A := \bigoplus_{i\in\Irr(\cC)} i^* \boxtimes i,
        \qquad D := \sqrt{\sum_{i\in\Irr(\cC)} d_i^2}.
    \]
    Let $q^A_{i^*\boxtimes i}\colon i^*\boxtimes i\to A$ be the
    orthonormal summand embeddings. For each $i,j,k\in\Irr(\cC)$, choose
    an orthonormal basis
    $\{q^{i\otimes j}_{k,\alpha}\}_{\alpha=1}^{N^{ij}_k}$ of
    $\cC(k,i\otimes j)$. Define $m\colon A\otimes A\to A$ by
\begin{align}
	&\left(q^A_{k^*\boxtimes k}\right)^\dagger\circ m\circ
    \left(q^A_{i^*\boxtimes i}\otimes q^A_{j^*\boxtimes j}\right) \nonumber\\
	&\quad := \frac{1}{D}\sqrt{\frac{d_id_j}{d_k}} \sum_{\alpha}\left(q^{i \otimes j}_{k,\alpha}\right)^*\boxtimes \left(q^{i\otimes j}_{k,\alpha}\right)^\dagger \nonumber\\
	&\quad = \frac{1}{D}\sqrt{\frac{d_id_j}{d_k}}\sum_{\alpha}
    \raisebox{-20.11682pt}{\includegraphics{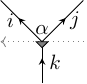}}
    \boxtimes
    \raisebox{-20.11682pt}{\includegraphics{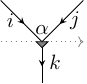}}.
\end{align}
    Here $*$ denotes the dual morphism, $\dagger$ denotes the adjoint,
    and the dotted arrows indicate the tensor-product directions. In particular,
    \[
    \left(q^{i \otimes j}_{k,\alpha}\right)^*
    :=
    \raisebox{-20.11682pt}{\includegraphics{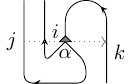}}.
    \]
    Define the unit $\iota\colon\one\boxtimes\one\to A$ by
    \[
        \iota := D\,q^A_{\one^*\boxtimes\one}.
    \]
    Then $(A,m,\iota)$ is a Q-system in $\cC^{\rev}\boxtimes\cC$,
    with normalization
    \[
        m\circ m^\dagger=\id_A,
        \qquad \iota^\dagger\circ\iota=D^2\,\id_{\one\boxtimes\one}.
    \]
\end{proposition}
\begin{proof}
    A unitary change of fusion basis acts on the dual and adjoint factors
    by mutually conjugate matrices, so the sum defining $m$ is unchanged and therefore $m$ is well-defined.

    We depict $m$ and $m^\dagger$ by black trivalent vertices, and
    $\iota$ and $\iota^\dagger$ by white univalent vertices. In the
    Q-system relations below, all strands carry $A$ and compositions
    are read from top to bottom.

    Associativity is the graphical identity
    \[
        \raisebox{-18pt}{\includegraphics{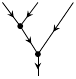}}
        =
        \raisebox{-18pt}{\includegraphics{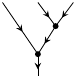}}.
    \]
    To verify it, fix the external labels $i,j,k,l$.
    Both compositions of $m$ have the same normalization factor
    $D^{-2}\sqrt{d_id_jd_k/d_l}$. Keeping this factor explicit,
    the paired recoupling gives
    \begin{align*}
        &\frac{1}{D^2}\sqrt{\frac{d_id_jd_k}{d_l}}
        \sum_{e,\alpha,\beta}
        \raisebox{-29pt}{\includegraphics{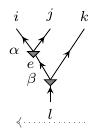}}
        \boxtimes
        \raisebox{-29pt}{\includegraphics{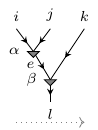}}
        \\
        &\quad =\frac{1}{D^2}\sqrt{\frac{d_id_jd_k}{d_l}}
        \sum_{f,\gamma,\delta}\sum_{g,\rho,\sigma}
        \\*
        &\qquad\times\Biggl[
        \sum_{e,\alpha,\beta}
        F^{ijk}_{l,(f,\gamma,\delta),(e,\alpha,\beta)}
        \overline{F^{ijk}_{l,(g,\rho,\sigma),(e,\alpha,\beta)}}
        \Biggr]
        \\*
        &\qquad\times
        \raisebox{-29pt}{\includegraphics{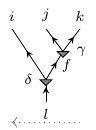}}
        \boxtimes
        \raisebox{-29pt}{\includegraphics{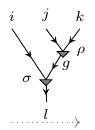}}
        \\
        &\quad =\frac{1}{D^2}\sqrt{\frac{d_id_jd_k}{d_l}}
        \sum_{f,\gamma,\delta}
        \raisebox{-29pt}{\includegraphics{figures/tikz/qsystem-f-right-dual-q.pdf}}
        \boxtimes
        \raisebox{-29pt}{\includegraphics{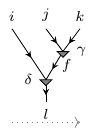}}.
    \end{align*}
    As in the multiplication diagram, reversed strand arrows and the
    dotted tensor-product direction indicate the dual factor. Duality
    preserves the recoupling coefficient, while taking the adjoint
    complex conjugates it. The last equality follows from unitarity
    \[
    \begin{aligned}
        &\sum_{e,\alpha,\beta}
        F^{ijk}_{l,(f,\gamma,\delta),(e,\alpha,\beta)}
        \overline{F^{ijk}_{l,(g,\rho,\sigma),(e,\alpha,\beta)}}\\
        &\qquad=\delta_{fg}\,\delta_{\gamma\rho}\,\delta_{\delta\sigma}.
    \end{aligned}
    \]
    Thus the intermediate object and both multiplicity labels agree
    in the two factors, proving associativity.

    When either input is the tensor unit, the coefficient in $m$ is $D^{-1}$.
    The factor $D$ in $\iota$ therefore gives both unit laws
    \[
        \raisebox{-14.42636pt}{\includegraphics{figures/tikz/frobenius-unit-insertion-left.pdf}}
        =
        \raisebox{-14.42636pt}{\includegraphics{figures/tikz/frobenius-unit-insertion-right.pdf}}
        =
        \raisebox{-14.42636pt}{\includegraphics{figures/tikz/frobenius-unit-reduced-strand.pdf}}.
    \]

    Orthogonality of the fusion bases gives, on the $k^*\boxtimes k$ summand,
    \[
        \left.m\circ m^\dagger\right|_{k^*\boxtimes k}
        =\frac{1}{D^2}\sum_{i,j}N^{ij}_k\frac{d_id_j}{d_k}
        \,\id_{k^*\boxtimes k}
        =\id_{k^*\boxtimes k}.
    \]
    The last equality uses
    $\sum_j N^{ij}_k d_j=d_i d_k$, which follows from rigidity and
    multiplicativity of quantum dimensions. The normalization of $\iota$
    follows from its definition, and $d_A=\sum_i d_i^2=D^2$.
    Graphically, these normalization conditions are
    \[
        \raisebox{-21.53955pt}{\includegraphics{figures/tikz/frobenius-bubble-removal.pdf}}
        =
        \raisebox{-21.53955pt}{\includegraphics{figures/tikz/frobenius-identity-strand.pdf}},
        \qquad
        \raisebox{-16pt}{\includegraphics{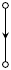}}
        =D^2.
    \]

    In a $C^*$ tensor category, unitality, associativity, and
    $m\circ m^\dagger=\id_A$ imply the Frobenius relations
    \[
    \begin{aligned}
        (m\otimes\id_A)\circ(\id_A\otimes m^\dagger)
        &=m^\dagger\circ m,\\
        (\id_A\otimes m)\circ(m^\dagger\otimes\id_A)
        &=m^\dagger\circ m.
    \end{aligned}
    \]
    The two Frobenius relations are represented by
    \[
        \raisebox{-21.53955pt}{\includegraphics{figures/tikz/frobenius-left-associated-tree.pdf}}
        =
        \raisebox{-21.53955pt}{\includegraphics{figures/tikz/frobenius-associativity-tree.pdf}}
        =
        \raisebox{-21.53955pt}{\includegraphics{figures/tikz/frobenius-right-associated-tree.pdf}}.
    \]
    Thus $(A,m,\iota)$ satisfies the Q-system relations in the normalization
    used in Sec.~\ref{sec:random-cluster-models}. In particular, $m^\dagger$
    is an $A$-bimodule section of $m$, so $A$ is separable.
\end{proof}

The underlying object can also be written as $A\simeq\otimes^\vee(\one)$,
where $\otimes^\vee\colon\cC\to\cC^{\rev}\boxtimes\cC$ is the right
adjoint of the linear functor
$\otimes\colon\cC^{\rev}\boxtimes\cC\to\cC$ sending
$X\boxtimes Y$ to $X\otimes Y$.

Below we suppress the dotted arrows and draw all string diagrams in the
original category.

\begin{proposition}[Condensability and the Lagrangian property]
\label{prop:diagonal-condensable}
    Let $\cC$ be a unitary braided fusion category with braiding $\beta$.
    Equip $\cC^{\rev}$ with the braiding
    \[
        \beta^-_{X,Y}\colon
        X\otimes^{\rev}Y=Y\otimes X
        \xrightarrow{\beta_{X,Y}^{-1}}
        X\otimes Y=Y\otimes^{\rev}X.
    \]
    Then the canonical Q-system $A$ of
    Proposition~\ref{prop:diagonal-algebra} is a condensable algebra,
    namely a connected commutative separable algebra, in
    $(\cC^{\rev},\beta^-)\boxtimes\cC$.
    Under the braided equivalence
    \[
        (\cC^{\rev},\beta^-)\cong\overline{\cC},
        \qquad \overline\beta_{X,Y}:=\beta_{Y,X}^{-1},
    \]
    it defines a condensable algebra in $\overline{\cC}\boxtimes\cC$.
    If $\cC$ is a UMTC, this algebra is Lagrangian.
\end{proposition}
\begin{proof}
    The tensor unit $\one\boxtimes\one$ occurs in $A$ with multiplicity
    one, so $A$ is connected. Separability follows from
    Proposition~\ref{prop:diagonal-algebra}.
    For commutativity, fix $i,j,k\in\Irr(\cC)$ and define the
    braiding coefficients in the chosen orthonormal fusion bases by
    \[
        \raisebox{-29pt}{\includegraphics{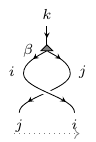}}
        =\sum_{\alpha}
        R^{ij}_{k,\alpha\beta}
        \raisebox{-22pt}{\includegraphics{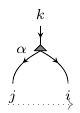}}.
    \]
    On these summands, the braiding of $A$ is
    $\beta_{i^*,j^*}^{-1}\boxtimes\beta_{i,j}$.
    The inverse braiding and braiding act on the two factors as
    \[
    \begin{aligned}
        &\left(q^{j\otimes i}_{k,\alpha}\right)^*
        \circ\beta_{i^*,j^*}^{-1}=\sum_{\gamma=1}^{N^{ij}_k}
        (R^{-1})^{ij}_{k,\gamma\alpha}
        \left(q^{i\otimes j}_{k,\gamma}\right)^*,\\[3pt]
        &\left(q^{j\otimes i}_{k,\alpha}\right)^\dagger
        \circ\beta_{i,j}=\sum_{\beta=1}^{N^{ij}_k}
        R^{ij}_{k,\alpha\beta}
        \left(q^{i\otimes j}_{k,\beta}\right)^\dagger.
    \end{aligned}
    \]
    The corresponding graphical calculation is
    \begin{align*}
        &\frac{1}{D}\sqrt{\frac{d_id_j}{d_k}}
        \sum_{\alpha=1}^{N^{ji}_k}
        \raisebox{-29pt}{\includegraphics{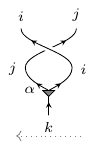}}
        \boxtimes
        \raisebox{-29pt}{\includegraphics{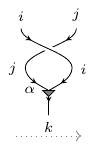}}
        \\
        &\quad=\frac{1}{D}\sqrt{\frac{d_id_j}{d_k}}
        \sum_{\gamma,\beta}
        \Biggl[\sum_{\alpha}
        (R^{-1})^{ij}_{k,\gamma\alpha}
        R^{ij}_{k,\alpha\beta}\Biggr]\\*
        &\qquad\times
        \raisebox{-22pt}{\includegraphics{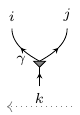}}
        \boxtimes
        \raisebox{-22pt}{\includegraphics{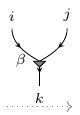}}
        \\
        &\quad=\frac{1}{D}\sqrt{\frac{d_id_j}{d_k}}
        \sum_{\beta=1}^{N^{ij}_k}
        \raisebox{-22pt}{\includegraphics{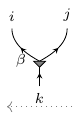}}
        \boxtimes
        \raisebox{-22pt}{\includegraphics{figures/tikz/qsystem-braid-adjoint-plain-beta.pdf}}.
    \end{align*}
    The first diagram uses the reversed strand arrows and tensor-product
    direction of the dual factor. The coefficient in brackets reduces to
    \[
        \sum_{\alpha=1}^{N^{ji}_k}
        (R^{-1})^{ij}_{k,\gamma\alpha}
        R^{ij}_{k,\alpha\beta}
        =\delta_{\gamma\beta}.
    \]
    The remaining paired sum is the corresponding component of $m$.
    Hence $m\circ\beta_{A,A}=m$, and $A$ is condensable.

    If $\cC$ is modular, $\overline{\cC}\boxtimes\cC$ is a UMTC, and
    \[
        d_A=\sum_{i\in\Irr(\cC)}d_i^2
        =D_{\cC}^2=D_{\overline{\cC}\boxtimes\cC}.
    \]
    Hence $A$ is Lagrangian.
\end{proof}

Applying Proposition~\ref{prop:diagonal-condensable} to the chiral category
$\cR$ gives the diagonal Lagrangian algebra
\begin{equation}\label{eq:diag-lagrangian-algebra}
    L_{\diagsector}
    =
    \bigoplus_{i\in\Irr(\cR)}
    i^*\boxtimes i
    \;\in\;
    \overline{\cR}\boxtimes\cR
    \cong
    \fZ(\cR).
\end{equation}
Its physical interpretation depends on whether the two factors are realized
topologically or conformally.

\paragraph{Topological interpretation---a wormhole.}

Condensing $L_{\diagsector}$ in the doubled topological order
$\overline{\cR}\boxtimes\cR\cong\fZ(\cR)$ produces a gapped boundary to
the vacuum~\cite{Kong:2013aya}. After unfolding, this boundary becomes the
identity domain wall of the topological order $\cR$, across which an anyon in $\cR$ passes without changing its
type~\cite{Huston_2023}. In the folded description, this transmission pairs
the sector $i$ in $\cR$ with $i^*$ in the oppositely oriented layer
$\overline{\cR}$. A circular boundary joining the two layers therefore forms
a \emph{wormhole} through the double.
\begin{center}
    \includegraphics[width=.62\columnwidth]{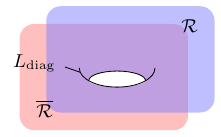}
\end{center}

In this geometric picture, shrinking the circular boundary to a point
replaces the hole by a bulk insertion labeled by
$\bigoplus_{i\in\Irr(\cR)}i^*\boxtimes i=L_{\diagsector}$.
This gives a dimensional-reduction interpretation of the diagonal algebra
within the framework of factorization homology~\cite{Ai_2017}.
The condensation also admits a microscopic realization in an extension of
the string-net model~\cite{christian2023latticemodelcondensationlevinwen}.

\paragraph{Conformal interpretation---the charge-conjugation full CFT.}

For a strongly rational unitary VOA $\scrV$ with $\Mod_\scrV\cong\cR$,
the diagonal algebra $L_{\diagsector}$ defines the charge-conjugation full
CFT. It pairs each chiral conformal family $\scrV_i$ with the antichiral
family $\overline{\scrV_{i^*}}$, giving the graded state space
\begin{equation}\label{eq:diagonal-full-field-space}
    \scrL_{\diagsector}
    =
    \bigoplus_{i\in\Irr(\cR)}
    \overline{\scrV_{i^*}}\otimes \scrV_i .
\end{equation}
The multiplication of $L_{\diagsector}$, realized by intertwining
operators in the vertex tensor category $\Mod_\scrV$, equips this space with
its genus-zero OPE and hence with the structure of a full field
algebra~\cite{Huang_2007,Kong_2007}. The associativity of $L_{\diagsector}$
ensures that the four-point functions obtained by sewing two pairs of pants
are invariant under changes of channel. The braided commutativity of
$L_{\diagsector}$ ensures the commutativity of the full-field OPE.

Under the standard rationality and nondegeneracy assumptions, the
Lagrangian condition of the algebra is equivalent to genus-one modular
invariance~\cite{HuangKong2010,KongRunkel2009}. The resulting torus
partition function is
\[
    Z_{\mathrm{cc}}(\tau,\bar\tau)
    =
    \sum_{i\in\Irr(\cR)}
    \chi_i(\tau)\,\overline{\chi_{i^*}(\tau)}.
\]
This is the charge-conjugation modular invariant.

\paragraph{Mixed interpretation---a gapless hole.}

In our construction, $i^*$ in each summand $i^*\boxtimes i$ of
$L_{\diagsector}$ is regarded as the antichiral conformal family
$\overline{\scrV_{i^*}}$, while $i$ remains an anyon in $\cR$.
We denote the resulting bond degrees of freedom by
\begin{equation}\label{eq:mixed-bond-object}
    \mathfrak L
    =
    \bigoplus_{i\in\Irr(\cR)}
    \overline{\scrV_{i^*}}\odot i.
\end{equation}
Here $\odot$ indicates that the conformal family supplies the multiplicity
space of the anyon $i$. The infinite tower of local states in each anyon
sector reflects the gapless modes in that sector. In
Sec.~\ref{sec:geometric-picture}, we interpret $\mathfrak L$ as a gapless
hole in a chiral topological order.

Let $\phi_{i,\Lambda}$ denote the field corresponding to
a state $\Lambda$ in the conformal family paired with sector $i$.
For a chosen pair-of-pants geometry, we define the rank-three tensor
$m_{\text{pants}}$ by specifying its components using the orthonormal
embeddings of Subsec.~\ref{sec:r-tensor-basis}:
\begin{multline}\label{eq:mixed-pants-map}
    \left(q^{\mathfrak{L}}_{k,\Theta}\right)^\dagger \circ m_{\text{pants}}\circ \left(q^{\mathfrak{L}}_{i,\Lambda} \otimes q^{\mathfrak{L}}_{j,\Xi}\right) \\
=\frac{1}{D}\sum_{\alpha}
\raisebox{-45.3187pt}{\includegraphics{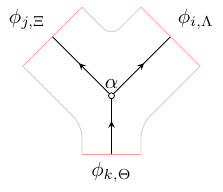}}
\xt
\raisebox{-28.65274pt}{\includegraphics{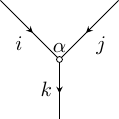}}
\end{multline}
Each summand pairs an antichiral three-point conformal block with a basis
morphism in the corresponding fusion space in $\cR$.

For a fixed $\alpha$, the first diagram
represents a three-point function of the
fields $\phi_{i,\Lambda}$, $\phi_{j,\Xi}$, and $\phi_{k,\Theta}$ on a flattened pair
of pants, with the three-point conformal block labeled by $\alpha$.
The gray lines mark the folds obtained by flattening, while the red lines
mark equal-time spatial slices on which the states corresponding to
$\phi_{i,\Lambda}$, $\phi_{j,\Xi}$, and $\phi_{k,\Theta}$ are specified.
Equivalently, the amplitude on this geometry can be interpreted as an
open-string amplitude with the trivial Cardy boundary condition on the
open-string boundaries.
The lines labeled by $i,j,k$ become Verlinde lines in the bulk of the doubled
CFT.
The second diagram represents the element labeled by $\alpha$ of an
isotropic basis in $\cR(i \otimes j, k)$.

For the square tensor network, we use the flat-pants geometry shown below,
whose central region is an isosceles right triangle. The open-string boundaries
carry the trivial Cardy condition and have infinitesimal length.
The conformal map and induced local coordinates for this geometry were computed in
Ref.~\cite{Cheng:2023kxh} and are reviewed in
Appendix~\ref{app:geometric-factors}.
Exactly sharp corners give rise to convergence issues and require
regularization~\cite{Brehm:2022,BrehmRunkel2024}. We discuss this further
in Sec.~\ref{sec:two-zero-move}.
\[
\raisebox{-36.98572pt}{\includegraphics{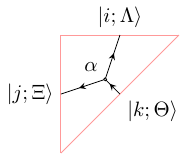}}.
\]

Sewing two three-point amplitudes, or equivalently contracting two rank-three
tensors along an intermediate leg, produces a four-point amplitude. The two
choices of diagonal in the square give two contraction orders. If the result
is independent of the choice of diagonal, the rank-three tensor satisfies
the 2--2 move invariance condition.
\[
\raisebox{-34.5434pt}{\includegraphics{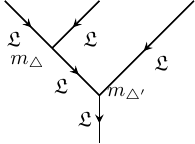}}
=
\raisebox{-34.5434pt}{\includegraphics{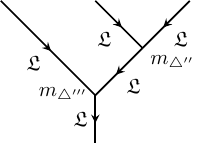}}
\]

Algebraically, this is the same dual-basis cancellation used to prove the
associativity of the Longo--Rehren algebra in
Proposition~\ref{prop:diagonal-algebra}. Recoupling transforms the conformal
block and its paired fusion-space basis morphism by mutually inverse
matrices, leaving their product invariant.

Ref.~\cite{Cheng:2023kxh} uses the same dual-basis cancellation. The difference
lies in the tensor realization. That construction passes to a tensor complex
whose plaquettes carry
module-category labels, which are summed over upon contraction. Here the
conformal blocks are paired directly with basis morphisms in the fusion spaces
of $\cR$, so no plaquette indices are introduced. The relation between the two descriptions
is explained in Sec.~\ref{sec:tensor-complex}.

Projecting each conformal family onto the primary state selected by the
bootstrap reduces $\mathfrak L$ to the finite bond object
$V=\bigoplus_{i\in\Irr(\cR)}i$ and $m_{\text{pants}}$ to the pants-gauge
topological-bootstrap tensor in Eq.~\eqref{eq:initial-tensor-bootstrap}.
For the selected primary fields $\phi_i$ of conformal weights $h_i$,
\[
    \raisebox{-30pt}{\includegraphics{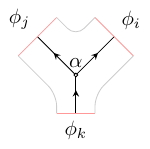}}
    =C_{ijk;\alpha}\,G_{\triangle}(h_i,h_j,h_k).
\]
Here $C_{ijk;\alpha}$ is the chiral three-point coefficient in the chosen
fusion basis, and $G_{\triangle}$ is the flat-pants geometric factor in
Eq.~\eqref{eq:flat-pants-gauge-app}. Thus
$\cN_{ijk;\alpha}=C_{ijk;\alpha}$, with the overall factor $1/D$
absorbed into the tensor normalization.
After this projection, the exact 2--2 relation becomes an approximate
fixed-point equation.

\subsection{2--0 move invariance and the cloaking boundary condition}
\label{sec:two-zero-move}

For the zero-collar pants geometry, higher-descendant amplitudes can decay
slowly, and sewing sums can require regularization. Brehm and Runkel
regularize the geometry by excising neighborhoods of the triangle vertices,
producing triangular pieces with clipped corners~\cite{Brehm:2022,BrehmRunkel2024}.
At a finite hole radius, high-energy contributions to amplitudes are
suppressed. In our geometric description, this suppression can be modeled
approximately by
attaching finite collars to the triangle edges, whose Euclidean evolution
damps higher-energy states.

However, when several such pieces are sewn together, the excised corners form holes
in the resulting world sheet. A hole with an elementary Cardy boundary
condition can contribute relevant bulk fields through its boundary-state
expansion. These contributions can destabilize the original CFT and drive
the flow away from its fixed point.

Brehm and Runkel introduced the cloaking boundary
condition~\cite{Brehm:2022,BrehmRunkel2024}, and Cheng et al.\ used a weighted
sum over boundary conditions in their fixed-point tensor
construction~\cite{Cheng:2023kxh}. For the full Cardy defect category, this
weighted sum is proportional to the vacuum Ishibashi state, whose subleading
contributions are irrelevant. Therefore, at IR fixed points, the hole shrinks
to a vacuum insertion, and the 2--0 move holds~\cite{Cheng:2023kxh}.

We now show how the categorical contraction implements this projection and yields the 2--0 move in the shrinking-hole limit.
Note that the 4--1 move discussed in Ref.~\cite{Cheng:2023kxh} is a specific realization of the 2--0 move,
related by a sequence of 2--2 moves.
The following argument does not require a specific choice of collar or
pants geometry.

Let $m_{\text{pants}}$ and $m_{\text{pants}'}$ denote the corresponding
collar-regularized pants maps, with the adjoint represented by reflection.
Contracting two rank-three $\cR$-tensors produces a tensor $T$
\[
    T:= m_{\text{pants}} \circ m_{\text{pants}'}^\dagger =
\raisebox{-43.0791pt}{\includegraphics{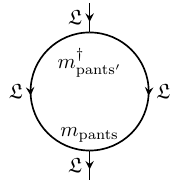}}.
\]
The tensor $T$ factorizes as
\begin{multline}
    \left(q^{\mathfrak{L}}_{k,\Xi}\right)^\dagger\circ T \circ q^{\mathfrak{L}}_{k,\Lambda} \\
	= \frac{1}{D^2} \sum_{ij}\sum_{\mu\nu} 
    \raisebox{-53.99657pt}{\includegraphics{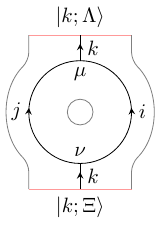}}
        \;\times\; 
    \raisebox{-37.74738pt}{\includegraphics{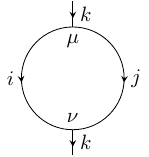}}   .
\end{multline}
In contrast to the purely algebraic derivation based on separability, the
gapless hole obstructs bubble cancellation in the CFT part.
Instead, we perform the bubble move of the symmetry part first
\[
\frac{1}{D^2}\sum_{ij,\mu} 
    \raisebox{-59.54506pt}{\includegraphics{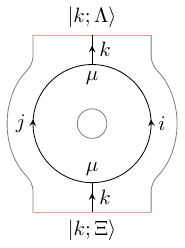}} 
    \times
    \sqrt{\frac{d_id_j}{d_k}} \id_k.
\]

We then use the corresponding graphical calculation from the
string-net model~\cite{Levin_2005}:
\begin{align*}
    &\sum_{ij,\mu}\sqrt{\frac{d_id_j}{d_k}}
    \raisebox{-42.8791pt}{\includegraphics{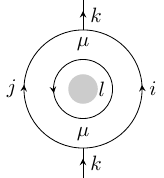}}\\
    &= 
     \sum_{im}
    \sum_{j,\mu\nu} \sqrt{\frac{d_id_m}{d_kd_l}}
    \raisebox{-42.8791pt}{\includegraphics{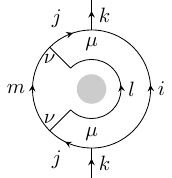}}\\
    &=
    \sum_{im}
    \sum_{n,\rho\sigma} \sqrt{\frac{d_id_m}{d_kd_l}}
    \raisebox{-42.8791pt}{\includegraphics{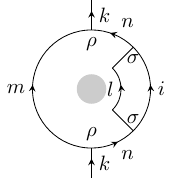}}\\
    &=
    \sum_{mn,\rho} \sqrt{\frac{d_m}{d_kd_n}}\sum_{i}N^{li}_{n}d_i
    \raisebox{-42.8791pt}{\includegraphics{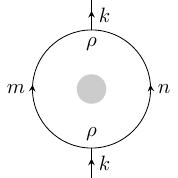}}\\
    &= d_l
    \sum_{mn,\rho} \sqrt{\frac{d_md_n}{d_k}}
    \raisebox{-42.8791pt}{\includegraphics{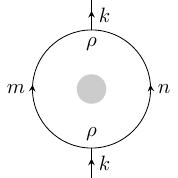}}.
\end{align*}
This identity implies that, under the sum, an $l$-labeled loop surrounding
the shaded region evaluates to a factor of $d_l$.

Let $O_l$ denote the insertion of an $l$-labeled TDL around the hole,
\[
O_l := \raisebox{-15.42638pt}{\includegraphics{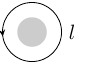}},
\]
where the shaded disk denotes the hole. Acting on the weighted sum above,
the projector $\frac{1}{D^2}\sum_l d_l O_l$ multiplies the amplitude by
$\frac{1}{D^2}\sum_l d_l^2=1$. It can therefore be inserted without changing
this amplitude. Bringing the $l$-loop to the hole boundary changes the
trivial Cardy condition to the $l$-Cardy condition.
\begin{align*}
    &\frac{1}{D^2}\sum_{ij,\mu} 
    \raisebox{-59.54506pt}{\includegraphics{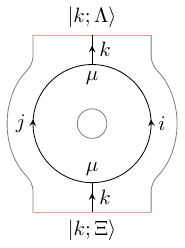}} 
    \times
    \sqrt{\frac{d_id_j}{d_k}} \id_k\\
    &=
    \frac{1}{D^2}\sum_{ij,\mu} \frac{1}{D^2}\sum_{l}d_l
    \raisebox{-59.54506pt}{\includegraphics{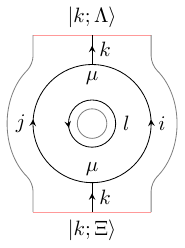}} 
    \times
    \sqrt{\frac{d_id_j}{d_k}} \id_k\\
    &=
    \frac{1}{D^2}\sum_{l}d_l
    \raisebox{-59.54506pt}{\includegraphics{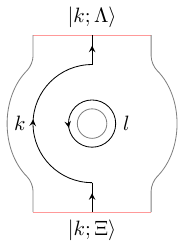}} 
    \times \id_k.
\end{align*}

The weighted projector thus cloaks the hole~\cite{Brehm:2022}.
After dividing out the overall factor $R^{-c/6}$, the resulting
vacuum-normalized boundary state is the cloaking boundary state
of Brehm and Runkel~\cite{Brehm:2022}:
\begin{equation}
    |\psi, R\rangle = R^{L_0 + \bar{L}_0}\,|0\rangle\!\rangle,
\end{equation}
where $|0\rangle\!\rangle$ is the vacuum Ishibashi state, i.e.\ a specific
superposition of Cardy states $|l\rangle$ weighted by the modular $S$-matrix.
$|0\rangle\!\rangle \propto \sum_l d_l|l\rangle$. Here $R$ is the radius of the hole.
The key property of the cloaking boundary condition is that it makes the hole
\emph{transparent} to all topological line defects of the CFT \cite{Brehm:2022}.

Throughout this argument, tensors are understood up to an overall nonzero
scalar. In this sense, the hole shrinks to a vacuum insertion as $R\to0$.
Only the vacuum $|0,0\rangle$, whose scaling dimension is $\Delta = 0$, survives.
For minimal models, the leading correction is the stress tensor $T\bar{T}$ at $\Delta = 4$,
which is highly irrelevant.
Hence the hole collapses and the 2--0 move holds.

\subsection{The killing ring}

Although the preceding hole-cancellation argument explains the 2--0 move
in the shrinking-hole limit, it does not explain the general numerical
stability of the Loop-TNR flow at
small cutoff dimensions, nor its apparent insensitivity to machine-precision
errors.

Before discussing this stability, we recall the killing-ring identity for
UMTCs. The basic identity appears in
Ref.~\cite[Corollary~3.1.11]{bakalov2001lectures}. The graphical formulation
and its extension to general objects used below are given in
Ref.~\cite[Lemma~4.5 and Corollary~4.6]{Hardiman_2019_graphical}.
\begin{proposition}[Killing ring]
    Let $\cC$ be a UMTC and let $X$ be an object in $\cC$.  The following
    identity holds:
    \[
    \frac{1}{D^2}\sum_{i\in\Irr\cC} d_i
    \raisebox{-29.90274pt}{\includegraphics{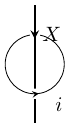}}
    =
    \sum_{\alpha}\quad
    \raisebox{-28.85274pt}{\includegraphics{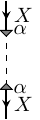}}.
    \]
    Here $\{q_\alpha: \one \rightarrow X\}$ is an orthonormal basis of $\cC(\one, X)$.
\end{proposition}

The killing-ring property states that the operator
$\frac{1}{D^2}\sum_l d_l O_l$ acts as a projector onto the tensor-unit
sector. It annihilates all nontrivial simple-object sectors while leaving the
tensor-unit sector invariant.

The cloaking-defect construction of Brehm and Runkel, which corresponds to
the killing ring here, also applies to regions containing general field insertions~\cite{Brehm:2022,BrehmRunkel2024}.
Here we show that, for a sufficiently large fusion-category symmetry
$\cR$, the corresponding projection is automatically implemented by the
contraction of $\cR$-tensors.

Now suppose that the level truncation of $m_{\text{pants}}$, together with other
perturbations induced by numerical roundoff and algorithmic approximations, can
be effectively represented by a deformed pair of pants with additional field
insertions:
\[
\frac{1}{D}\sum_{\alpha}
\raisebox{-45.3187pt}{\includegraphics{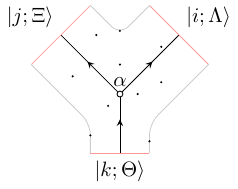}}
\xt
\raisebox{-28.65274pt}{\includegraphics{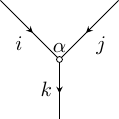}}.
\]

Here each point represents an unknown field insertion induced by algorithmic
and numerical perturbations.  The anyon lines remain topological away from
these insertion points, but they cannot, in general, be freely deformed across
them, since the inserted fields need not be transparent to the anyon lines.

Contracting two such pairs of pants produces a hole together with many unknown
field insertions:
\[
\frac{1}{D^2} \sum_{ij}\sum_{\mu\nu}
    \raisebox{-53.99657pt}{\includegraphics{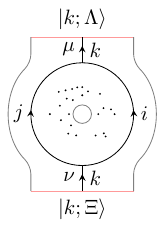}}
        \;\times\;
    \raisebox{-37.74738pt}{\includegraphics{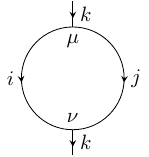}}.
\]
We assume that the hole and all additional insertions can be enclosed by a
common contour, with the surrounding fusion network lying in an annular
region containing no insertions.
This assumption is not restrictive, since insertions elsewhere can be
handled using nearby plaquettes of the $\cR$-tensor network.
The recouplings in the preceding subsection
can then be performed in this annular region without moving a TDL across
an insertion. This calculation shows that the weighted ring can be inserted
without changing the summed amplitude. Here the ring surrounds the hole
and the additional insertions together, giving
\[
\frac{1}{D^2} \sum_{ij}\sum_{\mu}\sum_{l} \frac{1}{D^2}d_l
    \raisebox{-53.99657pt}{\includegraphics{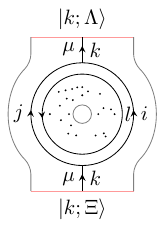}}
    \times
    \sqrt{\frac{d_id_j}{d_k}}\id_k.
\]

Under the folding/doubling interpretation of the open channel, the bulk field
insertions in the doubled CFT are represented by anyon worldlines between
the two chiral boundary layers. The ring links all these worldlines and projects
their total fusion channel onto the tensor unit. When the enclosed region
is expanded as an effective local insertion, this channel corresponds to
the vacuum conformal family in this interpretation.

The following discussion is restricted to unitary Virasoro minimal models or other VOAs without level-1 states.
More general CFTs lie outside the scope of this work. After normalizing the
vacuum coefficient to one, the effective insertion therefore has the form
\[
    |\psi\ra
=
\left(\one+a_1L_{-2}+a_2\overline{L}_{-2}
+a_3L_{-2}\overline{L}_{-2}
+\cdots\right)|0,0\ra.
\]

These coefficients are further constrained if the tensor has spatial $D_4$
symmetry, generated by $\pi/2$ rotations and reflections of its legs.
Under a $\pi/2$ rotation, an operator of conformal spin $s$ acquires a phase
$e^{\mathrm{i}\pi s/2}$. Exact spatial $D_4$ symmetry therefore excludes the
spin-$\pm2$ components $T$ and $\overline{T}$, giving $a_1=a_2=0$.
For an approximately $D_4$-symmetric Loop-TNR unit cell, these coefficients
can remain small but nonzero, with reflection symmetry and reality
conditions relating them. Schematically,
\[
|a_1|,|a_2| \approx 0,
\qquad
a_2\approx a_1^*.
\]
For a translation-invariant tensor network, the residual stress-tensor
terms describe a uniform deformation of the effective geometry. In
particular, changes in aspect ratio and shear modify the torus modular
parameter, so the tensor may represent a parallelogram rather than square
geometry. Such a deformation preserves the gaplessness of the CFT.

Within the Virasoro vacuum family, $D_4$ symmetry allows the scalar descendant
$L_{-2}\overline{L}_{-2}|0,0\ra$, corresponding to $T\overline{T}$, as well as
the spin-$\pm4$ descendants $L_{-4}|0,0\ra$, $\overline{L}_{-4}|0,0\ra$,
$L_{-2}^2|0,0\ra$, and $\overline{L}_{-2}^2|0,0\ra$.  All of these operators have
scaling dimension four and are therefore irrelevant near the CFT fixed point.
This provides a mechanism for the numerical stability of the gapless RG flow.

\subsection{Relation to the geometric picture}
\label{sec:geometric-picture}

The CFT-input tensor described above admits a geometric interpretation
through the symTO correspondence
\cite{Ji_2020,Kong_2020_phy,Kong_2020_math}. We interpret its bond degrees
of freedom as a chiral gapless hole and its three-point amplitudes as
pair-of-pants junctions between such holes.

Consider a chiral $\cR\cong\Mod_\scrV$ topological order with a hole whose
boundary supports the chiral gapless theory described by the VOA $\scrV$.
We denote this gapless hole by $\mathsf{hole}(\scrV)$ and ask what
$\cR$-anyon content it carries.
\[
    \raisebox{-21.33955pt}{\includegraphics{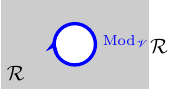}}.
\]

To determine this anyon content, place $\mathsf{hole}(\scrV)$ on a sphere
and insert a bulk anyon $i^*$.  The state space of the resulting
configuration is the antichiral conformal family
$\overline{\scrV_{i^*}}$, which therefore appears as the multiplicity
space of the sector $i$.  This gives
\[
    \mathsf{hole}(\scrV)
    \simeq
    \bigoplus_{i\in\Irr(\cR)}
    \overline{\scrV_{i^*}}\odot i.
\]
The infinite-dimensional multiplicity spaces encode the gapless boundary
modes. This is the geometric realization of the bond degrees of freedom
$\mathfrak L$ in Eq.~\eqref{eq:mixed-bond-object}.
We view $\mathsf{hole}(\scrV)$ as a vertex operator algebra internal to
the braided fusion category $\cR$. The mathematical details will be
presented in a companion paper.

Folding the chiral $\cR$-topological order produces the doubled order
$\fZ(\cR)\cong\overline{\cR}\boxtimes\cR$.  Under this folding,
$\mathsf{hole}(\scrV)$ becomes a gapless slit along the fold.  We
conjecture that condensing this slit opens the fold.

In spacetime, this
condensation proliferates the histories of $\mathsf{hole}(\scrV)$ into a
network whose junctions are given by the pair-of-pants operations, precisely
reproducing the fixed-point gapless tensor-network construction
above~\cite{kong2025highercondensationtheory}.
The two sides of the opened fold carry the chiral and antichiral theories
$\Mod_\scrV$ and $\overline{\Mod_\scrV}$.  Their stacking realizes the
doubled target CFT.

The numerical transfer spectra in
Figs.~\ref{fig:stable-ising-spin-spec},
\ref{fig:triising-spin-spec}, and~\ref{fig:w3-spin-spec} are consistent
with this interpretation.  They resolve the sectors $x\boxtimes y$ with
$\Delta_{x\boxtimes y}=h_x+h_y+n+\bar n$ and
$s_{x\boxtimes y}=h_x-h_y+n-\bar n$, where $n,\bar n$ are the descendant levels,
as summarized in
Eq.~\eqref{eq:doubled-sector-spectrum}, exactly as expected for a stack
of two oppositely oriented chiral boundaries.
\begin{equation}
    \raisebox{-22.22256pt}{\includegraphics{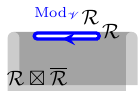}}\xrightarrow{\text{condensation}}
\raisebox{-23.07251pt}{\includegraphics{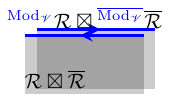}}
\end{equation}

The picture above suggests that the fixed-point $\cR$-tensor may define
a gapless analogue of a Lagrangian algebra.  More broadly, it points
toward a gapless anyon-condensation theory in which
$\mathsf{hole}(\scrV)$, its propagation, and its pair-of-pants junctions
constitute the elementary condensation data.  Such a theory should
provide a unified description of arbitrary admissible gapless boundaries
and domain walls, together with their condensation, fusion, and junctions.
We leave its systematic formulation to future work.

\subsection{Tensor-complex realization}
\label{sec:tensor-complex}

The tensor-complex construction of Ref.~\cite{Cheng:2023kxh} gives a
(gravitational) anomaly-free realization of a CFT whose torus partition
function is modular invariant rather than modular covariant.  It can be
obtained from an $\cR$-tensor network by applying a generalized reference
frame reviewed in Sec.~\ref{appsec:changing-symmetry}, namely a map between
the corresponding theory spaces induced by a monoidal functor
\[
F:\cR\longrightarrow\Fun(\cM,\cM),
\]
where $\cM$ is a semisimple $\cR$-module category.  Its target is a
Morita-trivial matrix category whose simple objects are labeled by
ordered pairs $(m,n)\in\Irr(\cM)\times\Irr(\cM)$.

When $\cM$ is semisimple, the multi-fusion category $\Fun(\cM,\cM)$ is
equivalent to a 2-matrix algebra, the categorification of a matrix algebra.
{\setlength{\arraycolsep}{2pt}
\[
\Fun(\cM, \cM)\cong \Mat_n(\Hilb) := \begin{pmatrix}
    \Hilb & \Hilb & \cdots & \Hilb \\
    \Hilb & \Hilb & \cdots & \Hilb \\
    \vdots & \vdots & \ddots & \vdots\\
    \Hilb & \Hilb & \cdots & \Hilb \\
\end{pmatrix}_{n\xt n},
\]
}
where $n$ is the number of simple objects in $\cM$.
An object $V$ in $\Mat_n(\Hilb)$ is a matrix of vector spaces:
\[
V := \begin{pmatrix}
    V_{11} & V_{12} & \cdots & V_{1n} \\
    V_{21} & V_{22} & \cdots & V_{2n} \\
    \vdots & \vdots & \ddots & \vdots\\
    V_{n1} & V_{n2} & \cdots & V_{nn} \\
\end{pmatrix}
\]
and a morphism $f:V\rightarrow W$ is a matrix of linear maps:
\[
f :=
\begin{pmatrix}
    f_{11} & f_{12} & \cdots & f_{1n} \\
    f_{21} & f_{22} & \cdots & f_{2n} \\
    \vdots & \vdots & \ddots & \vdots\\
    f_{n1} & f_{n2} & \cdots & f_{nn} \\
\end{pmatrix}.
\]
The composition of $f:V\rightarrow W$ and $g:W\rightarrow X$ is defined elementwise:
\[
(g\circ f)_{ij} := g_{ij}\circ f_{ij}.
\]
The dagger is defined similarly by $(f^\dagger)_{ij} := (f_{ij})^\dagger$.
The tensor product $V\otimes W$ is the matrix-multiplication version of the tensor product:
\[
(V\otimes W)_{ij} := \bigoplus_{k = 1}^{n}V_{ik}\otimes W_{kj}.
\]
The same rule applies to morphisms:
\[
(f\otimes g)_{ij} := \bigoplus_{k = 1}^{n}f_{ik}\otimes g_{kj}.
\]
The tensor product is natural:
\[
(f\otimes g) \circ (f'\otimes g') = (f\circ f')\otimes (g\circ g').
\]
Therefore, a tensor complex is nothing but a morphism in $\Mat_n(\Hilb)$.

We now give an example to see the connection more explicitly. Consider a morphism
\[
f : V \otimes W \rightarrow X
\]
in $\Mat_n(\Hilb)$. By definition, $f$ is a collection of linear maps
\[
f_{ikj} : V_{ik} \otimes W_{kj} \rightarrow X_{ij},\quad i,j,k = 1,\cdots, n.
\]
For fixed $i,j,k$, the three regions between the legs carry the
plaquette labels, while each leg carries the Hilbert space specified by
its two neighboring labels,
\[
f_{ikj}=\raisebox{-28pt}{\includegraphics{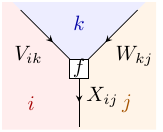}}.
\]
For example, composing $f$ with $f^\dagger$ contracts the $V$ and $W$
legs and encloses a plaquette and become a matrix in $\Mat_n(\Hilb)$.
By definition of the composition rule of $\Mat_n(\Hilb)$,
for fixed $i$ and $j$, the composition gives the $(i,j)$ entry of the resulting matrix by
\[
\begin{aligned}
    (f\circ f^\dagger)_{ij}
    &=\sum_{k=1}^{n}
    \raisebox{-40pt}{\includegraphics{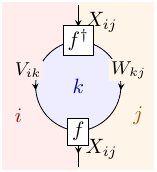}}\\
    &=\sum_{k=1}^{n}f_{ikj}\circ f_{ikj}^\dagger.
\end{aligned}
\]
For each fixed $k$, the tensor indices on $V_{ik}$ and $W_{kj}$
are contracted in the usual way. The additional sum over $k$ is the
sum over the internal plaquette label. In a larger network, each internal
plaquette carries one such label, shared by all tensors along its boundary
and summed independently of the other plaquette labels.

The module action
\[
    \triangleright: \cR \xt \cM \rightarrow \cM
\]
comes with a module associator, which for $V,W\in\cR$ and $q\in\cM$ is
\[
\alpha^{\cM}_{V,W,q}:(V\otimes W)\triangleright q
\longrightarrow V\triangleright(W\triangleright q).
\]
In orthonormal embedding-tree bases, the associator defines the module
$F$-symbols by
\[
\vcenter{\hbox{\includegraphics[width=.98\columnwidth]{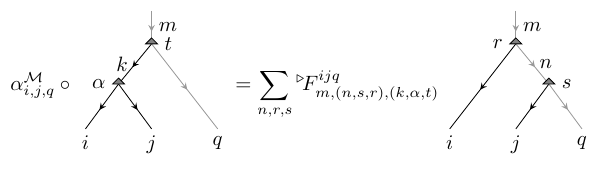}}}.
\]
Here $i,j,k\in\Irr(\cR)$ and $m,n,q\in\Irr(\cM)$ are simple objects.
The index $\alpha$ labels the embedding $q^{i\otimes j}_{k,\alpha}$, while
$r,s,t$ label the module-action embeddings
$q^{i\triangleright n}_{m,r}:m\to i\triangleright n$,
$q^{j\triangleright q}_{n,s}:n\to j\triangleright q$, and
$q^{k\triangleright q}_{m,t}:m\to k\triangleright q$, respectively.
Gray lines denote module objects, and black lines denote objects of $\cR$.
The index order follows the ordinary $F$-symbol convention of
Sec.~\ref{sec:r-tensor-presentation}.

We now describe the image of an $\cR$-tensor under the monoidal functor $F$.
For $X\in\cR$, the matrix entries of $F(X)$ are
\[
F(X)_{mn} = \cM(m, X \triangleright n),
\qquad m,n\in\Irr(\cM).
\]
A morphism $f:X\rightarrow Y$ in $\cR$ is mapped to $F(f)$, whose entries
are the linear maps
\begin{gather*}
    F(f)_{mn}: F(X)_{mn} \rightarrow F(Y)_{mn},\\
    g_{mn} \mapsto F(f)_{mn}(g_{mn}) := (f \triangleright \id_n) \circ g_{mn}.
\end{gather*}

For a rank-3 $\cR$-tensor $f:V\otimes W\rightarrow X$, the induced map
from $F(V)\otimes F(W)$ to $F(X)$ has components
\[
\bigoplus_{n\in\Irr(\cM)}F(V)_{mn}\otimes F(W)_{nq}
\longrightarrow F(X)_{mq}.
\]
For $g\in F(V)_{mn}$ and $h\in F(W)_{nq}$, it is defined by
\[
\begin{aligned}
g\otimes h\longmapsto{}&
(f\triangleright\id_q)\circ\left(\alpha^{\cM}_{V,W,q}\right)^{-1}\\
&\circ(\id_V\triangleright h)\circ g.
\end{aligned}
\]
The inverse module associator combines the successive actions of $V$ and
$W$ before applying $f\triangleright\id_q$.

Applying the map to a rank-3 $\cR$-tensor yields the tensor-complex construction of Ref.~\cite{Gu_2009_PEPS} for PEPS representations of string-net ground states:
\[
\raisebox{-28.65274pt}{\includegraphics{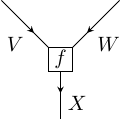}}\mapsto
\bigoplus_{m,n,q\in\Irr(\cM)}
\raisebox{-38.54242pt}{\includegraphics{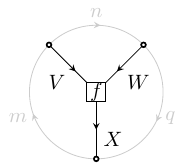}}.
\]
The gray lines carry the simple objects $m,n,q\in\cM$, which label the three
plaquettes adjacent to the tensor. For each fixed $m,n,q$, the diagram
represents a linear map from $F(V)_{mn}\otimes F(W)_{nq}$ to $F(X)_{mq}$.
The white vertices indicate the module-action spaces on these three legs,
whose state indices remain open.
The direct sum collects these maps for all choices of the plaquette labels.

To express this map in components, we need to choose bases in the three Hilbert spaces.
The simple-object decomposition of $V$ gives
\begin{align*}
F(V)_{mn}&=\cM(m,V\triangleright n)\\
&\cong\bigoplus_i V_i\otimes\cM(m,i\triangleright n),
\end{align*}
where $V_i=\cR(i,V)$ is the multiplicity space of $i$ in $V$.
The same decomposition applies to $F(W)_{nq}$ and $F(X)_{mq}$.
A basis state on the first leg is therefore labeled by a triple
$(i,\beta,r)$, where $\beta$ labels a basis of $V_i$ and $r$ a basis of
$\cM(m,i\triangleright n)$ with embedding $q^{i\triangleright n}_{m,r}$.
Using the orthonormal embeddings of
Sec.~\ref{sec:r-tensor-presentation}, this state is represented by
\[
\raisebox{-38pt}{\includegraphics{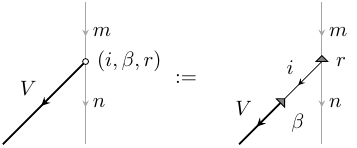}}.
\]
The gray module line runs from $m$ to $n$ through the triangle labeled $r$,
which represents $q^{i\triangleright n}_{m,r}$.
The branch labeled $i$ is mapped into $V$ by the embedding
$q^V_{i,\beta}$, represented by the triangle labeled $\beta$.
Equivalently,
$(i,\beta,r)=(q^V_{i,\beta}\triangleright\id_n)\circ q^{i\triangleright n}_{m,r}$.
The other two legs similarly have basis states
$(j,\gamma,s)\in F(W)_{nq}$ and $(k,\delta,t)\in F(X)_{mq}$.

Taking the dagger reverses the map, giving the dual basis element
$(i,\beta,r)^\dagger:V\triangleright n\to m$,
\[
\raisebox{-30pt}{\includegraphics{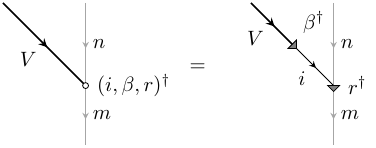}}.
\]
The triangles labeled $\beta^\dagger$ and $r^\dagger$ represent the
adjoints of the two orthonormal embeddings. Their composition is
\[
(i,\beta,r)^\dagger
=(q^{i\triangleright n}_{m,r})^\dagger\circ
\bigl((q^V_{i,\beta})^\dagger\triangleright\id_n\bigr).
\]

Expanding the three basis states separates the embeddings on the
$\cR$-lines from the module-action vertices.
\[
\raisebox{-43.5pt}{\includegraphics{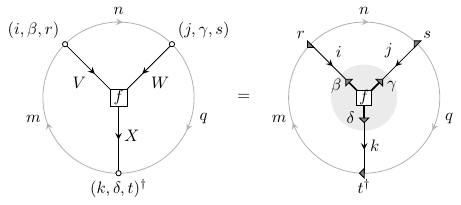}}.
\]
The shaded part is the composition of $f$ with the input embeddings and
the adjoint output embedding. It defines a morphism $i\otimes j\to k$,
whose expansion in the trivalent basis is
\begin{multline}
\left(q^X_{k,\delta}\right)^\dagger\circ f\circ
\left(q^V_{i,\beta}\otimes q^W_{j,\gamma}\right)\\
=\sum_\alpha \left(f^{k\delta}_{i\beta,j\gamma}\right)_\alpha
\left(q^{i\otimes j}_{k,\alpha}\right)^\dagger.
\label{eq:tensor-complex-multiplicity-coefficients}
\end{multline}
Here $\alpha$ labels the fusion channel $k\to i\otimes j$.

Substituting the expansion of $f$ into the diagram converts
successive actions into the fused channel. Since the module $F$-matrix is
unitary, the corresponding coefficient is the complex conjugate of the
module $F$-symbol defined above.
The triangles labeled
$r$, $s$, and $t^\dagger$ represent $q^{i\triangleright n}_{m,r}$,
$q^{j\triangleright q}_{n,s}$, and $(q^{k\triangleright q}_{m,t})^\dagger$, respectively.
The diagram denotes the coefficient in the $t$ output basis, with that
vertex taken in the dual basis:
\[
\raisebox{-46pt}{\includegraphics{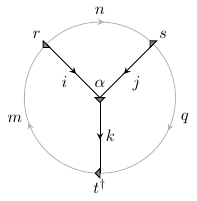}}
=\left({}^{\triangleright}\!F^{ijq}_{m,(n,s,r),(k,\alpha,t)}\right)^*.
\]
Evaluating the map in the input states $(i,\beta,r)$, $(j,\gamma,s)$
and the dual output state $(k,\delta,t)^\dagger$ gives
\begin{multline}
\mathcal T^{mnq}_{(i,\beta,r)(j,\gamma,s)(k,\delta,t)}\\*
=\sum_\alpha
\left(f^{k\delta}_{i\beta,j\gamma}\right)_\alpha
\left({}^{\triangleright}\!F^{ijq}_{m,(n,s,r),(k,\alpha,t)}\right)^*.
\label{eq:tensor-complex-factorization}
\end{multline}
The first factor contains the coefficients of the original $\cR$-tensor,
while the second depends only on the module category and the chosen bases.

For $f=m_{\text{pants}}$ in Eq.~\eqref{eq:mixed-pants-map}, the first
factor contains the conformal three-point blocks, with
$\beta,\gamma,\delta$ corresponding to the conformal-state labels
$\Lambda,\Xi,\Theta$. The plaquette labels $m,n,q$ specify conformal
boundary conditions, and $r,s,t$ distinguish boundary-changing primary
fields in the sectors $i,j,k$. In matching bases and normalizations, the
complex conjugates of the module $F$-symbols are the boundary OPE structure
constants, with the outgoing field paired with its
dual~\cite{Runkel1999BoundaryStructureConstants,FuchsRunkelSchweigert2002}.
Thus
Eq.~\eqref{eq:tensor-complex-factorization} is the decomposition of a
boundary CFT three-point function into conformal blocks and boundary
structure constants.

This separation gives the strange-correlator description of the CFT
tensor complex~\cite{Bal_2018,chen2024cftdtqftd1holographictensor,Cheng:2023kxh,bao2025tensorcomplextopologicalbootstrap}.
Contracting the numerical
coefficients gives the conformal part, while the complex conjugates of the
module $F$-symbols form the topological part. Summing their shared fusion labels, together with
the face weights of the tensor-complex construction, expresses the
partition function as their overlap. The fixed-point tensor of
Ref.~\cite{Cheng:2023kxh} thus combines the conformal three-point blocks
with the boundary structure constants supplied by the module category.

This tensor-complex description is obtained by choosing a generalized
reference frame. Allowing arbitrary deformations in $\Mat_n(\Hilb)$
relaxes the symmetry constraints imposed by $\cR$. In our approach, we
take tensors and their approximations directly in the morphism spaces
of $\cR$, preserving its symmetry constraints throughout the RG flow.

\section{Conclusion and outlook}
\label{sec:conclusion-and-outlook}

We formulated $\cR$-TNR, which preserves the categorical constraints encoded
by a unitary fusion category (UFC) $\cR$ throughout the RG flow. From the
fixed-point tensors, we extracted conformal spectra and structure constants
resolved by the bulk topological sectors. For the Ising, tricritical Ising,
and three-state Potts CFTs, the topological bootstrap ansatz produces stable
flows to gapless fixed points over a broad range of initial parameters.
These examples show how categorical constraints can support gapless phases
without fine-tuning and improve the accuracy of the extracted conformal data.
Of course, our framework and numerical method can also be used to study ordinary fine-tuned
critical points. Perturbing the corresponding fixed-point tensors and
following their RG flows would allow us to investigate how relevant
perturbations gap out the critical theory and which phases they lead to.
Relating these flows to the perturbing operators of the CFT
\cite{Ueda_2023_finite_bond_dimension} could improve
our understanding of phase transitions in two-dimensional systems.
Studying such examples may also help develop a mathematical theory of
these transitions.

The stable topological bootstrap examples studied here are restricted to minimal models and extensions of their VOAs. Although the framework described in Sec.~\ref{appsec:ffa-extensions-and-emergent-symto} is more general, extending the topological bootstrap beyond minimal models requires addressing the relevant perturbations carried by dimension-one currents in affine Lie algebra CFTs. These currents lie in the chiral-factorized vacuum sector, so the direct construction does not satisfy our stability criterion. A suitable orbifold can remove the weight-one fields from the chiral algebra, as discussed by Brehm and Runkel~\cite[Sec.~2.2, footnote~2]{BrehmRunkel2024}. We can then use the resulting charge category $\cR_{\mathrm{orb}}$ to impose the symmetry constraints. In the symTO description, $\fZ(\cR_{\mathrm{orb}})$ admits an anyon condensation to the symTO before orbifolding. Implementing this construction in $\cR$-Loop-TNR would provide a route to extending the topological bootstrap beyond minimal models.

The boundary-CFT construction of fixed-point tensors in
Ref.~\cite{Cheng:2023kxh} motivates studying a boundary version of
$\cR$-Loop-TNR and its fixed-point tensors. In general, boundary tensor
legs can be objects of an $\cR$-module category $\cM$, with boundary tensors
given by morphisms in $\cM$ and coupled to bulk tensors through the module
action.
Following the coupled bulk and boundary RG flows could clarify the
boundary-CFT structure of fixed-point tensors and help compare their
analytic and numerical constructions. It could also reveal how boundary
behavior changes across bulk phase transitions.
These boundary questions also motivate a more general algebraic
description of the gapless slits used in the fixed-point construction.
In Sec.~\ref{sec:geometric-picture}, we view a chiral gapless hole as a
vertex operator algebra internal to a braided fusion category and obtain
a gapless slit by folding. Extending this description to gapless slits
that need not arise from folding a chiral hole could lead to a
complete theory of gapless anyon condensation that unifies the description
of gapless boundaries and domain walls.

Further generalizations of $\cR$-TNR include relaxing unitarity, changing the base field, allowing infinitely many simple objects, and extending the framework to higher categories.
For nonunitary theories such as the Yang--Lee minimal model, a
central problem is to define a suitable norm for tensor approximation errors
and formulate a corresponding loop-optimization procedure.
Another direction is to allow
unitary tensor categories with infinitely many simple objects to develop
a symTO description of continuous symmetries.
A further related question is whether an AdS/CFT description can be recovered from
full field algebra data alone, and which classes of full field algebras admit such a description~\cite{ChenHungJiangLao_2025_Liouville,BaoHungJiangLiu_2024_SymQRG}.
Exploring this possibility, together with the questions above, motivates
the development of systematic algorithms and software for computing full field algebra data.

Finally, extending $\cR$-TNR to unitary fusion 2-categories could provide a framework for studying higher-dimensional symTOs and gapless systems. In particular, such an extension could be applied to three-dimensional classical systems and their CFT limits, whose bulk symTO would be four-dimensional TQFTs, described by the unitary braided fusion 2-category $\fZ(\cR)$.

\begin{acknowledgments}
We thank Dong-Yu Bao, Gong Cheng, Ling-Yan Hung, Kaixin Ji, Liang Kong,
Chenjie Wang, Ying-Jie Wei, Holiverse Yang, Shunyao Yu, Gen Yue, and
Jing-Ren Zhou for helpful discussions. 
C.Q.M. thanks Jacob Bridgeman, Boris De Vos, Lukas Devos, Jutho Haegeman,
and Zheng-Yuan Yue for guidance on using the TensorKit ecosystem.
C.Q.M. also thanks Adwait Naravane, Atsushi Ueda, and Victor Vanthilt for
their collaboration on TNRKit and discussions on Loop-TNR, and Hang Chen,
Haoxin Wang, and Zheng-Yuan Yue for discussions on programming concepts.
This work is supported by funding from Hong Kong’s Research Grants Council (RFS2324-4S02, CRF C7015-24G, CRS HKU701/24).
\end{acknowledgments}

\appendix

\section{Review of key results of the symTO}\label{app:symto-review}
The symTO describes a system with symmetry as a boundary of a topological
order in one higher dimension. This appendix reviews the correspondence
and the results used in this paper.

The description of symmetry by a topological theory in one higher dimension
has been studied through categorical descriptions of defects and
boundary--bulk relations~%
\cite{davydov2011fieldtheoriesdefectscentre,Kitaev_2012,kong2015boundarybulkrelationtopologicalorders},
noninvertible anomalies and categorical symmetries~%
\cite{Ji_2019,Ji_2020,Kong_2020_phy,Chatterjee_2023},
and symmetry topological field theories~%
\cite{Fiorenza_2015,Gaiotto_2015,Freed_2022,freed2024topologicalsymmetryquantumfield,Kaidi:2022cpf}.

For finite symmetries, the symmetry/topological-order
correspondence relates symmetry constraints and the allowed gapped and
gapless phases to bulk topological data and boundary conditions~%
\cite{Kong_2020_math,Kong_2022_liquids,chatterjee2024emergentgeneralizedsymmetrymaximal,KONG2018140,Kong_2020,Kong_2021,Chen_2020,Lu_2023}.

From a microscopic perspective, a quantum system with symmetry can be
characterized by a symmetric local algebra.
The superselection structure of such algebras has been studied using patch
operators, quantum currents, and DHR bimodules~%
\cite{Chatterjee_2023,Lan_2024_Quantum_Current,jones2024dhrbimodulesquasilocalalgebras}.
For fusion spin chains with charge category $\cR$, the DHR category is
braided equivalent to the Drinfeld center $\fZ(\cR)$~%
\cite{jones2024dhrbimodulesquasilocalalgebras,jones2025quantumcellularautomatacategorical},
reflecting the boundary--bulk relation of the symTO.

Tensor network methods have been extended to study anyon chains protected
by fusion category symmetries~%
\cite{Feiguin_2007,Pfeifer_2012,Aasen_2016,Lootens_2023,Ning:2023qbv,Sahinoglu_2021,Bultinck_2017,You_2014,Bal_2018,Cheng:2023kxh}.
The tensor network software package \texttt{TensorKit.jl} implements tensors
as general morphisms in unitary fusion categories and unitary braided fusion
categories, together with operations on these morphisms~%
\cite{Devos_2025_TensorKit,jutho_2025_15243493}.
Within this framework, a numerical study of a lattice model built from the
Haagerup fusion category resolved the conformal spectrum into bulk anyon
sectors and provided evidence for a CFT with $c=2$~\cite{Vanhove_2022}.

\subsection{Map between theories}

The symmetry/topological-order correspondence relates quantum systems whose
charges form $\cR$ to boundary theories of the topological order
$\fZ(\cR)$ \cite{Ji_2020,Kong_2020_phy,Chatterjee_2023},
\begin{equation}\label{eq:1d2d-correspondence}
\begin{gathered}
\text{(1+1)-dimensional quantum systems}\\
\text{with charge category $\cR$}\\
\longleftrightarrow\\
\text{(1+1)-dimensional boundary theories}\\
\text{on the boundary $(\cR,\fgt)$}\\
\text{of the (2+1)-dimensional topological order $\fZ(\cR)$}.
\end{gathered}
\end{equation}
Here
\[
\fgt\colon\fZ(\cR)\longrightarrow\cR,
\qquad
(X,\gamma_{X,-})\longmapsto X,
\]
is the canonical central functor that sends a center object to its underlying
object in $\cR$ by forgetting the half-braiding.

\subsubsection{Symmetric operators and boundary operators}

Let $\cR$ be the fusion category of symmetry charges. For example, an
ordinary finite-group symmetry $G$ has charge category $\cR=\Rep(G)$
\cite{Ji_2020,Kong_2020_phy}.
Consider a one-dimensional lattice model with sites indexed by
$i\in\mathbb Z$, whose local degrees of freedom are specified by
objects $D_i\in\cR$. The symmetric operators supported on a finite
lattice interval $I=[m,n]$ form
$\End_{\cR}(D_m\otimes\cdots\otimes D_n)$. Their inductive limit defines
the quasi-local algebra $\mathfrak R_{\cR}$
\cite{jones2024dhrbimodulesquasilocalalgebras}.
A center object $(X,\gamma_{X,-})\in\fZ(\cR)$ determines a quantum-current
MPO with virtual bond $X$ and local tensor $\gamma_{X,D_i}$. It commutes
with symmetric operators supported strictly inside its interval,
satisfying the transparency condition of patch operators
\cite{Lan_2024_Quantum_Current,Chatterjee_2023}.
For a uniform chain with a strong tensor-generating object $D$, the DHR
category is braided equivalent to the center
\cite{jones2024dhrbimodulesquasilocalalgebras,jones2025quantumcellularautomatacategorical},
\[
\mathsf{DHR}(\mathfrak R_{\cR},D)\cong\fZ(\cR).
\]

On the canonical boundary of $\fZ(\cR)$, whose excitations form $\cR$,
the local operators that do not create bulk excitations form a boundary algebra
$\mathfrak B_{\cR}$ \cite{Ji_2020,jones2025localtopologicalorderboundary}.
The symTO correspondence gives an isomorphism between the two local operator algebras
\cite{Kong_2020_phy},
\[
\mathfrak R_{\cR}\simeq\mathfrak B_{\cR}.
\]
Under this isomorphism, charges in $\cR$ become
boundary excitations, symmetric local operators become boundary
interactions, and quantum-current sectors become bulk anyons in
$\fZ(\cR)$ \cite{Kitaev_2012,Lan_2024_Quantum_Current}.

Fixing $\cR$ while varying the microscopic degrees of freedom and symmetric
couplings defines the $\cR$-symmetric theory space. The algebra isomorphism
maps symmetry-preserving deformations of these systems to local boundary
deformations \cite{Lan:2023oyu,jones2024dhrbimodulesquasilocalalgebras}.

\subsubsection{Changing symmetry}\label{appsec:changing-symmetry}

A monoidal functor between charge categories
\begin{equation}\label{eq:symmetry-change-monoidal-functor}
    \Omega:\cR\longrightarrow\cR'
\end{equation}
implements a change of symmetry from $\cR$ to $\cR'$ by sending the site
objects $D_i$ and local operators of an $\cR$-symmetric system to their
images in $\cR'$. Preservation of composition and tensor
products makes these maps compatible with the local operator algebras.
For example, restriction
$\operatorname{Res}^G_H:\Rep(G)\to\Rep(H)$ for a subgroup $H\subset G$
regards a $G$-symmetric system as an $H$-symmetric one and allows
perturbations that preserve $H$ but break $G$.

The bulk description of $\Omega$ uses its relative center
\cite{Kong_2017,Lan_2024},
\[
    \fZ(\Omega)=
    \Fun_{\cR\mid\cR'}({}_{\Omega}\cR',{}_{\Omega}\cR').
\]
Here ${}_{\Omega}\cR'$ is the $\cR$--$\cR'$ bimodule category whose
left action is induced by $\Omega$ and whose right action is tensor
product in $\cR'$. The relative center describes a gapped wall between
$\fZ(\cR)$ and $\fZ(\cR')$. Fusing this wall with the canonical
$\cR$ boundary gives the canonical $\cR'$ boundary, with the induced
map on boundary excitations equal to $\Omega$. Following
Ref.~\cite{Lan_2024}, we depict this fusion as
\[
    \includegraphics{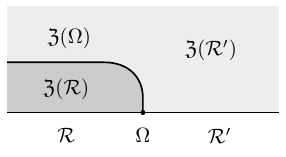}
\]

In general, an indecomposable gapped wall between $\fZ(\cR)$
and $\fZ(\cR')$ is specified by condensable algebras
$A\in\fZ(\cR)$ and $B\in\fZ(\cR')$, together with a braided
equivalence $\varphi$ between $\fZ(\cR)_A^{\localmodules}$ and
$\fZ(\cR')_B^{\localmodules}$
\cite{Huston_2023,Xu_Yang_2025}.
The wall can be pulled open into two condensation interfaces joined
by the invertible interface determined by this equivalence, as in
Fig.~\ref{fig:symmetry-wall-factorization}.
This is the wall decomposition for Witt-equivalent bulks, applied
here to Drinfeld centers.

\begin{figure}[!htb]
    \centering
    \includegraphics[width=0.98\linewidth]{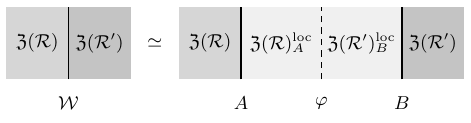}
    \caption{Spatial decomposition of an indecomposable gapped wall
    $\mathcal W$ between $\fZ(\cR)$ and $\fZ(\cR')$, following
    Ref.~\cite{Xu_Yang_2025}.
    The solid interfaces are the condensation walls specified by $A$
    and $B$. The dashed interface is the invertible wall determined
    by the braided equivalence $\varphi$ between the two condensed
    bulks. The symbol $\simeq$ denotes the equivalence of the two wall configurations.}
    \label{fig:symmetry-wall-factorization}
\end{figure}

We first consider the change of symmetry induced by an anyon
condensation wall. Let $A\in\fZ(\cR)$ be a condensable algebra and
let $U:\fZ(\cR)\to\cR$ be the forgetful functor. Let $\cR_A$
be the category of right $U(A)$-modules in $\cR$. The half-braiding
of $A$ equips $\cR_A$ with a tensor product over $U(A)$, and the
free-module functor is monoidal,
\[
    \cR\longrightarrow\cR_A,\qquad X\longmapsto X\otimes U(A).
\]
The category $\cR_A$ can be multifusion, and its center is the
condensed bulk \cite{Kong:2013aya,Davydov_Muger_Nikshych_Ostrik_2013},
\[
    \fZ(\cR_A)\cong\fZ(\cR)_A^{\localmodules}.
\]
We call this change of symmetry \emph{symmetry relaxation}.
For Lagrangian $A$, the condensed bulk is trivial.

We now consider special changes of symmetry for which $A$ is Lagrangian
and the resulting symTO is trivial. In this case, the condensation
interface becomes a boundary to the vacuum. Together with the dynamical
boundary, it forms a sandwich enclosing the original topological bulk,
as in the symTFT construction~\cite{freed2024topologicalsymmetryquantumfield}.
The other boundary serves as a \emph{reference frame} against which
different phases on the dynamical boundary are measured.

A choice of fiber functor $\fiberfunctor:\cR\to\Hilb$ equips
$\cR$ with a \emph{local structure}
\cite{Lan_2024,Kong_2020_phy}.
In the sandwich picture, this choice specifies an ordinary
reference frame. The target $\Hilb$ imposes no symmetry constraints.
The allowed local deformations therefore extend from symmetric operators
to the full local operator algebra, including perturbations that break
the original symmetry
\cite{Kong_2020_phy,Zhang_Cordova_2023,Meng:2024nxx}.
Such a functor exists precisely for anomaly-free fusion-category
symmetries \cite{thorngren2019fusioncategorysymmetryi,Thorngren:2021yso}.
For $\cR=\Rep(G)$, the standard fiber functor forgets the $G$-action.

More generally, a finite semisimple $\cR$-module category $\cM$ defines a
generalized reference frame through its action
$F:\cR\to\Fun(\cM,\cM)$, even when no ordinary fiber functor exists~%
\cite{Kong_2017,Lan_2024}.
In the lattice construction of Lootens et al., different choices of
module category give dual realizations of the same local symmetric
operator algebra~\cite{Lootens_2023,Lootens_2024_Topological_Sectors}.
Taking $\cM=\Hilb$ recovers an ordinary reference frame.

\subsection{Map between gapped phases}\label{appsec:gapped-phase-map}

The symTO duality in Eq.~\eqref{eq:1d2d-correspondence} maps
symmetry-preserving deformations of the lattice system to local deformations
of the corresponding boundary theory.  It therefore induces a correspondence
between gapped phases in the two descriptions.

\subsubsection{The duality map}

A gapped phase in the $\cR$-symmetric theory space is specified by an
indecomposable separable algebra $B\in\cR$. In the boundary description, $B$ condenses
boundary defects, leaving the defect category ${}_B\cR_B$
\cite{Lan_2024_Quantum_Current,Meng:2024nxx,Inamura_2022,Kong:2013aya}.
Equivalently, its full center is a Lagrangian algebra in the bulk,
\[
    L:=Z_{\mathrm{full}}(B)\in\fZ(\cR),
    \qquad \fZ(\cR)_L\cong{}_B\cR_B.
\]
The algebras $B$ and $L$ therefore specify the same gapped phase in the
boundary and bulk descriptions \cite{Kong:2013aya,Ai_2017,Lan_2024_Quantum_Current}.
The bulk $\fZ(\cR)$ is fixed by the symmetry, whereas $L$ specifies the
choice of phase.

Geometrically, shrinking a hole carrying this gapped boundary maps its
vacuum sector to $L$. The multiplication $m:L\otimes L\to L$ is the
pair-of-pants operation merging two holes
\cite{Ai_2017,kong2025highercondensationtheory,gaiotto2025condensationshighercategories}.
Figure~\ref{fig:gapped-duality-map} shows the shrinking operation when the
condensed bulk need not be trivial.
\begin{figure}[!ht]
    \centering
    \[
    \raisebox{-28.45274pt}{\includegraphics{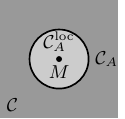}}
    \xrightarrow{\text{shrinking}}
    \raisebox{-28.45274pt}{\includegraphics{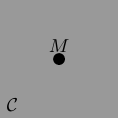}}
    \]
    \caption{Let $\cC$ be a bulk UMTC and $A\in\cC$ a condensable algebra.
    Shrinking a hole containing the $A$-condensed phase maps
    $M\in\cC_A^{\localmodules}$ to its underlying object in $\cC$.
    In particular, the vacuum maps to $A\in\cC$.}
    \label{fig:gapped-duality-map}
\end{figure}

\subsubsection{The effect of choosing a reference frame}

A reference frame fixes the local structure and the local order
parameters used to interpret a gapped phase as SPT or SSB
\cite{Kong_2020_phy,Meng:2024nxx,Lan_2024}.
Consider $G=\Z_2\times\Z_2$. Let $\beta\in Z^2(G,U(1))$ represent the
nontrivial class in $H^2(G,U(1))\simeq\Z_2$, and let $W$ carry the
two-dimensional irreducible $\beta$-projective representation of $G$,
with $\rho(g)\rho(h)=\beta(g,h)\rho(gh)$. The algebras
$B_0=\Fun(G)$ and $B_\beta=\End(W)$ in $\Rep(G)$ have the same underlying
object $\bigoplus_{\chi\in\widehat G}\chi$ but different multiplications.
Here $G$ acts on $\End(W)$ by conjugation,
$g\cdot T=\rho(g)T\rho(g)^{-1}$
\cite{yue2025condensationcompletiondefects21d}.
Despite their different multiplications, the two algebras are related by
a monoidal autoequivalence $\pi$ of $\Rep(G)$, with
$\pi(B_0)\simeq B_\beta$ as algebras. However, this equivalence
does not preserve the reference frame specified by the standard fiber
functor $\fgt:\Rep(G)\to\Hilb$ because $\fgt\circ\pi$ is not monoidally
isomorphic to $\fgt$. Under $\fgt$, the two algebras become
$\Fun(G)$ and $\Mat_2(\C)$, with four and one simple blocks, respectively.
These blocks label the vacuum sectors of the respective phases.
The four symmetry-related vacua of $B_0$ give the fully symmetry-breaking phase.
The algebra $B_\beta$ gives a unique vacuum on a closed chain and a
nontrivial projective edge representation, and thus describes the
nontrivial SPT phase
\cite{Meng:2024nxx,Gu_2009}.

The ground-state space on a circle depends on both the phase and the
reference frame. Let
$L,N\in\fZ(\cR)$ be the Lagrangian algebras specifying the phase and the
reference frame, and set
$L_x=\fZ(\cR)(x^*,L)$ and $N_x=\fZ(\cR)(x^*,N)$.
The annulus with these two boundary conditions has state space
\begin{equation}\label{eq:gapped-annulus-state-space}
    \cV_{S^1}(L;N)
    \simeq\bigoplus_{x\in\Irr\fZ(\cR)}L_x\otimes N_{x^*}
    \simeq\fZ(\cR)(\one,L\otimes N).
\end{equation}
The dual $x^*$ accounts for the opposite boundary orientations. The
ground-state degeneracy is
$\operatorname{GSD}_{S^1}(L;N)=\sum_x\dim L_x\,\dim N_{x^*}$
\cite{Ai_2017,Lan_2015,BhardwajBoundarySymTFT2025}.

\subsection{Map between gapless phases}\label{appsec:gapless-phase-map}
Under the symmetry--topological-order correspondence in
Eq.~\eqref{eq:1d2d-correspondence}, bosonic gapless phases in the
$\cR$-symmetric theory space correspond to gapless boundary phases of the
non-chiral topological order $\fZ(\cR)$ on its canonical boundary
$(\cR,\fgt)$~\cite{Ji_2020}.

\subsubsection{Enriched category description}
\label{appsec:enriched-category-description}

We denote a bulk topological order by $(\cC,c)$, where the UMTC $\cC$
describes the bulk anyons and $c$ is the chiral central charge.
A CFT-like gapless boundary is described by a pair
$(\scrV,{}^{\Mod_\scrV}\!\cS)$
\cite{KONG2018140,Kong_2020,Kong_2021}.
The first component, $\scrV$, is a unitary rational full field algebra
specifying the local quantum symmetry. The second component,
${}^{\Mod_\scrV}\!\cS$, is a fusion category enriched in the category
$\Mod_\scrV$ of $\scrV$-modules. The objects of $\cS$ describe topological
boundary defects, and the enriched morphism objects describe local operators
between them.
A VOA
$\scrV$ is included as the full field algebra $\scrV\otimes\C$ with
trivial antichiral factor. The bulk anyon category is the enriched center
of the boundary defect category~\cite{Kong_2017,Kong_2024_enriched},
\[
    \fZ({}^{\Mod_\scrV}\!\cS)\cong\cC.
\]

On a gapped boundary, the local operators between two defects form a
finite-dimensional vector space. On a gapless boundary they instead
form a $\scrV$-module, with composition given by the OPE. These modules
are the morphism objects of the enriched category. For the canonical
chiral Ising boundary, the operators that change the trivial defect into
a defect labeled by an anyon $i$ form the conformal family
$\scrV^{\mathrm I}_i$ of the Ising VOA $\scrV^{\mathrm I}$
\cite{Kitaev_2006,Kong_2020}.
For trivial local quantum symmetry $\scrV=\C$, the enrichment reduces to
ordinary vector spaces and $\cS$ is the gapped boundary defect category.

A systematic way to construct an enriched fusion category is to start
with a fusion category $\cS$ and a braided functor
$\cD\to\fZ(\cS)$. The canonical construction then equips $\cS$ with
a $\cD$-enrichment, yielding ${}^{\cD}\!\cS$~%
\cite{morrison2017monoidalcategoriesenrichedbraided}.
This construction covers the enriched fusion categories describing the
rational gapless boundaries considered here~\cite{Kong_2020,Kong_2021}.
With $\cD=\Mod_\scrV$, start from a gapped wall $\cS$ between
$\Mod_\scrV$ and $\cC$. The canonical construction yields the
enriched gapless boundary ${}^{\Mod_\scrV}\!\cS$ with bulk $\cC$.
Geometrically, this construction can be pictured as rotating the
$\Mod_\scrV$ layer into an internal space, a picture known as
\emph{topological Wick rotation}~%
\cite{KONG2018140,Kong_2020,Kong_2021}.
The enriched center of ${}^{\Mod_\scrV}\!\cS$ recovers $\cC$,
so the construction satisfies the boundary--bulk relation.
Geometrically,
\[
\raisebox{-18.2194pt}{\includegraphics{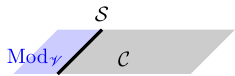}}
\leadsto
\raisebox{-21.73956pt}{\includegraphics{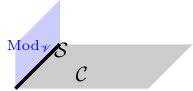}}.
\]

The classification of gapped walls reviewed in
Sec.~\ref{appsec:changing-symmetry} expresses any indecomposable wall as
a condensation interface, an invertible interface implementing a braided
equivalence, and a reversed condensation interface~%
\cite{Huston_2023,Xu_Yang_2025}.
Applied to the wall from $\Mod_\scrV$ to $\cC$, this decomposition gives
a corresponding description of gapless boundaries through topological
Wick rotation. We focus on the first step, condensation from
$\Mod_\scrV$, and call the resulting gapless boundary \emph{injective}.
The remaining two operations, the braided equivalence and the
reversed condensation, can both be regarded as gapped processes
within the gapped bulk. The study of gapless boundaries can
therefore be reduced to the injective case, with these two
bulk operations treated separately.

\begin{definition}[Injective gapless boundary]
\label{def:injective-gapless-boundary}
A boundary $(\scrV,{}^{\Mod_\scrV}\!\cS)$ of $\cC$ is \emph{injective} if
\[
    \cC\cong(\Mod_\scrV)_{\scrA}^{\localmodules},
    \qquad \cS\cong(\Mod_\scrV)_{\scrA}
\]
for a condensable algebra $\scrA\in\Mod_\scrV$.
\end{definition}
By the extension theorem of Huang--Kirillov--Lepowsky
\cite{Huang_2015} and the full-field algebra formulation of
Kong~\cite{Kong_2007}, an extension
$\scrV\hookrightarrow\scrA$ is described by a commutative
algebra object $\scrA$ in $\Mod_\scrV$.
For the condensable algebras considered here, this gives an
anyon condensation in $\Mod_\scrV$~\cite{Kong:2013aya}.
Modules of the extended algebra are local $\scrA$-modules,
so
\[
    \Mod_\scrA\cong(\Mod_\scrV)_{\scrA}^{\localmodules}.
\]
Thus the representation category of the extended full field
algebra describes the condensed bulk, while all $\scrA$-modules form the
condensation wall $(\Mod_\scrV)_{\scrA}$ whose enrichment describes the
gapless boundary. Consequently, to study an injective gapless boundary,
it suffices to study a pair of full field algebras related by an extension
$\scrV\hookrightarrow\scrA$.

A physical boundary may admit several choices of $\scrV$ with
the same trivial-sector algebra $\scrA$ and bulk $\Mod_{\scrA}$.
Extending this local quantum symmetry can reduce the allowed TDLs, which
must commute with it \cite{Kong_2020,Kong_2021}.

\subsubsection{Partition vector and modular covariance}
\label{appsec:partition-vector-and-modular-covariance}

For an injective gapless boundary, a bulk anyon $x\in\Irr(\cC)$ determines
an $\scrA$-module $\scrV_x$. This module is the Hilbert space of
low-energy edge modes on a disk with $x$ inserted in the bulk
\cite{Ji_2020,Chen_2020,chatterjee2024emergentgeneralizedsymmetrymaximal}.
In particular, the trivial sector is $\scrV_\one=\scrA$.

The trace over $\scrV_x$ gives the sector partition function
\[
    Z_x(\tau,\bar\tau)
    =\Tr_{\scrV_x}\!\left(q^{L_0-c_L/24}
      \bar q^{\overline L_0-c_R/24}\right),
\]
where $q=\ee^{2\pi\ii\tau}$, $\bar q=\ee^{-2\pi\ii\bar\tau}$,
and $c_L,c_R$ are the left- and right-moving central charges of the boundary CFT.
The components of the partition vector $Z=(Z_x)_{x\in\Irr(\cC)}$ obey
the modular covariance relations \cite{Levin_2013,Ji_2019}
\begin{align*}
    Z_x(\tau+1,\bar\tau+1)&=\ee^{-2\pi\ii(c_L-c_R)/24}\theta_x Z_x(\tau,\bar\tau),\\
    Z_x(-1/\tau,-1/\bar\tau)&=\frac1{D_\cC}\sum_y S_{xy}Z_y(\tau,\bar\tau).
\end{align*}
Here $S$ is the Hopf-link matrix and $D_\cC^2=\sum_a d_a^2$.
For the doubled theories considered below, $c_L=c_R=c$, so the phase arising
from the difference between the left- and right-moving central charges is unity.

Table~\ref{tab:gapless-boundary-examples} collects the extensions of full field algebras and
partition vectors for chiral, Ising, and Potts conformal field theories with different symmetries ~\cite{Kong_2007,Kong:2013aya}.

\begin{table*}[t]
\caption{Examples of full field algebra extensions and their enriched
boundary categories, bulks, and partition functions.
Category and sector labels follow Appendix~\ref{app:categorical-data}.
Here $\scrV^{\mathrm I}$ and $\scrV^{\mathrm P}$ are the Ising and
$\cM(6,5)$ Virasoro VOAs. The irreducible module $\scrV_h$ has conformal weight
$h$, with $\scrV_0=\scrV$.
The corresponding Virasoro characters are $\chi_h^{\mathrm I}$ and
$\chi_h^{\mathrm P}$, with argument $\tau$ omitted.
Overlines denote antichiral counterparts.
In the last two rows, $\scrV$ is a chiral VOA and $\scrL$ is a
Lagrangian full field algebra in $\fZ(\Mod_\scrV)$.}
\label{tab:gapless-boundary-examples}
\centering
\small
\setlength{\abovedisplayskip}{4pt}
\setlength{\belowdisplayskip}{4pt}
\begin{minipage}{\textwidth}
The extension algebras are
\begin{align}
\scrW^{\mathrm P}&=\scrV^{\mathrm P}_0\oplus\scrV^{\mathrm P}_3,
\qquad
\scrA=\scrV^{\mathrm I}_0\otimes\overline{\scrV^{\mathrm I}_0}
\oplus\scrV^{\mathrm I}_{1/2}\otimes\overline{\scrV^{\mathrm I}_{1/2}},\notag\\
\scrL_{\diagsector}&=\bigoplus_{h\in\{0,1/2,1/16\}}
\scrV^{\mathrm I}_h\otimes\overline{\scrV^{\mathrm I}_h},\qquad
H_{6,5}=\{0,1/40,1/15,1/8,2/5,21/40,2/3,7/5,13/8,3\},\notag\\
\scrL_{\mathrm{tI}}&=\bigoplus_{h\in H_{6,5}}
\scrV^{\mathrm P}_h\otimes\overline{\scrV^{\mathrm P}_h},
\label{eq:tetracritical-ising-full-field-algebra}\\
\scrA_{\mathrm P}&=\scrW^{\mathrm P}\otimes\overline{\scrW^{\mathrm P}}
\oplus(\scrV^{\mathrm P}_{2/5}\oplus\scrV^{\mathrm P}_{7/5})
\otimes(\overline{\scrV^{\mathrm P}_{2/5}}\oplus\overline{\scrV^{\mathrm P}_{7/5}}),\notag\\
\scrL_{\mathrm{3Potts}}&=\scrA_{\mathrm P}
\oplus2\scrV^{\mathrm P}_{2/3}\otimes\overline{\scrV^{\mathrm P}_{2/3}}
\oplus2\scrV^{\mathrm P}_{1/15}\otimes\overline{\scrV^{\mathrm P}_{1/15}},
\label{eq:potts-full-field-algebra}
\end{align}
The $\cW_3$ characters are
\begin{align*}
\widehat\chi_0&=\chi_0^{\mathrm P}+\chi_3^{\mathrm P},\qquad
\widehat\chi_{2/5}=\chi_{2/5}^{\mathrm P}+\chi_{7/5}^{\mathrm P},\qquad
\widehat\chi_{2/3}=\chi_{2/3}^{\mathrm P},\qquad
\widehat\chi_{1/15}=\chi_{1/15}^{\mathrm P}.
\end{align*}
\end{minipage}
\smallskip
\setlength{\tabcolsep}{2pt}
\renewcommand{\arraystretch}{1.1}
\setlength{\aboverulesep}{2pt}
\setlength{\belowrulesep}{2pt}
\begin{tabular}{@{}>{\raggedright\arraybackslash}p{0.13\textwidth}p{0.18\textwidth}p{0.20\textwidth}p{0.10\textwidth}p{0.335\textwidth}@{}}
\toprule
Theory & Extension & Enriched category & Bulk & Partition functions \\
\midrule
Chiral Ising
& $\scrV^{\mathrm I}\hookrightarrow\scrV^{\mathrm I}$
& ${}^{\Ising}\!\Ising$ & $\Ising$
& $(Z_\one,Z_\psi,Z_\sigma)=(\chi_0^{\mathrm I},\chi_{1/2}^{\mathrm I},\chi_{1/16}^{\mathrm I})$ \\
\midrule[0.2pt]
Chiral $\cW_3$ (Virasoro symmetry)
& $\scrV^{\mathrm P}\hookrightarrow\scrW^{\mathrm P}$
& ${}^{\cC_{6,5}}\!(\cC_{6,5})_{\scrW^{\mathrm P}}$
& $\mathsf{6}_{4/5}$
& $\begin{aligned}[t]
Z_\one&=\widehat\chi_0,\quad Z_\epsilon=\widehat\chi_{2/5},\\
Z_{\psi_1}=Z_{\psi_2}&=\widehat\chi_{2/3},\\
Z_{\sigma_1}=Z_{\sigma_2}&=\widehat\chi_{1/15}.
\end{aligned}$ \\
\midrule[0.2pt]
Chiral $\cW_3$ ($\cW_3$ symmetry)
& $\scrW^{\mathrm P}\hookrightarrow\scrW^{\mathrm P}$
& ${}^{\mathsf{6}_{4/5}}\!\mathsf{6}_{4/5}$ & $\mathsf{6}_{4/5}$
& $\begin{aligned}[t]
Z_\one&=\widehat\chi_0,\quad Z_\epsilon=\widehat\chi_{2/5},\\
Z_{\psi_1}=Z_{\psi_2}&=\widehat\chi_{2/3},\\
Z_{\sigma_1}=Z_{\sigma_2}&=\widehat\chi_{1/15}.
\end{aligned}$ \\
\midrule
Stable Ising symTO
& $\begin{gathered}[t]\scrV^{\mathrm I}\otimes\overline{\scrV^{\mathrm I}}\\
\hookrightarrow\scrV^{\mathrm I}\otimes\overline{\scrV^{\mathrm I}}\end{gathered}$
& ${}^{\fZ(\Ising)}\!\fZ(\Ising)$ & $\fZ(\Ising)$
& $\begin{gathered}[t]Z_{h,h'}=\chi_h^{\mathrm I}\overline{\chi_{h'}^{\mathrm I}},\\h,h'\in\{0,1/2,1/16\}.\end{gathered}$ \\
\midrule[0.2pt]
$\Z_2$-Ising symTO
& $\begin{gathered}[t]\scrV^{\mathrm I}\otimes\overline{\scrV^{\mathrm I}}\\\hookrightarrow\scrA\end{gathered}$
& ${}^{\fZ(\Ising)}\!\fZ(\Ising)_{\scrA}$
& $\fZ(\Rep(\Z_2))$
& $\begin{aligned}[t]
Z_\one&=|\chi_0^{\mathrm I}|^2+|\chi_{1/2}^{\mathrm I}|^2,\\
Z_e=Z_m&=|\chi_{1/16}^{\mathrm I}|^2,\\
Z_f&=\chi_0^{\mathrm I}\overline{\chi_{1/2}^{\mathrm I}}+\chi_{1/2}^{\mathrm I}\overline{\chi_0^{\mathrm I}}.
\end{aligned}$ \\
\midrule[0.2pt]
$\Z_3$-Potts symTO
& $\begin{gathered}[t]
\scrV^{\mathrm P}\otimes\overline{\scrV^{\mathrm P}}\\
\hookrightarrow\scrW^{\mathrm P}\otimes\overline{\scrW^{\mathrm P}}\\
\hookrightarrow\scrA_{\mathrm P}
\end{gathered}$
& ${}^{\fZ(\cC_{6,5})}\!\fZ(\cC_{6,5})_{\scrA_{\mathrm P}}$
& $\fZ(\Rep(\Z_3))$
& $\begin{aligned}[t]
Z_\one&=|\widehat\chi_0|^2+|\widehat\chi_{2/5}|^2,\\
Z_e=Z_{e^2}&=Z_m=Z_{m^2}\\
&=|\widehat\chi_{2/3}|^2+|\widehat\chi_{1/15}|^2,\\
Z_{em}=Z_{e^2m^2}&=\widehat\chi_{2/3}\overline{\widehat\chi_0}
+\widehat\chi_{1/15}\overline{\widehat\chi_{2/5}},\\
Z_{e^2m}=Z_{em^2}&=\widehat\chi_0\overline{\widehat\chi_{2/3}}
+\widehat\chi_{2/5}\overline{\widehat\chi_{1/15}}.
\end{aligned}$ \\
\midrule
Double Ising CFT
& $\begin{gathered}[t]\scrV^{\mathrm I}\otimes\overline{\scrV^{\mathrm I}}\\\hookrightarrow\scrL_{\diagsector}\end{gathered}$
& ${}^{\fZ(\Ising)}\!\Ising$ & $\Hilb$
& $Z=\displaystyle\sum_{h\in\{0,1/2,1/16\}}|\chi_h^{\mathrm I}|^2$ \\
\midrule
Tetracritical Ising CFT
& $\begin{gathered}[t]\scrV^{\mathrm P}\otimes\overline{\scrV^{\mathrm P}}\\\hookrightarrow\scrL_{\mathrm{tI}}\end{gathered}$
& ${}^{\fZ(\cC_{6,5})}\!\cC_{6,5}$ & $\Hilb$
& $Z=\displaystyle\sum_{h\in H_{6,5}}|\chi_h^{\mathrm P}|^2$ \\
\midrule[0.2pt]
Three-state Potts CFT
& $\begin{gathered}[t]
\scrV^{\mathrm P}\otimes\overline{\scrV^{\mathrm P}}\\
\hookrightarrow\scrW^{\mathrm P}\otimes\overline{\scrW^{\mathrm P}}\\
\hookrightarrow\scrL_{\mathrm{3Potts}}
\end{gathered}$
& ${}^{\fZ(\cC_{6,5})}\!\fZ(\cC_{6,5})_{\scrL_{\mathrm{3Potts}}}$ & $\Hilb$
& $\begin{aligned}[t]
Z&=|\chi_0^{\mathrm P}+\chi_3^{\mathrm P}|^2
+|\chi_{2/5}^{\mathrm P}+\chi_{7/5}^{\mathrm P}|^2\\
&\quad+2|\chi_{2/3}^{\mathrm P}|^2+2|\chi_{1/15}^{\mathrm P}|^2.
\end{aligned}$ \\
\midrule
Anomaly-free CFT
& $\scrV\otimes\overline{\scrV}\hookrightarrow\scrL$
& ${}^{\fZ(\Mod_\scrV)}\!\fZ(\Mod_\scrV)_{\scrL}$ & $\Hilb$ & $Z=\Tr_{\scrL}\!\left(q^{L_0-c/24}\bar q^{\overline L_0-c/24}\right)$ \\
\midrule[0.2pt]
Anomaly-free CFT (full local symmetry)
& $\scrL\hookrightarrow\scrL$
& ${}^{\Hilb}\!\Hilb$ & $\Hilb$ & $Z=\Tr_{\scrL}\!\left(q^{L_0-c/24}\bar q^{\overline L_0-c/24}\right)$ \\
\bottomrule
\end{tabular}
\end{table*}

Equal partition vectors need not imply equal defect categories.
For example, the two chiral $\cW_3$ descriptions in
Tab.~\ref{tab:gapless-boundary-examples} have the same partition vector.
Their defect categories differ because, after enlarging the local quantum
symmetry from $\scrV^{\mathrm P}$ to $\scrW^{\mathrm P}$, only local
$\scrW^{\mathrm P}$-modules are allowed.
Likewise, the toric-code $e$ and $m$ sectors have equal characters but
different $\scrA$-module structures~\cite{Chen_2020}.
With $\scrV^{\mathrm P}\otimes\overline{\scrV^{\mathrm P}}$ fixed as the
local quantum symmetry, the Lagrangian extensions $\scrL_{\mathrm{tI}}$ and
$\scrL_{\mathrm{3Potts}}$ specify different reference boundaries, yielding
the tetracritical Ising and three-state Potts CFTs, respectively.
The two-step Potts extension corresponds to successive chiral and
diagonal condensations~\cite{Huang_2015,Kong:2013aya}.

\subsubsection{Changing the symmetry of a gapless boundary}
\label{appsec:changing-gapless-boundary}

Let $(\scrV,{}^{\Mod_\scrV}\!\cS)$ describe a gapless boundary of
$\cC=\fZ(\cR)$. Changing symmetry is implemented by fusing the boundary
with a gapped wall, using the relative Deligne tensor product over
$\cC$~\cite{KONG2018140,Kong_2020,Kong_2021}.
For a condensable algebra $B\in\cC$, the wall is $\cC_B$, and fusion gives
\[
    ({}^{\Mod_\scrV}\!\cS)\boxtimes_{\cC}\cC_B
    \cong{}^{\Mod_\scrV}\!(\cS\boxtimes_{\cC}\cC_B).
\]
The new bulk is $\cC_B^{\localmodules}$, while the local quantum symmetry
$\scrV$ is unchanged.

For injective boundaries, successive changes of symmetry are encoded by
an extension tower of full field algebras
\[
    \scrV\hookrightarrow\scrA\hookrightarrow\scrB.
\]
The first extension gives the boundary
${}^{\Mod_\scrV}\!(\Mod_\scrV)_{\scrA}$ of $\Mod_{\scrA}$.
The second gives the wall $(\Mod_{\scrA})_{\scrB}$ from
$\Mod_{\scrA}$ to $\Mod_{\scrB}$. Their composition simplifies to
\begin{equation}\label{eq:injective-gapless-wall-fusion}
\begin{aligned}
    \bigl({}^{\Mod_\scrV}\!(\Mod_\scrV)_{\scrA}\bigr)
    &\boxtimes_{\Mod_{\scrA}}(\Mod_{\scrA})_{\scrB}\\
    &\cong{}^{\Mod_\scrV}\!(\Mod_\scrV)_{\scrB}.
\end{aligned}
\end{equation}
Here $\scrB$ is viewed as an algebra in $\Mod_\scrV$ by restriction along
$\scrV\hookrightarrow\scrA$. Thus composing extensions gives the
injective boundary associated with the composite extension
$\scrV\hookrightarrow\scrB$~%
\cite{Davydov_Muger_Nikshych_Ostrik_2013,Kong:2013aya}.
The same rule iterates along a longer extension tower, with enrichment
in $\Mod_\scrV$ and final bulk $\Mod_{\scrB}$.

Let $\scrV_x$ denote the sector Hilbert space
labeled by $x\in\Irr\cC$. Let
$U_B:\cC_B^{\localmodules}\to\cC$ be the forgetful functor and
$M_{x^*}=\cC(x,U_B(M))$ the multiplicity space of $x$.
The same wall fusion gives
\begin{equation}\label{eq:gapless-sector-after-condensation}
    \scrV_M\simeq\bigoplus_{x\in\Irr\cC}\scrV_x\otimes M_{x^*},
\end{equation}
the gapless analogue of Eq.~\eqref{eq:gapped-annulus-state-space}.

As an example, consider the maximal Ising gapless boundary~\cite{chatterjee2024emergentgeneralizedsymmetrymaximal}
\[
    \left(\scrV^{\mathrm I}\otimes\overline{\scrV^{\mathrm I}},
    {}^{\cC}\!\cC\right),
    \qquad \cC=\fZ(\Ising).
\]
The two changes of symmetry are described by the extension tower
\[
    \scrV^{\mathrm I}\otimes\overline{\scrV^{\mathrm I}}
    \hookrightarrow\scrA\hookrightarrow\scrL_{\diagsector},
\]
where
\begin{align*}
    \scrA&=\scrV^{\mathrm I}_0\otimes\overline{\scrV^{\mathrm I}_0}
    \oplus\scrV^{\mathrm I}_{1/2}\otimes\overline{\scrV^{\mathrm I}_{1/2}},\\
    \scrL_{\diagsector}&=\scrA\oplus
    \scrV^{\mathrm I}_{1/16}\otimes\overline{\scrV^{\mathrm I}_{1/16}}.
\end{align*}
Their module categories give the bulk sequence
\[
    \cC\condensesTo{\scrA}\Mod_{\scrA}\cong\mathsf{TC}
    \condensesTo{\scrL_{\diagsector}}\Mod_{\scrL_{\diagsector}}\cong\Hilb.
\]
Applying Eq.~\eqref{eq:injective-gapless-wall-fusion} at each step gives
\begin{align*}
    ({}^{\cC}\!\cC)\boxtimes_{\cC}\cC_{\scrA}
    &\cong{}^{\cC}\!\cC_{\scrA},\\
    ({}^{\cC}\!\cC_{\scrA})\boxtimes_{\mathsf{TC}}
    (\Mod_{\scrA})_{\scrL_{\diagsector}}
    &\cong{}^{\cC}\!\cC_{\scrL_{\diagsector}}
    \cong{}^{\cC}\!\Ising.
\end{align*}
The first boundary is the $\Z_2$-Ising gapless boundary~\cite{Chen_2020}.
In the second step, $\scrL_{\diagsector}$ viewed as a
$\scrA$-module is the Lagrangian algebra $B=\one\oplus e$ in
$\mathsf{TC}$, so its wall is a reference boundary to the vacuum.
The composite is the direct extension
$\scrV^{\mathrm I}\otimes\overline{\scrV^{\mathrm I}}
\hookrightarrow\scrL_{\diagsector}$, with the same local quantum
symmetry and enrichment as the initial boundary.

For the second step, $U_B(B)=\one\oplus e$, so
Eq.~\eqref{eq:gapless-sector-after-condensation} gives
\[
    \scrV_B=\scrV_\one\oplus\scrV_e
    =\bigoplus_{h\in\{0,1/2,1/16\}}
      \scrV^{\mathrm I}_h\otimes\overline{\scrV^{\mathrm I}_h}
    =\scrL_{\diagsector}.
\]
This gives the anomaly-free Ising CFT with partition function
$Z=|\chi_0^{\mathrm I}|^2+|\chi_{1/2}^{\mathrm I}|^2+|\chi_{1/16}^{\mathrm I}|^2$.
The composite wall realizes the generalized reference frame associated
with the regular module category $\cM=\Ising$,
$F:\Ising\to\Fun(\Ising,\Ising)$, acting by $x\mapsto x\otimes-$.
The same rule combines the sectors for the Potts reference boundary in
Eq.~\eqref{eq:potts-full-field-algebra}.

\subsection{Full field algebra extensions and emergent symTO}
\label{appsec:ffa-extensions-and-emergent-symto}

An injective gapless boundary is specified by a full field algebra
extension $\scrV\hookrightarrow\scrA$, as in
Definition~\ref{def:injective-gapless-boundary}. Its bulk topological order is
described by $\Mod_{\scrA}$. Symmetry-preserving local perturbations
lie in the trivial bulk-anyon sector
\cite{chatterjee2024emergentgeneralizedsymmetrymaximal}, which here is
$\scrA$. If this sector contains relevant fields,
their couplings must be tuned to reach the critical boundary theory.

Every injective gapless boundary specified by an extension
$\scrB\hookrightarrow\scrA$ can be obtained from the self-enriched
boundary ${}^{\Mod_{\scrB}}\!\Mod_{\scrB}$, with trivial condensation,
by the condensation-type symmetry change associated with this extension.
We interpret $\Mod_{\scrB}$ as a larger emergent symTO of the gapless
boundary with bulk $\Mod_{\scrA}$. On the self-enriched boundary,
symmetry-preserving local perturbations belong to $\scrB$. Requiring
this emergent symmetry therefore excludes relevant fields of $\scrA$
outside $\scrB$. Smaller full field subalgebras give larger emergent
symTOs in the condensation order and can reduce the space of allowed
relevant perturbations. This motivates studying the stability of gapless
boundaries through full field algebra extension towers.

We assume that the module categories carry the required braided vertex
tensor structures, but do not require them to be finite.

\subsubsection{Extension tower}
\label{appsec:full-field-extension-tower}

A full field algebra $\scrB$ is \emph{chiral-factorized} if
$\scrB\simeq\scrV_L\otimes\overline{\scrV_R}$ for chiral VOAs
$\scrV_L$ and $\scrV_R$. We also call its symTO $\Mod_{\scrB}$
chiral-factorized.

Let $L(-1)$ and $\overline{L(-1)}$ denote the holomorphic and
antiholomorphic translation operators of $\scrA$.
By Proposition~3.12 of Moriwaki~\cite{Moriwaki_2023_full_vertex_algebra},
the kernels
\[
    \scrV_L^{\max}=\ker\left(\overline{L(-1)}|_{\scrA}\right),
    \qquad \scrV_R^{\max}=\ker\left(L(-1)|_{\scrA}\right)
\]
form chiral and antichiral vertex subalgebras of $\scrA$.
Every chiral or antichiral subalgebra lies in the corresponding kernel,
since its fields are annihilated by the opposite translation operator.
Proposition~3.18 of the same reference gives the canonical full vertex
algebra homomorphism
\[
    \scrV_L^{\max}\otimes\overline{\scrV_R^{\max}}
    \longrightarrow\scrA,
    \qquad a\otimes\overline b\longmapsto a(-1,-1)b.
\]
Assuming its source is simple, this map is injective.
The two factors of any chiral-factorized subalgebra lie in these kernels,
so that subalgebra is contained in the image of this map.
The image is therefore the largest chiral-factorized subalgebra of $\scrA$.
The energy--momentum tensors $\omega_L$ and $\omega_R$ of $\scrA$
generate its minimal full field subalgebra. Together, these give the
extension tower
\begin{equation}\label{eq:full-field-extension-tower}
\begin{aligned}
    \operatorname{Vir}(\omega_L)\otimes
      \overline{\operatorname{Vir}(\omega_R)}
    &\hookrightarrow\scrV_L^{\max}\otimes\overline{\scrV_R^{\max}}\\
    &\hookrightarrow\scrA.
\end{aligned}
\end{equation}
The left and right chiral algebras need not be isomorphic.
Other full field algebras may lie between the minimal full field
subalgebra and the largest chiral-factorized subalgebra. Of particular
interest are those with no relevant fields, the stable full field algebras
defined below.

The largest chiral-factorized subalgebra gives what we call the
\emph{minimal chiral-factorized symTO} associated with $\scrA$,
\[
    \Mod_{\scrV_L^{\max}}\boxtimes
      \overline{\Mod_{\scrV_R^{\max}}}.
\]
The minimal full field subalgebra in Eq.~\eqref{eq:full-field-extension-tower}
defines what we call the \emph{maximal symTO associated with $\scrA$}~\cite{chatterjee2024emergentgeneralizedsymmetrymaximal},
\[
    \Mod_{\operatorname{Vir}(\omega_L)}\boxtimes
      \overline{\Mod_{\operatorname{Vir}(\omega_R)}}.
\]
The condensation order on UMTCs has no largest element.
Here we restrict to condensations induced by full field algebra
extensions $\scrB\hookrightarrow\scrA$ with $\scrA$ fixed.
The minimal full field subalgebra of $\scrA$ then determines
the largest symTO in this restricted order.
For unitary Virasoro minimal models, this module category is a finite
UMTC and agrees with the maximal symTO of
Ref.~\cite{chatterjee2024emergentgeneralizedsymmetrymaximal}.
Outside the Virasoro minimal models, this minimal full field algebra
already has infinitely many inequivalent irreducible modules.

For a full field algebra
$\scrB=\bigoplus_{h,\bar h}\scrB_{h,\bar h}$ graded by chiral and antichiral conformal weights, define
\[
    \operatorname{Rel}(\scrB)
    :=\bigoplus_{0<h+\bar h<2}\scrB_{h,\bar h}.
\]

\begin{definition}[Stable full field algebra]
\label{def:stable-full-field-algebra}
The algebra $\scrB$ is \emph{stable} if
$\operatorname{Rel}(\scrB)=0$. When $\Mod_{\scrB}$ has the required
braided vertex tensor structure, we call the symTO $\Mod_{\scrB}$ stable
if $\scrB$ is stable.
\end{definition}

Here stability refers to the absence of relevant perturbations. Marginal
perturbations are outside the scope of this definition.

Although a stable full field algebra need not be chiral-factorized, for a
chiral-factorized algebra $\scrB=\scrV_L\otimes\overline{\scrV_R}$ of
CFT type, the stability condition can be simplified since both chiral components are nonnegatively $\Z$-graded.
The only possible nonidentity relevant fields then have weights $(1,0)$
or $(0,1)$, so stability is equivalent to
\[
    (\scrV_L)_1=(\scrV_R)_1=0.
\]

\subsubsection{Ising and Potts}

\paragraph{Ising.}\label{para:ising-extension-tower}
For the anomaly-free Ising full field algebra $\scrL_{\diagsector}$ of
Sec.~\ref{appsec:changing-gapless-boundary}, the Virasoro and maximal chiral
subalgebras coincide. The tower reduces to
\[
    \scrV^{\mathrm I}\otimes\overline{\scrV^{\mathrm I}}
    \hookrightarrow\scrL_{\diagsector}.
\]
The factorized vacuum algebra has no weight-one states and hence no
nonidentity field of scaling dimension below two. Its symTO
$\fZ(\Ising)$ is therefore stable, and the minimal chiral-factorized and
maximal symTOs coincide~\cite{chatterjee2024emergentgeneralizedsymmetrymaximal}.

\paragraph{Potts.}\label{para:potts-extension-tower}
For the three-state Potts full field algebra in
Eq.~\eqref{eq:potts-full-field-algebra}, the tower is
\[
    \scrV^{\mathrm P}\otimes\overline{\scrV^{\mathrm P}}
    \hookrightarrow\scrW^{\mathrm P}\otimes\overline{\scrW^{\mathrm P}}
    \hookrightarrow\scrL_{\mathrm{3Potts}}.
\]
The maximal and minimal chiral-factorized symTOs~\cite{chatterjee2024emergentgeneralizedsymmetrymaximal} are, respectively,
\[
    \cC_{6,5}\boxtimes\overline{\cC_{6,5}},
    \qquad \mathsf{6}_{4/5}\boxtimes\overline{\mathsf{6}_{4/5}}.
\]
The $\cW_3$ vacuum algebra has no weight-one states, so the minimal
chiral-factorized symTO is already stable.

\subsection{Summary of three uses of condensation}
\label{appsec:three-condensations-summary}

Table~\ref{tab:three-condensations} summarizes the uses of condensation
in this appendix. A condensation wall can either describe a change of
bulk topological order or act on an existing boundary to change its
symmetry protection.
\begin{table*}[t]
    \centering
    \small
    \renewcommand{\arraystretch}{1.2}
    \begin{tabular}{@{}>{\raggedright\arraybackslash}p{0.20\textwidth}
                        >{\raggedright\arraybackslash}p{0.39\textwidth}
                        >{\raggedright\arraybackslash}p{0.34\textwidth}@{}}
        \toprule
        Construction & Categorical data & Physical role \\
        \midrule
        Gapped condensation
        & A condensable algebra $A\in\cC$, with condensed bulk $\cC_A^{\localmodules}$
          and wall $\cC_A$.
        & Changes the bulk topological order. For Lagrangian $A$, the wall
          is a gapped boundary to the vacuum. \\
        \addlinespace
        Injective gapless boundary
        & $\scrV\hookrightarrow\scrA$, with
          $\cC\cong(\Mod_\scrV)_{\scrA}^{\localmodules}$.
        & Topological Wick rotation turns $(\Mod_\scrV)_{\scrA}$
          into an enriched gapless boundary with trivial sector $\scrA$. \\
        \addlinespace
        Change of symmetry
        & Composition of an existing boundary with a condensation wall.
        & Changes the symmetry protection and recombines sector Hilbert
          spaces by Eq.~\eqref{eq:gapless-sector-after-condensation}. \\
        \bottomrule
    \end{tabular}
    \caption{Three uses of condensation in the symTO description.}
    \label{tab:three-condensations}
\end{table*}

\section{Categorical data and factorizations for the examples}
\label{app:categorical-data}

This appendix collects the fusion-category and UMTC data used by the
examples in the main text.

We use $\cM(p,q)$ for the Virasoro minimal model and $\cC_{p,q}$ for
the UMTC of its chiral sectors. When no ambiguity arises,
$\mathsf{N}_c$ denotes a rank-$N$ UMTC with chiral central charge
$c$ modulo $8$. The symbol $\boxtimes$ denotes the Deligne tensor product.
The Drinfeld center of $\cC$ is denoted by $\fZ(\cC)$.
For a braided category $\cC$, let $\overline{\cC}$ denote its braided reverse.
It has the same objects and fusion data as $\cC$,
but inverse braiding and twists. Note that the bar is attached to the
category, not to its simple-object labels.

\subsection{Ising category}
\label{appsec:ising-category}

The Ising UMTC
$\Ising\cong\cC_{4,3}\cong\mathsf{3}_{1/2}$ has central charge
$c=\frac{1}{2}$ and total quantum dimension $D=2$. Its three simple
objects are $\one$, $\psi$, and $\sigma$. Their fusion rules and
topological data are listed in Tab.~\ref{tab:ising-mtc}.

\begin{table}[tb]
    \centering
    \caption{Simple objects, topological data, and fusion rules of the
    Ising UMTC.}
    \label{tab:ising-mtc}
    \begin{tabular}{cccc}
        \toprule
        Simple object & $d$ & Topological spin $s$ & Conformal weight $h$ \\
        \midrule
        $\one$   & $1$   & $0$        & $0$ \\
        $\psi$   & $1$   & $1/2$      & $1/2$ \\
        $\sigma$ & $\sqrt{2}$ & $1/16$ & $1/16$ \\
        \bottomrule
    \end{tabular}
    \qquad
    \begin{tabular}{c}
        \toprule
        Fusion rules \\
        \midrule
        $\psi\otimes\psi = \one$ \\
        $\sigma\otimes\sigma = \one\oplus\psi$ \\
        $\psi\otimes\sigma = \sigma\otimes\psi = \sigma$ \\
        \bottomrule
    \end{tabular}
\end{table}

In an orthonormal gauge, the only nonscalar $F$-matrix is
$F^{\sigma\sigma\sigma}_\sigma$. In the basis $\{\one,\psi\}$ of
intermediate fusion channels,
\begin{equation}
    F^{\sigma\sigma\sigma}_\sigma
    = \frac{1}{\sqrt{2}}
    \begin{pmatrix}
        1 & 1 \\
        1 & -1
    \end{pmatrix}.
\end{equation}
The remaining nonzero $F$-symbols are scalars. We use the standard
braiding convention
\begin{equation}
\begin{aligned}
    R^{\sigma\sigma}_\one &= \ee^{-\ii\pi/8}, &
    R^{\sigma\sigma}_\psi &= \ee^{3\ii\pi/8},\\
    R^{\psi\psi}_\one &= -1, &
    R^{\psi\sigma}_\sigma=R^{\sigma\psi}_\sigma &= -\ii.
\end{aligned}
\end{equation}
Braiding with the tensor unit is trivial.

Since $\Ising$ is modular,
$\fZ(\Ising)\cong\Ising\boxtimes\overline{\Ising}$. Its nine simple objects
are $x\boxtimes y$, where $x,y\in\{\one,\psi,\sigma\}$. Here $y$ belongs to
the braided-reversed factor $\overline{\Ising}$. The object itself is
not decorated with a bar. The forgetful functor sends
$x\boxtimes y$ to $x\otimes y$. Its half-braiding with $a\in\Ising$ combines
the braiding of $x$ with $a$ and the inverse braiding of $a$ with $y$.

The modular matrices factorize accordingly:
\[
    S^{\fZ(\Ising)} = S^{\Ising}\otimes \overline{S^{\Ising}},\qquad
    T^{\fZ(\Ising)} = T^{\Ising}\otimes \overline{T^{\Ising}}.
\]
The topological spin of the sector $x\boxtimes y$ is
$s_x-s_y$ modulo $1$, and agrees with the conformal spin of the
corresponding doubled primary. The anyon-insertion operators described
in Sec.~\ref{sec:modular-covariance} implement the modular $S$ action.

\subsection{Fibonacci category}
\label{appsec:fibonacci-category}

The Fibonacci UMTC $\Fib$ has simple objects
$\one$ and $\tau$. Both are self-dual, and
\[
    \tau\otimes\tau=\one\oplus\tau,
    \qquad
    d_\tau=\varphi=\frac{1+\sqrt{5}}{2}.
\]
In an orthonormal gauge, its nonscalar $F$-matrix and nontrivial
braiding eigenvalues are
\begin{align}
    F^{\tau\tau\tau}_\tau & =
    \begin{pmatrix}
        \varphi^{-1} & \varphi^{-1/2} \\[2pt]
        \varphi^{-1/2} & -\varphi^{-1}
    \end{pmatrix},\\
    R^{\tau\tau}_\one & = \ee^{-4\pi\ii/5}, &
    R^{\tau\tau}_\tau & = \ee^{3\pi\ii/5}.
\end{align}
The topological spin of $\tau$ is $s_\tau=2/5$. Reversing the
braiding gives $\overline{\Fib}$, for which
$s_\tau=-2/5$.

\subsection{Tricritical Ising category}
\label{appsec:tricritical-ising-category}

The tricritical Ising CFT $\cM(5,4)$ has central charge $c=7/10$
and six Virasoro primaries. Its rank-$6$ UMTC factorizes as
\[
\cC_{5,4}\cong\mathsf{6}_{7/10}
\cong\overline{\Fib}\boxtimes\mathsf{3}_{7/2}.
\]
The factor $\mathsf{3}_{7/2}$ is one of the eight Ising-type UMTCs
in Kitaev's 16-fold way. It has the same fusion rules and $F$-symbols
as $\Ising$, but different braiding. Its simple objects carry spins
\[
s_{\one} = 0,\quad s_{\psi} = \frac{1}{2},\quad s_{\sigma} = \frac{7}{16},
\]
and has central charge $c=7/2$.

We use the notation and primary--object correspondence listed in
Tab.~\ref{tab:triising-mtc}.

The UMTC factorization does not extend to the chiral CFT. Indeed,
$\sigma'$ and $\psi'$ correspond to $\tau\boxtimes\one$ and
$\one\boxtimes\psi$, whereas $\epsilon$ corresponds to $\tau\boxtimes\psi$.
Their conformal weights satisfy
\[
    h_{\epsilon}=\frac{1}{10}\neq
    h_{\sigma'}+h_{\psi'}=\frac{3}{5}+\frac{3}{2}
    =\frac{21}{10}.
\]
The UMTC determines $h_i$ only modulo $1$ through the topological
spin $s_i$. The integer shifts $n_i=h_i-s_i$ must therefore be
specified separately in the bootstrap.

\begin{table}[tb]
    \centering
    \caption{The six Virasoro primaries of the tricritical Ising CFT
    $\cM(5,4)$ and their correspondence with simple objects under
    $\cC_{5,4}\cong\overline{\Fib}\boxtimes\mathsf{3}_{7/2}$.}
    \label{tab:triising-mtc}
    \begin{tabular}{cccccc}
        \toprule
        $\cM(5,4)$ field & Kac label & $h$ & $s$ & $d$ & Factorization \\
        \midrule
        $\one$      & $(1,1)$ & $0$    & $0$      & $1$               & $\one\boxtimes\one$ \\
        $\epsilon$  & $(1,2)$ & $1/10$ & $1/10$   & $\varphi$         & $\tau\boxtimes\psi$ \\
        $\sigma'$   & $(1,3)$ & $3/5$  & $-2/5$   & $\varphi$         & $\tau\boxtimes\one$ \\
        $\psi'$     & $(1,4)$ & $3/2$  & $1/2$    & $1$               & $\one\boxtimes\psi$ \\
        $\sigma''$  & $(2,2)$ & $3/80$ & $3/80$   & $\sqrt{2}\varphi$ & $\tau\boxtimes\sigma$ \\
        $\sigma'''$ & $(2,4)$ & $7/16$ & $7/16$   & $\sqrt{2}$        & $\one\boxtimes\sigma$ \\
        \bottomrule
    \end{tabular}
\end{table}

\subsection{Three-state Potts categories}
\label{appsec:potts-categories}

The critical three-state Potts CFT has central charge
$c=\frac{4}{5}$. We use its six irreducible $\cW_3$ modules, whose
UMTC factorizes as
\[
    \mathsf{6}_{4/5}
    \cong \Fib\boxtimes\overline{\mathsf{3}_2}.
\]
The extension from the $\cM(6,5)$ Virasoro algebra to the $\cW_3$
algebra is reviewed in Sec.~\ref{appsec:partition-vector-and-modular-covariance}.

The simple objects of $\mathsf{3}_2$ are labelled by
$a\in\Z_3$, with
\[
    a\otimes b=a+b \pmod 3.
\]
In the convention used here, the associator is trivial. Forgetting the
braiding therefore gives
\[
    \mathsf{3}_2\cong\Hilb_{\Z_3}\cong\Rep(\Z_3)
\]
as fusion categories. This equivalence is not braided.
$\Rep(\Z_3)$ has symmetric braiding and $\theta_a=1$, whereas
$\mathsf{3}_2$ has
\[
    c_{a,b}=\ee^{2\pi\ii ab/3},
    \qquad
    \theta_a=\ee^{2\pi\ii a^2/3}.
\]
Thus $\theta_1=\theta_2=\ee^{2\pi\ii/3}$, and the corresponding
$S$-matrix is non-degenerate. The braided reverse
$\overline{\mathsf{3}_2}$ has conjugate twists
and $s_1=s_2=-\frac{1}{3}$. The remaining data are listed in
Tab.~\ref{tab:3-2-data}.

\begin{table}[tb]
    \centering
    \caption{Topological data of the pointed UMTC $\mathsf{3}_2$.}
    \label{tab:3-2-data}
    \begin{tabular}{cccc}
        \toprule
        Charge $a$ & $d_a$ & $\theta_a$ & $s_a$ \\
        \midrule
        $0$ & $1$ & $1$               & $0$ \\
        $1$ & $1$ & $\ee^{2\pi\ii/3}$ & $\frac{1}{3}$ \\
        $2$ & $1$ & $\ee^{2\pi\ii/3}$ & $\frac{1}{3}$ \\
        \bottomrule
    \end{tabular}
\end{table}

We use the $\cW_3$ primary notation and primary--object
correspondence listed in Tab.~\ref{tab:potts-mtc}. The factorization
of $\mathsf{6}_{4/5}$ is a statement about the UMTC, not the chiral
CFT. The fields $\epsilon$ and $\psi_1$ correspond to
$\tau\boxtimes 0$ and $\one\boxtimes 1$, while their tensor product
$\tau\boxtimes 1$ corresponds to $\sigma_1$. However,
\[
    h_{\sigma_1}=\frac{1}{15},
    \qquad
    h_\epsilon+h_{\psi_1}
    =\frac{2}{5}+\frac{2}{3}
    =\frac{16}{15}.
\]
These conformal weights differ by $1$ and hence give the same
topological spin. The integer shifts $n_i=h_i-s_i$ must therefore be
specified separately in the bootstrap.

For comparison, the ten Virasoro sectors form
\begin{equation}
\label{eq:m65-category-factorization}
    \cC_{6,5}\cong\Fib\boxtimes\JonesKauffman_4,
\qquad
    \JonesKauffman_4\cong\overline{\mathsf{5}^{\mathrm b}_2},
\end{equation}
where $\JonesKauffman_4$ is the rank-$5$ Jones--Kauffman category
\cite{Levaillant:2015nia}. Its fusion rules and quantum dimensions
agree with those of $\mathrm{SU}(2)_4$, but the braiding is different.
the charge-$2$ twist is $\ee^{-2\pi\ii/3}$ in $\JonesKauffman_4$ and
$\ee^{2\pi\ii/3}$ in $\mathrm{SU}(2)_4$. The doubled category
$\JonesKauffman_4\boxtimes\overline{\JonesKauffman_4}$ also appears as the
enhanced symTO of self-dual $S_3$ and $\Rep(S_3)$ spin chains
\cite{Chatterjee_2024_JK4}. Its simple objects are listed in
Tab.~\ref{tab:jk4data}~\cite{Wen_2015}.

\begin{table}[tb]
    \centering
    \caption{Quantum dimensions and topological twists of
    $\JonesKauffman_4\cong\overline{\mathsf{5}^{\mathrm b}_2}$.}
    \label{tab:jk4data}
    \begin{tabular}{ccc}
        \toprule
        Charge $a$ & $d_a$ & $\theta_a$ \\
        \midrule
        $0$ & $1$        & $1$ \\
        $1$ & $\sqrt{3}$ & $-\ee^{\ii\pi/4}$ \\
        $2$ & $2$        & $\ee^{-2\pi\ii/3}$ \\
        $3$ & $\sqrt{3}$ & $\ee^{\ii\pi/4}$ \\
        $4$ & $1$        & $1$ \\
        \bottomrule
    \end{tabular}
\end{table}

The Potts modular invariant is diagonal relative
to the $\cW_3$ algebra and non-diagonal relative to the Virasoro
algebra. The associated condensation is described in
Sec.~\ref{appsec:partition-vector-and-modular-covariance}. In Tab.~\ref{tab:potts-mtc},
the repeated Virasoro sectors $(1,3)$ and $(3,3)$ split into distinct
$\cW_3$ modules.

\begin{table}[tb]
    \centering
    \caption{The six irreducible $\cW_3$ modules of the critical
    three-state Potts CFT, their Virasoro content in $\cM(6,5)$, and
    their simple-object labels under
    $\mathsf{6}_{4/5}\cong\Fib\boxtimes\overline{\mathsf{3}_2}$.}
    \label{tab:potts-mtc}
    \begin{tabular}{cccccc}
        \toprule
        $h$ & $s$ & $d$ & $\cW_3$ field & $\cM(6,5)$ field & Factorization \\
        \midrule
        $0$        & $0$      & $1$        & $\one$      & $(1, 1) \oplus (4, 1)$ & $\one\boxtimes 0$ \\
        $2/5$      & $2/5$    & $\varphi$  & $\epsilon$  & $(2, 1) \oplus (3, 1)$ & $\tau\boxtimes 0$ \\
        $2/3$      & $-1/3$   & $1$        & $\psi_1$    & $(1, 3)$               & $\one\boxtimes 1$ \\
        $2/3$      & $-1/3$   & $1$        & $\psi_2$    & $(1, 3)$               & $\one\boxtimes 2$ \\
        $1/15$     & $1/15$   & $\varphi$  & $\sigma_1$  & $(3, 3)$               & $\tau\boxtimes 1$ \\
        $1/15$     & $1/15$   & $\varphi$  & $\sigma_2$  & $(3, 3)$               & $\tau\boxtimes 2$ \\
        \bottomrule
    \end{tabular}
\end{table}

\section{The \texorpdfstring{$4\xt 4$}{4 x 4} patch for computing structure constants}\label{app:four-by-four-structure-constants}
This appendix gives the implementation details of the enlarged triangular
patch used to construct $\mathcal A_{abc}$ in
Sec.~\ref{sec:structure-constants}.

To further suppress irrelevant perturbations, we choose a larger triangular
patch than in Ref.~\cite{Li_2022}. Starting from the tensors $T_A,T_B$ at the
chosen RG step, we solve the loop-optimization problem in
Eq.~\eqref{eq:loop} and keep all eight tensors $S_1,\ldots,S_8$.
We use these tensors to represent one triangular half of a $4\xt4$ block,
\[
\raisebox{-59.95091pt}{\includegraphics{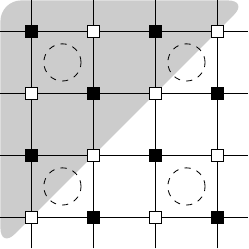}} \approx
\raisebox{-54.40263pt}{\includegraphics{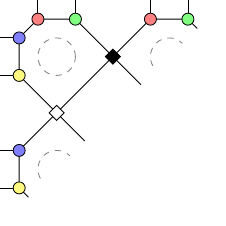}}.
\]
Here the black and white tensors are the local Loop-TNR tensors at the next RG
step, as in Eq.~\eqref{eq:next-ta-tb}. The dashed circles mark the plaquettes
used in the loop optimization. Colors distinguish the eight tensor types,
and arrows on tensor legs are omitted.

For the contractions below, we group the optimized tensors into three
four-leg tensors $U,T,D$, with the numbering of Eq.~\eqref{eq:loop},
\[
\vcenter{\hbox{\includegraphics[width=.98\linewidth]{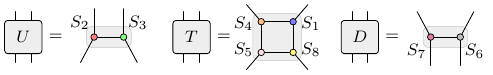}}}.
\]
The shaded regions specify the contractions.
We call this grouping UTD. The tensor $T$ joins adjacent plaquettes and is
one of the coarse-grained tensors in Eq.~\eqref{eq:next-ta-tb}.
The up and left groupings differ by a quarter turn and use the similar
optimized tensors.

The transfer eigenstates in the horizontal and vertical directions are
defined by
\[
\raisebox{-41.36166pt}{\includegraphics{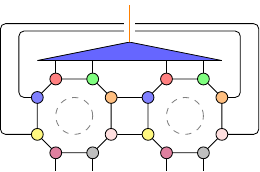}} = \lambda^{y}_{b}
\raisebox{-18.08296pt}{\includegraphics{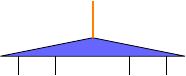}},
\]
\[
\raisebox{-62.44035pt}{\includegraphics{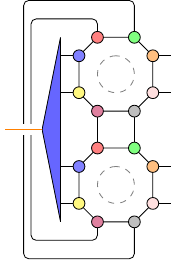}} = \lambda^{x}_{a}
\raisebox{-44.65741pt}{\includegraphics{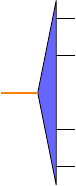}},
\]

The transfer action contracts the shaded regions with the state in the
numbered order,
\[
\vcenter{\hbox{\includegraphics[width=.90\linewidth]{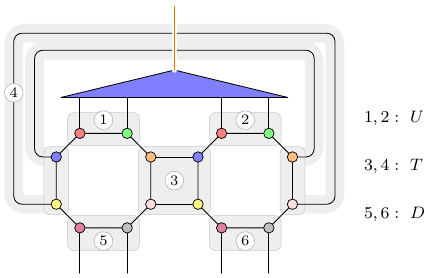}}}.
\]
Region 4 is a single $T$ across the periodic cut. Each step contracts
two lattice legs and leaves four, giving a cost $\cO(\chi^6)$ at fixed
anyon label.
This action is passed to the Krylov eigensolver of
\texttt{KrylovKit.jl}~\cite{haegeman_2025_14911568}.

The plaquette choice for the tilted direction follows Bao's observation
that Loop-TNR gives more accurate results when optimization loops at
successive RG steps are concentric~\cite[Sec.~2.1 and Fig.~2.2(g)]{Bao:2019hfc}.
We therefore require each loop in the second optimization to share its
center with a loop in the first. This fixes the second-step plaquettes
and gives the staggered arrangement below, shifted relative to a rotated
copy of the horizontal or vertical transfer diagram. One optimized loop
lies within the strip, while the other closes across the periodic cut.

The two $T$ tensors from the up and left groupings form the coarse-grained
lattice, with legs rotated as in Eq.~\eqref{eq:next-ta-tb}. We optimize
the selected plaquettes with an independent cutoff and group the resulting
tensors into $U,T,D$, obtaining the tilted transfer operator,
\[
\raisebox{-71.49681pt}{\includegraphics{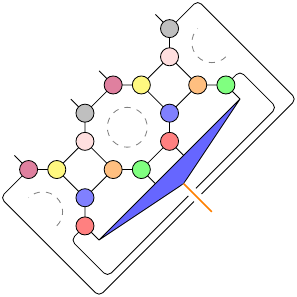}} =
\lambda^{z}_{c}
\raisebox{-37.54594pt}{\includegraphics{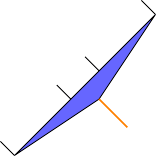}}
\]
The second loop optimization is an additional approximation in the tilted
transfer operator.

\Needspace{12\baselineskip}
The contraction order is marked directly in
Eq.~\eqref{eq:four-by-four-structure-amplitude}. In each shaded region 1,
we contract the eigenstate with two $U$ tensors and then the adjacent $T$.
Joining the two shared lattice bonds of these prepared states gives
region 2, bounded by the dashed contour. Its four remaining lattice bonds
are then contracted with the tilted state in step 3. A half-braiding places
the tilted state's periodic cut at the position shown. The leading cost is
$\cO(\chi^6)$ at fixed anyon labels. Finally, the surrounding $i$ loop
applies the center projector,
\begin{equation}\label{eq:four-by-four-structure-amplitude}
\cA_{abc^*}
=\frac{1}{D_\cR^2}\sum_{i\in\Irr\cR}d_i
\raisebox{-77.02759pt}{\includegraphics{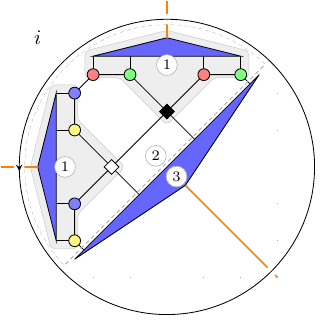}}
\end{equation}
Projecting Eq.~\eqref{eq:four-by-four-structure-amplitude} onto the basis
$v_{xy}^{z,\alpha}$ gives the components
$\mathcal A_{abc^*;\alpha}$ in
Eq.~\eqref{eq:structure-amplitude-decomposition}.

\section{Computation of geometric factors}
\label{app:geometric-factors}

In this appendix we derive the planar geometric factors $G_{\text{nut}}$ and
$G_{\text{pants}}$ used for the initial tensors in the topological-bootstrap
construction of Sec.~\ref{sec:stable-symto}.  We also derive the distinct
closed-string-pants factor $G_{\mathrm{cpants}}$ used only to extract bulk
structure constants from fixed-point tensors in
Sec.~\ref{sec:structure-constants}.  The derivation follows
Refs.~\cite{ueda2023fixedpointtensorfourpointfunction,Cheng:2023kxh}.

We use $\mathrm{pants}$ to label the flat chiral closed-string amplitudes,
equivalently described as open-string amplitudes in the construction
considered here. We reserve $\mathrm{cpants}$ for the full nonchiral
closed-string amplitudes.

\subsection{General principle}

The geometric factors in this appendix arise from chiral three-point
amplitudes on a sphere with specified local coordinates.  Let
\begin{equation*}
    (\Sigma;p_1,p_2,p_3;\xi_1,\xi_2,\xi_3)
\end{equation*}
be a Riemann sphere with three ordered punctures.  Each local
coordinate is defined only in a neighborhood $U_i$ of its puncture:
\begin{equation*}
    \xi_i:U_i\longrightarrow\C,
    \qquad
    \xi_i(p_i)=0.
\end{equation*}
Choose the unique biholomorphism
\begin{equation*}
    \psi:\widehat{\C}\longrightarrow\Sigma,
    \qquad
    \psi(0)=p_1,\quad
    \psi(1)=p_2,\quad
    \psi(\infty)=p_3.
\end{equation*}
Throughout this appendix, $\phi_i$ denotes an inserted state, or the
corresponding field under the state--field correspondence, whereas $\xi_i$
denotes an abstract local coordinate and $\varphi_i$ its explicit realization
in a particular geometry.  We use $h_i$ for chiral conformal weights and
$\Delta_i=h_i+\bar h_i$ for bulk scaling dimensions.  We reserve $\cO(\cdot)$
exclusively for Landau remainder terms.  We use the same formal parameter $x$
for the three local disks.

At
the first puncture, $f_1(x)$ denotes the germ of
$\xi_1\circ\psi$ expressed by setting $z=x$ near $0$.  At the
second puncture, $f_2(x)$ denotes the germ of
$\xi_2\circ\psi$ expressed by setting $z=1+x$ near $1$.  At the
outgoing puncture, $f_3(x)$ denotes the germ of
$\xi_3\circ\psi$ expressed by setting $z=x^{-1}$ near $\infty$.
Thus each $f_i$ is a holomorphic germ from $(\C,0)$ to $(\C,0)$,
although the three germs describe neighborhoods of different points
on $\widehat{\C}$.  

We use the same symbol $f_i$ for the substitution
operator that this germ induces on $\C[[x]]$. Its action on the
coordinate function $x$ is
\[
    f_i(x) = a^{(i)}_0\left(x+a^{(i)}_1x^2+a^{(i)}_2x^3+\cdots\right)
\]
where $a^{(i)}_0\neq0$.

To express the nonlinear part of each germ in Virasoro coordinates, define
$A_n^{(i)}$ recursively by
\begin{equation}\label{eq:formal-coordinate-change}
    x+a^{(i)}_1x^2+a^{(i)}_2x^3+\cdots
    =
    \exp\left(-\sum_{n\geq1}A_n^{(i)}l_n\right)x,
\end{equation}
where
\begin{equation*}
    l_n=-x^{n+1}\frac{\partial}{\partial x},
    \qquad
    [l_n,l_m]=(n-m)l_{n+m}.
\end{equation*}
Since
\[
l_nx^k=-k x^{k+n},
\]
the generator $l_n$ raises the $x$-adic order by $n$.  Consequently, the
coefficient of any fixed power $x^N$ in
Eq.~\eqref{eq:formal-coordinate-change} receives contributions from only
finitely many generators and their products.  The coefficients
$A_n^{(i)}$ can therefore be determined recursively, to any prescribed
finite order, by a finite triangular computation.

Equivalently, as an identity of substitution operators on $\C[[x]]$,
\[
f_i(x) = \exp\left(-\sum_{n\geq1}A_n^{(i)}l_n\right)\left(a^{(i)}_0\right)^{-l_0}x.
\]
More generally, every invertible formal germ
\[
f(z)=a_0z+a_1z^2+\cdots,\qquad a_0\neq0,
\]
determines a unique element of the completed nonnegative Witt group.  We
denote its pullback action on $\C[[z]]$ by $\cW(f)$:
\[
\cW(f)H(z):=H(f(z)),\qquad H\in\C[[z]].
\]
In particular, $\cW(f)z=f(z)$.  If $g$ is another such germ, then
\[
\begin{aligned}
\cW(f)\cW(g)H(z)
&=\cW(f)H(g(z))\\
&=H(g(f(z)))\\
&=\cW(g\circ f)H(z).
\end{aligned}
\]
Thus the pullback action is contravariant:
\[
\cW(f)\cW(g)=\cW(g\circ f).
\]
The formal composition is well defined because both germs have zero
constant term. The coefficient of every fixed power of $z$ receives only
finitely many contributions.

The nonnegative Witt algebra acts on the chiral state space through the
corresponding Virasoro operators.  The Virasoro central term vanishes on
$L_m,L_n$ with $m,n\geq0$, so this action exponentiates on the graded
completion without a projective correction.  The pullback operator
$\cW(f_i)$ is represented by the chiral coordinate-change operator
\begin{equation}\label{eq:chiral-coordinate-change-operator}
    \hat{\cW}(f_i)
    :=
    \exp\left(-\sum_{n\geq1}A_n^{(i)}L_n\right)\left(a^{(i)}_0\right)^{-L_0}.
\end{equation}
Here $L_n$ are the Virasoro operators, and
$\left(a^{(i)}_0\right)^{-L_0}$ is defined using a chosen logarithm.
For positive $a^{(i)}_0$, we use the real logarithm.
To keep the logarithm choices compatible with composition, we define the
logarithmic lift of a composite coordinate change as the sum of the
individual lifts.
With the same contravariant convention,
\[
\hat{\cW}(f)\hat{\cW}(g)=\hat{\cW}(g\circ f).
\]

We use the overall-normalization convention of Sec.~\ref{sec:cft-fixed-point-mechanism},
with amplitudes understood up to a nonzero scalar independent of the
external state labels. State-dependent local-coordinate factors are kept
explicit.

Let $\phi_1\in W_1$, $\phi_2\in W_2$, and
$\phi_3\in W_3^\prime$, where the third puncture is outgoing.  For an
intertwining operator $\cY$ of type
$\binom{W_3}{W_2\;W_1}$, the amplitude associated with the decorated sphere is
\[
\left\langle
    \hat{\cW}(f_3)\phi_3,
    \cY\bigl(\hat{\cW}(f_2)\phi_2,1\bigr)
    \hat{\cW}(f_1)\phi_1
\right\rangle_{\widehat{\C}}.
\]
The pairing contracts the output of $\cY$ with the state in the contragredient
module at the outgoing puncture.

If $\phi_i$ are primary states of conformal weights $h_i$, then
$L_n\phi_i=0$ for every $n>0$, and the nonlinear coordinate changes
drop out.  With
\begin{equation*}
    C_{123}
    :=
    \langle\phi_3,\cY(\phi_2,1)\phi_1\rangle_{\widehat{\C}},
\end{equation*}
the amplitude reduces to
\begin{equation}
    \prod_{i = 1}^3\left(a_0^{(i)}\right)^{-h_i}C_{123}.
\end{equation}
The geometric factor is
\[
G_{\triangle}(h_1, h_2, h_3) = \prod_{i = 1}^3\left(a_0^{(i)}\right)^{-h_i}.
\]

\subsection{Closed-string pants}
\label{appsec:closed-string-pants}

Following \cite{ueda2023fixedpointtensorfourpointfunction}, we review a flat
realization of the thrice-punctured sphere and derive the closed-string pants
factor in Eq.~\eqref{eq:pants-gauge-app}.

This subsection concerns bulk three-point functions in a doubled,
nonchiral CFT.  The closed-string pants are used in this work as a readout geometry for the
doubled-CFT structure constants of a fixed-point tensor.

Since any three points on the Riemann sphere can be mapped to
$\{0,1,\infty\}$ by a global conformal transformation, we may choose, without
loss of generality,
\begin{equation}
z = \frac{L}{2\pi}\bigl(- \log w - \ii \log (w - 1)\bigr),
\end{equation}
where $L>0$ fixes an overall scale that is immaterial in a scale-invariant
theory.  The logarithms make $z$ multivalued.  Their periods specify the
circumferences of the flat cylinders, while the branch cuts chosen below
select a fundamental domain for this geometry.

The asymptotic behavior is
\begin{align*}
w&\rightarrow 0 \;\Longrightarrow\; z\rightarrow \infty,\\
w&\rightarrow 1 \;\Longrightarrow\; z \rightarrow \ii \infty,\\
w&\rightarrow\infty \;\Longrightarrow\; z\rightarrow (-1-\ii)\infty.
\end{align*}
The differential
\[
\dd z = \frac{L}{2\pi}\Bigl(-\frac{1}{w} - \frac{\ii}{w-1}\Bigr)\dd w
\]
has simple poles at the three punctures.  Its periods around positively
oriented puncture loops are
\[
\begin{array}{c|ccc}
\text{puncture}&0&1&\infty\\ \hline
\Delta z&-\ii L&L&-(1-\ii)L\\
\end{array}.
\]
The three period vectors sum to zero, as required by the residue theorem.
They determine the directions and circumferences of the three semi-infinite
flat cylinders. Two have circumference $L$, while the diagonal one has
circumference $\sqrt{2}L$. These period data fix the shape used here.

To fix the branch cuts, let $z(w)=u(x,y)+\ii v(x,y)$ on a chosen branch.
Near $w=0$, the level sets of $u$ are the closed circles around the first
cylindrical end. Near $w=1$, it is instead the level sets of $v$ that form
closed circles, while the level sets of $u$ run radially.  The corresponding
critical level sets meet at the point $w_0$ where $\nabla u=0$.  By the
Cauchy--Riemann equations,
$\nabla u = 0$ if and only if $z'(w) = 0$, giving
\[
w_0 = \frac{1}{2}(1-\ii).
\]
The same $w_0$ is critical for $v$ and for $u+v$.
The critical level sets provide a natural choice of branch cuts.
\begin{itemize}
    \item the iso-$u$ line connecting $w_0$ to $1$;
    \item the iso-$v$ line connecting $w_0$ to $0$;
    \item and the iso-$u+v$ line connecting $w_0$ to $\infty$.
\end{itemize}

After cutting along these three lines, the resulting fundamental domain is
mapped by $z(w)$ to a flat pair of pants
(Fig.~\ref{fig:pants-conformal-map}).  The critical point $w_0$ is the trivalent junction
of the critical graph.  Changing branches merely selects adjacent copies of
the same flat fundamental domain.

\begin{figure}[ht]
\centering
\includegraphics[width=0.8\linewidth]{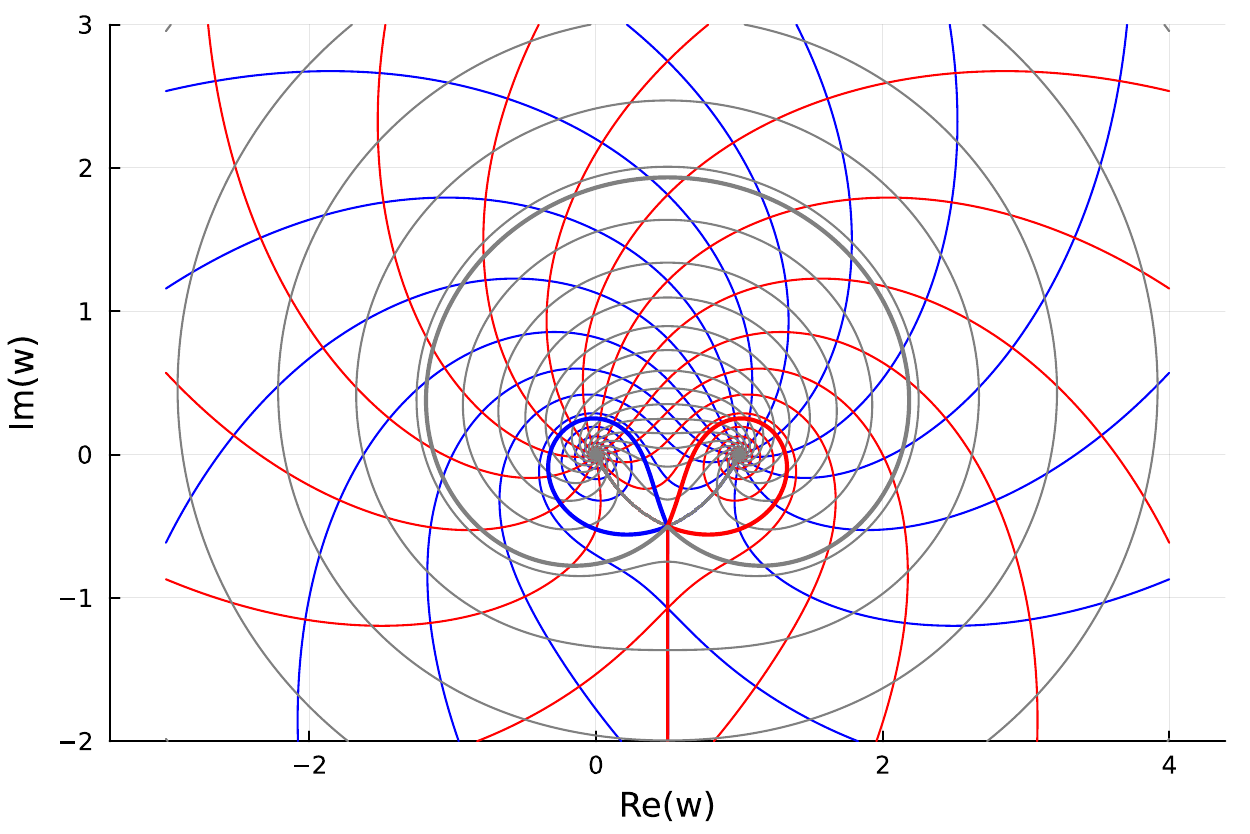}
\qquad$\mapsto$\qquad
\raisebox{-71.53186pt}{\includegraphics{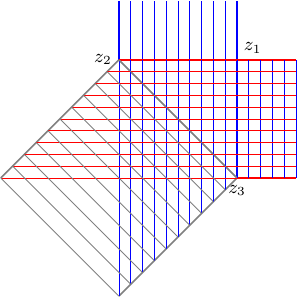}}
\caption{Cutting the Riemann sphere along the three colored lines (left)
produces a fundamental domain that is mapped conformally to a flat pair of
pants (right).  The two banks of each cut map to a pair of identified seams on
the corresponding cylindrical leg.}
\label{fig:pants-conformal-map}
\end{figure}

The three vertices $z_1,z_2,z_3$ of the pants are the images of $w_0$ approached
through the three sectors separated by the cuts:
\begin{align*}
z_1 &= z(w_0+\ii 0^+)
     = \frac{L}{2\pi}\Bigl[\frac{1}{2}(1+\ii)\log 2 +\frac{\pi}{4}(5+\ii)\Bigr],\\
z_2 &= z(w_0 + 0^+),\qquad z_3 = z(w_0 - 0^+),
\end{align*}
satisfying $z_2 = z_1 - L$ and $z_3 = z_1 - \ii L$.
Thus the compact core is a right isosceles triangle with side vectors $L$,
$-\ii L$, and $-(1-\ii)L$. 

We next normalize the local coordinates at the three cylindrical ends
$+\infty$, $\ii\infty$, and $(-1-\ii)\infty$ so that the selected cutoff
circles at their mouths are mapped to unit circles. Using the offsets
$z_1,z_2$, we choose the local coordinates as follows.
\begin{align*}
\varphi_1(z) &= \ee^{-\frac{2\pi (z-z_1)}{L}},\\
\varphi_2(z) &= \ee^{\frac{2\pi\ii (z-z_1)}{L}},\\
\varphi_3(z) &= \ee^{-\frac{2\pi \ii(z - z_2)}{(1-\ii)L}}.
\end{align*}
The offsets fix the remaining scale and phase information.
Expanding in the standard local parameters $w$, $w-1$, and $1/w$ at the three
punctures, respectively, gives
\begin{align*}
\varphi_1(z(w)) &= \exp\Bigl[\frac{1}{4}(1+\ii)\pi+\frac{1}{2}(1+\ii)\log 2\Bigr]\,w
            + \cO(w^2),\\
\varphi_2(z(w)) &= (-1+\ii)\,2^{-\ii/2}\ee^{\pi/4}\,(w-1) + \cO((w-1)^2),\\[4pt]
\varphi_3(z(w)) &= \exp\Bigl[\frac{1}{4}(1+6\ii)\pi-\frac{1}{2}\log 2\Bigr]
            \,\frac{1}{w} + \cO(w^{-2}).
\end{align*}
For closed-string primary fields
$\phi^{\closedsector}_1,\phi^{\closedsector}_2,\phi^{\closedsector}_3$, the norm of the three-point
function picks up the geometric factor
\[
\left|\la\phi^{\closedsector}_1\phi^{\closedsector}_2\phi^{\closedsector}_3\ra_{\mathrm{cpants}}\right|
= \left|C_{123}^{\closedsector}\right|G_{\mathrm{cpants}}(\Delta_1, \Delta_2, \Delta_3).
\]
Here we define
\begin{equation}\label{eq:pants-gauge-app}
    G_{\mathrm{cpants}}(\Delta_1, \Delta_2, \Delta_3) = \ee^{-\frac{\pi}{4}(\Delta_1+\Delta_2+\Delta_3)}\,
  (\sqrt{2})^{-\Delta_1-\Delta_2 + \Delta_3}.
\end{equation}
The geometric factor $G_{\mathrm{cpants}}$ derived for this continuum geometry
is used in the approximate prescription for extracting primary-field
structure constants from numerical fixed-point tensors in
Sec.~\ref{sec:structure-constants}.

\subsection{Nut}
\label{appsec:nut}

We now connect the conformal map of App.~\ref{appsec:closed-string-pants} to the
fixed-point tensor construction of \cite{Cheng:2023kxh}. The ``nut'' diagram is a
specific triangulation of the plane by isosceles right triangles, shown in
Fig.~\ref{fig:nut-diagram-app}.

\begin{figure}[ht]
\centering
\raisebox{-60.10948pt}{\includegraphics{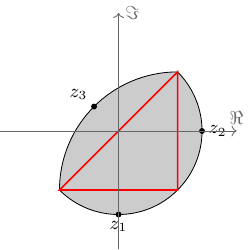}}
\caption{The ``nut'' diagram has three field insertions at $z_1,z_2,z_3$ on
the boundary of a region triangulated by isosceles right triangles.}
\label{fig:nut-diagram-app}
\end{figure}

The three insertion points are
\[
z_1 = -\sqrt{2}\ii,\qquad z_2 = \sqrt{2},\qquad
z_3 = 1-\sqrt{2} + (\sqrt{2}-1)\ii.
\]
We choose local coordinates such that the upper rim of the arc is mapped to the
unit semicircle.
\[
\raisebox{-24.84215pt}{\includegraphics{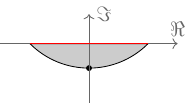}}
\rightarrow
\raisebox{-31.95532pt}{\includegraphics{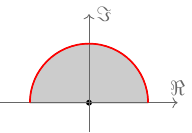}}
\]

Define
\begin{align*}
\xi &= \frac{1 + z}{1-z},\qquad
\eta = \ii\xi^2,\qquad
w = \frac{\eta-1}{\eta+1},\\
\varphi(z) &= w\circ\eta\circ\xi(z)
     = \ii + \frac{4z}{1 + z(2\ii + z)},
\end{align*}
which sends the arc to the unit circle. The three local coordinates are then
\begin{align*}
\varphi_1(z) &= \varphi(z+\ii),\qquad
\varphi_2(z) = \varphi(-\ii(z-1)),\\
\varphi_3(z) &= \varphi\Bigl(\frac{1}{\sqrt{2}}\ee^{\frac{3\pi\ii}{4}}z\Bigr).
\end{align*}

Following \cite{Cheng:2023kxh}, we map the nut to the upper half-plane.
Define
\begin{align*}
u &= -\ii\frac{z+1+\ii}{z-1-\ii},\qquad
v = u^{\frac{4}{3}},\\
\chi(z) &= v\circ u(z)
        = \Bigl(-\ii\frac{z+1+\ii}{z-1-\ii}\Bigr)^{\!\frac{4}{3}}.
\end{align*}
The insertion points map to
\[
v_1 = (\sqrt{2}-1)^{\frac{4}{3}},\qquad
v_2 = (\sqrt{2}+1)^{\frac{4}{3}},\qquad
v_3 = -1.
\]
The inverse map is
\[
z = \chi^{-1}(v) = (1+\ii)\frac{v^{\frac{3}{4}}-\ii}{v^{\frac{3}{4}}+\ii}.
\]
The local coordinates on the upper half-plane are
$f_i = \varphi_i\circ\chi^{-1}$ ($i=1,2,3$). Their derivatives at the insertion
points are
\begin{align*}
f_1'(v_1) &= \frac{3}{4}\Bigl(2+\frac{3}{\sqrt{2}}\Bigr)(\sqrt{2}+1)^{\frac{1}{3}},\\
f_2'(v_2) &= \frac{3}{4}\frac{1}{\sqrt{2}}(\sqrt{2}-1)^{\frac{1}{3}},\\
f_3'(v_3) &= \frac{3}{4}.
\end{align*}
Assembling the conformal factors from the three local coordinates gives the geometric
factor for the nut gauge:
\begin{multline}\label{eq:nut-gauge-app}
G_{\text{nut}}(h_1,h_2,h_3)
= \Bigl[\frac{4\sqrt{2}}{3}\frac{1}{A(A^{\frac{4}{3}}-A^{-\frac{4}{3}})}
   \Bigr]^{h_1+h_2} \\
   \times\Bigl[\frac{4}{3}\frac{A^{\frac{4}{3}}-A^{-\frac{4}{3}}}
               {(1+A^{-\frac{4}{3}})(1+A^{\frac{4}{3}})}
   \Bigr]^{h_3},
\qquad A = \sqrt{2}+1.
\end{multline}
Together with the pants gauge in Eq.~\eqref{eq:flat-pants-gauge-app}, this yields
the two planar gauges listed in the topological-bootstrap construction of
Sec.~\ref{sec:stable-symto}.

\subsection{Flat pants}
\label{appsec:flat-pants}

We consider another geometry of a three-punctured Riemann surface:
\[
\raisebox{-74.63445pt}{\includegraphics{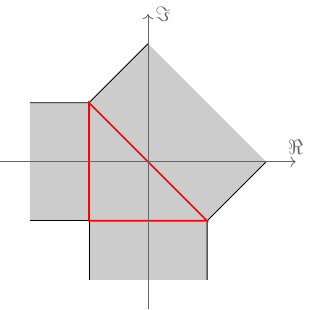}}
\]
Throughout this subsection, $\zeta$ denotes the upper-half-plane
coordinate, $t=t(\zeta)$ is the auxiliary coordinate used to rationalize
the Schwarz--Christoffel differential, and $u$ denotes the flat coordinate
on the polygon.  The global conformal map is $u=\rho(\zeta)$.  Its expression
in the auxiliary coordinate will be denoted by $g(t)$, so that
$\rho(\zeta)=g(t(\zeta))$.  The explicit local
coordinates $\varphi_i$ below are functions of $u$.

In the flat coordinate $u$, the three punctures are at $-\infty$,
$-\ii\infty$, and $\ee^{\pi\ii/4}\infty$, respectively.  The exponential
map sends each semi-infinite strip to a half-disk:
\[
\raisebox{-31.95532pt}{\includegraphics{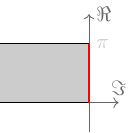}}
\rightarrow
\raisebox{-31.95532pt}{\includegraphics{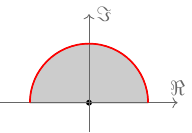}}
\]

Three local coordinates are
\begin{align*}
    \varphi_1(u) &= \exp\left(\frac{\pi}{2}(u+1+\ii)\right),\\
    \varphi_2(u) &= \exp\left(\frac{\pi}{2}(-\ii u+1+\ii)\right),\\
    \varphi_3(u) &= \exp\left(\frac{\pi}{2}\left(\frac{\ee^{\frac{3\pi\ii}{4}}}{\sqrt{2}}u + \ii\right)\right).
\end{align*}

\paragraph{Global map and normalization}

The upper half-plane can be mapped to the shaded polygon by a
Schwarz--Christoffel map of the form
\[
\rho(\zeta) = B+A\int^\zeta
\frac{s^{\frac12}(s-a)^{\frac14}(s+a)^{\frac14}}{s^2-1}\dd s,
\]
where $a>1$.  The fractional powers are defined by analytic continuation in
the upper half-plane.  Since all their branch points lie on the real axis,
the integrand is single-valued and holomorphic in the open upper half-plane.
Consequently, the integral is independent of the contour within this domain.
The constants $A$ and $B$ fix the overall scale, rotation, and translation.
We impose the following boundary correspondence:
\[
\begin{array}{c|cccccc}
\zeta
&-a&-1&0&1&a&\infty\\ \hline
\rho(\zeta)
&-1+\ii&-\infty&-1-\ii&-\ii\infty&1-\ii
&\ee^{\pi\ii/4}\infty
\end{array}
\]

The Schwarz--Christoffel differential becomes rational after the substitution
\[
t=t(\zeta)=\frac{(\zeta-a)^{\frac14}(\zeta+a)^{\frac14}}{\zeta^{\frac12}},
\qquad t^4=1-\frac{a^2}{\zeta^2}.
\]
With the upper-half-plane branches chosen above, this map sends the upper
half plane to
\[
\Omega_t=
\left\{t:-\frac{\pi}{4}<\arg t<\frac{\pi}{4}\right\}
\setminus(0,1].
\]
\begin{align*}
    \raisebox{-28.65273pt}{\includegraphics{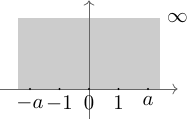}}\xrightarrow{t}
&\raisebox{-46.0567pt}{\includegraphics{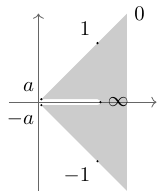}}\\
& \downarrow\\
&
\raisebox{-58.07178pt}{\includegraphics{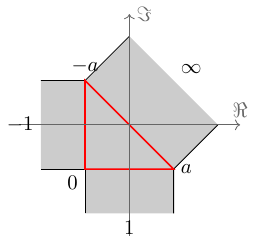}}
\end{align*}

Let
\[
b:=\lim_{\epsilon\downarrow0}t(1+\ii\epsilon)
=\ee^{\pi\ii/4}(a^2-1)^{\frac14}.
\]
In the $t$ coordinate it becomes
\[
\dd u
=A\left(\frac{2}{1-t^4}
-\frac{2}{1-(t/b)^4}\right)\dd t.
\]
Hence, with the choice of base point absorbed into $B$,
\[
\rho(\zeta)=B+A\int^{t(\zeta)}
\left(\frac{2}{1-\tau^4}
-\frac{2}{1-(\tau/b)^4}\right)\dd\tau,
\]
where the contour lies in $\Omega_t$ and the boundary values are obtained by
continuation from its interior.  The poles $t^4=1$ and $t^4=b^4$ encode the
logarithmic ends of the flat pants, while the phase of the differential along
each boundary component determines the direction of the corresponding edge.

To fix $a$, we require the regularized oriented displacement along the
direction of the ray $\{\ee^{\frac{\pi\ii}{4}}\tau:\tau>0\}$ to vanish.
Equivalently, after subtracting the common logarithmic divergence at
$\tau=(a^2-1)^{1/4}$, we impose the principal-value condition
\[
\begin{aligned}
0
&=\operatorname{PV}\int_0^{\infty}
\left(\frac{2}{1+\tau^4}
-\frac{2}{1-\tau^4/(a^2-1)}\right)\dd\tau\\
&=\frac{\pi}{\sqrt{2}}
-\frac{\pi}{2}(a^2-1)^{\frac14}.
\end{aligned}
\]
We therefore obtain $a=\sqrt{5}$ and hence $b=1+\ii$.  Comparing the
transverse displacement across the logarithmic end then fixes the
multiplicative constant, and the condition $\rho(0)=-1-\ii$ fixes the additive
constant.  The normalized map is
\begin{align*}
g(t):={}&-1-\ii+\frac{2}{\pi}(1+\ii)
\int_{\infty}^{t}\\
&\quad\left(\frac{2}{1-\tau^4}
-\frac{2}{1+\tau^4/4}\right)\dd\tau,\\
\rho(\zeta)={}&g(t(\zeta)).
\end{align*}
The contour is defined as a boundary limit from the interior of $\Omega_t$.
In particular, when oriented from $t=0$ to $t=\infty$, it passes below the
pole $t=1+\ii$. The contour in the displayed integral has the reverse
orientation.

\paragraph{Common primitive and residues}

For the local expansions, define the rational integrand once and for all by
\[
\omega(t):=\frac{2}{1-t^4}-\frac{2}{1+t^4/4}.
\]
It has simple poles at
\[
\mathcal P=\{\pm1,\ \pm\ii,\ \pm(1+\ii),\ \pm(1-\ii)\}.
\]
The residues $r_p:=\operatorname{Res}_{t=p}\omega(t)$ are
\[
\begin{array}{|c|c|c|c|c|c|c|c|c|}
\hline
p&1&-1&\ii&-\ii&1+\ii&-1-\ii&1-\ii&-1+\ii\\ \hline
r_p&-\frac12&\frac12&-\frac{\ii}{2}&\frac{\ii}{2}
&\frac{1+\ii}{2}&-\frac{1+\ii}{2}
&\frac{1-\ii}{2}&\frac{-1+\ii}{2}\\
\hline
\end{array}.
\]
Since $\sum_{p\in\mathcal P}r_p=0$, the primitive normalized at infinity is
\[
g(t)=-1-\ii+\frac{2}{\pi}(1+\ii)
\sum_{p\in\mathcal P}r_p
\log_{\Omega_t}\left(1-\frac{p}{t}\right),
\]
where all logarithms are continued from infinity inside $\Omega_t$.
Only three poles are relevant for the local coordinates.  The three
calculations below follow the rows of this table:
\[
\begin{array}{c|c|c|c}
\zeta_\ast&t_\ast&\text{local parameter}
&\text{leading local coordinate}\\ \hline
-1&1-\ii&x_1=\zeta+1&\varphi_1\circ\rho\sim\kappa_1x_1\\
1&1+\ii&x_2=\zeta-1&\varphi_2\circ\rho\sim\kappa_2x_2\\
\infty&1&x_3=-1/\zeta&\varphi_3\circ\rho\sim\kappa_3x_3
\end{array}.
\]
Label these rows by \(i=1,2,3\), in the displayed order, and set
\[
A:=\frac{2(1+\ii)}{\pi},
\qquad B:=-1-\ii,
\qquad r_i:=r_{t_i}.
\]
Each local-coordinate map has the form
\[
\varphi_i(u)=\exp(s_i u+\delta_i),
\]
where
\[
\begin{aligned}
(s_1,s_2,s_3)
&=\left(
\frac{\pi}{2},-\frac{\pi\ii}{2},
\frac{\pi\ee^{3\pi\ii/4}}{2\sqrt2}
\right),\\
(\delta_1,\delta_2,\delta_3)
&=\left(
\frac{\pi}{2}(1+\ii),\frac{\pi}{2}(1+\ii),
\frac{\pi\ii}{2}
\right).
\end{aligned}
\]

\paragraph{Local data at the three punctures}

The changes of variables at the three punctures have the common leading form
\[
t(\zeta)-t_i=c_i x_i^{m_i}
\bigl[1+\cO(x_i)\bigr],
\]
with
\[
\begin{aligned}
(t_1,t_2,t_3)&=(1-\ii,\,1+\ii,\,1),\\
(r_1,r_2,r_3)&=\left(\frac{1-\ii}{2},\,
\frac{1+\ii}{2},\,-\frac12\right),\\
(c_1,c_2,c_3)&=\left(
\frac{5(1-\ii)}8,\,-\frac{5(1+\ii)}8,\,-\frac54
\right),\\
(m_1,m_2,m_3)&=(1,1,2).
\end{aligned}
\]

After inserting \(t-t_i=c_i x_i^{m_i}[1+\cO(x_i)]\), we have
\[
g(t(\zeta))
=K_i+A r_i m_i\log x_i
+\cO(x_i),
\]
where
\[
\begin{aligned}
K_i
={}&B+A\Biggl[\sum_{q\in\mathcal P\setminus\{t_i\}}r_q
\log_{\Omega_t}\left(1-\frac q{t_i}\right)\\
&\quad+r_i\log_{\Omega_t}\left(\frac{c_i}{t_i}\right)
\Biggr].
\end{aligned}
\]
After composing with the local coordinates, we have
\[
\begin{aligned}
\varphi_i\circ g \circ t(\zeta)
&=\kappa_i x_i
\bigl[1+\cO(x_i)\bigr],\\
\log\kappa_i&=s_i K_i+\delta_i.
\end{aligned}
\]
Here the identity
\[
s_i A r_i m_i=1\qquad(i=1,2,3)
\]
ensures the linearity of the leading term of the expansion.

It remains only to evaluate the regular parts.  All logarithms inherit their
branches from analytic continuation in \(\Omega_t\). The points \(t_1\) and \(t_2\) are
approached along
\(\epsilon\ee^{\pi\ii/4}\) and
\(\epsilon\ee^{-\pi\ii/4}\), respectively, while the branch at \(t_3=1\)
comes from the upper half-plane.  In particular,
\(\log(-5/4)=\log(5/4)-\pi\ii\).  The residue sums reduce to
\[
(K_1, K_2, K_3)
=\bigl(
\lambda-\ii,\,
1+\ii\lambda,\,
(1+\ii)(\lambda_\infty+\ii)
\bigr),
\]
where the two remaining quantities are real:
\[
\begin{aligned}
\lambda
&:=-\frac12+
\frac{\log(125/256)-2\arctan(1/3)}{\pi},\\
\lambda_\infty
&:=-\frac12+
\frac{\log(8/25)+2\arctan(1/2)}{\pi}.
\end{aligned}
\]
Substituting these residue sums into the expressions for $\log\kappa_i$
shows that their imaginary parts cancel:
\[
\begin{aligned}
\log\kappa_1
&=\frac{\pi}{2}(\lambda-\ii)
+\frac{\pi}{2}(1+\ii)
=\frac{\pi}{2}(1+\lambda)\in\R,\\
\log\kappa_2
&=-\frac{\pi\ii}{2}(1+\ii\lambda)
+\frac{\pi}{2}(1+\ii)
=\frac{\pi}{2}(1+\lambda)\in\R,\\
\log\kappa_3
&=-\frac{\pi}{2}(\lambda_\infty+\ii)
+\frac{\pi\ii}{2}
=-\frac{\pi}{2}\lambda_\infty\in\R.
\end{aligned}
\]
Finally,
\[
\begin{aligned}
\kappa_1=\kappa_2
&=\ee^{\pi(1+\lambda)/2}
=\frac{5\sqrt5}{16}\exp\left(\arctan\frac12\right),\\
\kappa_3
&=\ee^{-\pi\lambda_\infty/2}
=\frac{5}{2\sqrt2}\exp\left(\arctan\frac13\right).
\end{aligned}
\]

\paragraph{Three-point geometric factor}

We can now evaluate the three-point function.  For chiral primary fields of
conformal weights $h_1,h_2,h_3$ inserted at
$\zeta=-1,1,\infty$, respectively, our normalization on the Riemann sphere is
\[
\left\langle
\phi_1(-1)\phi_2(1)\phi_3(\infty)
\right\rangle_{\widehat{\C}}
=2^{-h_1-h_2+h_3}C_{123}.
\]
Pulling the fields back through the three local coordinates contributes the
inverse linear factors $\kappa_i^{-h_i}$.  The pants three-point function is
\[
\left\langle
\phi_1\phi_2\phi_3
\right\rangle_{\mathrm{pants}}
=G_{\mathrm{pants}}(h_1,h_2,h_3)C_{123},
\]
where
\begin{equation}\label{eq:flat-pants-gauge-app}
\begin{split}
G_{\mathrm{pants}}(h_1,h_2,h_3)
={}&
\left[
\frac{8\ee^{-\arctan(1/2)}}{5\sqrt{5}}
\right]^{h_1+h_2}\\
&\times
\left[
\frac{4\sqrt{2}\,\ee^{-\arctan(1/3)}}{5}
\right]^{h_3}\\
\approx{}&
(0.450065146826)^{h_1+h_2}\\
&\times
(0.820106954075)^{h_3}.
\end{split}
\end{equation}

\paragraph{Adding finite legs}

We finally add finite legs whose lengths are measured directly in the flat
$u$-metric.  Let $d_i>0$ be the Euclidean distance by which the mouth of the
$i$th leg is displaced outward along the axis of that leg.  The transverse
widths of the three semi-infinite strips are
\[
W_1=W_2=2,\qquad W_3=2\sqrt{2}.
\]
Since a strip of width $W_i$ is mapped to a half-disk by an exponential with
radial rate $\pi/W_i$, extending the leg by $d_i$ rescales its local coordinate
as
\[
\varphi_i\longmapsto \ee^{\pi d_i/W_i}\varphi_i.
\]
Equivalently, the new unit contour in the original local coordinate is
\[
|\varphi_i|=\ee^{-\pi d_i/W_i}.
\]
The resulting geometry is represented schematically by
\[
\raisebox{-77.86612pt}{\includegraphics{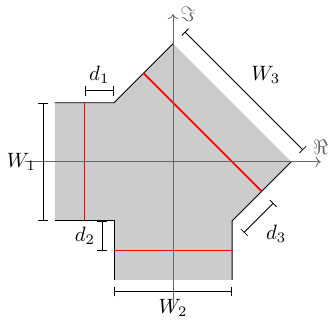}}.
\]
Multiplying the target coordinate by $\ee^{\pi d_i/W_i}$ composes $f_i$ on the
left with a dilation.  The contravariant composition law therefore gives
\[
\hat{\cW}(\ee^{\pi d_i/W_i}f_i)
=\hat{\cW}(f_i)\ee^{-(\pi d_i/W_i)L_0}.
\]
Consequently, if $\phi_i$ is a homogeneous descendant of a primary of
weight $h_i$, at level $N_i$, the rescaled amplitude obeys
\[
\begin{aligned}
&\left\langle\phi_1\phi_2\phi_3\right\rangle_{d_1,d_2,d_3}\\
&\quad=
\prod_{i=1}^3\exp\left[-\frac{\pi(h_i+N_i)d_i}{W_i}\right]
\left\langle\phi_1\phi_2\phi_3\right\rangle_{0,0,0}.
\end{aligned}
\]
In particular, for primary fields,
\[
\begin{split}
&G_{\mathrm{pants}}
(h_1,h_2,h_3;d_1,d_2,d_3)\\
&\qquad=
\exp\left[-\frac{\pi}{2}(h_1d_1+h_2d_2)\right]\\
&\qquad\quad\times
\exp\left(-\frac{\pi h_3d_3}{2\sqrt{2}}\right)
G_{\mathrm{pants}}(h_1,h_2,h_3),
\end{split}
\]
where the last factor is the zero-leg-length result computed above.  In terms
of the doubled cylinders, whose circumferences are $C_i=2W_i$, the same radial
factor is $\exp(-2\pi d_i/C_i)$.  This makes explicit the usual $2\pi$
normalization while keeping $d_i$ as the geometric leg length.

\section{Directional tube traces and modular \texorpdfstring{$S$}{S}-covariance}
\label{app:directional-tube-traces}

\subsection{Setup}

This appendix supplements the componentwise proof in the main text with a
basis-independent change-of-cut argument. Let $\cR$ be a unitary fusion
category with its canonical spherical structure, and use the duality
conventions fixed in Sec.~\ref{sec:r-tensor}. For
\begin{equation*}
    f:U\otimes X\longrightarrow V\otimes X,
    \qquad
    g:X\otimes U\longrightarrow X\otimes V,
\end{equation*}
define the right and left partial traces by
\begin{align*}
    \ptr_X^r(f)
    &:={}
      (\id_V\otimes\ev_X)
      \circ(f\otimes\id_{X^*})
      \circ(\id_U\otimes\ev_X^\dagger),\\
    \ptr_X^\ell(g)
    &:={}
      (\coev_X^\dagger\otimes\id_V)
      \circ(\id_{X^*}\otimes g)
      \circ(\coev_X\otimes\id_U).
\end{align*}
Graphically, these are
\begin{equation*}
    \includegraphics[scale=.8]{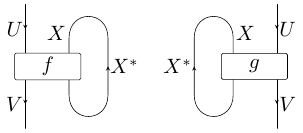}
\end{equation*}
We use $\Tr$ for both the resulting spherical trace and the ordinary trace of
a finite-dimensional vector space. Its meaning is determined by the object
being traced.  We use the total quantum dimension convention
\begin{equation*}
    D_\cR^2:=\sum_{i\in\Irr\cR}d_i^2,
\end{equation*}
All coends below are taken in $\mathsf{Vect}_{\C}$.

\begin{definition}[Directional tube categories~\cite{lan2025tubecategorytensorrenormalization}]
The objects of $\Tubev(\cR)$ are the objects of $\cR$, and its morphism
spaces are
\begin{equation*}
    \Tubev(\cR)(A,B)
    :=\int^{X\in\cR}\cR(A\otimes X,X\otimes B).
\end{equation*}
Denote the universal dinatural maps by
\begin{equation*}
    \pi_X^{A,B}:\cR(A\otimes X,X\otimes B)
    \longrightarrow\Tubev(\cR)(A,B).
\end{equation*}
Vertical composition is determined by the universal property of the coends
and the commutative diagram

\begin{equation*}
\resizebox{\columnwidth}{!}{%
\includegraphics{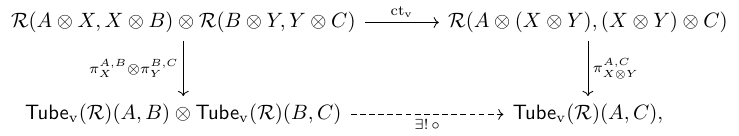}%
}
\end{equation*}

where
\begin{equation*}
    \contraction_{\mathrm v}(f\otimes g)
    :=(\id_X\otimes g)\circ(f\otimes\id_Y).
\end{equation*}
The identity morphism of $A$ is
\begin{equation*}
    \id_A^{\mathrm v}:=\pi_\one^{A,A}(\id_A),
\end{equation*}
where the unitors are suppressed.  Dinaturality makes the induced composition
well defined, while monoidal coherence gives associativity and the unit laws.
As usual, we suppress all associators and unitors.

Equip the reverse category $\cR^{\rev}$ with the tensor product
$A\otimes_{\cR^{\rev}}B:=B\otimes_\cR A$, and define the horizontal
tube category by
\begin{equation*}
    \Tubeh(\cR):=\Tubev(\cR^{\rev}).
\end{equation*}
In particular,
\begin{equation*}
    \Tubeh(\cR)(X,Y)
    \simeq\int^{A\in\cR}\cR(A\otimes X,Y\otimes A).
\end{equation*}
Denote the corresponding universal dinatural maps by
\begin{equation*}
    \varpi_A^{X,Y}:\cR(A\otimes X,Y\otimes A)
    \longrightarrow\Tubeh(\cR)(X,Y).
\end{equation*}
The horizontal composition and identities are those of
$\Tubev(\cR^{\rev})$. In particular,
\begin{equation*}
    \id_X^{\mathrm h}:=\varpi_\one^{X,X}(\id_X).
\end{equation*}
Thus passing from one direction to the other reverses the tensor product. It
is not merely a relabeling of objects.
\end{definition}
The two compositions have the following cut-square representatives. Opposite
dashed edges are identified, as in Eq.~\eqref{eq:tube-equivalence-relation}.
\begin{equation*}
    \includegraphics[scale=.8]{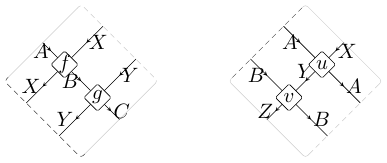}
\end{equation*}
On the right, $u:A\otimes X\to Y\otimes A$ and
$v:B\otimes Y\to Z\otimes B$. The running objects of the vertical and
horizontal composites are $X\otimes Y$ and $B\otimes A$, respectively.

Henceforth $d\in\{\mathrm v,\mathrm h\}$ denotes a direction, and a symbol
carrying the superscript or subscript $d$ refers to the corresponding vertical
or horizontal construction.

\subsection{Directional categorical traces}

Define
\begin{align*}
    \Torusv(\cR)
    &:=\int^{A\in\Tubev(\cR)}\Tubev(\cR)(A,A),\\
    \Torush(\cR)
    &:=\int^{X\in\Tubeh(\cR)}\Tubeh(\cR)(X,X).
\end{align*}
Denote the universal dinatural maps by
\begin{align*}
    \iota_A^{\mathrm v}:\End_{\Tubev(\cR)}(A)
    &\longrightarrow\Torusv(\cR),\\
    \iota_X^{\mathrm h}:\End_{\Tubeh(\cR)}(X)
    &\longrightarrow\Torush(\cR).
\end{align*}
Equivalently, they satisfy the cyclicity relation~\cite{Beliakova_2014_trace}
\begin{equation*}
    \iota_U^d(QP)=\iota_V^d(PQ),
    \qquad
    P:U\longrightarrow V,
    \quad
    Q:V\longrightarrow U
\end{equation*}
in $\Tube_d(\cR)$.
For $d=\mathrm v$, this relation is shown below on representatives
$P=\pi_X^{U,V}(p)$ and $Q=\pi_Y^{V,U}(q)$, with arbitrary running
objects $X,Y\in\cR$. Dashed edges are identified in the tube coend,
and dotted edges by the categorical trace.
\begin{equation*}
    \includegraphics[scale=.8]{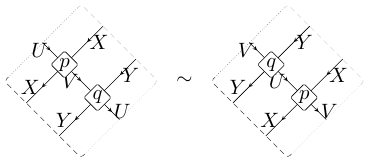}
\end{equation*}

\subsection{Algebraic change of cut}

A square morphism
\begin{equation*}
    u:A\otimes X\longrightarrow X\otimes A
\end{equation*}
has both a vertical and a horizontal closure.  We construct the change of
cut from the two universal coends, rather than first declaring a map on their
quotients.  For $A,X\in\cR$, set
\begin{equation*}
\begin{aligned}
    \Theta_{A,X}:\cR(A\otimes X,X\otimes A)&\longrightarrow\Torush(\cR),\\
    u&\longmapsto\iota_X^{\mathrm h}\!\left(\varpi_A^{X,X}(u)\right).
\end{aligned}
\end{equation*}

\begin{proposition}[Universal change of cut]
For fixed $A$, the family $\Theta_{A,X}$ is dinatural in $X$ and hence
induces a unique linear map
\begin{equation*}
\begin{aligned}
    \Theta_A:\End_{\Tubev(\cR)}(A)&\longrightarrow\Torush(\cR),\\
    \Theta_A\pi_X^{A,A}&=\Theta_{A,X}.
\end{aligned}
\end{equation*}
The family $(\Theta_A)_{A\in\Tubev(\cR)}$ is a cowedge for the categorical
trace of $\Tubev(\cR)$.  Consequently the universal property of that trace
induces a unique isomorphism
\begin{equation*}
    \scrS_{\cutlabel}:
    \Torusv(\cR)\xrightarrow{\ \simeq\ }\Torush(\cR),
\end{equation*}
characterized on square generators by
\begin{equation}
    \scrS_{\cutlabel}
    \!\left(\iota_A^{\mathrm v}\!\left(\pi_X^{A,A}(u)\right)\right)
    =\iota_X^{\mathrm h}\!\left(\varpi_A^{X,X}(u)\right).
    \label{eq:appendix-change-of-cut}
\end{equation}
Its inverse is induced by the same construction with the two directions
interchanged.
\end{proposition}

The two closures of $u$ are shown below. Dashed edges are identified in the
tube coend, and dotted edges are identified by the categorical trace. The
vertical arrows apply $\iota_A^{\mathrm v}$ and $\iota_X^{\mathrm h}$,
and the lower arrow is $\scrS_{\cutlabel}$.
\begin{equation*}
    \includegraphics[scale=.8]{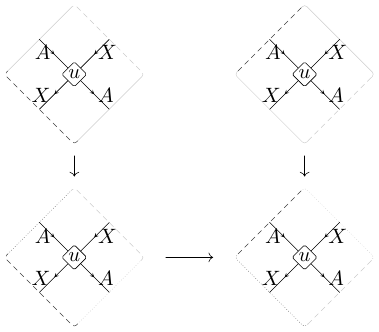}
\end{equation*}

\begin{proof}
For $f:X\to Y$ and $u:A\otimes Y\to X\otimes A$, horizontal trace
cyclicity, applied to the horizontal tube morphisms represented by $f$ with
running label $\one$ and by $u$ with running label $A$, gives
\begin{equation*}
    \Theta_{A,X}\bigl(u(\id_A\otimes f)\bigr)
    =\Theta_{A,Y}\bigl((f\otimes\id_A)u\bigr).
\end{equation*}
Thus $\Theta_{A,X}$ is dinatural in $X$, and the first assertion follows
from the universal property of the tube coend.

It remains only to verify the cowedge condition for $\Theta_A$.  By the
universal property of the tube coends and linearity, it suffices to use
the tube morphisms
\begin{align*}
    P&:=\pi_X^{A,B}(p)\in\Tubev(\cR)(A,B),\\
    Q&:=\pi_Y^{B,A}(q)\in\Tubev(\cR)(B,A),
\end{align*}
represented respectively by
\begin{align*}
    p&:A\otimes X\longrightarrow X\otimes B,\\
    q&:B\otimes Y\longrightarrow Y\otimes A.
\end{align*}
By the definition of vertical tube composition,
\begin{align*}
    Q\circ_{\mathrm v}P
    &=\pi_{X\otimes Y}^{A,A}\!\left(
        \contraction_{\mathrm v}(p\otimes q)\right),\\
    P\circ_{\mathrm v}Q
    &=\pi_{Y\otimes X}^{B,B}\!\left(
        \contraction_{\mathrm v}(q\otimes p)\right).
\end{align*}
By the definition of $\Theta_A$,
\begin{align*}
    \Theta_A\!\left(Q\circ_{\mathrm v}P\right)
    &=\iota_{XY}^{\mathrm h}\!\left(
        \varpi_A^{XY,XY}\!\left(\contraction_{\mathrm v}(p\otimes q)\right)
      \right),\\
    \Theta_B\!\left(P\circ_{\mathrm v}Q\right)
    &=\iota_{YX}^{\mathrm h}\!\left(
        \varpi_B^{YX,YX}\!\left(\contraction_{\mathrm v}(q\otimes p)\right)
      \right).
\end{align*}

Consider the canonical cyclic swap
\begin{equation*}
    \swap_{X,Y}
    :=\varpi_X^{YX,XY}(\id_{XYX})
    \in\Tubeh(\cR)(YX,XY).
\end{equation*}
Choose a right dual $X^*$.  Its inverse is
\begin{equation*}
    \swap_{X,Y}^{-1}
    :=\varpi_{X^*}^{XY,YX}\!\left(
        (\id_Y\otimes\ev_X^\dagger)
        \circ(\coev_X^\dagger\otimes\id_Y)
      \right).
\end{equation*}
Their cut-square representatives are
\begin{equation*}
    \includegraphics[scale=.8]{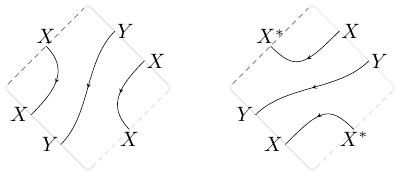}
\end{equation*}
The first uses $\id_{XYX}$, and the second uses the indicated evaluation
and coevaluation. No braiding in $\cR$ is involved.
Indeed, horizontal dinaturality together with the two snake identities gives
\begin{align*}
    \swap_{X,Y}^{-1}
      \circ_{\mathrm h}\swap_{X,Y}
    &=\id_{YX}^{\mathrm h},\\
    \swap_{X,Y}
      \circ_{\mathrm h}\swap_{X,Y}^{-1}
    &=\id_{XY}^{\mathrm h}.
\end{align*}

Apply horizontal dinaturality to
\begin{equation*}
    p:A\otimes X\longrightarrow X\otimes B
\end{equation*}
and the morphism
\begin{equation*}
    \id_X\otimes q\otimes\id_X:
    (XB)\otimes(YX)\longrightarrow(XY)\otimes(AX).
\end{equation*}
This gives the intertwining relation
\begin{multline*}
    \swap_{X,Y}\circ_{\mathrm h}
    \varpi_B^{YX,YX}\!\left(\contraction_{\mathrm v}(q\otimes p)\right)\\
    =\varpi_A^{XY,XY}\!\left(\contraction_{\mathrm v}(p\otimes q)\right)
    \circ_{\mathrm h}\swap_{X,Y}.
\end{multline*}
Thus the two horizontal endomorphisms are conjugate by
$\swap_{X,Y}$.  Since categorical trace is invariant under
conjugation,
\begin{equation*}
    \Theta_A\!\left(Q\circ_{\mathrm v}P\right)
    =\Theta_B\!\left(P\circ_{\mathrm v}Q\right).
\end{equation*}
In the cut-square picture, this equality of horizontal trace classes is
\begin{equation*}
    \includegraphics[scale=.8]{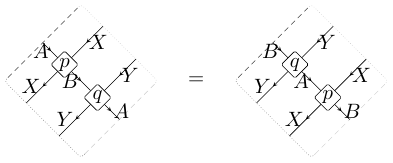}
\end{equation*}
This is the cowedge condition for $(\Theta_A)_A$.  The universal property of
the vertical torus coend therefore induces the unique map
$\scrS_{\cutlabel}$ with Eq.~\eqref{eq:appendix-change-of-cut}.

Interchanging the horizontal and vertical variables gives a cowedge in the
opposite direction and hence an inverse map.  Both composites are the identity
on the universal square generators, so $\scrS_{\cutlabel}$ is an
isomorphism.
\end{proof}

\subsection{Tube amplitudes}

The trace coend is natural under $\C$-linear functors. Every
$F:\mathcal C\to\mathcal D$ induces
\begin{equation*}
\begin{aligned}
    F_*:
    \int^{A\in\mathcal C}\mathcal C(A,A)
    &\longrightarrow
    \int^{V\in\mathcal D}\mathcal D(V,V),\\
    F_*\!\left(\iota_A^{\mathcal C}(u)\right)
    &:=\iota_{F(A)}^{\mathcal D}\!\left(F(u)\right).
\end{aligned}
\end{equation*}

For $\mathcal D=\Hilb$, the ordinary vector-space trace
gives the canonical isomorphism
\begin{equation*}
    \int^{V\in\Hilb}\End(V)
    \simeq\C,
    \qquad
    \iota_V^{\mathrm{fd}}(T)\longmapsto\Tr_V(T).
\end{equation*}
Under this isomorphism, for a finite-dimensional $\C$-linear functor
$F:\Tube_d(\cR)\to\Hilb$, define
\begin{align*}
    \Torus_d(\cR)&:=\int^{X}\Tube_d(\cR)(X,X),\\
    Z_F^d:=F_*&:\Torus_d(\cR)\longrightarrow\C.
\end{align*}
Explicitly,
\begin{equation*}
    Z_F^d\!\left(\iota_A^d(u)\right)
    =\Tr_{F(A)}\!\left(F(u)\right).
\end{equation*}

For $x=(X,\gamma_{X,-})\in\fZ(\cR)$, the vertical anyon-insertion functor is
defined on objects by $x^{\sharp_{\mathrm v}}(A)=\cR(X,A)$.  Its action on a
coend generator is the basis-free partial trace
\begin{equation*}
    x^{\sharp_{\mathrm v}}\!\left(\pi_V^{A,B}(u)\right)(\phi)
    :=\ptr^{\,\ell}_V\!\left[
        u\circ(\phi\otimes\id_V)\circ\gamma_{X,V}^{-1}
    \right].
\end{equation*}
Half-braiding naturality makes this family dinatural, and the hexagon gives
functoriality.  Applying the same construction to
$\Tubeh(\cR)\cong\Tubev(\cR^{\rev})$ gives horizontal simple
functors labeled by $\Irr\fZ(\cR^{\rev})$.
For either direction, abbreviate
\begin{equation*}
    Z_x^d:=Z_{x^{\sharp_d}}^d.
\end{equation*}

In the present covariant convention, the unitary form of Hardiman's
equivalence is
\begin{equation*}
    (-)^{\sharp_{\mathrm v}}:
    \fZ(\cR)^{\mathrm{op}}
    \xrightarrow{\ \cong\ }
    \Fun_*\bigl(\Tubev(\cR),\Hilb\bigr).
\end{equation*}
Here the opposite category records that these functors are covariant, whereas
Hardiman uses contravariant representations~\cite{Hardiman_2019_graphical}.  The
horizontal statement follows by applying the same equivalence to
$\cR^{\rev}$.

For $x=(X,\gamma_{X,-})\in\Irr\fZ(\cR)$, Hardiman's idempotent, translated
to the present convention, is
\begin{equation*}
    \varepsilon_x^{\mathrm v}
    :=\frac{1}{D_\cR^2}
      \sum_{i\in\Irr\cR}d_i\,
      \pi_i^{X,X}\!\left(\gamma_{X,i}\right)
    \in\Tubev(\cR)(X,X).
\end{equation*}
The anyon-insertion action and a coend generator appearing in
$\varepsilon_x^{\mathrm v}$ are depicted by
\begin{equation*}
    \includegraphics[scale=.8]{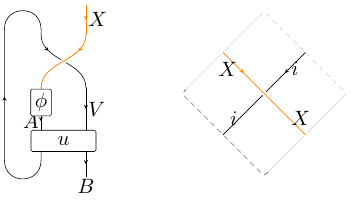}
\end{equation*}
Define the vertical trace class of $\varepsilon_x^{\mathrm v}$ by
\begin{equation*}
    e_x^{\mathrm v}:=
    \iota_X^{\mathrm v}\!\left(\varepsilon_x^{\mathrm v}\right)
    \in\Torusv(\cR).
\end{equation*}
For $y=(Y,\gamma_{Y,-})\in\Irr\fZ(\cR)$, the operator
$y^{\sharp_{\mathrm v}}(\varepsilon_x^{\mathrm v})$ is the projection
\begin{equation*}
    \cR(Y,X)\longrightarrow\fZ(\cR)(y,x).
\end{equation*}
Consequently,
\begin{equation*}
    Z_y^{\mathrm v}(e_x^{\mathrm v})
    =\dim\fZ(\cR)(y,x)=\delta_{x,y}.
\end{equation*}
Completeness of the Hardiman idempotents implies that the classes
$\{e_x^{\mathrm v}\}_{x\in\Irr\fZ(\cR)}$ form a basis of
$\Torusv(\cR)$.

With the reversed tensor product specified above, inverse half-braidings give
a canonical bijection of simple isomorphism classes
\begin{equation*}
\begin{aligned}
    \kappa:\Irr\fZ(\cR)&\longrightarrow\Irr\fZ(\cR^{\rev}),\\
    (X,\gamma_{X,-})&\longmapsto(X,\gamma_{X,-}^{-1})
\end{aligned}
\end{equation*}

\subsection{Modular \texorpdfstring{$S$}{S}-covariance}

Orient the meridional and longitudinal cycles as $(m,\ell)$.  With the tensor
orders fixed in the directional tube categories, Eq.~\eqref{eq:appendix-change-of-cut}
represents the oriented quarter-turn
\begin{equation*}
    \scrS_{\cutlabel}:(m,\ell)\longmapsto(\ell,-m).
\end{equation*}
and hence induces the standard modular transformation
$\tau\mapsto-1/\tau$ on the complex modulus.

For $x=(X,\gamma_{X,-})$ and $y=(Y,\gamma_{Y,-})$ in
$\Irr\fZ(\cR)$, define the Hopf-link matrix by
\begin{equation*}
    S_{yx}
    :=\Tr_{Y\otimes X}
      \!\left(\gamma_{X,Y}\circ\gamma_{Y,X}\right).
\end{equation*}
The total quantum dimensions satisfy
$D_{\fZ(\cR)}=D_\cR^2$.

The explicit form of $e_x^{\mathrm v}$ and
Eq.~\eqref{eq:appendix-change-of-cut} give
\begin{equation*}
    \scrS_{\cutlabel}(e_x^{\mathrm v})
    =\frac{1}{D_\cR^2}
      \sum_{i\in\Irr\cR}d_i\,
      \iota_i^{\mathrm h}\!\left(
        \varpi_X^{i,i}(\gamma_{X,i})
      \right).
\end{equation*}
For $\phi:Y\to i$, half-braiding naturality gives
\begin{equation*}
\begin{aligned}
    \kappa(y)^{\sharp_{\mathrm h}}\!\left(
      \varpi_X^{i,i}(\gamma_{X,i})
    \right)(\phi)
    &=\ptr_X^r\!\left[
        \gamma_{X,i}\circ(\id_X\otimes\phi)
        \circ\gamma_{Y,X}
      \right] \\
    &=\phi\circ h_x^y,\\
    h_x^y&:=\ptr_X^r\!\left[
      \gamma_{X,Y}\circ\gamma_{Y,X}
    \right].
\end{aligned}
\end{equation*}
Thus $h_x^y\in\End_\cR(Y)$. Closing its remaining $Y$ strand gives the
Hopf link defining $S_{yx}$.
\begin{equation*}
    \includegraphics[scale=.8]{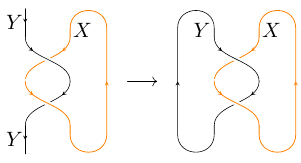}
\end{equation*}
The crossings are the half-braidings $\gamma_{Y,X}$ and $\gamma_{X,Y}$,
read from top to bottom. For every $h\in\End_\cR(Y)$,
semisimplicity gives the canonical trace
decomposition
\begin{equation*}
    \Tr_Y(h)
    =\sum_{i\in\Irr\cR}d_i\,
      \Tr_{\cR(Y,i)}(\phi\longmapsto\phi\circ h).
\end{equation*}
Consequently, sphericality gives
\begin{equation}
\begin{aligned}
    Z_{\kappa(y)}^{\mathrm h}\!\left(
      \scrS_{\cutlabel}e_x^{\mathrm v}
    \right)
    &=\frac{1}{D_\cR^2}\Tr_Y(h_x^y) \\
    &=\frac{1}{D_\cR^2}
      \Tr_{Y\otimes X}\!\left(
        \gamma_{X,Y}\circ\gamma_{Y,X}
      \right)\\
    &=\frac{S_{yx}}{D_\cR^2}.
\end{aligned}
    \label{eq:appendix-mixed-sector-pairing}
\end{equation}

\begin{proposition}[Modular $S$-covariance]
For every $y\in\Irr\fZ(\cR)$,
\begin{equation*}
    Z_{\kappa(y)}^{\mathrm h}\circ\scrS_{\cutlabel}
    =\sum_{x\in\Irr\fZ(\cR)}
      \frac{S_{yx}}{D_\cR^2}\,Z_x^{\mathrm v}.
\end{equation*}
\end{proposition}

\begin{proof}
It suffices to evaluate both sides on the sector basis.  For
$x\in\Irr\fZ(\cR)$, Eq.~\eqref{eq:appendix-mixed-sector-pairing} gives the value of the
left-hand side on $e_x^{\mathrm v}$, while the right-hand side gives
\begin{equation*}
    \sum_{z\in\Irr\fZ(\cR)}
      \frac{S_{yz}}{D_\cR^2}Z_z^{\mathrm v}(e_x^{\mathrm v})
    =\frac{S_{yx}}{D_\cR^2}.
\end{equation*}
Thus the two linear functionals agree on a basis of $\Torusv(\cR)$.
\end{proof}

\clearpage
\onecolumngrid
\section{Notation table}
\label{app:notation-table}

This appendix collects recurring notation. The meaning of a bar or capital
letter depends on its context. The source column gives the corresponding
definition or convention.

Categorical algebra objects are usually denoted by $A,B,\ldots$,
Q-systems by $Q$, and Lagrangian algebras by $L,N,\ldots$.
Full field algebras and VOAs use script letters, also used when these
algebras are viewed as objects of module categories.

Throughout the manuscript, $\cong$ denotes equivalence or isomorphism of
categories, and $\simeq$ denotes isomorphism of objects, including algebra
objects and vector spaces. Natural isomorphisms between functors are also
written $\simeq$.

\begingroup
\small
\renewcommand{\arraystretch}{1.16}
\setlength{\LTpre}{0.5em}
\setlength{\LTpost}{0.5em}
\begin{longtable}{@{}>{\raggedright\arraybackslash}p{0.18\textwidth}
                        >{\raggedright\arraybackslash}p{0.50\textwidth}
                        >{\raggedright\arraybackslash}p{0.22\textwidth}@{}}
\caption{Recurring notation and its defining equations or conventions.}
\label{tab:notation-table}\\
\toprule
Symbol & Meaning & Definition or convention \\
\midrule
\endfirsthead
\multicolumn{3}{@{}l}{\tablename~\thetable\ (continued)}\\[2pt]
\toprule
Symbol & Meaning & Definition or convention \\
\midrule
\endhead
\midrule
\multicolumn{3}{r@{}}{Continued on next page}\\
\endfoot
\bottomrule
\endlastfoot

\multicolumn{3}{@{}l}{\textbf{Categories, objects, and categorical operations}}\\*
\addlinespace[2pt]
$\cR$ & Unitary fusion category of symmetry charges in an
$\cR$-tensor network; no braiding is assumed unless stated. &
Sec.~\ref{sec:r-tensor} \\
$\cC,\cD,\ldots$ & Generic categories; in braided constructions,
$\cC$ is usually a braided fusion category or a UMTC. &
App.~\ref{app:categorical-data} \\
$\cC(A,B)$ & Morphism (Hom) space
$\operatorname{Hom}_{\cC}(A,B)$ between objects $A,B\in\cC$. &
Sec.~\ref{sec:r-tensor} \\
$\cong$; $\simeq$ & Equivalence/isomorphism of categories; isomorphism of
objects, algebra objects, vector spaces, or morphism spaces, respectively. &
App.~\ref{app:notation-table} (opening convention) \\
$\Hilb$ & Category of finite-dimensional Hilbert spaces. &
Subsec.~\ref{sec:r-tensor-basis} \\
$\Rep(G)$ & Category of finite-dimensional unitary representations of a finite group
$G$, also used for the extension to compact Lie groups. &
Sec.~\ref{sec:r-tensor-network} \\
$\one$ & Tensor unit of the category currently in scope. &
Sec.~\ref{sec:r-tensor} \\
$\Irr(\cR)$ & Fixed set of representatives of the simple isomorphism classes
of $\cR$. &
Eq.~\eqref{eq:orthonormal-completeness} \\
$i,j,k,l,m,n,\ldots$ & Simple-object labels in $\cR$ used in the
$\cR$-tensor, $\cR$-Loop-TNR, and holographic constructions. &
Subsec.~\ref{sec:r-tensor-basis} \\
$X,Y,Z,A,B,V,\ldots$ & General, possibly nonsimple objects; capital letters
also label tensor-network bonds or tube objects locally. &
Sec.~\ref{sec:r-tensor} \\
$\otimes$, $\oplus$, $\boxtimes$ & Monoidal product, direct sum, and Deligne
tensor product.  A subscript, as in $\otimes_A$ or $\boxtimes_{\cC}$,
specifies the base algebra, object, or category for a relative product. &
Sec.~\ref{sec:r-tensor} and App.~\ref{app:categorical-data} \\
$X^*$; $\bar i$ & Chosen categorical dual; in the duality subsection,
$\bar i\in\Irr(\cR)$ is the representative unitarily isomorphic to $i^*$.
The bar here does not mean reversed braiding. &
Sec.~\ref{sec:r-tensor} and Subsec.~\ref{sec:r-tensor-presentation} \\
$\cC^{\rev}$ & Monoidal reverse with
$X\otimes^{\rev}Y:=Y\otimes X$. &
Prop.~\ref{prop:diagonal-algebra} \\
$\overline{\cC}$ & Braided reverse of a braided category. The objects and
fusion data are unchanged, while braiding and twists are inverted. &
App.~\ref{app:categorical-data} \\
$d_X$ (e.g. $d_i,d_A$), $D_\cC$, $D$ & Quantum dimension of an object and total
quantum dimension of a fusion or modular category, with
$D_\cC^2=\sum_{i\in\Irr(\cC)}d_i^2$; $D=D_\cR$ in the holographic
and cloaking calculations. &
Prop.~\ref{prop:diagonal-algebra} and App.~\ref{app:directional-tube-traces} \\
$\fZ(\cR)$ & Drinfeld center of $\cR$, interpreted as the bulk UMTC or symTO
associated with the charge category. &
Sec.~\ref{sec:holographic-picture} \\
symTO; symTFT & Symmetry topological order and symmetry topological field
theory, respectively.  Both abbreviations use a lowercase ``sym'' throughout
the manuscript. &
Appendix~\ref{app:symto-review} \\
$x=(X,\gamma_{X,-})$ & Center object/anyon sector with underlying object
$X\in\cR$ and half-braiding
$\gamma_{X,A}:X\otimes A\rightarrow A\otimes X$. &
Eq.~\eqref{eq:anyon-insertion-rep} \\
$\cM(p,q)$, $\cC_{p,q}$ & Virasoro minimal model and the UMTC of its chiral
sectors, respectively. &
App.~\ref{app:categorical-data} \\
$\mathsf{N}_c$ & Rank-$N$ UMTC with chiral central charge $c$ modulo $8$.
Named symbols such as $\Ising$ and $\Fib$ denote the underlying fusion
categories when used as charge categories. &
App.~\ref{app:categorical-data} \\
$\JonesKauffman_4$ & Jones--Kauffman category appearing in the Potts
factorization $\cC_{6,5}\cong\Fib\boxtimes\JonesKauffman_4$. &
App.~\ref{appsec:potts-categories} \\
$[m]_k$, $\varphi$ & Level-$k$ quantum integer and the golden ratio
$\varphi=(1+\sqrt5)/2$. &
Sec.~\ref{sec:random-cluster-models} \\

\addlinespace[4pt]
\multicolumn{3}{@{}l}{\textbf{String diagrams, bases, and tensor operations}}\\*
\addlinespace[2pt]
$f:X\to Y$, $f^\dagger$ & Morphism and dagger.  Diagrams are read from top to
bottom; horizontal juxtaposition is tensor product and vertical stacking is
composition. &
Sec.~\ref{sec:r-tensor} \\
$\ev_X$, $\coev_X$ & Evaluation
$X\otimes X^*\to\one$ and coevaluation $\one\to X^*\otimes X$ in the chosen
spherical duality convention. &
Sec.~\ref{sec:r-tensor} \\
$\Tr$, $\ptr^{\ell,r}_X$ & Spherical trace (or,
when the traced object is a vector space, ordinary trace) and left/right
partial traces. &
App.~\ref{app:directional-tube-traces} \\
$N_i^X$ & Multiplicity of the simple object $i$ in $X$. &
Eq.~\eqref{eq:orthonormal-completeness} \\
$q^X_{i,\alpha}$ & Orthonormal embedding $i\to X$; reflected triangles in the
diagrams denote its dagger. &
Eq.~\eqref{eq:orthonormal-completeness} \\
$\alpha,\beta,\gamma,\delta$; $\mu,\nu$; $I,J$ & In the rank-four decomposition,
the first four indices label simple-object multiplicities in the bond
objects, while $\mu,\nu$ label fusion-space multiplicities. $I,J$ collect
the labels of an embedding tree. Index letters are reused locally. &
Eq.~\eqref{eq:rank-4-tensor-decomposition} \\
$f^m_{I,J}$ & Matrix block of a morphism in total-charge sector $m$; the
$\Irr(\cR)$-graded blocks are the numerical data stored on a computer. &
Subsec.~\ref{sec:r-tensor-basis} \\
$v_{ij}^{k,\alpha}$ & Isotopy-invariant trivalent basis, equal to the
orthonormal vertex multiplied by $(d_id_j/d_k)^{1/4}$. &
Eq.~\eqref{eq:isotopy-invariant-fusion-basis} \\
$F^{ijk}_{l}$ & Unitary associator matrix relating two fusion-tree bases; its
channel and multiplicity indices are displayed when needed. &
Subsec.~\ref{sec:r-tensor-basis} \\
$Z_i$, $\varkappa_i$ & Unitary identification $i^*\simeq\bar i$ and the
Frobenius--Schur phase relating $Z_i^*$ to $Z_{\bar i}$. &
Eqs.~\eqref{eq:first-duality-identity} and~\eqref{eq:other-two-dual} \\
$A^{ij}_{k,\beta\alpha}$,
$B^{ij}_{k,\beta\alpha}$ & Unitary $A$- and $B$-move matrices; the
leg-bending formulas carry separate quantum-dimension factors. &
Subsec.~\ref{sec:r-tensor-presentation} \\
$U,S,V$; $\sigma_K^m$ & Sectorwise singular-value decomposition and singular
values in sector $m$. &
Sec.~\ref{sec:decomposition-and-truncation} \\
$\chi$, $r_m$, $M_{\mathrm{trunc}}$ & Quantum-dimension budget, retained rank in
sector $m$, and truncated bond object. The ranks maximize the retained
Frobenius norm subject to $\sum_m d_m r_m\leq\chi$. &
Sec.~\ref{sec:decomposition-and-truncation} \\
$\odot$ & Copower: $\C^n\odot i$ consists of $n$ copies of $i$.
More generally, the factor on the left specifies the multiplicity space. &
Subsec.~\ref{sec:r-tensor-presentation} \\
$[h,L,\xi]$ & Lattice-tube geometry consisting of Euclidean-time height, circumference,
and horizontal gluing displacement. &
Sec.~\ref{sec:transfer-morphism-diagonalization} \\
$\chi_{\mathrm{RG}}$,
$\chi_{\mathrm{transfer}}$ & Bond cutoffs used during the RG flow and in the
final transfer morphism, respectively. &
Sec.~\ref{sec:transfer-morphism-diagonalization} \\
$N_{\mathrm{RG}}$ & Number of RG steps. &
Sec.~\ref{sec:stable-ising} \\
$\cN_4^i$ & Environment matrix for tensor $S_4$ in charge sector $i$ in
the alternating least-squares optimization. &
Subsec.~\ref{sec:r-loop-tnr} \\

\addlinespace[4pt]
\multicolumn{3}{@{}l}{\textbf{Holographic expansion, tubes, and modular covariance}}\\*
\addlinespace[2pt]
$\Gamma$, $\Lambda$ & Square graph on the torus and its complete coloring by
simple objects and fusion-basis labels. &
Sec.~\ref{sec:holographic-picture} \\
$|\Gamma,\Lambda\ra$, $\cH_\Gamma$ & Colored graph state and the colored
graph space for fixed $\Gamma$. &
Eq.~\eqref{eq:frozen-graph-state} \\
$|Z_\Gamma\ra$ & Frozen-graph state
$\sum_\Lambda A_\Gamma(\Lambda)|\Gamma,\Lambda\ra$. &
Eq.~\eqref{eq:frozen-graph-state} \\
$A_\Gamma(\Lambda)$ & Ordinary finite-dimensional contraction of the
coefficient tensors for a fixed coloring. &
Sec.~\ref{sec:holographic-picture} \\
$\Gamma_n$, $U_n$ & Graph after $n$ coarse-graining steps and the evolution
from $\Gamma_{n-1}$ to $\Gamma_n$, comprising $2$--$2$ and $2$--$0$ moves. &
Sec.~\ref{sec:holographic-picture} \\
$v_*$ & Common bulk apex in the Turaev--Viro tent representation of a
finite contractible patch. &
Fig.~\ref{fig:holographic-dual-bulk} \\
$\cH_{T^2}$ & Torus ground-state space of the Turaev--Viro theory. &
Sec.~\ref{sec:holographic-picture} \\
$Z_{\cR\text{-TV}}$, $E_\Gamma^x(\Lambda)$ & Turaev--Viro evaluation map and
its amplitude from a colored graph to torus sector $x$. &
Eq.~\eqref{eq:tv-evaluation-amplitude} \\
$Z_x$, $Z^{\sym}_\cR$ & Sector component and the resulting relative
partition-function vector $\sum_x Z_x|x\ra$. &
Eq.~\eqref{eq:evaluated-partition-state} \\
$\Tubev(\cR)$, $\pi_V^{A,B}$ & Vertical tube category, defined by an annular
coend, and its universal dinatural tube-closure map. This closure is the
pro-tensor trace, distinct in general from the spherical quantum trace. &
Sec.~\ref{sec:transfer-morphism-diagonalization} \\
$\Tubeh(\cR)$, $\varpi_A^{X,Y}$ & Horizontal tube category
$\Tubev(\cR^{\rev})$ and its universal dinatural map. &
App.~\ref{app:directional-tube-traces} \\
$(-)^\sharp$, $x^\sharp$ & The sharp superscript denotes the anyon-insertion
construction.  It sends $x=(X,\gamma_{X,-})$ to the $*$-functor with
$x^\sharp(A)=\cR(X,A)$ and inner product
$\la\phi|\psi\ra=d_X^{-1}\Tr(\phi^\dagger\psi)$. &
Eqs.~\eqref{eq:inner-product-def} and~\eqref{eq:anyon-insertion-rep} \\
$x^{\sharp_d}$ & Anyon-insertion functor for direction
$d=\mathrm v,\mathrm h$. &
App.~\ref{app:directional-tube-traces} \\
$\Phi_{A,B}$ & Simultaneous decomposition of a tube morphism into all simple
$\fZ(\cR)$ sector representations. &
Eq.~\eqref{eq:tube-sector-decomposition} \\
$Z_x(u)$, $Z(u)$ & Trace of a tube morphism in sector $x$ and the vector of all
sector partition functions. &
Sec.~\ref{sec:modular-covariance} \\
$\cT$, $\theta_x$, $s_x$ & Lattice translation, topological twist, and
topological spin, with $\theta_x=\ee^{2\pi\ii s_x}$ and one full circuit acting
as the positive twist. &
Sec.~\ref{sec:modular-covariance} \\
$S_{xy}$ & Unnormalized Hopf-link matrix.  In the tube formulas the normalized
mapping-class-group coefficient is $S_{xy}/D_\cR^2$. &
Sec.~\ref{sec:modular-covariance} \\
$\kappa(x)$ & Horizontal-sector label
$(X,\gamma_{X,-}^{-1})\in\fZ(\cR^{\rev})$ corresponding to
$x=(X,\gamma_{X,-})$. &
App.~\ref{app:directional-tube-traces} \\
$\Torusv(\cR)$, $\Torush(\cR)$ & Categorical traces of the vertical and
horizontal tube categories. &
App.~\ref{app:directional-tube-traces} \\
$\iota_A^d$ & Universal map from $\End_{\Tube_d(\cR)}(A)$ to
$\Torus_d(\cR)$, sending a tube endomorphism to its trace class. &
App.~\ref{app:directional-tube-traces} \\
$Z_x^d$ & Sector amplitude on $\Torus_d(\cR)$, with
$Z_x^d(\iota_A^d(u))=\Tr(x^{\sharp_d}(u))$. &
App.~\ref{app:directional-tube-traces} \\
$\scrS_{\cutlabel}$ & Universal change of cut, implementing
$(m,\ell)\mapsto(\ell,-m)$ and hence $\tau\mapsto-1/\tau$. &
Eq.~\eqref{eq:appendix-change-of-cut} \\
$\varepsilon_x^{\mathrm v}$,
$e_x^{\mathrm v}$ & Hardiman idempotent for sector $x$ and its vertical trace
class; the latter form a basis of $\Torusv(\cR)$. &
App.~\ref{app:directional-tube-traces} \\

\addlinespace[4pt]
\multicolumn{3}{@{}l}{\textbf{CFT spectra and three-point data}}\\*
\addlinespace[2pt]
$h$, $\bar h$, $\Delta$, $s$ & Chiral/antichiral conformal weights, scaling
dimension $\Delta=h+\bar h$, and conformal spin $s=h-\bar h$. &
Sec.~\ref{sec:modular-covariance} \\
$L_0$, $\overline L_0$, $c$ & Chiral and antichiral Virasoro zero modes and
central charge. &
Sec.~\ref{sec:modular-covariance} and
App.~\ref{appsec:partition-vector-and-modular-covariance} \\
$\tau$, $q$, $\bar q$ & Modular parameter and nomes
$q=\ee^{2\pi\ii\tau}$,
$\bar q=\ee^{-2\pi\ii\bar\tau}$. &
Sec.~\ref{appsec:gapless-phase-map} \\
$\scrV$, $\Mod_\scrV$, $\scrV_i$, $\scrV_x$ & Full field algebra (including
a chiral VOA) specifying the local quantum symmetry, its module category,
an irreducible chiral module labeled by $i$, and the module in bulk-anyon
sector $x$, respectively. Concrete chiral algebras and their modules
carry model superscripts: $\mathrm I$ (Ising), $\mathrm{TCI}$
(tricritical Ising), and $\mathrm P$ (Potts/$\cM(6,5)$); $\scrW^{\mathrm P}$
denotes the Potts $\cW_3$ extension. &
Apps.~\ref{appsec:enriched-category-description}
and~\ref{appsec:partition-vector-and-modular-covariance} \\
$\chi_i(\tau)$; $\chi_h^{\mathrm I}$, $\chi_h^{\mathrm P}$;
$\widehat\chi_h$ & Chiral character; Ising and Potts Virasoro characters;
Potts $\cW_3$ characters. The argument $\tau$ is often omitted. &
App.~\ref{appsec:partition-vector-and-modular-covariance} \\
${}^{\cD}\!\cS$ & Fusion category $\cS$ enriched in $\cD$.
For a gapless boundary, $\cD=\Mod_\scrV$; objects describe boundary
defects and morphism objects describe local operators between them. &
App.~\ref{appsec:enriched-category-description} \\
$\operatorname{Rel}(\scrB)$ & Relevant-field subspace
$\bigoplus_{0<h+\bar h<2}\scrB_{h,\bar h}$; its vanishing defines stability,
without restricting conformal spin or excluding marginal fields. &
Definition~\ref{def:stable-full-field-algebra} \\
$\cT$ (CFT) & Boundary CFT for which $\fZ(\cR)$ is called a stable symTO
when its trivial bulk-anyon sector contains no nonidentity relevant local
operator. This use of $\cT$ is distinct from lattice translation. &
Definition~\ref{def:stable-symto} \\
$|x;a\ra$; $|i;\Lambda\ra$ & State $a$ in center-anyon sector $x$;
conformal-family state $\Lambda$ paired with charge/defect sector $i$. &
Secs.~\ref{sec:structure-constants} and~\ref{sec:two-two-move} \\
$\cA_{abc^*}$ & Numerical three-point amplitude in
$\fZ(\cR)(x\otimes y,z)$; the star records the outgoing orientation of the third
state. &
Eq.~\eqref{eq:structure-amplitude-hom-space} \\
$\cA_{abc^*;\alpha}$ & Component of the numerical amplitude in the
isotopy-invariant fusion basis $v_{xy}^{z,\alpha}$. &
Eq.~\eqref{eq:structure-amplitude-decomposition} \\
$\phi_a^{\closedsector}$; $C^{\closedsector}_{abc;\alpha}$ & Closed-string/full-CFT
primary field and its fusion-space-resolved structure constant. &
Sec.~\ref{sec:structure-constants} and App.~\ref{appsec:closed-string-pants} \\
$\cA_{000}$ & Vacuum amplitude in the trivial anyon sector and its unique
vacuum fusion channel; division by it removes the overall tensor
normalization. &
Eq.~\eqref{eq:extract-closed-structure-constants} \\
$G_{\mathrm{cpants}}$ & Nonchiral closed-string geometric factor used in the
approximate extraction of bulk primary-field structure constants. &
Eq.~\eqref{eq:pants-gauge-app} \\

\addlinespace[4pt]
\multicolumn{3}{@{}l}{\textbf{Algebras, topological bootstrap, and model-specific notation}}\\*
\addlinespace[2pt]
$\Mat_n(\C)$ & Full matrix algebra of all linear endomorphisms of $\C^n$,
$\Mat_n(\C):=\End(\C^n)$; in particular, its vector-space dimension is
$n^2$. &
Eq.~\eqref{eq:tube-alg-iso} and Sec.~\ref{appsec:gapped-phase-map} \\
$\Fun(\cM,\cM)$,\newline
$\Mat_n(\Hilb)$ & Endofunctor category of a semisimple $\cM$ and its realization
by $n\times n$ matrices of Hilbert spaces, where $n=|\Irr(\cM)|$. &
Sec.~\ref{sec:tensor-complex} \\
$Q$, $m$, $\iota$ & Q-system in the charge category, with multiplication and
unit. &
Sec.~\ref{sec:random-cluster-models} \\
$A,B$; $L,N$ & Categorical algebra objects; Lagrangian algebras.
Conditions such as separability are stated where needed. &
App.~\ref{app:notation-table} (opening convention) \\
$\scrA$, $\scrB$, $\scrL$ & Full-field/VOA-side algebra or sector Hilbert
space, including their conformal-family content; $\scrL$ is reserved for a
Lagrangian full field algebra. &
App.~\ref{app:notation-table} (opening convention) and
Sec.~\ref{appsec:ffa-extensions-and-emergent-symto} \\
$L_{\diagsector}$ & Canonical diagonal Lagrangian algebra
$\bigoplus_i i^*\boxtimes i$ in
$\overline{\cR}\boxtimes\cR\cong\fZ(\cR)$, for modular $\cR$. &
Eq.~\eqref{eq:diag-lagrangian-algebra} \\
$\scrL_{\diagsector}$ & Charge-conjugation full-field space
$\bigoplus_i\overline{\scrV_{i^*}}\otimes\scrV_i$. &
Eq.~\eqref{eq:diagonal-full-field-space} \\
$\mathfrak L$ & Bond degrees of freedom
$\bigoplus_i\overline{\scrV_{i^*}}\odot i$, assigning a conformal
family as the multiplicity space of each charge. &
Eq.~\eqref{eq:mixed-bond-object} \\
$\mathsf{hole}(\scrV)$ & Chiral gapless hole whose boundary is described by
$\scrV$; the geometric realization of $\mathfrak L$. &
Sec.~\ref{sec:geometric-picture} \\
$m_{\text{pants}}$; $m_\triangle$, $m_{\triangle'}$ &
Tensor pairing conformal three-point blocks with basis morphisms in fusion
spaces; subscripts distinguish the pants geometries. &
Eq.~\eqref{eq:mixed-pants-map} and Sec.~\ref{sec:two-zero-move} \\
$O_l$ & Insertion of an $l$-labeled topological defect loop around a hole;
the weighted projector $D^{-2}\sum_l d_l O_l$ implements cloaking. &
Sec.~\ref{sec:two-zero-move} \\
$R$, $|\psi,R\rangle$, $|0\rangle\!\rangle$ & Hole radius, vacuum-normalized
cloaking boundary state $|\psi,R\rangle=R^{L_0+\bar L_0}|0\rangle\!\rangle$,
and vacuum Ishibashi state. The radius $R$ is unrelated to $\cR$. &
Sec.~\ref{sec:two-zero-move} \\
$\cC_A$, $\cC_A^{\localmodules}$ & Categories of $A$-modules and local
$A$-modules. &
Sec.~\ref{appsec:gapped-phase-map} \\
$V$, $T$ & In the topological-bootstrap ansatz,
$V=\bigoplus_{i\in\Irr(\cR)}i$ and $T\in\cR(V\otimes V,V)$ is the initial
rank-three tensor. &
Eq.~\eqref{eq:initial-tensor-bootstrap} \\
$\cN_{ijk;\alpha}$ & Continuous, $S_3$-symmetric trial parameters for
admissible fusion triples in the initial tensor. &
Eq.~\eqref{eq:initial-tensor-bootstrap} \\
$G_\triangle$, $n_i$ & Chiral geometric factor and the integer lift
$h_i=s_i+n_i$ of a topological spin. &
Eq.~\eqref{eq:initial-tensor-bootstrap} \\
$E$, $\Omega$, $p$, $v$ & Random-cluster edge set, bond configuration,
bond-occupation parameter, and $v=p/(1-p)$. &
Sec.~\ref{sec:random-cluster-models} \\
$C(\Omega)$, $l(\Omega)$,
$q=d_Q$ & Numbers of clusters and open edges, and random-cluster weight.  This
$q$ is unrelated to the modular nome. &
Sec.~\ref{sec:random-cluster-models} \\
\addlinespace[4pt]
\multicolumn{3}{@{}l}{\textbf{Geometric factors and conformal data}}\\*
\addlinespace[2pt]
$\Sigma$; $p_i$, $\xi_i$ & Three-punctured Riemann sphere, its ordered
punctures, and abstract local coordinates vanishing at those punctures. &
App.~\ref{app:geometric-factors} \\
$\phi_i$, $\varphi_i$, $f_i$ & Inserted state/field, explicit local coordinate
in a chosen geometry, and pulled-back coordinate germ/substitution operator.
& App.~\ref{app:geometric-factors} \\
$l_n$, $A_n^{(i)}$ & Nonnegative Witt generators and the recursively defined
Virasoro coordinates of a formal germ. &
Eq.~\eqref{eq:formal-coordinate-change} \\
$\cW(f)$, $\widehat{\cW}(f)$ & Contravariant substitution operator and its
Virasoro representation. &
Eq.~\eqref{eq:chiral-coordinate-change-operator} \\
$G_{\mathrm{nut}}$, $G_{\mathrm{pants}}$ & Planar/chiral geometric factors
used to initialize topological-bootstrap tensors; $G_{\mathrm{pants}}$ describes
the flat chiral closed-string amplitude, equivalently expressed as an
open-string amplitude. &
Eq.~\eqref{eq:nut-gauge-app} and App.~\ref{appsec:flat-pants},
``Three-point geometric factor'' \\
$\zeta$, $t$, $u$, $\rho$, $g$ & Upper-half-plane coordinate, rationalizing
auxiliary coordinate, flat polygon coordinate, and the two forms
$u=\rho(\zeta)=g(t(\zeta))$ of the global flat-pants map. &
App.~\ref{app:geometric-factors} \\
$\scrV$, $\Mod_\scrV$ & Strongly rational unitary VOA and its category
of grading-restricted modules. &
Sec.~\ref{sec:two-two-move} \\
$L_n$ & Virasoro modes acting on conformal states. &
Eq.~\eqref{eq:chiral-coordinate-change-operator} \\

\end{longtable}
\endgroup

\clearpage
\twocolumngrid

\bibliography{ref}

\end{document}